\documentclass[12pt]{article}
\usepackage[T1]{fontenc}
\usepackage[utf8]{inputenc}
\usepackage{amssymb,amsmath}
\usepackage{amsfonts}
\usepackage{mathtools}
\usepackage{amsthm}
\usepackage{upgreek}
\usepackage{enumitem}
\usepackage{amsfonts}
\usepackage{chngcntr}

\usepackage{graphicx,psfrag}
\usepackage{amsfonts}
\usepackage{mathtools}

\usepackage{geometry}
\usepackage[colorlinks=true,linkcolor=blue, citecolor=blue, urlcolor = blue]{hyperref}
\usepackage{lmodern}
\usepackage{avant}
\usepackage[round,sort,authoryear]{natbib}
\usepackage{comment}
\usepackage{setspace}
\usepackage{mathptmx}

\counterwithin*{equation}{section}

\theoremstyle{definition}
\newtheorem{theorem}{Theorem}
\newtheorem{definition}{Definition}
\newtheorem{assumption}{Assumption}
\newtheorem{lemma}{Lemma}
\newtheorem{example}{Example}

\newtheorem{remark}{Remark}
\newtheorem{corollary}{Corollary}

\newcommand\norm[1]{\left\lVert#1\right\rVert}

\DeclarePairedDelimiter{\floor}{\lfloor}{\rfloor}

\usepackage{amsmath,amssymb}

\newcommand{\vertiii}[1]{{\left\vert\kern-0.25ex\left\vert\kern-0.25ex\left\vert #1 
    \right\vert\kern-0.25ex\right\vert\kern-0.25ex\right\vert}}

\newif\ifshow 
\showfalse 

\ifshow
  \includecomment{wrap}
\else
  \excludecomment{wrap} 
\fi

\newcommand{\nc}{\newcommand}
\nc{\mbb}{\mathbb}\nc{\bb}{\mathbb}
\nc{\mbf}{\mathbf}\nc{\mb}{\mathbf}
\nc{\mc}{\mathcal}
\nc{\msf}{\mathsf}\nc{\ms}{\mathsf}
\nc{\acc}{\ms{acc}}
\nc{\ack}{\ms{ack}}
\nc{\alp}{\alpha}\nc{\al}{\alpha}\nc{\gka}{\alpha}
\nc{\ap}{\ms{ap}}
\nc{\apd}{\ms{apd}}
\nc{\base}{\ms{base}}\nc{\ba}{\ms{base}}
\nc{\bet}{\beta}\nc{\gkb}{\beta}
\nc{\boucle}{\ms{loop}}\nc{\Loop}{\ms{loop}}\nc{\lo}{\ms{loop}}
\nc{\bu}{\bullet}
\nc*{\cc}{\raisebox{-3pt}{\scalebox{2}{$\cdot$}}}
\nc{\centre}{\ms{center}}\nc{\Center}{\ms{center}}\nc{\cen}{\ms{center}}\nc{\ce}{\ms{center}}
\nc{\ci}{\circ}
\nc{\code}{\ms{code}}\nc{\cod}{\ms{code}}\nc{\decode}{\ms{decode}}\nc{\encode}{\ms{encode}}
\nc{\de}{:\equiv}
\nc{\dr}{\right}\nc{\ga}{\left}
\nc{\ds}{\displaystyle}
\nc{\ep}{\varepsilon}
\nc{\eq}{\equiv}
\nc{\ev}{\ms{ev}}
\nc{\fib}{\ms{fib}}
\nc{\funext}{\ms{funext}}\nc{\fu}{\ms{funext}}
\nc{\gam}{\gamma}
\nc{\glue}{\ms{glue}}\nc{\gl}{\ms{glue}}
\nc{\happly}{\ms{happly}}\nc{\ha}{\ms{happly}}
\nc{\id}{\ms{id}}
\nc{\ima}{\ms{im}}
\nc{\inc}{\subseteq}
\nc{\ind}{\ms{ind}}
\nc{\inl}{\ms{inl}}
\nc{\inr}{\ms{inr}}
\nc{\isContr}{\ms{isContr}}\nc{\co}{\ms{isContr}}\nc{\iC}{\ms{isContr}}\nc{\ic}{\ms{isContr}}
\nc{\isequiv}{\ms{isequiv}}\nc{\iseq}{\ms{isequiv}}\nc{\ieq}{\ms{isequiv}}
\nc{\ishae}{\ms{ishae}}\nc{\ish}{\ms{ishae}}\nc{\ih}{\ms{ishae}}
\nc{\isProp}{\ms{isProp}}\nc{\prop}{\ms{isProp}}\nc{\iP}{\ms{isProp}}\nc{\ip}{\ms{isProp}}
\nc{\isSet}{\ms{isSet}}\nc{\isS}{\ms{isSet}}\nc{\iss}{\ms{isSet}}\nc{\iS}{\ms{isSet}}\nc{\is}{\ms{isSet}}
\nc{\lam}{\lambda}
\nc{\LEM}{\ms{LEM}}\nc{\lem}{\ms{LEM}}\nc{\LE}{\ms{LEM}}
\nc{\lv}{\lvert}\nc{\rv}{\rvert}\nc{\lV}{\lVert}\nc{\rV}{\rVert}
\nc{\Map}{\ms{Map}}
\nc{\merid}{\ms{merid}}\nc{\meri}{\ms{merid}}\nc{\mer}{\ms{merid}}\nc{\me}{\ms{merid}}
\nc{\N}{\bb N}
\nc{\na}{\ms{nat}}
\nc{\nn}{\noindent}
\nc{\one}{\mb1}
\nc{\oo}{\operatorname}
\nc{\pd}{\prod}
\nc{\ps}{\mc P}
\nc{\pa}{\ms{pair}^=}
\nc{\ph}{\varphi}
\nc{\ppmap}{\ms{ppmap}}
\nc{\pr}{\ms{pr}}
\nc{\Prop}{\ms{Prop}}
\nc{\qinv}{\ms{qinv}}\nc{\qin}{\ms{qinv}}\nc{\qi}{\ms{qinv}}
\nc{\rec}{\ms{rec}}
\nc{\refl}{\ms{refl}}
\nc{\seg}{\ms{seg}}
\nc{\Set}{\ms{Set}}
\nc{\sm}{\scriptstyle}
\nc{\sms}{\ms s}
\nc{\sq}{\square}
\nc{\suc}{\ms{succ}}\nc{\su}{\ms{succ}}
\nc{\tb}{\textbf}
\nc{\then}{\Rightarrow}
\nc{\tms}{\ms t}
\nc{\tx}{\text}
\nc{\transport}{\ms{transport}}\nc{\tr}{\ms{transport}}
\nc{\two}{\mb2}
\nc{\Type}{\text-\ms{Type}}\nc{\type}{\text-\ms{Type}}\nc{\ty}{\text-\ms{Type}}
\nc{\U}{\mc U}
\nc{\ua}{\ms{ua}}
\nc{\uniq}{\ms{uniq}}
\nc{\univalence}{\ms{univalence}}
\nc{\vide}{\varnothing}
\nc{\ws}{\ms{sup}}
\nc{\zero}{\mb0}

\title{\textbf{Structural Analysis of Vector Autoregressive Models}\footnote{I am grateful to Professor Markku Lanne and Professor Mika Meitz from the Faculty of Social Sciences, University of Helsinki for helpful conversations. I also wish to thank Professor Jose Olmo and Professor Tassos Magdalinos from the School of Economics as well as Professor Zudi Lu and Dr. Chao Zheng from the School of Mathematical Sciences, University of Southampton for helpful discussions. Financial support from the Research Council of Finland (Grant 347986) is also gratefully acknowledged. Address correspondence to Christis Katsouris, Faculty of Social Sciences, University of Southampton, United Kingdom; Email Address: \textcolor{blue}{c.katsouris@soton.ac.uk}.
 } }

\author{\textbf{Christis Katsouris}\thanks{Dr. Christis Katsouris, is currently a Postdoctoral Researcher at the Faculty of Social Sciences, University of Helsinki. } \\  \textit{University of Southampton} \\ $\&$ \textit{University of Helsinki}\\  \\ (Work-in-progress)}

\date{\today}

\begin{document}

\maketitle

\begin{abstract}
\vspace*{-0.8 em}
This set of lecture notes discuss key concepts for the Structural Analysis of Vector Autoregressive models for the teaching of a course on Applied Macroeconometrics with Advanced Topics.
\end{abstract}

\maketitle

\newpage 
   
\begin{small}
\begin{spacing}{0.9}
\tableofcontents
\end{spacing}
\end{small}

\newpage

\section{Introduction}

This set of lecture notes are meant to discuss fundamental concepts on the mechanisms of structural vector autoregressive models with many economic theory driven examples which are drawn from several seminal studies presented in the existing literature. The particular perspective we follow permits to present several mathematical and theoretical concepts in an informal way while discussing important concepts by presenting how some long-standing problems in the applied macroeconometrics literature have been tackled. Moreover, our discussion of selected studies from the literature in the form of examples that present key aspects in the identification and estimation of SVAR models, main theoretical econometric concepts such as matrix representations and asymptotic theory convergence results; are motivated from a research-based pedagogical approach that encourages the reader to seek a comprehensive understanding of these concepts while fostering further discussions on current state-of-the art methodological approaches and debates in the literature. Note that a full treatment of these topics are presented in the "Time Series Analysis" book of \cite{hamilton1994time}. Additional material can be found in the books of  \cite{gourieroux1997time}, \cite{lutkepohl2005new} and  \cite{kilian2017structural}. 

More specifically, \cite{stock2001vector} describe the key job tasks of applied macro-econometricians, which shall also serve as the main learning objectives motivating the preparation of this set of lecture notes. There are summarized as below: 
\begin{itemize}

\item To be able to describe/summarize and infer about the presence of stochastic trends, cointegration dynamics, unit root behaviour as well as structural breaks in macroeconomic time series.   

\item To be able to recover the structure of the macroeconomy based on data and theoretical-driven questions using structural analysis (i.e., identification of structural shocks) as well as impulse response analysis and construction of related test statistics and confidence intervals. 

\item To be able to make forecasts on key macroeconomic variables based on appropriate forecasting schemes and formulations of VAR and SVAR models. 

\item To be able to advise macroeconomic policy-makers (e.g., using counterfactual analysis and related methods). 

\end{itemize} 
Vector Autoregressive models are a popular econometric tool for estimating the dynamic interaction between the variables included in the system. Moreover, several test statistics are constructed using VAR models such that: Wald tests for Granger Causality, impulse response functions (IRFs) and forecast error variance decomposition (FEVDs). Therefore, inference on these statistics is typically based on either on first-order asymptotic approximations or on bootstrap methods. However, the deviation from \textit{i.i.d} innovations as in the case of conditional heteroscedasticity, invalidates a number of standard inference procedures such that the application of these methods may lead to conclusions that are not in line with the true underlying dynamics. Consequently, in many VAR applications there is need for inference methods that are valid when innovations are only serially correlated but not independent. Nevertheless, the VAR representation of the economy is suitable for estimation/testing as well as forecasting purposes but is not adequate for policy analysis which is the purpose of this course (thus, our interest in SVARs).

\newpage

\subsection{The Identification Problem}

Roughly speaking, the \textit{identification problem} of structural parameters in linear simultaneous equations models are closely related\footnote{The notion of coefficient restrictions was originally extended to show the equivalence relationship between identifiability and instrumental variables estimation, that is, the restrictions required for identification  give rise to instrumental variables required for estimation.}. Seminal papers discussing the identification problem include among others \cite{sargan1983identification} and \cite{dufour2003identification}. According to \cite{hausman1983identification}, necessary and sufficient conditions for identification with linear coefficient and covariance restrictions are developed in a limited information context. In particular, imposing covariance restrictions facilitate identification \textit{iff} they imply that a set of endogenous variables is predetermined in the equation of interest - which generalizes the notion of recursiveness for structural learning and causal recovery. Under full information, covariance restrictions imply that residuals from other equations are predetermined in a particular equation, and, under certain conditions, can facilitate system identification. This implies that in the general case, FIML first order conditions show that if a system of equations is identifiable as a whole, covariance restrictions cause residuals to behave as instruments. 

Furthermore, imposing exclusion restrictions and the normalization $\beta_{ii} = 1$, then the classical structural econometric specification for the simultaneous equation model $\boldsymbol{Y} \boldsymbol{B}^{\prime} + \boldsymbol{Z} \boldsymbol{\Gamma}^{\prime} = \boldsymbol{U}$, where the vector $\boldsymbol{y}_i$, for $i \in \left\{ 1,..., G \right\}$, includes $G$ jointly dependent random variables. A key result for system identification purposes is the fact that the relative triangularity of equations $(i,j)$ is precisely equivalent to a zero in the $(i,j)-$th position of $\boldsymbol{B}^{-1}$, denoted as $\boldsymbol{B}^{-1}_{ij}$. As a result, equations $(i,j)$ are relatively recursive \textit{if and only if} there are no paths by which a shock to $\boldsymbol{u}_j$ can be transmitted to $\boldsymbol{y}_i$.     
\begin{lemma}
Zero restrictions on $( \boldsymbol{B}, \boldsymbol{\Sigma}_1 )$ are sufficient for identification if and only if they induced the equivalence relation:
\begin{align}
\boldsymbol{\Psi} \boldsymbol{B}^{-1} \boldsymbol{\Sigma}_1^{\prime} = \boldsymbol{0}    
\end{align}
for some selection matrix $\boldsymbol{\Psi}$. 
\end{lemma}

\begin{corollary}
If $\boldsymbol{y}_i$ is predetermined in the first equation, then every endogenous variable in the $i-$th structural equation $( \boldsymbol{y}_j )$ for which $\sigma_{j1} = 0$ is predetermined in the first structural equation.     
\end{corollary}

\begin{remark}
Notice that in the case of a diagonal disturbance covariance matrix, an endogenous variable is predetermined in the first equation implying a special structure. In particular, if $\boldsymbol{y}_i$ is predetermined in the first equation, every endogenous variable in the $i-$th equation is also predetermined in the first equation. Consequently, every endogenous variable in their respective equations is predetermined in the first equation. In other words, in the case of a diagonal variance matrix $\boldsymbol{\Sigma}$, a relatively recursive situation is one in which a set of endogeneous variables determines itself independently from the first equation variable, and thus all elements of the set are predetermined in the first equation.      
\end{remark}

Next, we consider the equivalence between instrumental variables and covariance restrictions. 

\newpage

\begin{lemma}[Rank]
The parameters of the first structural equation are identifiable if and only if the 2SLS estimator is well-defined, using $\boldsymbol{W}$ as the matrix of instruments.     
\end{lemma}

\begin{proof}
In particular, asymptotically we have that 
\begin{align}
\underset{ T \to \infty }{ \mathsf{plim} } \frac{1}{T} \boldsymbol{W}^{\prime} \boldsymbol{V}_1 = 
\begin{bmatrix}
\underset{ T \to \infty }{ \mathsf{plim} } \frac{1}{T} \boldsymbol{\Psi} \boldsymbol{Y}^{\prime} \boldsymbol{V}_1   
\\
\underset{ T \to \infty }{ \mathsf{plim} } \frac{1}{T} \boldsymbol{Z}^{\prime} \boldsymbol{V}_1
\end{bmatrix}
=
\begin{bmatrix}
\boldsymbol{\Psi} \boldsymbol{\Omega}_1
\\
\boldsymbol{0}
\end{bmatrix}
\end{align}
where $\boldsymbol{V}_1$ is the reduced form disturbance corresponding to $\boldsymbol{Y}_1$. This implies that, 
\begin{align}
\boldsymbol{Y}_1 = \boldsymbol{Z}_1 \boldsymbol{\Pi}_{11}^{\prime} + \boldsymbol{Z}_2 \boldsymbol{\Pi}_{12}^{\prime}  + \boldsymbol{V}_1.   
\end{align}
Notice that rank conditions for identification are commonly used as identifying conditions for static and dynamic factor models as well as cointegrating regression models (e.g., VECM forms). 
\end{proof}
A well-known problem to the literature of  structural VAR models is that imposing the assumption of Gaussian errors renders non-identifiable structural parameters due to the nonlinearities and complex dependence structure of the system. Additional identifying restrictions are needed to ensure robust identification and estimation of SVARs. Based on these concerns, the framework of \cite{lanne2017identification} shows that in practice the Gaussian case is an exception in that a SVAR model whose error vector consists of independent non-Gaussian components is, without any additional restrictions, identified and leads to essentially unique responses. The identification scheme employs the non-Gaussianity of structural shocks while an MLE estimator for the parameters of the non-Gaussian SVAR model is consistent and asymptotically normally distributed. Moreover, the scheme allows to impose additional identifying restrictions which can be formulated into testable hypotheses. Further extensions of the non-Gaussian identification strategy includes cases such as when the system includes only strictly positive components (see, \cite{nyberg2022structural}), cases when a GMM estimation approach is employed (see, \cite{lanne2021gmm} and \cite{keweloh2021generalized}) as well as specific cases when structural shocks are leptokurtic (e.g. see \cite{lanne2023identifying}) or heavy-tailed (e.g., see \cite{anttonen2023statistically}).

Furthermore, while the reduced-form VAR model can be seen as a convenient description of the joint dynamics of a number of time series that facilitates forecasting, the SVAR model is more appropriate for answering economic questions of theoretical and practical interest. The main tools in analyzing the dynamics in SVAR models are the impulse response function and the forecast error variance decomposition. Moreover, usually when fitting a SVAR model, the joint distribution of the error terms is almost always (either explicitly or implicitly) assumed to follow a multivariate Gaussian distribution which might render the system non-identifiable especially under the presence of non-linearities, higher-order dynamics and heavy-tails. The positive news is that the joint distribution of the reduced-form errors is fully determined by their covariances only but this which implies that the structural errors of the system cannot be identified unless we impose additional linear restrictions (see, \cite{lanne2017identification}). However the use of suitable orthogonal transformations or even permutations (due to being exchangeable random variables), can be shown to induce observationally equivalent SVAR processes.

\newpage

Furthermore, we discuss some key identification strategies in the following sections. Generally, speaking understanding all the available tools in our disposal is important when considering relevant empirical and theoretical driven questions in relation to the impact of structural shocks on economic outcomes. Therefore, in this set of lecture notes we present some key concepts and recent developments in the literature with focus on the implementation of econometric and statistical tools for the identification and estimation of  structural vector autoregressive models, (SVAR). In particular, we discuss several different approaches that are commonly employed for the structural identification of these interdependent systems in relation to economic phenomena such as business cycle fluctuations and macroeconomic conditions in order to obtain policy recommendations. On the other hand, exogenous variation to macroeconomic conditions such as the adverse effect of climate change across interconnected economies requires novel econometric methods to ensure robust identification and estimation\footnote{In particular, the omission of such an important component as the climate when considering the structural analysis of time series models could lead to non-causality, leading to inaccurate model estimates and forecasts (see, \cite{lutkepohl1982non}). Further properties on Noncausal vector AR processes are presented by \cite{davis2020noncausal}.} (e.g., see \cite{pesaran2000structural}).

The stability of cointegrated VAR processes can be assessed with respect to the long-run relations as well as with respect to the short-run parameters. Usually the long-run equilibrium relations are based on economic theory and so we expect them to be relatively more stable after a perturbation of structural shocks (such as a change in the monetary policy framework), than the short-run coefficients which are data-driven. Consider for example the study of the transmission of monetrary policy shocks to macroeconomic variables using structural cointegrated VAR models. Roughly speaking the identification of SVAR models is achieved by imposing restrictions on the covariance matrix of the residuals of a reduced-form VAR to provide an economic interpretation of shocks while a structural cointegrated VAR model is identified by imposing restrictions on the cointegrating space. Consequently, based on a structural identification scheme of the covariance matrix then metrics such as impulse-responses, variance-decompositions and historical decompositions can be constructed.

In particular, \cite{rubio2010structural} verifies the identification for linear Gaussian models by checking whether two parameter points are observetionally equivalent \textit{iff} they have the same reduced-form representation. Moreover, a global identification can be verified by checking whether the model is identified at different points in the parameter space prior to the estimation step. On the other hand, a well-know fact in the literature of nonstationary autoregressive models, asymptotic distribution discontinuities arise in the case when the parameter is at the boundaries of the parameter space. 

\begin{definition}
Consider a $\mathsf{SVAR(p)}$ process with restrictions imposed by $R$. Then, the $\mathsf{SVAR(p)}$ is exactly identified $\textit{iff}$, for almost any reduced-form parameter point $( \boldsymbol{B}, \boldsymbol{\Sigma} )$, there exists a unique structural parameter point $( \boldsymbol{A}_0, \boldsymbol{A}_{+} ) \in \mathbb{R}$ such that $g ( \boldsymbol{A}_0, \boldsymbol{A}_{+} ) = ( \boldsymbol{B}, \boldsymbol{\Sigma} )$.    
\end{definition}

Following \cite{rubio2010structural}
\begin{align}
\mathcal{K} := \big\{  ( \boldsymbol{A}_0, \boldsymbol{A}_{+} ) \in \mathbb{R} \ \big| \ \mathsf{rank} \big( \boldsymbol{M}_j \left( f ( \boldsymbol{A}_0, \boldsymbol{A}_{+} ) \right) \big) = n, \ \ \ 1 \leq j \leq n \big\}.    
\end{align}

\newpage

Generally, it is well-known that identification in parametric models implies that a specific model structure is identifiable if there is no other structure which is observationally equivalent (see, also \cite{sargent1976observational}). For instance, the identification scheme proposed by \cite{lanne2017identification} requires the construction of observationally equivalent SVAR processes when comparing two $MA(\infty)$ representations of the model. In conjecture with conventional rank conditions that ensure the valid identification of SVAR models under non-Gaussianity the identification of structural parameters can be extended into more complex model configurations regardless of the presence of nonlinear or nonstationary dynamics. 

More specifically, the main intuition of the identification methodology of \cite{rubio2010structural}, is to establish general rank conditions for global identification of \textit{identified} and \textit{exactly identified} models. Moreover, \cite{bognanni2018class} highlights that the key conceptual properties which ensure SVAR tractability include: $\textit{(i)}$ observational equivalence of points $( \boldsymbol{A}, \boldsymbol{F} )$ and $( \boldsymbol{A Q}, \boldsymbol{F Q} )$ for $\boldsymbol{Q} \in \mathcal{O}_n$ in the structural model, and $\textit{(ii)}$ the ability to reparametrize the structural model from $( \boldsymbol{A}, \boldsymbol{F} )$ to $( \boldsymbol{B}, \boldsymbol{H}, \boldsymbol{Q} )$ and separate identifiable elements $( \boldsymbol{B}, \boldsymbol{H} )$ and non-identifiable elements $\boldsymbol{Q}$. Thus, the above definitions imply the existence of a unique mapping of reduced-form parameters which is invariant to scaling and reordering.

Another issue discussed by \cite{forni2014sufficient} is the presence of sufficient information to ensure that the conditions of existence and uniqueness are ensure for identification purposes. Specifically \cite{forni2014sufficient} present necessary and sufficient conditions under which a VAR contains sufficient information to estimate the structural shocks. Based on this theoretical result the authors propose a test statistic for detecting information deficiency and a procedure to amend a deficient VAR model. A relevant example is the bivariate VAR with unemployment and labour productivity which is found to be informentionally deficient. The building block of this framework is a representation of the economy where $q$ mutually orthogonal structural shocks affect macro variables through square-summable impulse-response functions (see, also \cite{kocikecki2018global, kocikecki2023solution}). 

\begin{assumption}[MA representation]
The $n-$dimensional vector $x_t$ of stationary series satisfies 
\begin{align}
\boldsymbol{x}_t = \boldsymbol{F}(L) \boldsymbol{u}_t,  \ \ \ \boldsymbol{F}(L) := \sum_{j=0}^{ \infty } \boldsymbol{F}_j L^j,   
\end{align}
where $\boldsymbol{u}_t$ is a $q-$dimensional white noise vector of structural macroeconomic shocks and $\boldsymbol{F}(L)$ is an $( n \times q)$ matrix of square-summable linear filters in the non-negative powers of the lag operator $L$. The above formulation also corresponds o the steady-state equilibrium of a DSGE model. 
\end{assumption}

\begin{assumption}[Information Set]
We shall say that the information set $\mathcal{X}_t^{*}$ is given by the closed linear space spanned by present and past values of the variables in $x^{*}_t$, where
\begin{align}
\boldsymbol{x}_t^{*} = \boldsymbol{x}_t + \boldsymbol{\xi}_t \equiv \boldsymbol{F}(L) \boldsymbol{u}_t + \boldsymbol{\xi}_t,    
\end{align}
\end{assumption}
In practice the number of observable variables $n$ is very large, so that the econometrician needs to reduce it in order to estimate a VAR. The VAR information set is then spanned by an $s-$dimensional sub-vector of $\boldsymbol{x}_t^{*}$, or, more generally, an $s-$dimensional linear combination of $\boldsymbol{x}_t^{*}$, say $\boldsymbol{z}_t^{*} = \boldsymbol{C} \boldsymbol{x}_t^{*}$.






     

\newpage

\section{Linear Vector Autoregressions}
\label{Section2}

\begin{example}
Consider the model $\boldsymbol{y}_t = \boldsymbol{A} \boldsymbol{y}_{t-1} + \boldsymbol{e}_t$, which implies that (see, \cite{lai1982least})   
\begin{align}
\boldsymbol{y}_t &= \boldsymbol{A}^n \boldsymbol{y}_0 + \sum_{j=1}^n \boldsymbol{A}^{n-j}  \boldsymbol{e}_j, \ \ \ t = 1,...,n. 
\\
\boldsymbol{A} &= 
\begin{pmatrix}
\alpha_1 & \hdots & \alpha_{k-1} & \alpha_k
\\
 &  \boldsymbol{I}_{ k-1} &  &  \boldsymbol{0}
\end{pmatrix}
\end{align}
which is equivalent to the matrix coefficient of a $\mathsf{VAR}(1)$ process. Moreover, expressing the matrix $\boldsymbol{A}$ into its Jordan form we obtain $\boldsymbol{A} = \boldsymbol{C} \boldsymbol{D} \boldsymbol{C}^{-1}$, where $\boldsymbol{D} = \mathsf{diag} \left\{ \boldsymbol{D}_1,..., \boldsymbol{D}_q  \right\}$  
\begin{align}
\boldsymbol{D}_j = 
\begin{pmatrix}
z_j & 1 & 0 & \hdots & 0
\\
0 & z_j & 1 & \hdots & 0
\\
\vdots & \vdots & \ddots &  \ddots & 0  
\\
0 & 0 & \hdots & \hdots & z_j 
\end{pmatrix}_{ ( m_j \times m_j  )  }
\end{align}
with $z_j$ being the root of $\varphi (z)$ with multiplicity $m_j$ and $\boldsymbol{C}$ a nonsingular matrix. Moreover, let $M = \mathsf{max}_{j} m_j$, then it holds that 
\begin{align}
\norm{ \boldsymbol{A}^n } = \norm{ \boldsymbol{C} \boldsymbol{D}^n \boldsymbol{C}^{-1}  }  \leq \norm{ \boldsymbol{C}  }  \norm{ \boldsymbol{C}^{-1} } \sum_{j=1}^q \norm{  \boldsymbol{D}_j^n } = \mathcal{O} ( n^{M-1} ).  
\end{align}
In particular, when the roots of the characteristic polynomial lie strictly inside the unit circle, then with the independent white noise $\boldsymbol{e}_t$, the system is stable. Moreover, by allowing the roots of the characteristic polynomial to lie on the unit circle, we can formulate unstable but non-explosive systems related to the ARIMA models. Notice the analysis of time series regression models includes unit-roots and partially nonstationary (cointegrated) models.   
\end{example}

\subsection{Cointegration and Vector Autoregressive Processes}

\subsubsection{Integration Order}

An important property of $I(1)$ variables is that there can be linear combinations of these variables that are $I(0)$. If this is so then these variables are said to be cointegrated. Notice that econometric cointegration analysis can be used to overcome difficulties associated with stochastic trends in time series and applied to test whether there exist combinations of non-stationary series that are themselfs stationary. 
\begin{definition}
Suppose that $y_t$ is $I(1)$. Then $y_t$ is cointegrated if there exists an $N \times r$ matrix $\beta$, of full column rank and where $0 < r < N$, such that the $r$ linear combinations, $\beta^{\prime} y_t = u_t$, are $I(0)$. We say that the dimension $r$ is the \textit{cointegration rank} and the columns of $\beta$ are the \textit{cointegrating vectors}. 
\end{definition}

\newpage 

\begin{itemize}

\item Testing for cointegrating relations that economic theory predicts should exist, implies that the null hypothesis of noncointegration is not rejected. However, under the presence of breaks there is a need for implementing tests of the null hypothesis of non-cointegration, against alternatives allowing cointegrating relations subject to breaks.

\item Relevant asymptotic theory results on the efficient estimation of the parameter path in unstable time series models can be found in the study of \cite{muller2010efficient}.  Further issues include the use of exogenous regressors (e.g., see \cite{pesaran2000structural}) in VECM representations.  

\end{itemize}

\subsubsection{Cointegrated Vector Autoregressive Models}

A VAR has several equivalent representations that are valuable for understanding the interactions between exogeneity, cointegration and economic policy analysis. To start, the levels form of the $s-$th order Gaussian VAR for $x$ is
\begin{align}
x_t = K q_t + \sum_{j=1}^s A_j x_{t-j} + \varepsilon_t, \ \ \ \varepsilon_t \sim \mathcal{N} \left( 0, \Sigma \right).     
\end{align}
where $K$ is an $N \times N_0$ matrix of coefficients of the $N_0$ deterministic variables $q_t$. Suppose that we have a vector $Y_t = \left[ y_{1t}, y_{2t},..., y_{nt} \right]^{\prime}$ that does not satisfy the conditions for stationarity. One way to achieve stationarity might be to model $\Delta y_t$, rather than $y_t$ itself. However, differencing can discard important information about the equilibrium relationships between the variables. This is because another way to achieve stationarity can be through linear combinations of the levels of the variables. Thus, if such linear combinations exist then we have cointegration and the variables are said to be cointegrated. The notion of Cointegration has some important implications: (i) It implies a set of dynamic long-run equilibria between the variables, (ii) Estimates of the cointegrating relationships are super-consistent, they converge at rate $T$ rather than $\sqrt{T}$, and (iii) Modelling cointegrated variables allows for separate short-run and long-run dynamic responses. Further details on Vector Autoregression and Cointegration can be found in the corresponding Chapter of \cite{watson1994vector}. 

\begin{example}
Let $\boldsymbol{x}_t$ be an $I(1)$ vector of $n$ components, each with possibly deterministic trend in mean. Suppose that the system can be written as a finite-order vector autoregression:
\begin{align}
\boldsymbol{x}_t = \boldsymbol{\mu} + \boldsymbol{\pi}_1 \boldsymbol{x}_{t-1}  + \boldsymbol{\pi}_2 \boldsymbol{x}_{t-2} + ... + \boldsymbol{\pi}_k \boldsymbol{x}_{t-k} + \boldsymbol{\varepsilon}_t, \ \ \ t = 1,..., T    
\end{align}
Then, the model can be re-written in error-correction form as below
\begin{align*}
\Delta \boldsymbol{x}_t 
= 
\boldsymbol{\mu} + \boldsymbol{\Gamma}_1 \Delta \boldsymbol{x}_{t-1} + \boldsymbol{\Gamma}_2 \Delta \boldsymbol{x}_{t-2} + ... + \boldsymbol{\Gamma}_{k-1} \Delta \boldsymbol{x}_{t-k+1} + \boldsymbol{\pi} \boldsymbol{x}_{t-k} + \boldsymbol{\varepsilon}_t   
\equiv
\boldsymbol{\mu} + \sum_{i=1}^{k-1} \boldsymbol{\Gamma}_i (1 - L) L^i \boldsymbol{x}_i + \boldsymbol{\pi} \boldsymbol{x}_{t-k} + \boldsymbol{\varepsilon}_t  
\end{align*}
Therefore, we get the following system equation representation $\boldsymbol{\pi} (L) \boldsymbol{x}_t = \boldsymbol{\mu} + \boldsymbol{\varepsilon}_t$, $t = 1,...,T$, where $
\boldsymbol{\pi} (L) = (1 - L) \boldsymbol{I}_n - \sum_{i=1}^{k-1} \boldsymbol{\Gamma}_i (1-L) L^i - \boldsymbol{\pi} L^k,$ and $
\boldsymbol{\Gamma}_i = - \boldsymbol{I}_n + \boldsymbol{\pi}_1 + \boldsymbol{\pi}_2 + ... + \boldsymbol{\pi}_i$, for $i = 1,...,k$.
\end{example}

\newpage 

\begin{remark}
Within a cointegrated analysis context, it has been proved to be advantageous for both theoretical and practical purposes to separate the long-run behaviour of the system from the more transient dynamics by using the error correction form of the model which is useful when measuring the impact of structural shocks to economic outcomes such as monetary policy, fiscal policy and macro aggregates. Thus when employing vector autoregressions to model the comovements of macroeconomic aggregates time series stationarity is assumed (e.g., by using a first-difference transformations). Moreover, cointegration dynamics can be correctly specified using a Cointegrated VAR representation. Our main interest is the structural analysis of VAR models under time series stationarity as we explain below. 
\end{remark}

\begin{example}
Consider the following data generating process as below
\begin{align}
\Delta X_t = \alpha \beta^{\prime} X_{t-1} + \sum_{i=1}^{k-1} \Gamma_i \Delta X_{t-i} + \varepsilon_t, \ \ \ \text{for} \ t = 1,...,T, 
\end{align}
where $\left\{ \varepsilon_t \right\}$ is \textit{i.i.d} with mean zero and full-rank covariance matrix $\Omega$, and where the initial values $X_{1-k},..., X_0$ are fixed. We are interested in the null hypothesis $H_0: \beta = \beta_0$. Thus, when $\beta$ is a known $(p \times r)$ matrix of full column rank $\beta_0$, the subspace spanned by $\beta$ and $\beta_0$ are identical.  
\end{example}

\begin{example}[Wage formation with Cointegrated VAR, see \cite{petursson2001wage}]
A Gaussian VAR$(k)$ model is used can be rewritten in the usual error correction form in terms of stationary variables
\begin{align}
\Delta x_t = \sum_{j=1}^{ k-1} \Gamma_j \Delta x_{t-j} + \alpha \beta^{\top} x_{t-1} + \Phi \Delta + \varepsilon_t, \ \ \ t = 1,..., T    
\end{align}
Using the cointegrated VAR model we can characterize the cointegrating relations based on a model of wage bargaining between trade unions and firms. Thus the economic theory implies that: (i) the marginal productivity condition for labour (downward sloping demand curve of firms product), and (ii) real wage relation derived from the bargaining between trade unions and firms over wage. In other words, two stationary combinations of the non-stationary data exist and that these can be identified based on the prior knowledge of these two economic relations. The associated imposed parameter restrictions are not rejected by the data and simplified analysis is achieved via the partial VAR formulation which corresponds to the conditional system for $\Delta x_{1t}$ given the past and $\Delta x_{2t}$.
\end{example}

\begin{remark}
Let $\alpha$ and $\beta$ be defined as below
$
\beta = 
\begin{pmatrix}
I_r & 0
\\
- \beta_2 & I_{k-r}
\end{pmatrix}
\ \ \ \text{and}
\ \ \
\alpha = 
\begin{pmatrix}
\alpha_{11} & \alpha_{12}
\\
0 & \alpha_{2}
\end{pmatrix}
$, then the cointegration matrix becomes as below
\begin{align}
\Pi = \beta \alpha
\begin{pmatrix}
\alpha_{11} & \alpha_{12}
\\
- \beta_2 \alpha_{11} & - \beta_2 \alpha_{12} + \alpha_{22}
\end{pmatrix},
\end{align}
where all the above parameter coefficients are unrestricted. Moreover, according to \cite{kleibergen1994shape} the behaviour of the likelihood is due to the nonidentifiedness of certain parameters, which occurs when the model is a difference stationary one. In other words, flat priors are informative in cointegration models because difference stationary models are infinitely favoured.   
\end{remark}

\newpage 

\begin{example}
Imposing long-run restrictions is a commonly used approach for the identification of structural shocks. However, usually identification assumptions for cointegrating vectors can be complicated, however following the approach of \cite{zha1999block}, which applies block recursive assumptions to the structural VECM with long-run restrictions. The block recursive system (see, also \cite{keweloh2023monetary}) is well developed in structural VAR and VECM models with short-run restrictions but not so developed in the structural VECM with long-run restrictions (see, \cite{hecq2000permanent}). 

Consider the structural VAR model as below
\begin{align*}
B(L) x_t = \mu + u_t, \ \ B(L) = B_0 - \sum_{j=1}^p B_j L^j
\end{align*}
where $u_t$ is a vector of structural disturbances and $x_t$ is an $n-$dimensional parameter of interest. Moreover, consider the structural Wold representation $\Delta x_t = \delta + \Gamma(L) u_t$. Then the corresponding reduced-form of the model is expressed as below
\begin{align*}
\Delta x_t = \delta + C(L) \varepsilon_t, \ \ \ \varepsilon_t = \Gamma_0 u_t \ \ \text{where} \ \ \Gamma (L) = \left( \Gamma_0 + \sum_{j=1}^{\infty} \Gamma_j L^j \right) \ \text{and} \ C(L) = \Gamma (L) \Gamma_0^{-1}.
\end{align*}
Suppose that $x_t$ is a vector of cointegrated time series where $u_t$ is an $( n \times 1)$ vector of serially uncorrelated structural disturbances with a mean zero and a covariance matrix $\Sigma_{u}$ and $\varepsilon_t$ is an $( n \times 1)$ vector of serially uncorrelated linear forecast errors with a mean of zero and a covariance matrix $\Sigma_{\varepsilon}$. 

In addition, assume that all variables are cointegrated such that there is a reduced number of common stochastic trends $( k = n - r)$ and a number of transitory components $r$. Moreover, these common trends are assumed to be generated by permanent shocks such that the vector of error terms $u_t$ is decomposed into $u_t \equiv ( u_{kt}^{\prime}, u_{rt}^{\prime}  )^{\prime}$, where $u_{kt}$ is a $k-$dimensional vector of permanent shocks and $u_{rt}$ is an $r-$dimensional vector of transitory shocks.

\begin{itemize}

\item Empirical evidence: The authors investigate the effects of contractionary shock to the monetary policy on economic variables in a seven-variable VECM with long-run restrictions. In particular, the permanent shocks include a shock that affects the long-run level of real exchange rates (a real-exchange-rate shock) and a US monetary policy shock that affects the long-run level of US prices. Thus a Japanese monetary policy shock can be considered as a transitory shock since the model does not include the Japanese price, while it can be considered as a permanent shock it it affects the long-run level of Japanese interest rates. 

\item Theoretical Evidence: When the monetary policy and exchange rates are simultaneously determined contemporaneously, then the recursive ordering assumptions are not satisfied. On the other hand, identification with long-run restrictions does not suffer from this simultaneity problem because no zero restrictions are imposed on the structural parameter $B_0$.
    
\end{itemize}

\end{example}

\newpage 

\begin{example}[The causal effects of fiscal policy shocks, see \cite{caldara2017analytics}]

More recently attention has been paid in the role of fiscal policy for stabilizing business cycles. Since empirical studies have not reached a consensus about the effects of fiscal policy on macroeconomic variables to assess the effects of fiscal policy the SVAR methodology is commonly used (see also \cite{boiciuc2015effects}). The structural analysis of VARs allows to estimate the effects of fiscal policy shocks on economic activity.

Consider the structural representation of a VAR model is given by 
\begin{align}
A_0 x_t = A(L) x_{t-1} + B \varepsilon_t    
\end{align}
\begin{itemize}
    \item $A_0$ is the matrix of contemporaneous influence between the variables,

    \item $x_t$ is a vector of the endogenous macroeconomic variables such as government expenditures, real output, inflation, tax revenues and short-term interest rates.

    \item $A(L)$ is an $( n \times n)$ matrix of lag-length $L$, representing impulse-response functions of the shocks to the elements of $x_t$, and $B$ is an $( n \times n )$ matrix that captures the linear relations between the structural shocks and those in the reduced form.
    
\end{itemize}
To estimate the SVAR the reduced form is given by $x_t = C(L) x_{t-1} + u_t$ where $u_t = A_0^{-1} B \varepsilon_t$. The relation between structural shocks and reduced form shocks is $A_0 u_t = B \varepsilon_t$. 
Recall that a $\mathsf{SVAR}(p)$ model with the following representation
\begin{align*}
y_t = \mu + A_1 y_{t-1} + A_2 y_{t-2} + ... + A_p y_{t-p} + \varepsilon_t
\end{align*}
can be written as $y_t = \big( I_n \otimes x_t^{\prime} \big) \beta + \varepsilon_t$, where $I_n$ is the identity matrix and $\beta = \mathsf{vec} \left( [ \mu, A_1, A_2,..., A_p ]^{\prime} \right)$.

\end{example}

\medskip

\begin{example}
Consider the study of \cite{bilgili2012impact} who attempts to reveal explicitly whether or not biomass consumption can mitigate carbon dioxide $( CO_2 )$ emissions. In other words, in order to correctly capture the underline features in the data (and produce unbiased and efficient estimators), a cointegrating regression specification with regime shifts (structural breaks) are essential to understand the long-run equilibrium of $CO_2$ emissions with biomass consumption as well as fossil fuel consumption. Their main findings include the presence of a statistical positive impact of fuel's consumption and a statistical negative impact of biomass consumption on $CO_2$ emissions.  A climate-economic related event is a change in the policy of a Country's Energy Authority by the introduction of a policy act to introduce measures for diminishing $CO_2$ emissions (e.g., an increase in biomass consumption). In particular, this can be detected in the data by identifying (dating) the presence of a regime shift using the cointegration model with structural breaks. Furthermore, a statistical negative impact on $CO_2$ emissions (or equivalently a statistical positive impact in $CO_2$ emissions reductions), might be expected to increase/decrease through possible government incentives for research and development on biomass plants (assuming that the magnitudes of other parameters, such as population growth and growth in demand for energy, will not increase beyond the expectations) (see, also \cite{phillips2020econometric}).  
\end{example}

\newpage 

\begin{example}[see, \cite{moon2002minimum}]
Suppose that $\phi = 1$ and $\mu = 0$. Define with $y_{1,t} = C_t$ and $\boldsymbol{y}_{2t} = [ W_t, I_t ]$. According to the permanent-income model all three variables are integrated of order one $I(1)$ which implies the following cointegration regression model 
\begin{align}
\begin{bmatrix}
y_{1,t}
\\
\boldsymbol{y}_{2t}
\end{bmatrix} 
=
\begin{bmatrix}
\boldsymbol{A}^{\prime}
\\
\boldsymbol{I}_2
\end{bmatrix}
\boldsymbol{y}_{2t-1} + \boldsymbol{u}_t, 
\end{align}
In particular, the distribution theory of estimators of the unrestricted cointegration vector $\boldsymbol{A}$ is well-developed in the literature which typically have a $T-$convergence rate. Moreover, the MLE and FM-OLS estimators of \cite{phillips1991optimal} and \cite{phillips1990statistical} respectively have a mixed-Gaussian limit distribution with a random covariance matrix. In addition, several studies concerning the estimation of the restricted cointegration vectors are also presented in the literature. In particular, \cite{saikkonen1995problems} extends the analysis for the estimation of cointegration vectors with linear restrictions to the case in which the restriction function is nonlinear and twice differentiable. More precisely, he provides stochastic equicontinuity conditions to make the conventional Taylor approximation approach valid.  Even if the income process is stationary such that $0 \leq \phi < 1$ and $\mu > 0$, both consumption and wealth are $I(1)$ processes under the optimal consumption choice. Therefore, the optimal decision rule creates restrictions between parameters that are associated with long-run relationships and parameters that control the short-run dynamics (see, also \cite{blanchard1988dynamic}). Define with $
y_{1,t} = C_t$, $\boldsymbol{y}_{2,t} = [ \Delta W_t, I_t ]^{\prime}$, $x_{1,t} = W_{t-1}$, $\boldsymbol{x}_{2,t} = [ 1, I_{t-1} ]^{\prime}$, $\boldsymbol{y}_t = [ y_{1,t}$, $\boldsymbol{y}_{2,t}   ]^{\prime}$, $\boldsymbol{x}_t = [ x_{1,t}, \boldsymbol{x}_{2,t}   ]^{\prime}$.
Thus, the consumption model is nested in the following general specification
\begin{align}
\begin{bmatrix}
y_{1,t}
\\
\boldsymbol{y}_{2t}
\end{bmatrix} 
=
\begin{bmatrix}
\boldsymbol{A}_{11}^{\prime} & \boldsymbol{A}_{21}^{\prime} 
\\
0 & \boldsymbol{A}_{22}^{\prime}
\end{bmatrix}
\begin{bmatrix}
x_{1,t}
\\
\boldsymbol{x}_{2,t}
\end{bmatrix} + 
\begin{bmatrix}
u_{1,t}
\\
\boldsymbol{u}_{2,t}
\end{bmatrix},
\end{align}
Define with $a_{ij} = \mathsf{vec} ( A_{ij} )$. Then the unrestricted parameter vector $a = [ a_{11}^{\prime}, a_{21}^{\prime}, a_{22}^{\prime}   ]^{\prime}$ and $b = [ r, \mu, \phi ]^{\prime}$ is composed of the structural parameters. Assume that the partial sum process of $\Delta W_t$ converges to a vector Brownian motion such that 
\begin{align}
\frac{1}{ \sqrt{T} } \sum_{t=1}^{ \floor{Tr} } \Delta W_t \Rightarrow \boldsymbol{B}(r) \equiv BM (\boldsymbol{\Omega}),    
\end{align}
where $\boldsymbol{\Omega}$ is the long-run covariance matrix of $\Delta W_t$ defined by\footnote{Related asymptotic theory with examples can be found in  \cite{katsouris2023limit} who present unit-root dynamics and weak convergence arguments to a suitable topological space. Moreover, asymptotic theory results for a cointegrating predictive regression model with near units roots when modelling the term structure of interest rates is presented by \cite{lanne2000near}. Notice that \cite{moon2002minimum} show the consistency of the estimator using a Skorohod representation of the weakly converging objective function and derive the limit distribution of the MD estimator for smooth restriction functions. }
\begin{align}
\boldsymbol{\Omega} := \underset{ T \to \infty }{ \mathsf{lim} }\  \frac{1}{T} \mathbb{E} \left[ \left( \sum_{t=1}^T \Delta W_t \right) \left( \sum_{t=1}^T \Delta W_t \right)^{\prime} \right].    
\end{align}
\end{example}

\newpage

\begin{example}[A simple climate-economic system, see \cite{pretis2021exogeneity}]
Climate and economic variables are observed over time and space. Denote with $y_t = \big( y_{1t}^{\prime}, y_{2t}^{\prime} )$ denotes the relevant climate and socio-economic variables. Denote with $Y_j^i = ( y_i,..., y_j )$ for $i \leq j$, such that $Y_T^1 = (y_1,..., y_T)$.  Then, the model of interest can be characterized as
\begin{align}
f_Y \left( Y_T^1 | Y_0, \theta  \right) = \prod_{t=1}^T f_y \big( y_t | Y_{t-1}, \theta  \big), \ \ \ \theta \in \Theta \subset \mathbb{R}^{n},   
\end{align}
where $f_y \big( y_t | Y_{t-1}, \theta  \big)$ denoting the sequentially-conditioned, joint-density for $y_t$, with $( n \times 1 )$ parameter vector $\theta$ lying in parameter space $\Theta$. A vector autoregression process within a cointegration framework due to the fact that economic and climate time-series are pre-dominantly non-stationary time series due to the presence of stochastic trends and structural breaks. Therefore, climate-economic systems can be well-approximated by cointegrated econometric models, although in addition we are interested to measure the weather shocks into the macro-economy (see, \cite{pretis2020econometric, pretis2021exogeneity}).  
\end{example}

\begin{example}
Consider the Energy Balance System studied by \cite{pretis2020econometric} (see, also \cite{carrion2021statistical} and \cite{bruns2020multicointegration}), such that the law of motion is captured by the following system of differential equations 
\begin{align}
C_m \frac{ d T_m }{ dt } &= - \lambda T_m + F - \gamma ( T_m - T_d )
\\
C_d \frac{ d T_d }{ dt } &= \gamma ( T_m - T_d )
\end{align}
The feedback parameter $\lambda$ is of particular interest as it determines the equilibrium response of surface temperatures $T_m$ to a change in the forcing $F$ (e.g., from increased $CO_2$ concentrations). Furthermore, focusing on a cointegration analysis which allows to examine whether there exists a combinations of time-series that are themselves stationary, it can be shown that the system of differential equations of the two-component EBM is equivalent to a cointegrated systen with restrictions on the parameters. Thus, a climate-economic system is formulated as a cointegrated vector-autoregression, where the regressand is decomposed into two variables, that is, $y_t = (e_t, c_t)$, such that $e_t$ represents a univariate economic variable and $c_t$ represents a univariate climate variable. 
\begin{align}
y_t &= \sum_{j=1}^s A_j y_{t-j} + \mu + \varepsilon_t, \ \ \  \varepsilon_t \sim \mathcal{N} \left( 0, \Sigma \right)  
\\
\Delta y_t &= \alpha \beta^{\prime} y_{t-1} + \Gamma \Delta y_{t-1} + \mu + \varepsilon_t,
\end{align}
where $y_t = ( e_t, c_t )^{\prime}, \varepsilon_t = ( \varepsilon_{e,t}, \varepsilon_{c,t} )$ and $\Delta y_t = y_t - y_{t-1}$. Thus, the full system can be written as 
\begin{align}
\begin{bmatrix}
\Delta e_t
\\
\Delta c_t 
\end{bmatrix}
=
\begin{bmatrix}
\alpha_1
\\
\alpha_2 
\end{bmatrix}
\begin{bmatrix}
\beta_1 & \beta_2
\end{bmatrix}
\begin{bmatrix}
e_{t-1}
\\
c_{t-1} 
\end{bmatrix}
+ 
\begin{bmatrix}
\Gamma_{11} & \Gamma_{12}
\\
\Gamma_{21} & \Gamma_{22}
\end{bmatrix}
\begin{bmatrix}
\Delta e_{t-1}
\\
\Delta c_{t-1}  
\end{bmatrix}
+ 
\begin{bmatrix}
\mu_e
\\
\mu_c
\end{bmatrix}
+ 
\begin{bmatrix}
\varepsilon_{e,t}
\\
\varepsilon_{c,t}
\end{bmatrix}
\end{align}
A relevant discussion on the interaction between unit roots and exogeneity can be found in \cite{hendry1994interactions} (see, also \cite{engle1983exogeneity}).

\newpage 

Therefore, the above economic and climate variables approximate the full climate-economic system with the links between climate and the economy given by both the short-run parameters $\Gamma$ and the equilibrium relationship $h_t$ given by the cointegrating vector $\beta^{\prime} y_t$ such that:
\begin{align}
h_t = 
\begin{bmatrix}
\beta_1 & \beta_2
\end{bmatrix} 
\begin{bmatrix}
e_{t}
\\
c_{t} 
\end{bmatrix}
= 
\begin{bmatrix}
\beta_1 e_{t}  & \beta_2 c_{t} 
\end{bmatrix}. 
\end{align}
The cointegrating relation is an equilibrium one, and does not necessarily reflect purely a climate-impact function, but rather an equilibrium between the two series, to which each series adjusts. The statistical properties of the two-stage least squares estimator under cointegration can be found in \cite{hsiao1997statistical}.
\end{example}

\subsubsection{Cointegration and Dynamic Inference from ADLM}

Consider a general autoregressive distributed lag ARDL (p,q) model where a series, $y_t$, is a function of a constant term, $\boldsymbol{\alpha}_0$, past values of itself stretching back $p$periods, contemporaneous and lagged values of an independent variable, $x_t$, of lag order $q$, and indepenendent, identically distributed error term:
\begin{align}
y_t = \boldsymbol{\alpha}_0 + \sum_{i=1}^p \boldsymbol{\alpha}_i y_{t-i} + \sum_{j=0}^q \boldsymbol{\beta}_j \boldsymbol{x}_{t-j} + \boldsymbol{\epsilon}_t,    
\end{align}

\begin{example}
A commonly used model is the ARDL (1,1) model given by 
\begin{align}
y_t = \boldsymbol{\alpha}_0 + \boldsymbol{\alpha}_1 y_{t-1} + \boldsymbol{\beta}_0 x_t + \boldsymbol{\beta}_1 x_{t-1} + \boldsymbol{\epsilon}_t,     
\end{align}
\end{example}
The contemporaneous effect of $x_t$ on $y_t$ is given by $\boldsymbol{\beta}_0$. Moreover, the magnitude of $\boldsymbol{\alpha}_1$ informs us about the memory property of $y_t$. Assuming that $0 < \alpha_1 < 1$, larger values indicate that movements in $y_t$ take longer to dissipate. The long-run effect (or long-run multiplier) is the total effect that a change in $x_t$ has on $y_t$. A simple model that incorporates such dynamic effects, is the distributed lag model.

\begin{example}[see, \cite{baumeister2021advances}]

\

The authors consider the special case of models in which only the effects of a single structural shock are identified and develop a new closed-form equation that could be used to estimate consistently the parameters of that structural equation by combining knowledge of the effects of the structural shock with the observed covariance matrix of the reduced-form residuals. Notice that exact prior information regarding the distributional assumptions of the structural model or the true ordering of variables is typically referred to as $\textit{identifying assumptions}$ (see, discussion in \cite{fry2011sign} on sign restrictions).  

In particular, consider a three-variable VAR system which is identified using a recursive structure (Cholesky Decomposition). Then, in this three-variable VAR system the order of variables matter for consistently estimating the structural parameters. Specifically, when the demand equation is ordered last in the system, then identifying assumptions imply that the parameters of the demand equation can be estimated by an OLS regression of price on current quantity, income and lagged values of the variables.  

\newpage 

Consider a demand equation system in which $q_t$ is a measure of the quantity of oil purchased, $p_t$ is a measure of the real price of oil, and $y_t$ is a measure of the real income such that 
\begin{align}
q_t = \delta y_t + \beta p_t + \textbf{b}^{\prime}_d \textbf{x}_{t-1} + u_t^d. 
\end{align}
Moreover, the demand structural system describes the behaviour of oil producers and the determinants of income such that (the order matters) 
\begin{align}
q_t  &= \delta y_t + \alpha p_t + \textbf{b}^{\prime}_s \textbf{x}_{t-1} + u_t^s \Rightarrow  q_t  - \delta y_t - \alpha p_t = \textbf{b}^{\prime}_s \textbf{x}_{t-1} + u_t^s
\\
y_t  &= \epsilon q_t + \beta p_t + \textbf{b}^{\prime}_y \textbf{x}_{t-1} + u_t^y \Rightarrow - \epsilon q_t  + y_t - \beta p_t = \textbf{b}^{\prime}_s \textbf{x}_{t-1} + u_t^s
\\
q_t  &= \zeta y_t + \gamma p_t + \textbf{b}^{\prime}_d \textbf{x}_{t-1} + u_t^d \Rightarrow  q_t  - \zeta y_t - \gamma  p_t = \textbf{b}^{\prime}_s \textbf{x}_{t-1} + u_t^s
\end{align}
where $\textbf{x}_{t-1} := \big(  1, \textbf{y}^{\prime}_{t-1}, \textbf{y}^{\prime}_{t-2},...., \textbf{y}^{\prime}_{t-p} \big)^{\prime}$ is a vector consisting of a constant term and $p$ lags of each of the three variables with $\textbf{y}_t = ( q_t, y_t, p_t )^{\prime}$. 
\begin{itemize}
\item $\alpha$ is the short-run price elasticity of oil supply,  

\item $u_t^s$ is a structural shock to oil production, 

\item $\beta$ is the contemporaneous effect of oil prices on economic activity. 

\end{itemize}
Then, the structural model can be written in the following form: 
\begin{align*}
\textbf{A} \textbf{y}_t &= \textbf{B}  \textbf{x}_{t-1} + \textbf{u}_t, 
\\
\underbrace{ 
\begin{bmatrix}
1 & - \delta & - \alpha
\\
- \epsilon & 1 & - \beta
\\
1 & - \zeta & - \gamma
\end{bmatrix}
}_{  \textbf{A}  }
\begin{bmatrix}
q_t
\\
y_t
\\
p_t
\end{bmatrix}
&= 
\begin{bmatrix}
\textbf{b}^{\prime}_s
\\
\textbf{b}^{\prime}_y
\\
\textbf{b}^{\prime}_d
\end{bmatrix}
\textbf{x}_{t-1}
+ 
\begin{bmatrix}
u_t^s
\\
u_t^y
\\
u_t^d
\end{bmatrix}.
\end{align*}
We assume that these structural shocks have mean zero and are serially uncorrelated as well as uncorrelated with each other such that 
\begin{align*}
\mathbb{E} \big[ \textbf{u}_t  \textbf{u}_t^{\prime} \big] = 
\begin{cases}
\textbf{D}, & \ \text{for} \ t = s, 
\\
0, & \ \text{for} \ t \neq s.
\end{cases}
\end{align*}
where $\textbf{D}$ is a diagonal covariance matrix.  Then, by pre-multiplying the structural VAR representation with the matrix $\textbf{A}^{-1}$, we obtain the corresponding reduced-form VAR equation which has the following dynamic structural model form $
\textbf{y}_t = \Pi \textbf{x}_t + \epsilon_t,  \ \ \ \mathbb{E} \left(  \epsilon_t \epsilon_t^{\prime}     \right) = \Omega$. Therefore the above reduced form specification can be expressed as below
\begin{align*}
\textbf{y}_t = \Pi \textbf{x}_t + \epsilon_t \equiv \textbf{c} + \Phi_1 \textbf{y}_{t-1} + \Phi_2 \textbf{y}_{t-2} + ... + \Phi_m \textbf{y}_{t-m}  + \epsilon_t.
\end{align*}

\newpage 

Then, the above reduced-form can be employed to construct impulse-response functions $\Psi_s = \partial \textbf{y}_{ t+s} / \partial \epsilon_t^{\prime}$.  
\begin{align*}
\hat{\Psi}_s = \hat{\Phi}_1 \hat{\Psi}_{s-1} + \hat{\Phi}_2 \hat{\Psi}_{s-2}  + ... + \hat{\Phi}_m \hat{\Psi}_{s-m}, \ \ \ \text{for} \ \ \ s = 1,2,3,... 
\end{align*}
starting from $\hat{\Phi}_0 = \textbf{I}_n$ and $\hat{\Psi}_s = 0$ for $s < 0$. 
\end{example}

\begin{remark}
Following \cite{kilian2017structural} we assume that the short-run income and price elasticities of supply as well as the contemporaneous coefficient relating oil prices to economic activity are all zero, although such conditions might need to be revisited especially within a climate-economy equilibrium setting without resorting to non-equilibrium dynamic states. Nevertheless, using the Cholesky factorization\footnote{Using a Cholesky decomposition is a standard approach in the literature and several studies employ such techniques although alternative parametrizations are of particular interest within a unified framework.} of $\Omega_{MLE}$ we can obtain the estimated effects on $\textbf{y}_{t+s}$ of a one-standard-deviation increase in one of the structural shocks at date $t$ (see, also \cite{hamilton1994time}). 
\end{remark}

\begin{definition}
A VAR process of order $p$, $VAR(p)$, is a multivariate process $\boldsymbol{y}_t$ specified as follows
\begin{align}
\boldsymbol{y}_t = \boldsymbol{\eta} + \sum_{j=1}^p \boldsymbol{A}_j \boldsymbol{y}_{t-j} + \boldsymbol{\varepsilon}_t, \ \ \  \boldsymbol{\varepsilon}_t \sim WN_{(n)}    
\end{align}
where $\boldsymbol{\eta}$ and $\boldsymbol{A}_1, \boldsymbol{A}_2,..., \boldsymbol{A}_p$, are a constant vector and constant matrices, respectively. 
\end{definition}

\begin{remark}
Such a process can be rewritten in operator form as below
\begin{align}
\boldsymbol{A}(L) \boldsymbol{y}_t = \boldsymbol{\eta} + \boldsymbol{\varepsilon}_t, \ \ \ \boldsymbol{A}(L) = \boldsymbol{I}_n - \sum_{j=1}^p \boldsymbol{A}_j L^j     
\end{align}
which is considered to be a stationary process provided all roots of $\mathsf{det} \left[ \boldsymbol{A}(z) \right] = 0$, lie outside the unit circle. Then, the process admits a causal VMA$(\infty)$ respresentation such that 
\begin{align}
\boldsymbol{y}_t = \boldsymbol{\omega} + \sum_{k = 0}^{\infty} \boldsymbol{C}_k \boldsymbol{\varepsilon}_{t-k}. 
\end{align}
Although in  Section \ref{Section2} we consider suitable formulations of cointegrated VAR models and their applications from related economic studies, one needs to revise the concepts presented in Chapter 10 of \cite{hamilton1994time} on Covariance-Stationary processes before proceeding to the material of Section \ref{Section3} 
\end{remark}

\textbf{Open Problems} Various open problems remain in the literature such as the use of the cross-section as a mechanism for beliefs elicitation and identification of structural shocks. Related studies from the finance literature where the cross-section is used for improving the predictability of portfolio returns include among others  \cite{chinco2019sparse}, \cite{freyberger2020dissecting} and \cite{kyle2023beliefs}, but in the case of modelling macroeconomic fundamentals using cross-sectional information is still a growing literature. Relevant studies in this direction include \cite{forni2017dynamic}, \cite{kong2019weak},   \cite{yamamoto2023cross} and \cite{hannadige2023forecasting} among others.

\newpage

\section{Vector Autoregressions: Prediction and Granger Causality}
\label{Section3}

\subsection{Stability Conditions in VAR$(p)$ Model}

\begin{example}
\label{example3}
Consider that $\boldsymbol{y}_t \sim \mathsf{VAR} (p)$ is defined recursively such that
\begin{align}
\label{expression}
\boldsymbol{y}_t = \boldsymbol{A}_1 \boldsymbol{y}_{t-1} + ... + \boldsymbol{A}_p \boldsymbol{y}_{t-p} + \boldsymbol{\varepsilon}_t, \ \ \ t = 1,...,n   
\end{align}
where $\boldsymbol{y}_t \in \mathbb{R}^d$ is a $d-$dimensional vector. Then, the above $\mathsf{VAR} (p)$ process can be expressed as a $\mathsf{VAR} (1)$ process using a companion matrix as below
\begin{align}
\begin{bmatrix}
\boldsymbol{y}_t 
\\
\boldsymbol{y}_{t-1}
\\
\vdots
\\
\boldsymbol{y}_{t - p + 1 } 
\end{bmatrix} 
= 
\begin{bmatrix}
\boldsymbol{A}_1 & \boldsymbol{A}_2 & \boldsymbol{A}_3 & \hdots & \boldsymbol{A}_p
\\
\boldsymbol{I}_d & \boldsymbol{0} & \boldsymbol{0} & \hdots & \boldsymbol{0}
\\
\boldsymbol{0} & \boldsymbol{I}_d  & \boldsymbol{0} &  \hdots & \boldsymbol{0} 
\\
\vdots & \vdots & \ddots & \hdots & \vdots
\\
\boldsymbol{0} & \boldsymbol{0} & \vdots & \boldsymbol{I}_d &  \boldsymbol{0}
\end{bmatrix}
\begin{bmatrix}
\boldsymbol{y}_{t-1} 
\\
\boldsymbol{y}_{t-2}
\\
\vdots
\\
\boldsymbol{y}_{t - p } 
\end{bmatrix}   
+
\begin{bmatrix}
\boldsymbol{\varepsilon}_t
\\
\boldsymbol{0}
\\
\boldsymbol{0}
\\
\vdots
\\
\boldsymbol{0}
\end{bmatrix}
\end{align}
such that $\boldsymbol{Y}_t = \boldsymbol{\mathcal{\mathcal{A}}} \boldsymbol{Y}_{t-1} + \boldsymbol{U}_t$, which implies that a unique solution can be determined using 
\begin{align}
\boldsymbol{Y}_t = \boldsymbol{\mathcal{\mathcal{A}}}^t \boldsymbol{Y}_{0}  + \sum_{j=1}^t \boldsymbol{\mathcal{\mathcal{A}}}^{t-j} \boldsymbol{U}_j. 
\end{align}
\end{example}

\begin{assumption}
\label{Assumption2}
(a) any non-zero eigenvalue $\lambda$ of $\mathcal{\mathcal{A}}$ satisfies the following stability condition
\begin{align}
\mathsf{det} \left( \boldsymbol{I}_d - \sum_{j=1}^p \frac{1}{ \lambda^j } \boldsymbol{ \mathcal{A} }_j \right) = 0.    
\end{align}     
(b) The spectral radius of the companion matrix $\boldsymbol{ \mathcal{A} }$ is smaller than 1, $\rho ( \boldsymbol{ \mathcal{A} } ) < 1$. 
\end{assumption}

\begin{lemma}
Consider the homoscedastic $\mathsf{VAR}(p)$ process defined in \eqref{expression}, where the roots of the characteristic polynomial satisfy the conditions of Assumption \ref{Assumption2}.  
\begin{itemize}

\item[(i).] Consider the process $\boldsymbol{y}_t \sim \mathsf{VAR} (p)$, and suppose that the spectral radius of the companion matrix $\boldsymbol{\mathcal{A} }$ satisfies $\rho( \boldsymbol{\mathcal{A} } ) < 1$. Then, since $\boldsymbol{Y}_0 = \sum_{j=0}^{ \infty } \boldsymbol{\mathcal{A} }^j \boldsymbol{U}_{t-j}$ almost surely then under the covariance stationarity assumption it holds that $\boldsymbol{\Gamma}_y (j) = \mathbb{E} \big[ \boldsymbol{Y}_t  \boldsymbol{Y}_{t-j}^{\prime} \big]$. 

\item[(ii).] An equivalent expression for the $\mathsf{VAR} (p)$ process using the lag operator is as below
\begin{align}
\left( \boldsymbol{I}_d - \boldsymbol{A}_1 L -  \boldsymbol{A}_2 L^2 - ... - \boldsymbol{A}_p L^p \right) \boldsymbol{y}_t  = \boldsymbol{\varepsilon}_t.   
\end{align}
Then the characteristic polynomial of the $\mathsf{VAR} (p)$ process is written as
\begin{align}
\boldsymbol{\Phi} (z) := \left( \boldsymbol{I}_d - \boldsymbol{A}_1 z -  \boldsymbol{A}_2 z^2 - ... - \boldsymbol{A}_p z^p \right) , \ \ \ \text{for some} \ \ z \in \mathbb{C}.    
\end{align}

\newpage

Consequently, the following are equivalent: $\textit{(i)}$ $\boldsymbol{y}_t$ is a stable process;  $\textit{(ii)}$  $\rho( \boldsymbol{\mathcal{A}} ) < 1$; $\textit{(iii)}$ all the roots of the characteristic polynomial $\mathsf{det} \left\{  \boldsymbol{\Phi} (z)  \right\} = 0$ lie outside the unit circle $\left\{ z \in \mathbb{C} : |z| < 1  \right\}$ such that $\mathsf{det} \left\{  \boldsymbol{\Phi} (z)  \right\} = 0$ \textit{iff} |z| > 1.
\end{itemize}

\item[(iii).] Hence, based on the results given in (i) and (ii) above, the $\mathsf{VAR} (p)$ process admits a linear process representation such that 
\begin{align}
\boldsymbol{y}_t  = \sum_{j=0}^{ \infty} \boldsymbol{C}_j \boldsymbol{\varepsilon}_{t-j}, \ \ \textit{almost surely}, \ \ \ \text{with} \ \ \boldsymbol{\varepsilon}_{t} \overset{\textit{i.i.d}}{\sim} ( \boldsymbol{0}, \boldsymbol{\Sigma}_{\varepsilon}),   
\end{align}
for a sequence of $\left\{ \boldsymbol{C}_j  \right\}_{ j \geq 0 }$ satisfying the summability condition $\sum_{ j = 0}^{ \infty } \norm{ \boldsymbol{C}_j } < \infty$.
\end{lemma}

\subsection{Relation between Dynamic Structural Models and Vector Autoregressions}

Consider that these interelated equations can be written in the following form:
\begin{align}
\boldsymbol{B}_0 \boldsymbol{y}_t = \boldsymbol{\mu} + \boldsymbol{B}_1 \boldsymbol{y}_{t-1} +   \boldsymbol{B}_2 \boldsymbol{y}_{t-2} + ... +  \boldsymbol{B}_p \boldsymbol{y}_{t-p} + \boldsymbol{u}_t, 
\end{align}
Therefore, by pre-multiplying the above dependent variable with $\boldsymbol{B}_0$ we can obtain a VAR representation given by the following expression: 
\begin{align}
\boldsymbol{y}_t = \boldsymbol{c} + \boldsymbol{\Phi}_1 \boldsymbol{y}_{t-1} +   \boldsymbol{\Phi}_2 \boldsymbol{y}_{t-2} + ... +  \boldsymbol{\Phi}_p \boldsymbol{y}_{t-p} + \boldsymbol{\varepsilon}_t, 
\end{align}
In other words, we can view the above representation as a special case of the Dynamic Structural system equation since we eliminate the interelations of the dependent variable to a reduced form of a VAR. 

\begin{example}
Let $\boldsymbol{z}_t = \begin{pmatrix}
z_{1t} \\ z_{2t}    
\end{pmatrix}$ be an $n-$dimensional vector stochastic process, where $z_{1t}$ is an $( n_1 \times 1 )$ and $z_{2t}$ is an $( n_2 \times 1 )$ such that $n = n_1 + n_2$. Assume a linear dynamic model (see, \cite{lastrapes2005estimating})
\begin{align}
\boldsymbol{A}_0 \boldsymbol{z}_t = \boldsymbol{A}_1 \boldsymbol{z}_{t-1} + ... + \boldsymbol{A}_p \boldsymbol{z}_{t-p} + \boldsymbol{u}_t,     
\end{align}
where $\boldsymbol{u}_t = \begin{pmatrix}
u_{1t} & u_{2t}    
\end{pmatrix}^{\prime}$ is a white noise vector process normalized such that $\mathbb{E} [ \boldsymbol{u}_t  \boldsymbol{u}_t^{\prime} ] = \boldsymbol{I}$, which implies that the reduced-form errors have a diagonal variance-covariance matrix equal to the identify matrix. Moreover, the reduced form of this structural model is formulated as below
\begin{align*}
\boldsymbol{z}_t 
&= 
\boldsymbol{A}_0^{-1} \boldsymbol{A}_1 \boldsymbol{z}_{t-1} + ... + \boldsymbol{A}_0^{-1} \boldsymbol{A}_p \boldsymbol{z}_{t-p} + \boldsymbol{A}_0^{-1} \boldsymbol{u}_t
\\
\boldsymbol{z}_t 
&= 
\boldsymbol{B}_1 \boldsymbol{z}_{t-1} + ... + \boldsymbol{B}_p \boldsymbol{z}_{t-p} + \boldsymbol{\varepsilon}_t.
\end{align*}
such that the structural errors have a variance-covariance matrix $\mathbb{E} \left( \boldsymbol{\varepsilon}_t \boldsymbol{\varepsilon}_t^{\prime} \right) = \boldsymbol{\Omega}$. Therefore, the MA representation of the structural model is
\begin{align}
z_t = \big( \boldsymbol{A}_0 - \boldsymbol{A}_1 L - ... - \boldsymbol{A}_p L^p \big) \boldsymbol{u}_t  \equiv \big( \boldsymbol{D}_0 + \boldsymbol{D}_1 L +  \boldsymbol{D}_2 L^2 + ...  \big) \boldsymbol{u}_t  = \boldsymbol{D} (L) \boldsymbol{u}_t.  
\end{align}

\newpage

Similarly, the reduced form MA has the following representation
\begin{align}
z_t = \big( \boldsymbol{I} - \boldsymbol{B}_1 L - ... - \boldsymbol{B}_p L^p \big) \boldsymbol{\epsilon}_t  \equiv \big( \boldsymbol{I} + \boldsymbol{C}_1 L +  \boldsymbol{C}_2 L^2 + ...  \big) \boldsymbol{u}_t  = \boldsymbol{C} (L) \boldsymbol{\epsilon}_t.    
\end{align}
In other words, the parameters of interest are the structural dynamic multiplies or impulse response functions such that $\frac{ \partial z_{t+k}  }{ \partial u_t } = D_k$. The common practice in the applied macroeconometrics literature is to obtain estimates for $\boldsymbol{B} (L)$ and $\boldsymbol{\Omega}$ and then impose suitable restrictions on the underline structure to identify the model parameters. 
\end{example}

\begin{remark}
Notice that while reduced-form parameters $A_j$ for $j = 1,..., p$ and the covariance matrix of residuals $\Sigma_u$ can be estimated consistently, the structural parameters collected in the matrix $B$ are not identified without further theoretical or statistical assumptions. One possible solution is to restrict the matrix $B$ to be a lower-triangular imposing a recursive causal structure among the model variables.  
\end{remark}

\subsubsection{Main Assumptions}

Consider an $n-$dimensional covariance stationary zero-mean vector stochastic process $\boldsymbol{x}_t$ of observable variables, driven by $q-$dimensional unobservable vector process $u_t$ of structural shocks. 
\begin{align}
\boldsymbol{x}_t = \boldsymbol{C} (L) \boldsymbol{u}_t,    
\end{align}
where $\boldsymbol{C} (L) = \displaystyle \sum_{j=0}^{\infty} C_j L^j$ is an one-sided polynomial in the lag operator $L$ in infinite order. These shocks are orthogonal white noises such that $u_t \sim ( 0, \boldsymbol{\Sigma}_u )$, where $\boldsymbol{\Sigma}_u$ is diagonal. Furthermore, notice that a VAR model based on a subset of the variables implies a different $MA(\infty)$ representation than its counterpart based on the joint distribution. We begin by considering various examples regarding the identification and estimation of SVAR models using properties of structural shocks. 
\begin{definition}[Fundamentalness in Systems]
Given a covariance stationary vector process $x_t$, the representation $\boldsymbol{x}_t = \boldsymbol{C} (L) \boldsymbol{u}_t$ is fundemental if 
\begin{itemize}

\item[\textit{(i).}] $\boldsymbol{u}_t$ is a white noise vector,

\item[\textit{(ii).}] $\boldsymbol{C} (L)$ has no poles of modules less or equal than unity, i.e., no poles inside the unit disc. 

\item[\textit{(iii).}] $\mathsf{det} [ \boldsymbol{C}(z) ]$ has no roots of modules less than unity, i.e., all its roots are outise the unit disc
\begin{align}
\boldsymbol{C}(z) \neq 0, \ \ \ \forall \ z \in \mathbb{C} \ \ s.t \ |z| < 1.     
\end{align}
    
\end{itemize}
\end{definition}
If the roots of $\mathsf{det} [ \boldsymbol{C}(z) ]$ are outside the unit disc, we have invertability in the past, that is, the inverse representation depends only on non-negative powers of $L$, and we have fundamentalness. If at least one of the roots of $\mathsf{det} [ \boldsymbol{C}(z) ]$  is inside the unit disc, then invertability and non-fundementalness holds.

\newpage 
   
\begin{example}[Partially nonstationary multivariate autoregressive AR(p)] Consider the partially nonstationary multivariate autoregressive AR(p) given by 
\begin{align}
\boldsymbol{\Phi} (L) \boldsymbol{Y}_t \equiv \left( \boldsymbol{I}_m - \sum_{j=1}^p \boldsymbol{\Phi}_j L^j \right) \boldsymbol{Y}_t  = \boldsymbol{\epsilon}_t   
\end{align}
where $\left\{ \boldsymbol{Y}_t \right\}$ is an $m-$dimensional process and for $\boldsymbol{\epsilon}_t$ it holds that $\mathbb{E} ( \boldsymbol{\epsilon}_t ) = 0$ and $\mathsf{cov} ( \boldsymbol{\epsilon}_t ) = \boldsymbol{\Omega}$. Moreover, it is assumed that $\mathsf{det} \big[ \boldsymbol{\Phi} (L) \big] = 0$ has $d < m$ unit roots and the remaining roots are outside the unit circle and that $\mathsf{rank} \big[ \Phi (1) \big] = r$, where $r = (m-d) > 0$ (see, \cite{samuelson1941conditions} and \cite{dickey1986unit}). Therefore, each component of the first difference $W_t = ( Y_t - Y_{t-1} )$ is assumed to be stationary.

Then the model has the following error-correction form below
\begin{align}
\boldsymbol{\Phi}^{*} (L) (1-L) \boldsymbol{Y}_t = \boldsymbol{C} \boldsymbol{Y}_{t-1} + \boldsymbol{\varepsilon}_t,  
\ \ \ \ 
\Delta \boldsymbol{Y}_t = \boldsymbol{C} \boldsymbol{Y}_{t-1} + \sum_{j=1}^{p-1} \boldsymbol{\Phi}^{*}  \Delta \boldsymbol{Y}_{t-j}  +  \boldsymbol{\varepsilon}_t  
\end{align}
where the matrix coefficients are defined as below
\begin{align}
\boldsymbol{\Phi}^{*} (L) = \left( \boldsymbol{I}_m - \sum_{j=1}^{p-1} \boldsymbol{\Phi}_j^{*} L^j  \right), \ \ \  \boldsymbol{\Phi}_j^{*} = - \sum_{k=j+1}^p \boldsymbol{\Phi}_k, \ \ \ \boldsymbol{C} = - \boldsymbol{\Phi} (1).   
\end{align}
Moreover, the Jordan Canonical form of $\sum_{j=1}^p \boldsymbol{\Phi}_j$ is given by
\begin{align}
\boldsymbol{P}^{-1} \left( \sum_{j=1}^p \boldsymbol{\Phi}_j \right) \boldsymbol{P} = \mathsf{diag} \big( \boldsymbol{I}_d, \boldsymbol{\Lambda}_r \big)  
\end{align}
One can define with $\boldsymbol{Z}_t = [ \boldsymbol{Z}_{1t}^{\prime}, \boldsymbol{Z}_{2t}^{\prime} ] = \boldsymbol{Q} \boldsymbol{Y}_t$, where $\boldsymbol{Z}_{1t} = \boldsymbol{Q}_1^{\prime} \boldsymbol{Y}_t$ and $\boldsymbol{Z}_{2t} = \boldsymbol{Q}_2^{\prime} \boldsymbol{Y}_t$.
\end{example}

The identification of structural vector autoregressions is a crucial step before proceeding with constructing test statistics and forecasts (see Section \ref{Section4}). A suitable approach found in the literature, is the method of recursive identification\footnote{The theoretical analysis of recursive identification methods is given by \cite{soderstrom1978theoretical} and \cite{chen2010new} } (see, \cite{chen2014recursive}) which implies exact identification of the structural parameter $B_0^{-1}$ (e.g., see \cite{kim2000exchange}).  Moreover, the presence of stochastic singularity are crucial on whether a system is identified or unidentifiable. Based on these considerations \cite{komunjer2011dynamic} employ restrictions implied by observational equivalence to establish conditions for avoiding the presence of stochastic singularity. To ensure that these structural econometric models can be still partially identified (see, also \cite{phillips1989partially}) regardless of the possible presence of stochastic singularity, then the econometrician can exploit the existence and uniqueness of local identification of DSGE models based on the corresponding linearized solution. Lastly, given the prominent role of expectations in news aspects such as the information flow and anticipated shocks can affect the identification of structural shocks  (e.g., optimal inter-temporal decisions discount future tax obligations, see \cite{leeper2013fiscal} and \cite{mertens2013dynamic}).

\newpage 

\subsubsection{Weak Exogeneity in $I(2)$ VAR Systems}

The notion of weak exogeneity is important when considering the structural analysis of cointegrating regression models. Under the assumption of weak exogeneity estimation of the cointegration parameters in conditional models is accordingly influenced. In particular, for the VAR model allowing for $I(1)$ variables under the assumption of Gaussian errors implies the use of a cointegrated VAR specification which is suitable for modelling long-run equilibrium dynamics for multivariate time series.  Within this stream of literature one is interested to analyze the conditions under which a subset of equations is weakly exogenous with respect to the cointegration parameters (see,  \cite{paruolo1997asymptotic, paruolo2000asymptotic}, \cite{paruolo1999weak},  \cite{dolado1992note} and \cite{pesaran2000structural}). Moreover, \cite{tchatoka2013finite} propose unified exogeneity test statistics and examine the pivotality property under strict exogeneity (see, also \cite{white2014granger}) and characterize the finite-sample distributions of the statistics under $H_0$, including when identification is weak and the errors are possibly non-Gaussian. Assume that the vectors $\boldsymbol{U}_t = [ \boldsymbol{u}_t, \boldsymbol{V}_t ]^{\prime}$ for $t = 1,..., T$, have the same nonsingular covariance matrix defined below
\begin{align}
\mathbb{E} \big[ \boldsymbol{U}_t \boldsymbol{U}_t^{\prime} \big] = \boldsymbol{\Sigma} =
\begin{bmatrix}
\sigma_u^2 & \boldsymbol{\delta}^{\prime}
\\
\boldsymbol{\delta}^{\prime} & \boldsymbol{\Sigma}_V
\end{bmatrix} > \boldsymbol{0}, \ \ \ t = 1,..., T, 
\end{align}
where $\boldsymbol{\Sigma}_V$ has dimension $G$ (see, also \cite{khalaf2014identification}). Then, the covariance matrix of the reduced-form disturbances $\boldsymbol{W}_t = [ \boldsymbol{v}_t, \boldsymbol{V}_t^{\prime} ]^{\prime} = [ \boldsymbol{u} + \boldsymbol{V} \boldsymbol{\beta}, \boldsymbol{V} ]$ takes the following form
\begin{align}
\boldsymbol{\Omega} := 
\begin{bmatrix}
\sigma_u^2 + \boldsymbol{\beta}^{\prime} \textcolor{blue}{\boldsymbol{\Sigma}_V }\boldsymbol{\beta} + 2 \boldsymbol{\beta}^{\prime} \boldsymbol{\delta} &  \boldsymbol{\beta}^{\prime} \textcolor{blue}{\boldsymbol{\Sigma}_V}  + \boldsymbol{\delta}^{\prime}
\\
\textcolor{blue}{ \boldsymbol{\Sigma}_V }\boldsymbol{\beta} + \boldsymbol{\delta} &  \textcolor{blue}{ \boldsymbol{\Sigma}_V } 
\end{bmatrix}
\end{align}
where $\boldsymbol{\Omega}$ is positive definite matrix. The exogeneity hypothesis can be expressed as $\mathcal{H}_0: \boldsymbol{\delta} = \boldsymbol{0}$. 

\begin{example}
Suppose that $\boldsymbol{W}_t = \boldsymbol{J} \bar{ \boldsymbol{W} }_t$, for $t = 1,..., T$, and suppose that $\bar{ \boldsymbol{W} }_t \overset{ \textit{i.i.d} }{ \sim } \mathcal{N} \left( \boldsymbol{0}, \boldsymbol{I}_{G+1} \right)$. Then, it holds that $\boldsymbol{\Omega} = \mathbb{E} \left[ \boldsymbol{W}_t \boldsymbol{W}_t \right] = \boldsymbol{J} \boldsymbol{J}^{\prime}$. Moreover, since $\boldsymbol{J}$ is upper triangular, then its inverse $\boldsymbol{J}^{-1}$ is also upper triangular. Let $\boldsymbol{P} = \left( \boldsymbol{J}^{-1} \right)^{\prime}$. Since $\boldsymbol{P}$ is a $( G \times 1) \times ( G \times 1)$ lower triangular matrix then we can orthogonalize the matrix $\boldsymbol{J} \boldsymbol{J}^{\prime}$ such that
\begin{align}
\boldsymbol{P}^{\prime} \boldsymbol{J} \boldsymbol{J}^{\prime} \boldsymbol{P} = \boldsymbol{I}_{G+1}, \ \ \  \left( \boldsymbol{J} \boldsymbol{J}^{\prime} \right)^{-1} = \boldsymbol{P} \boldsymbol{P}^{\prime}.  
\end{align}
The matrix $\boldsymbol{P}$ is the Cholesky factor of $\boldsymbol{\Omega}^{-1}$, so $\boldsymbol{P}$ is the unique lower triangular matrix. 
\begin{align}
\boldsymbol{P} =
\begin{bmatrix}
P_{11} & 0 
\\
P_{21} & P_{22}
\end{bmatrix}
\end{align}
Thus, an appropriate $\boldsymbol{P}$ matrix is obtained by taking
\begin{align}
P_{11} &:= \left( \sigma_u^2 - \boldsymbol{\delta}^{\prime} \boldsymbol{\Sigma}_V^{-1} \boldsymbol{\delta} \right)^{-1/2} \equiv \sigma_{\varepsilon}.   
\\
P_{21} &:= - \left( \boldsymbol{\beta} + \boldsymbol{\Sigma}_V^{-1} \boldsymbol{\delta} \right)  \left( \sigma_u^2 - \boldsymbol{\delta}^{\prime} \boldsymbol{\Sigma}_V^{-1} \boldsymbol{\delta} \right)^{-1/2}  \equiv - ( \beta + \alpha ) \sigma^{-1}_{\varepsilon}. 
\end{align}
\end{example}

\newpage

\subsection{Local Projections}

\begin{assumption}[Wold representation]
\

Suppose that $\varepsilon_t$ represents an innovation sequence which is strictly stationary and ergodic such that $\mathbb{E} ( \varepsilon_t | \mathcal{F}_{t-1} ) = 0$ almost surely, where $\mathcal{F}_{t-1} = \sigma ( \varepsilon_{t-1}, \varepsilon_{t-2},... )$. Then, $y_t$ satisfy a Wold representation if it can be expressed as $y_t = \sum_{j=0}^{\infty} \Phi_j u_{t-j}$ or $y_t = \Phi (L) u_t$. 
\end{assumption}

\subsubsection{Application: Measuring the Impact of Fiscal Policy}

This example is based on the study of \cite{mertens2010measuring} who consider the consequences of anticipation effects for VAR-based estimates of the impact of government spending shocks. We consider a bivariate time series representation of a vector consisting of control variables $z_t$, and government spending $\mathsf{g}_t$. We assume transitory government spending shocks, such that $\mu_{ \mathsf{g} } (L)$ is a stable polynomial. The identification is facilitated by assuming that information can arrive with any anticipation horizon between 1 and $q$ periods. Thus, the vector of observables $\boldsymbol{v}_t = [ \mathsf{g}_t, z_t ]^{\prime}$ is formulated as below
\begin{align}
\boldsymbol{v}_t &=
\begin{bmatrix}
0 & 1 - \mu_{ \mathsf{g} } (L)
\\
\phi_{ zk } & \phi_{ zk } 
\end{bmatrix}
\begin{bmatrix}
k_t
\\
\mathsf{g}_t
\end{bmatrix}
+ 
\begin{bmatrix}
1 & L^q
\\
0 & \phi_{z,1} \Theta (L)
\end{bmatrix}
\boldsymbol{\Sigma}_{e} \boldsymbol{e}_t,
\\
\boldsymbol{\Theta} (L) &= \omega^{q-1} + \omega^{q-2} L +    ... + \omega L^{q-2} + L^{q-1}, \ \ \boldsymbol{\Sigma}_{e} \equiv \sigma_{ \mathsf{g} } \begin{bmatrix}
1 & 0 
\\
0 & \lambda
\end{bmatrix}, 
\ \boldsymbol{e}_t \equiv 
\begin{bmatrix}
e_{0,t}^{ \mathsf{g} }
\\
e_{q,t}^{ \mathsf{g} }
\end{bmatrix} 
\end{align}
Substituting to the steady-state equilibrium solution of the system, the MA representation is
\begin{align*}
v_t 
= \boldsymbol{\mathcal{Y} } (L)  \boldsymbol{\Sigma}_{e} \boldsymbol{e}_t
=
\mu_{ \mathsf{g} } (L)^{-1} 
\begin{bmatrix}
1  & L^q
\\
1 & 0
\end{bmatrix}.
\end{align*}

\begin{example}
Consider that the vector $\boldsymbol{y}_t = [ GOV_t, GDP_t, CON_t ]$, where all three macroeconomic variables are in real terms and in logarithms. Then, the VECM is given by the following expression
\begin{align}
\Delta \boldsymbol{y}_t = \boldsymbol{\Pi} \boldsymbol{y}_{t-1} + \boldsymbol{C} (L) \boldsymbol{y}_{t-1} + \boldsymbol{D} \varepsilon_t,  
\end{align}
where $\boldsymbol{D} = \boldsymbol{Y} (0) \boldsymbol{\Sigma}_{e} \boldsymbol{B} (0)$ and $\boldsymbol{\varepsilon}_t \boldsymbol{B}(L)^{-1} \boldsymbol{e}_t$, where $\boldsymbol{e}_t$ contains the structural shocks of interest. Notice that despite the presence of permanent fiscal shocks, the variables in $\boldsymbol{y}_t$ cointegrate since the investment-output ratio is unaffected by the level of government spending in the long-run. Moreover, the unanticipated government spending shock is allowed to affect the level of government spending immediately, while the anticipated government spending shock is assumed not to affect government spending within one quarter. However, the two shocks are restricted to have the same long run impact on the level of government spending (see, \cite{mertens2010measuring}). Overall, the estimation procedure aims to uncover the response to an anticipated fiscal shock. On the other hand, when the anticipation rate is high which implies that the anticipated shocks are relatively important, then biased estimates for the unanticipated shock are obtained in small samples.  Further examples can be found in \cite{mertens2013dynamic}.   
\end{example}

\newpage 

Two popular methods to estimate impulse responses when a shock has already been identified include: local projections and distributed lag models. In particular, one can show that the dynamic response from a VAR with the shock embedded as an endogenous variable is equivalent to that of a VAR with the shock included as an exogenous variable only when that shock has no serial correlation (e.g., see \cite{plagborg2021local},  \cite{olea2021inference}). The result follows because a VAR with a shock as an exogenous variable (VAR-X) can be seen as a multivariate generalization of a DLM. A comparison of these cases can be found in \cite{alloza2020dynamic}. Moreover, \cite{lusompa2023local} propose an efficient estimation method for Local projections when residuals are autocorrelated.  Lastly, the simultaneous impact of monetary and fiscal policy is discussed by \cite{bruneau2003monetary} while the relation of fiscal policy and persistent inflation is studied by  \cite{bianchi2023fiscal}.

\begin{example}[see, \cite{alloza2020dynamic}]
Consider the following data generating mechanism:
\begin{align}
y_t &= \rho y_{t-1} + \delta_0 x_t + \delta_1 x_{t-1} + u_t,
\\
x_t &= \gamma x_{t-1} + v_t.
\end{align}
Then, the process can be formulated as a SVAR of the form $\boldsymbol{A}_0 \boldsymbol{Y}_t = \boldsymbol{B} \boldsymbol{Y}_{t-1}$ such that 
\begin{align}
\begin{bmatrix}
1 & 0 
\\
- \delta_0 & 1
\end{bmatrix} 
\begin{bmatrix}
x_t
\\
y_t
\end{bmatrix}
=
\begin{bmatrix}
\gamma & 0
\\
\delta_1 & \rho
\end{bmatrix}
\begin{bmatrix}
x_{t-1}
\\
y_{t-1}
\end{bmatrix}
+
\begin{bmatrix}
\varepsilon_t^x
\\
\varepsilon_t^y
\end{bmatrix}
\end{align}
Consider the following econometric specification 
\begin{align}
\boldsymbol{y}_t &= \boldsymbol{\mu} + \sum_{j=1}^p \boldsymbol{A}_j \boldsymbol{y}_{t-j}  + \boldsymbol{B} \boldsymbol{\varepsilon}_t  
\\
\boldsymbol{z}_t &=  \boldsymbol{\mu}_z + \boldsymbol{\Gamma} \boldsymbol{\varepsilon}_t  + \boldsymbol{\Sigma}^{1/2} \boldsymbol{\eta}_t 
\end{align}
Without further restrictions, the augmented SVAR model is only identified up to orthogonal rotations of the form $\boldsymbol{B} = \tilde{\boldsymbol{B}} \boldsymbol{Q}$. According to \cite{alloza2020dynamic} the presence of serial correlation can be thought as cross-sectional persistence as in the case of narrative shocks, which are usually serially correlated. This aspect can directly affect the identification and estimation of their dynamic effects. In particular, when estimating the dynamic response of some variable to a serially correlated shock, some part of the persistence may be spillover on the impulse response function. We discuss the use of a suitable parametrization for modelling persistence in macroeconomic data in Section \ref{Section4.6}.
\end{example}

\subsubsection{Estimation Method}

Recently, it has been shown that the method of local projections has advantages in drawing inference of impulse responses, especially with respect to uniform validity (e.g., see \cite{inoue2020uniform}). We demonstrate an example of the estimation method from the framework of  \cite{xu2023local}. The VAR$(p)$ model is of the form $y_t = A_1 y_{t-1} + ... + A_p y_{t-p} + u_t$, $t = 1,...,n$, where $u_t$ is serially uncorrelated shock.

\newpage

To measure the responses of the first endogenous variable $y_{1t}$ to shocks after $h$ propagation periods, where $h \geq 1$, the LP method runs the regression below (see, \cite{xu2023local}). 
\begin{align}
y_{1,t+h} &= \beta_1(h) y_t + \sum_{j=1}^{p-1} \theta_{1j}(h)^{\prime} y_{t-j} + \xi_{1t}(h) 
\\
y_{2,t+h} &= \beta_2(h) y_t + \sum_{j=1}^{p-1} \theta_{2j}(h)^{\prime} y_{t-j} + \xi_{2t}(h)
\\
\vdots \ \  &= \ \ \vdots
\\
y_{k,t+h} &= \beta_k(h) y_t + \sum_{j=1}^{p-1} \theta_{kj}(h)^{\prime} y_{t-j} + \xi_{kt}(h)
\end{align}
Alternative, we can estimate the following system of equations
\begin{align}
y_{1,t+h} &= \beta_1(h) y_t + \sum_{j=1}^{p} \theta_{1j}(h)^{\prime} y_{t-j} + \xi_{1t}(h) 
\\
y_{2,t+h} &= \beta_2(h) y_t + \sum_{j=1}^{p} \theta_{2j}(h)^{\prime} y_{t-j} + \xi_{2t}(h)
\\
\vdots \ \  &= \ \ \vdots
\\
y_{k,t+h} &= \beta_k(h) y_t + \sum_{j=1}^{p} \theta_{kj}(h)^{\prime} y_{t-j} + \xi_{kt}(h)
\end{align}
Notice that by definition it holds that $\theta_{1,p} (h) = 0$. Recall that for a matrix $x$, we denote its Frobenius norm, such that $|x| = [ \mathsf{trace} ( x^{\prime} x ) ]^{1/2}$. Suppose that $\left\{ \widehat{u}_t(h): t = p^{*},..., n - h \right\}$ are OLS residuals of the $\mathsf{VAR} ( p^{*} - 1 )$ regression, using the data $\left\{ y_t: t = p^{*},..., n - h \right\}$. In addition, related distributional assumptions for the error term of the model can be imposed. 

The OLS estimator is given by 
\begin{align}
\left( \widehat{\beta}^{LA}_1 (h) - \beta_1(h) \right) 
=     
\left( \sum_{t=p^{*}}^{n-h} \widehat{u}_t(h) \widehat{u}_t(h)^{\prime} \right)^{-1} \left( \sum_{t=p^{*}}^{n-h} \widehat{u}_t(h) \xi_{1t} (h) \right). 
\end{align}
Notice that these local projection estimates correspond to an in-sample estimates for the VAR$(p)$ model based on $p$ lags. On the other hand, if we are interested to obtain out-of-sample forecast sequences then we need to implement a forecasting scheme such as a moving window which is used to obtain pseudo-out-of-sample forecasts. The window size should be larger than the number of lags to have good statistical properties for the forecast sequences. Relevant applications of local projections include the construction of impulse response functions.

\newpage

\subsection{Impulse Response Functions}

A major advantage of modelling interdependent systems such as Structural Vector Autoregressions is that these system representations can be used to construct impulse response functions and forecast error variance decompositions, which are useful for investigating the dynamics within the system as well as the statistical properties of econometric specifications for forecasting purposes (see, \cite{baillie2013estimation}). Thus, reduced-form impulse responses allows us to find the response of one variable to an impulse in another variable. Usually impulse response functions can be constructed based on the stable VAR$(p)$ process. An alternative method for estimating impulse response functions has been proposed in the literature, using the method of local projections as in \cite{jorda2005estimation}, \cite{barnichon2019impulse}, \cite{gorodnichenko2020forecast}, \cite{montiel2021local} and \cite{plagborg2021local}. In addition, the use of exogenous instrumentation as an identification strategy of forecast errors is examined by \cite{olea2021inference} and \cite{plagborg2022instrumental}. The construction of robust confidence intervals for IV-based local projections is studied by \cite{noh2017impulse}, \cite{gafarov2018delta}, \cite{ganics2021confidence} and \cite{koo2022impulse} among others.

\subsubsection{Asymptotic Results for VAR Processes with Known Order}

Suppose that $\boldsymbol{\beta}$ is an $( n \times 1)$ vector of parameters and $\hat{\boldsymbol{\beta}}$ is an estimator such that 
\begin{align}
\sqrt{T} \left(  \boldsymbol{\beta} -  \hat{\boldsymbol{\beta}}  \right)  \overset{d}{\to} \mathcal{N} \left( 0, \boldsymbol{\Sigma}_{\beta} \right),   
\end{align}
Moreover, let $\big( g_1 ( \boldsymbol{\beta}),..., g_m ( \boldsymbol{\beta}) \big)^{\prime}$ be a continuously differentiable function with values in $m-$dimensional Euclidean space and $\displaystyle \frac{  \partial g_i }{ \partial \boldsymbol{\beta}^{\prime} } = \frac{  \partial g_i }{ \partial \beta_j }$, for $i \in \left\{ 1,..., m \right\}$. Then, it holds that (see, \cite{lutkepohl1990asymptotic})
\begin{align}
\sqrt{T} \left( g ( \hat{\boldsymbol{\beta} } ) -  g ( \boldsymbol{\beta}) \right) \overset{d}{\to} \mathcal{N} \left( 0, \frac{  \partial g }{ \partial \boldsymbol{\beta}^{\prime} }  \boldsymbol{\Sigma}_{\beta}  \frac{  \partial g^{\prime} }{ \partial \boldsymbol{\beta} }   \right).     
\end{align}

\begin{remark}
Notice that if the VAR$(p)$ process is $y_t$ is (covariance) stationary with 
\begin{align}
\mathsf{det}  \big( \boldsymbol{I}_K - \boldsymbol{A}_1 z - ... -  \boldsymbol{A}_p z^p  \big) \neq 0, \ \ \ \text{for} \ \ |z| \leq 1,    
\end{align}
and the $u_t$ are independent, identically distributed (i.i.d) with bounded forth moments. This implies that the usual OLS estimators have asymptotic covariance matrix given by 
\begin{align}
\boldsymbol{\Sigma}_{a} = \Gamma^{-1} \otimes \boldsymbol{\Sigma}_u,     
\ \ \
\boldsymbol{\Gamma} 
= 
\mathbb{E} \left\{  \big[ \boldsymbol{y}_t  \ \boldsymbol{y}_{t-1} \ ... \ \boldsymbol{y}_{t-p+1}  \big]^{\prime} \otimes  \big[ \boldsymbol{y}_t  \ \boldsymbol{y}_{t-1} \ ... \ \boldsymbol{y}_{t-p+1}  \big]  \right\}
\end{align}
In addition, if $\boldsymbol{y}_t$ is Gaussian, then $\hat{\boldsymbol{\alpha}}$ and $\hat{\boldsymbol{\sigma}}$ are asymptotically independent. Interval forecasts based on predictive distributions are studied by \cite{chatfield1993calculating} while the reliability of local projection estimators of impulse response functions is examined by \cite{kilian2011reliable}. The asymptotic distributions of IRFs and FEVD for VARs are established by \cite{lutkepohl1990asymptotic} (see, also \cite{lanne2016generalized}).
\end{remark}

\newpage

\begin{remark}
Regarding the construction of Impulse Response Functions and FVED, it can be shown that Impulse Responses that are obtained from unrestricted VARs with roots near unity have long period estimated impulse responses that are inconsistent. Moreover, FEVD are also inconsistent in unrestricted VAR models with near unit roots. The  inconsistency problem of IRFs from nonstationary SVAR models is discussed by \cite{phillips1998impulse} (see, also \cite{lutkepohl1997impulse}). Specifically, the OLS estimation approach of \cite{phillips1998impulse} shows that in nonstationary VAR models with some roots at or near unity the estimated impulse response matrices are inconsistent at long horizons and tend to random matrices rather than the true impulse responses. In this stream of literature another issue is the consistent estimation of the number of cointegrating relations in a system of equations (e.g., see, \cite{barigozzi2021large} and \cite{barigozzi2022testing}).    
\end{remark}

\begin{example}
Consider the VECM process below
\begin{align}
\Delta y_t = \alpha \beta^{\prime} y_{t-1} + \Gamma_1 \Delta y_{t-1} + ... + \Delta y_{t-p+1} + u_t, \ \ t = 1,2,...    
\end{align}
The residuals of $u_t$ are the one-step ahead forecast errors associated with the VECM representation. However, tracing the marginal effects of a change in one component of $u_t$ through the system may not reflect the actual responses of the variables since in practice an isolated change in a single component of $u_t$ is likely to occur if the component is correlated with the other components.
\end{example}

\begin{remark}
Identifying restrictions are imposed in 
order to identify the transitory and permanent shocks in the system. For example, \cite{castelnuovo2010monetary} present impulse response analysis evidence that the VAR model is robust to two different identification strategies based on zero restrictions and the sign restrictions when evaluating the response of prices on monetary policy shocks. Additional to an impulse response analysis as a mechanism for studying the impact of structural shocks, the VAR framework is often used for forecasting based on suitable forecasting schemes (e.g., see, \cite{ng1990recursive}, \cite{bhansali2002multi} and \cite{mariano2002testing}). Two commonly used forecasting approaches: 
\begin{itemize}

\item[(i).] Direct Multiperiod Forecasting (e.g., see \cite{kang2003multi}, \cite{greenaway2013multistep, greenaway2020multistep}).

\item[(ii).] Iterated Multiperiod Forecasting (e.g., see \cite{marcellino2006comparison}).

\end{itemize}    

\end{remark}

\begin{example}[Constructing Iterated Multistep Forecasts]
\begin{align}
y_{t+h} = y_t + \beta^{\prime} X_{t} + \varepsilon_t 
\end{align}
Suppose that the we have an auxiliary equation for the $X_t$ such that $
X_t = A X_{t-1} + u_t$. Then, estimating the vector autoregression for the regressors by OLS we obtain an estimate $\hat{A}$ and thus a forecast of $X_{t+j}$ is $\hat{A}_j X_t$ and a forecast of the quantity $\left( y_{t+h} - y_t \right)$ is given by $\sum_{j=0}^{ h-1 } \hat{A}_j X_t$. Moreover, an example, of constructing one-step ahead forecasted values based on information from the cros-section is presented by \cite{katsouris2021forecast}. An interesting extension would be to use such cross-sectional shrinkage approach within the aforementioned econometric environment. 
\end{example}

\newpage 

\subsubsection{Forecast Error Variance Decomposition}

Generally, the Forecast Error Variance Decomposition gives the proportional contribution of each of the structural shocks to the forecast error variance of each variable at different horizons. Recall that the $h-$period forecast error of $y_t$ is given by 
\begin{align}
y_{t+h} - y_{t+h | t} = \sum_{i=0}^{h-1} \Phi_i u_{t+h-i} = \sum_{i=0}^{h-1} \Phi_i B_0^{-1} w_{t+h-i} = \sum_{i=0}^{h-1} \Theta_i w_{t+h-i}.
\end{align}
Hence, the $h-$period ahead forecast error of the $j-$th component of $y_t$ is given by 
\begin{align*}
\sum_{i=0}^{h-1} \big( \theta_{j1,i} w_{1,t+h-i} + ... + \theta_{jK,i} w_{K,t+h-i} \big)
= 
 \sum_{k=0}^{K} \big( \theta_{jk,0} w_{k,t+h} + ... + \theta_{jk,h-1} w_{k,t+1} \big). 
\end{align*} 
Moreover, the variance of the $h-$period forecast error of the $j-$th component of $y_t$ is given by 
\begin{align*}
\mathbb{E} \left[ \sum_{k=1}^K \bigg( \theta_{jk,0} w_{k,t+h} + ... + \theta_{jk,h-1} w_{k,t+1} \bigg)^2 \right] \equiv 
\sum_{k=1}^K
\bigg( \theta_{jk,0}^2 + ... +  \theta_{jk,h-1}^2 \bigg). 
\end{align*}
which holds since $\Sigma_w = I_K$. Notice that the sum of the terms $\left\{  \theta_{jk,0}^2 + ... +  \theta_{jk,h-1}^2 \right\}$, is thus the $k-$th shock to the forecast error variance of the $h-$period forecast of the $j-$th variable. Then, we have that 
\begin{align}
\mathsf{FEVD}_j^k (j) := \frac{ \theta_{jk,0}^2 + ... +  \theta_{jk,h-1}^2   }{ \displaystyle \sum_{k=1}^K \bigg( \theta_{jk,0}^2 + ... +  \theta_{jk,h-1}^2 \bigg) }.
\end{align}
As an example of application of the above metrics see \cite{kilian2009impact} (and \cite{kilian2009not}) who consider the effects of three structural shocks (oil supply shock, aggregate demand shock, and oil-specific demand shock), in the US stock market. Another useful metric is the construction of historical decompositions which facilitate the quantification of the contribution of a given structural shock to the historically observed fluctuations in the variables included in the SVAR model. In particular, one might be interested to find to what extent the fiscal policy shock explains recessionary conditions (e.g., see \cite{cevik2023s}) during uncertain times (e.g., see \cite{bloom2009impact}). This is useful because we can concentrate on explaining a historical event which we know a priori might have induced a structural change to the macroeconomic system, instead of using for example the average contribution of the fiscal policy shock to the business cycle fluctuations as given by the FEVD (see, \cite{kilian2014quantifying}). 

Suppose that $y_t$ is weakly stationary, and we have observations from 1 to $t$, such that for any $t \in \mathbb{N}$ 
\begin{align*}
y_t = \sum_{s=0}^{t-1} \Theta_s w_{t-s} + \sum_{s=t}^{\infty} \Theta_s w_{t-s}.
\end{align*}
The MA coefficients $\Theta_s$ converge to zero the further to the past, so $\hat{y}_t = \sum_{s=0}^{t-1} \Theta_s w_{t-s}$ approximates $y_t$. 

\newpage 

\section{Identification of Structural Vector Autoregression Models}
\label{Section4}

While theory-based identification through sign restrictions guarantees economic interpretation of the identified shocks (e.g., see \cite{antolin2018narrative} and \cite{mavroeidis2021identification}), data-based identification offers diagnostic tools for testing the otherwise just identifying structural model specifications, while remain silent about the economic interpretation of the resulting structural shocks. 

\subsection{Identification using Short-run and Long-run Restrictions}

\begin{example}[Partially Identified SVARs]
Consider the dynamic structural models of the following form (see, \cite{baumeister2015sign})
\begin{align}
\boldsymbol{A} \boldsymbol{y}_t = \boldsymbol{B} \boldsymbol{x}_{t-1} + \boldsymbol{u}_t,     
\end{align}
where $\boldsymbol{y}_t$ is an $( n \times 1)$ vector of observed variables and $\boldsymbol{u}_t$ is an $( n \times 1)$ vector of structural disturbances assumed to be independent and identically distributed $\mathcal{N} ( 0, \boldsymbol{D} )$ and mutually uncorrelated. Then, the reduced VAR associated with the structural model is given by
\begin{align}
\boldsymbol{y}_t = \boldsymbol{\Phi} \boldsymbol{x}_{t-1} + \boldsymbol{\varepsilon}_t, \ \ \ \text{where} \ \ \ \boldsymbol{\Phi} = \boldsymbol{A}^{-1} \boldsymbol{B},
\ \ \ \ \text{and} \ \ \ \ 
\boldsymbol{\varepsilon}_t = \boldsymbol{A}^{-1} \boldsymbol{u}_t, \ \ \ \mathbb{E} \big[ \boldsymbol{\varepsilon}_t \boldsymbol{\varepsilon}_t^{\prime} \big] = \boldsymbol{\Omega} =  \boldsymbol{A}^{-1} \boldsymbol{D} \left( \boldsymbol{A}^{-1} \right)^{\prime}
\end{align}
where $\boldsymbol{x}_{t - 1}^{\prime} = \big( \boldsymbol{y}_{t-1}^{\prime},  \boldsymbol{y}_{t-2}^{\prime},..., \boldsymbol{y}_{t-m}^{\prime}, \boldsymbol{1} \big)^{\prime}$ is a $( k \times 1 )$ vector with $k = mn + 1$ containing a constant and $m$ lags of $\boldsymbol{y}$. Then, the MLE estimates of the reduced-form parameters are given by 
\begin{align}
\widehat{\boldsymbol{\Phi}}_T
= 
\left( \sum_{t=1}^T \boldsymbol{y}_t  \boldsymbol{x}_{t - 1}^{\prime} \right) \left( \sum_{t=1}^T \boldsymbol{x}_{t - 1}  \boldsymbol{x}_{t - 1}^{\prime} \right)^{-1}  
\ \ \ \ \text{and} \ \ \ \ 
\hat{\boldsymbol{\Omega}}_T 
=
\frac{1}{T} \sum_{t=1}^T \hat{\varepsilon}_t \hat{\varepsilon}_t^{\prime},
\end{align}
with $\hat{\varepsilon}_t = \boldsymbol{y}_{t} - \hat{\boldsymbol{\Phi}}_T \boldsymbol{x}_{t - 1}$. Moreover, 
\cite{baumeister2019structural} examines aspects related to the structural analysis of Vector Autoregressive models with incomplete identification based on the Bayesian approach rather than the classical frequentest approach.  
Generally to uniquely identify the structural parameter in a SVAR model implies of having $K(K-1)/2$ restrictions. In addition when rank conditions hold, then the model is assumed to be exactly identified (e.g., see \cite{rubio2005markov}).  
\end{example}

\begin{example}[Identification using Stability Restrictions]
According to \cite{magnusson2014identification}, often identification depends on the distributional assumptions imposed on $x_t$. Suppose that $x_t$ is a policy variable determined according to an underline stochastic process such that
\begin{align}
x_t = \rho x_{t-1} + (1 - \rho ) \phi y_t + \eta_t    
\end{align}
Furthermore, under a deterministic rational expectations equilibrium, the dynamics of $y_t$ and $x_t$ are as $
y_t = \beta x_{t-1} + u_{yt}$ and $
x_t = \varphi x_{t-1} + u_{xt}$, where $u_{yt}$ and $u_{xt}$ are innovation sequences. Moreover,  \cite{vlaar2004asymptotic} consider the asymptotic distribution of impulse responses in SVARs with long-run restrictions (see, also \cite{mittnik1993asymptotic}).

\end{example}

\newpage 

\paragraph{Recursive Identification} Usually recursive identification with $\boldsymbol{B}_0^{-1}$ being a lower diagonal matrix was particularly popular in the earlier SVAR literature which implies exact identification. In particular, when $B_0^{-1}$ is lower triangular, so it is the matrix $B_0$ and thus the SVAR is recursive.  

\begin{example}
Consider the following trivariate recursive $\mathsf{SVAR}(1)$ model as below
\begin{align*}
\begin{bmatrix}
b_{11,0} & 0 & 0 
\\
b_{12,0} & b_{22,0} & 0
\\
b_{31,0} & b_{32,0} & b_{33,0}
\end{bmatrix}
\begin{bmatrix}
y_{1t}
\\
y_{2t}
\\
y_{3t}
\end{bmatrix}
= 
\begin{bmatrix}
b_{11,1} & b_{11,2} & b_{13,1} 
\\
b_{21,1} & b_{22,1} & b_{23,1}
\\
b_{31,1} & b_{32,1} & b_{33,1}
\end{bmatrix}
\begin{bmatrix}
y_{1,t-1}
\\
y_{2,t-1}
\\
y_{3,t-1}
\end{bmatrix}
+ 
\begin{bmatrix}
\eta_{1t}
\\
\eta_{2t}
\\
\eta_{3t}
\end{bmatrix}
\end{align*}
Therefore, we have that 
\begin{align}
y_t = \sum_{i=0}^{\infty} \Phi_i B_0^{-1} \eta_{t-i} 
\equiv 
\sum_{i=0}^{\infty} \Theta_i \eta_{t-i} = 
\underbrace{ \begin{bmatrix}
\theta_{11,0} & 0 & 0
\\
\theta_{21,0} & \theta_{22,0} & 0
\\
\theta_{31,0} & \theta_{32,0} & \theta_{33,0}
\end{bmatrix} }_{\Theta_0 } 
\underbrace{
\begin{bmatrix}
\eta_{1t}
\\
\eta_{2t}
\\
\eta_{3t}
\end{bmatrix} }_{ \eta_t  } + \Theta_1 \eta_{t-1} + ... 
\end{align}
A relevant study considers the identification of monetary policy shocks using a trivariate SVAR(1) model where $y_t = ( x_t, \Delta p_t, \pi_t )^{\prime}$. However, a drawback of the recursive identification approach is that a different ordering of the variables can yields a different SVAR model with different identified structural shocks (i.e., does not produce uniquely identified structural shocks). On the other hand, it can be easily estimated since it involves estimating the reduced-form model using OLS or MLE, computing the covariance matrix and the computing the lower-triangular Cholesky decomposition. 
\end{example}

\paragraph{Non-Recursive Models}

A non-recursively identified SVAR model resembles a traditional system of simultaneous equation models. The associated identifying restrictions practically generate moment conditions that can be used to estimate the SVAR model. Furthermore, using such moment conditions one can employ suitably imposed rank conditions as a mechanism for identification (e.g., see \cite{blanchard2003empirical}). Although exact identification is not directly a testable hypothesis, we can employ the $J-$test for over identifying restrictions. 
\begin{itemize}

\item Provided that the order condition is satisfied, the rank condition for identification can be checked by assessing sequentially whether each of the equations can be estimated by the instrumental variables regression, using the variables excluded from each equation as instruments. 

\item Once an equation has been verified as being identified, then the corresponding structural shock is identified and thus it can be used as instrument since the shocks are uncorrelated. 

\end{itemize}

The GMM estimator is obtained by
\begin{align}
J = T \bigg( \mathsf{vech} \left( \hat{\Sigma}_u \right) -  \mathsf{vech} \left( B_0^{-1} B_0^{-1 \prime} \right)  \bigg)^{\prime} \hat{W} \bigg( \mathsf{vech} \left( \hat{\Sigma}_u \right) -  \mathsf{vech} \left( B_0^{-1} B_0^{-1 \prime} \right)  \bigg).  
\end{align}

\newpage

\subsubsection{Dynamic Structural Factor Models with Overidentifying Restrictions}

As an example of formal statistical testing  for overidentifying restrictions we present the framework proposed by \cite{han2018estimation} who develops a new estimator for the impulse response functions (IRFs) in structural factor models with a fixed number of over-identifying restrictions (see, also Chapter 12 in  \cite{dhrymes2013mathematics}). The proposed identification scheme nests the conventional just-identified recursive scheme as a special case. The first step before formulating the test statistic is to present the identification of structural shocks in the structural factor models setting (see, also \cite{koistinen2022estimation}).

Consider the following structural factor model for $t = 1,...,T$ such that $\boldsymbol{X}_t = \boldsymbol{\Lambda} \boldsymbol{F}_t + \boldsymbol{e}_t$,
\begin{align}
\boldsymbol{F}_t &= \sum_{j=1}^p \boldsymbol{\Phi}_j \boldsymbol{F}_{t-j} + \boldsymbol{G} \boldsymbol{\eta}_t, \ \ \ 
\boldsymbol{\eta}_t = \boldsymbol{A} \boldsymbol{\zeta}_t,
\end{align}
where $\boldsymbol{X}_t = \big[ X_{1t},..., X_{nt} \big]^{\top}$ is an $n-$dimensional vector, $\boldsymbol{F}_t$ is an $r-$dimensional static factor, $\boldsymbol{\Lambda}$ is an $( n \times r )$ factor loading matrix and $e_t = \big[ e_{1t},...., e_{nt} \big]^{\top}$ is an $n-$dimensional idiosyncratic error term. Moreover, $\boldsymbol{G}$ is an $( r \times q )$ matrix of rank $q$, $\boldsymbol{\eta}_t$ is the $q-$dimensional reduced-form shock and $\boldsymbol{\zeta}_t$ is the $q-$dimensional structural shock. We also have that $\mathbb{E} ( \boldsymbol{\eta}_t \boldsymbol{\eta}_t^{\top} ) = \boldsymbol{I}_q$ and $\boldsymbol{A}$ is a $q \times q$ nonsingular matrix. Setting $q \leq r$, then the proposed framework allows to to obtain dynamic factors. Furthermore, by assuming the stationarity of $\boldsymbol{F}_t$, we have that
\begin{align}
\big( \boldsymbol{I}_r - \boldsymbol{\Phi}_1 \boldsymbol{L} - ... - \boldsymbol{\Phi}_p \boldsymbol{L}^p \big)^{-1} =  \boldsymbol{I}_r + \sum_{j=1}^{\infty} \boldsymbol{\Psi}_j \boldsymbol{L}^j,    
\end{align}
where $\boldsymbol{\Psi}_j$ is the coefficient matrix in the vector moving average representation. Moreover, denote with $\boldsymbol{\mathcal{F}}_t = \big[ \boldsymbol{F}_{t-1}^{\top},...,  \boldsymbol{F}_{t-p}^{\top} \big]^{\top}$ and $\boldsymbol{\Phi} = \big[ \boldsymbol{\Phi}_1,..., \boldsymbol{\Phi}_p \big]$. Rearranging equations gives
\begin{align}
\boldsymbol{X}_t 
= \boldsymbol{\Pi} \boldsymbol{\mathcal{F}}_t + \boldsymbol{\Theta} \boldsymbol{\eta}_t + \boldsymbol{e}_t
\equiv \boldsymbol{\Pi} \boldsymbol{\mathcal{F}}_t + \boldsymbol{\Theta} \boldsymbol{\eta}_t + \boldsymbol{e}_t
\end{align}
where $\boldsymbol{\Pi}  = \boldsymbol{\Lambda} \boldsymbol{\Phi}$, $\boldsymbol{\Theta} = \boldsymbol{\Lambda}  \boldsymbol{G}$ and $\boldsymbol{\Gamma} = \boldsymbol{\Theta} \boldsymbol{A}$. In other words, for the $i-$th cross-section, the factor model can be represented as $\underline{X}_i = F \lambda_i + \underline{e}_i$, where $\underline{X}_i = [ X_{i1},...,  X_{iT} ]^{\top}$ and $F = [ F_1,..., F_T ]^{\top}$ and $\underline{e}_i = [ e_{i1},..., e_{iT} ]^{\top}$. Furthermore, let $\hat{F}_t$ denote the standard PCA estimator for $F_t$, such that $\frac{1}{T} \sum_{t=1}^T \hat{F}_t \hat{F}^{\top}_t = \boldsymbol{I}_r$. 

\begin{remark}
Suppose that the $F_t$ and $e_t$ are uncorrelated, then $
\boldsymbol{\Sigma} = \boldsymbol{\Lambda} \boldsymbol{\Lambda}^{\prime} + \boldsymbol{\Psi}$, where $\boldsymbol{\Sigma} = Var ( \boldsymbol{y}_t )$. By taking the variance of the static factor model we obtain the above expression for the variance, but we also need to impose necessary and sufficient conditions for the identification. Even if $\boldsymbol{\Lambda} \boldsymbol{\Lambda}^{\prime}$ is identified, $\boldsymbol{\Lambda}$ is only identified up to an $( r \times r )$ rotation. Suitable conditions under which $\boldsymbol{\Lambda}$ can be identified imply resolving the rotational indeterminacy (\textit{local identification}). The factor loading $\boldsymbol{\Lambda}$ is locally identified at $\boldsymbol{\Lambda}^{*}$ if in a small open neighborhood of $\boldsymbol{\Lambda}^{*}$, then $\boldsymbol{\Lambda}^{*}$ is the only matrix that satisfies all the identifying constraints   
 (see, \cite{bai2014identification}). According to \cite{han2015tests}, the existing literature has not addressed the overidentification problem in FAVAR models. Unlike the conventional structural VAR analysis where the number of restrictions for identifying structural shocks, the FAVAR models tend to involve a large number of identifying restrictions.
\end{remark}

\newpage

\subsection{Identification via Conditional Heteroscedasticity}

The presence of structural change in a sample can be exploited as an identification mechanism within a structural system context, especially when the heteroscedasticity of structural variances across the two regimes is exploited (see, \cite{hodoshima1988estimation}, \cite{choi2002structural}, \cite{lanne2008identifying} as well as \cite{bacchiocchi2011new}). In other words, the presence of heteroscedasticity in the error terms $\varepsilon_t$ of the system can be employed for identifiability purposes. Although, the conditional heteroscedastiticy approach often requires to use a pre-testing procedure before estimating the model as in \cite{meitz2021testing}, \cite{lutkepohl2021testing}, \cite{bertsche2022identification} and \cite{guay2021identification}; we follow the study of   \cite{bruggemann2016inference} who proposed a framework for the identification of SVARs with conditional heteroscedasticity of unknown form.

Let $( u_t, t \in \mathbb{Z} )$ be a $K-$dimensional white noise sequence defined on a probability space $( \Omega, \mathcal{F}, \mathbb{P} )$, such that each $u_t = \left( u_{1t},..., u_{2t}    \right)^{\prime}$ is assumed to be measurable with respect to $\mathcal{F}_t$, where $( \mathcal{F}_t )$ is a sequence of increasing $\sigma-$fields of $\mathcal{F}$. Suppose that we observe a data sample $( y_{-p+1},..., y_0, y_1,..., y_T )$ of sample size $T$ plus $p$ pre-sample values from the following DGP for the $K-$dimensional time series $y_t = \left( y_{1t},..., y_{Kt} \right)^{\prime}$ such that
\begin{align}
y_t = \mu + A_1 y_{t-1} + ... + A_p y_{t-p} + u_t, \ \ \ t \in \mathbb{Z},    
\end{align}
where $A(L) = I_k - A_1 L - A_2 L^2 - ... - A_p L^p, \ \ \ A_p \neq 0$. 

Consider a $K-$dimensional time series such that $y_t = \left( y_{1t},..., y_{Kt} \right)$ where
\begin{align}
y_t = \mu + A_1 y_{t-1} + ... + A_p y_{t-p} + u_t, \ \ \ t \in \mathbb{Z}    
\end{align}
or $A(L) y_t = \mu + u_t$, in compact representation. Denote with $\boldsymbol{y} = \mathsf{vec} \left( y_1,..., y_T \right)$ to be $(KT \times 1)$ vector. Moreover, the parameter $\boldsymbol{\beta}$ is estimated by $
\widehat{\boldsymbol{\beta}} = \mathsf{vec} \left( \widehat{A}_1,..., \widehat{A}_p \right)$, via the multivariate OLS estimator 
\begin{align}
\widehat{\boldsymbol{\beta}} = \left( \left( \boldsymbol{Z} \boldsymbol{Z}^{\prime} \right)^{-1}  \boldsymbol{Z} \otimes \boldsymbol{I}_K \right) \boldsymbol{y}.    
\end{align}
Assume that the process $\boldsymbol{y}_t$ is stable, then it has a vector moving-average representation (VMA) s.t.
\begin{align}
y_t = \sum_{j=0}^{\infty} \boldsymbol{\Phi}_j u_{t-j}, \ \ \ t \in \mathbb{Z},    
\end{align}
where $\Phi_j, j \in \mathbb{N}$, is a sequence of (exponentially fast decaying) $( K \times K)$ coefficient matrices with 
\begin{align}
\boldsymbol{\Phi}_0 = \boldsymbol{I}_K \ \ \ \text{and} \ \ \ \boldsymbol{\Phi}_i = \sum_{j=1}^i \Phi_{i-j} A_j, \ i = 1,2,...    
\end{align}
Notice that the standard estimator of $\boldsymbol{\Sigma}_u$ is given by $
\boldsymbol{\Sigma}_u = \frac{1}{T} \sum_{t=1}^T \widehat{u}_t  \widehat{u}_t^{\prime}$, where $\hat{u}_t = y_t - \widehat{A}_1 y_{t-1} - ... -  \widehat{A}_p y_{t-p}$ are the residuals obtained from the estimated VAR(p) model.

\newpage

Moreover, we set $\boldsymbol{\sigma} = \mathsf{vech} ( \boldsymbol{\Sigma}_u )$ and $\widehat{\boldsymbol{\sigma}} = \mathsf{vech} ( \widehat{\boldsymbol{\Sigma}}_u )$. The $\mathsf{vech}$-operator is defined to stack column wise the elements on and below the main diagonal of the matrix. A particular useful way to consider the development of the asymptotic theory is to obtain the joint limiting distribution using an unconditional central limit theorem as below
\begin{align}
\sqrt{T} 
\begin{pmatrix}
\widehat{\boldsymbol{\beta}} - \boldsymbol{\beta}
\\
\widehat{\boldsymbol{\sigma}}^2 - \boldsymbol{\sigma}^2 
\end{pmatrix}
\overset{d}{\to} \mathcal{N} \left( 0, \boldsymbol{V}  \right).
\end{align}
such that the covariance matrix is partitioned as below
\begin{align}
\boldsymbol{V} = 
\begin{pmatrix}
\boldsymbol{V}^{(1,1)} & \boldsymbol{V}^{(2,1)\prime}
\\
\boldsymbol{V}^{(2,1)} & \boldsymbol{V}^{(2,2)}
\end{pmatrix}
\end{align}
with the following analytical forms 
\begin{align}
\boldsymbol{V}^{(1,1)} 
&= 
\left(  \boldsymbol{\Gamma}^{-1} \otimes \boldsymbol{I}_K \right) \left( \sum_{i,j = 1}^{\infty} ( \boldsymbol{C}_i \otimes \boldsymbol{I}_K ) \sum_{h = - \infty}^{\infty} \tau_{i,h,h+j} ( \boldsymbol{C}_i \otimes \boldsymbol{I}_K )^{\prime} \right)  \left(  \boldsymbol{\Gamma}^{-1} \otimes \boldsymbol{I}_K \right)^{\prime}   
\\
\boldsymbol{V}^{(2,1)} 
&= 
L_k \left( \sum_{j=1}^{\infty} \sum_{h= - \infty}^{\infty} \tau_{0,h,h+j} ( \boldsymbol{C}_i \otimes \boldsymbol{I}_K )^{\prime}  \right) \left(  \boldsymbol{\Gamma}^{-1} \otimes \boldsymbol{I}_K \right)^{\prime}
\\
\boldsymbol{V}^{(2,1)} 
&= 
L_K \left(  \sum_{h= - \infty}^{\infty} \big[ \tau_{0,h,h} - \mathsf{vech} ( \boldsymbol{\Sigma}_u ) \mathsf{vech} ( \boldsymbol{\Sigma}_u )^{\prime} \big] \right) L_K^{\prime} 
\end{align}

\begin{remark}
Notice that we can also write the following expression 
\begin{align}
V^{(2,2)} = \mathsf{Var} ( \boldsymbol{u}_t^2 ) + \sum_{h = - \infty}^{\infty} \mathsf{Cov} \left( \boldsymbol{u}_t^2, \boldsymbol{u}_{t-h}^2 \right)    
\end{align}
such that $\boldsymbol{u}_t^2 = \mathsf{vech} ( u_t u_t^{\prime} )$. Hence, $V^{(2,2)}$ has a long-run variance representation in terms of $\boldsymbol{u}_t^2$ that captures the (linear) dependence structure in the underline stochastic sequence. In addition if the errors are \textit{i.i.d} then we have that 
$V^{(2,2)} =  \mathsf{Var} ( \boldsymbol{u}_t^2 ) = L_K \tau_{0,0,0} L_K^{\prime} - \boldsymbol{\sigma} \boldsymbol{\sigma}^{\prime}$.
\end{remark}
Next, we focus on the residual-based moving block bootstrap resampling method. Various studies in the literature have demonstrated that block bootstrap methods are suitable for capturing dependencies in time series data. Specifically, we are interested in applying the moving block bootstrap technique for the residuals obtained from a fitter VAR$(p)$ model to approximate the limiting distribution of 
\begin{align}
\sqrt{T} \left( \left( \widehat{\boldsymbol{\beta}} - \boldsymbol{\beta}  \right)^{\prime},  \left( \widehat{\boldsymbol{\sigma}} - \boldsymbol{\sigma} \right)^{\prime} \right)^{\prime}.
\end{align}
\begin{itemize}
    
\item[\textbf{Step 1.}] Fit a VAR$(p)$ model to the data to get $\widehat{A}_1,...., \widehat{A}_p$, and compute the residuals $\widehat{u}_t = y_t - \widehat{A}_1 y_{t-1} - \widehat{A}_p y_{t-p}$, for $t = 1,..., T$.

\newpage

\item[\textbf{Step 2.}] Choose a block length $\ell < T$ and let $T = \floor{ T / \ell }$ be the number of blocks needed such that $\ell N \geq T$. Moreover, define $( K \times \ell )-$dimensional blocks such that 
\begin{align}
B_{i, \ell} = \left( \widehat{u}_{i+1},...,  \widehat{u}_{i+\ell} \right), \ \ \ i \in \left\{ 0,..., T - \ell \right\}    
\end{align}
such that $i_0,..., i_{N-1}$ be \textit{i.i.d} random variables uniformly distributed on the set $\left\{ 0,1,..., T - \ell \right\}$. Moreover, we lay blocks $B_{i_0, \ell},..., B_{i_{N-1}, \ell}$ end-to-end together and discard the last $N \ell - T$ values to get bootstrap residuals $\widehat{u}_t^{*},..., \widehat{u}_T^{*}$.

\item[\textbf{Step 3.}] Centering this sequence of residuals $\widehat{u}_t^{*},..., \widehat{u}_T^{*}$ based on the following rule
\begin{align*}
u_{j \ell + s}  
=  \widehat{u}_{j \ell + s} - \mathbb{E}^{*} \left[ \widehat{u}_{j \ell + s} \right]  
= 
\widehat{u}^{*}_{j \ell + s} - \frac{1}{T - \ell + 1} \sum_{r=0}^{T-\ell} \widehat{u}_{s+r}
\end{align*}
for $s \in \left\{ 1,..., \ell \right\}$ and $j \in \left\{ 0,1,2,..., N - 1 \right\}$ and unconditional mean, $\mathbb{E}^{*} ( u_t^{*} ) = 0$, for all $t = 1,..., T$.

\item[\textbf{Step 4.}] Set bootstrap pre-sample values $y_{-p+1}^{*},..., y_{0}^{*}$ equal to zero and generate the bootstrap sample $y_1^{*},..., y_T^{*}$ according to the following 
\begin{align}
y_t^{*} = \widehat{A}_1 y^{*}_{t-1} + ... + \widehat{A}_p y^{*}_{t-p} + u_t^{*}.   
\end{align}

\item[\textbf{Step 5.}] Compute the bootstrap estimator

\begin{align}
\widehat{\boldsymbol{\beta}}^{*} = \mathsf{vec} \left( \widehat{A}_1 ,..., \widehat{A}_p \right) 
= 
\left( \left( \boldsymbol{Z}^{*} \boldsymbol{Z}^{* \prime} \right)^{-1} \boldsymbol{Z}^{*} \otimes \boldsymbol{I}_K \right) \boldsymbol{y}^{*}  
\end{align}
Moreover, we define the bootstrap analogue of $\boldsymbol{\Sigma}_u$ such that 
\begin{align}
\boldsymbol{\Sigma}_u^{*} = \frac{1}{T} \sum_{t=1}^T \widehat{u}_t^{*} \widehat{u}_t^{* \prime}    
\end{align}
where $\widehat{u}_t^{*} = y_t^{*} - \widehat{A}_1^{*} y_{t-1}^{*} - ... - \widehat{A}_p^{*} y_{t-p}^{*}$ are the bootstrap residuals obtained from the VAR$(p)$ fit. We set $\widehat{\boldsymbol{\sigma}}^{*} = \mathsf{vech} ( \widehat{\boldsymbol{\Sigma}}_u^{*} )$.
    
\end{itemize}

\begin{theorem}[Residual-based MBB Consistency, \cite{bruggemann2016inference}] Under regularity conditions and if $\ell^3 / T \to \infty$ as $T \to \infty$, it holds that
\begin{align}
\underset{ x \in \mathbb{R}^{ \tilde{K} } }{ \mathsf{sup} } \left| \mathbb{P}^{*} \bigg( \sqrt{T} \left( (  \widehat{\boldsymbol{\beta}}^{*} - \widehat{\boldsymbol{\beta}} )^{\prime},  (  \widehat{\boldsymbol{\sigma}}^{*} - \widehat{\boldsymbol{\sigma}} )^{\prime} \right)^{\prime} \leq x \bigg)  
-  
\mathbb{P} \bigg( \sqrt{T} \left( (  \widehat{\boldsymbol{\beta}} - \widehat{\boldsymbol{\beta}} )^{\prime},  (  \widehat{\boldsymbol{\sigma}} - \widehat{\boldsymbol{\sigma}} )^{\prime} \right)^{\prime} \leq x \bigg) \right| \to 0 
\end{align}
in \textit{probability}, where $\mathbb{P}^{*}$ denotes the probability measure induced by the residual-based MBB.     
\end{theorem}

\newpage

\subsection{Identification via Non-Gaussianity of Structural Shocks}

\begin{example}
Consider a standard $d-$dimensional VAR$(p)$ process given by (see, \cite{petrova2022asymptotically})
\begin{align}
\boldsymbol{Y}_t = \mu_0  + \sum_{i=1}^p B_{0,1} \boldsymbol{Y}_{t-1} + \boldsymbol{\varepsilon}_{t} \equiv \mu_0 + \boldsymbol{B}_0 \odot \boldsymbol{\mathcal{Y}}_{t-1}  + \boldsymbol{\varepsilon}_{t}, \ \ \    \boldsymbol{\varepsilon}_{t} \sim ( 0, \Omega_0 ).  
\end{align}
where $\boldsymbol{B}_0 = \big[ B_{0,1},..., B_{0,p} \big]$ and $\boldsymbol{\mathcal{Y}}_{t-1} = \big[ \boldsymbol{Y}_{t-1}^{\top},...,  \boldsymbol{Y}_{t-p}^{\top} \big]^{\top}$ is an $dp \times 1$ vector containing the lags of the vector $\boldsymbol{Y}_t$ and $\odot$ denotes the inner product between two vectors of the same dimension. Denote with $\mathcal{F}_t = \sigma \big( \boldsymbol{\varepsilon}_{t},..., \boldsymbol{\varepsilon}_{1} \big)$ the natural filtration of the innovation sequence and with $\mathbb{E}_{ \mathcal{F}_t }$ the conditional expectation operator.  Recall that covariance stationarity of $\boldsymbol{Y}_t$ yields a vector $MA(\infty)$ representation of the form $\boldsymbol{Y}_t = \sum_{j=- \infty}^{+\infty} \Phi_j \boldsymbol{\varepsilon}_{t-j}$.

\begin{assumption}[see, \cite{petrova2022asymptotically}]
Suppose that the following conditions hold: 
\begin{itemize}

\item[\textit{(i).}]  The VAR process is stable, so that all roots of the polynomial 
\begin{align}
\psi(z) = \mathsf{det} \left( \boldsymbol{I}_d - \sum_{j=1}^p z^j \boldsymbol{B}_{0,i}  \right), \ \textit{lie outside the unit circle}.  
\end{align}
 
\item[\textit{(ii).}]  The error process $( \boldsymbol{\varepsilon}_{t}, \mathcal{F}_t )_{ t \geq 1}$ has the following properties:
\begin{itemize}
    \item[(a)] is a martingale difference sequence satisfying $\mathbb{E}_{ \mathcal{F}_{t-1} } [ \boldsymbol{\varepsilon}_{t} \boldsymbol{\varepsilon}_{t}^{\top} ] = \boldsymbol{\Omega}_0$ for all $t$.

    \item[(b)] has time-invariant third and fourth conditional moments such that 
    \begin{align}
        \mathbb{E}_{ \mathcal{F}_{t-1} } \left[ \boldsymbol{\varepsilon}_{t} \odot \mathsf{vech} \left( \boldsymbol{\varepsilon}_{t} \boldsymbol{\varepsilon}_{t}^{\top} \right) \right] = \boldsymbol{\mathcal{S}} \ \ \ \mathbb{E}_{ \mathcal{F}_{t-1} } \left[ \mathsf{vech} \left( \boldsymbol{\varepsilon}_{t} \boldsymbol{\varepsilon}_{t}^{\top} \right) \odot \mathsf{vech} \left( \boldsymbol{\varepsilon}_{t} \boldsymbol{\varepsilon}_{t}^{\top} \right)  \right] = \boldsymbol{\mathcal{K}}, \ \forall \ t.
    \end{align}
 
\end{itemize}
    
\end{itemize}
    
\end{assumption}
Based on the conditions above (stability condition of the system), a suitable estimation method is the quasi-maximum likelihood method which is robust to distributional misspecifications. 
The QML estimators of the model parameters $\boldsymbol{B}_{0,\mu} = [ \mu, \boldsymbol{B}_0 ]$ are given as below:
\begin{align}
\widehat{\boldsymbol{B}}_{0,\mu} = \left( \sum_{t=1}^T \boldsymbol{X}_{t-1} \boldsymbol{X}_{t-1}^{\top} \right)^{-1}  \left( \sum_{t=1}^T \boldsymbol{Y}_{t-1} \boldsymbol{X}_{t-1}^{\top} \right), \ \ \ \hat{\Omega}_T = \frac{1}{T} \sum_{t=1}^T \boldsymbol{\varepsilon}_{t} \boldsymbol{\varepsilon}_{t}^{\top},
\end{align}
where $\boldsymbol{\varepsilon}_{t} = \boldsymbol{Y}_t - \widehat{\boldsymbol{B}}_{0,\mu} \boldsymbol{X}_{t-1}$. Estimating the VAR model without an intercept after demeaning leads to the same QML estimators for $\widehat{\boldsymbol{B}}_{0,\mu}$ and  $\boldsymbol{\varepsilon}_{t}$, which is a simple consequence of the Frisch-Waugh-Lovell theorem. Moreover, the asymptotic theory analysis can be developed using the 
Hessian matrix, denoted with $\mathcal{H}$, such that 
\begin{align}
\mathcal{H} ( \boldsymbol{Y}_t; \theta ) = \frac{ \partial^2 \ell ( \boldsymbol{Y}_t; \theta ) }{ \partial \theta \partial \theta^{\prime}}.    
\end{align}
\end{example}

\newpage

\begin{remark}
The result of distributional misspecification\footnote{ \cite{petrova2022asymptotically} mentions that inaccurately imposing Gaussian distributional assumptions in standard multivariate time series models does not affect inference on the autoregressive coefficients but distorts both classical and Bayesian inference on the volatility matrix whenever the true error distribution has excess kurtosis relative to the multivariate normal density. } is that Bayesian methods leads to asymptotically invalid posterior inference for the intercept and the volatility matrix and, consequently, invalid posterior credible sets for quantities such as impulse responses, variance decompositions and density forecasts. Therefore, a Bayesian procedure which delivers asymptotically correct posterior credible sets regardless of distributional assumptions is desirable for constructing the posterior distribution of quantities such as impulse response functions.  (see, \cite{petrova2022asymptotically}). 
\end{remark}
Generally, exploiting the presence of non-Gaussianity especially in a multivariate system which is more likely to include non-linearities and complex dependencies of unknown forms between variables is an identification method which uniquely determines the structural parameter. However, identification up to signed-permutations reduces the identified set and thus to achieve point identification it is necessary to have a mechanism for selecting a specific permutation as in \cite{lanne2017identification} (see, also \cite{hallin2015r}).  Recently, \cite{chan2023large} propose a fast algorithm for identifying the exact permutation scheme based on multiple sign and ranking restrictions. 

\begin{example}[see, \cite{lanne2017identification}]
Consider the structural VAR (SVAR) model as below 
\begin{align}
\boldsymbol{y}_t = \boldsymbol{\mu} + \boldsymbol{A}_1 y_{t-1} + ... + \boldsymbol{A}_p \boldsymbol{y}_{t-p} + \boldsymbol{B} \boldsymbol{\varepsilon}_t,    
\end{align}
where $\boldsymbol{y}_t$ is the $d-$dimensional time series of interest, $\boldsymbol{\mu}$, $( d \times 1 )$ is an intercept term such that $\left\{ \boldsymbol{A}_1,..., \boldsymbol{A}_p \right\}$ and $\boldsymbol{B}$ are the $(d \times d)$ parameter matrices with $\boldsymbol{B}$ non-singular, and $\boldsymbol{\varepsilon}_t$ is an $( d \times 1)$ temporally uncorrelated strictly stationary error term with zero mean and finite positive definite covariance matrix. Furthermore, since we consider stationary (or stable) time series, we assume 
\begin{align}
\mathsf{det} ( \boldsymbol{A} ) \overset{ \mathsf{def} }{=} \big( \boldsymbol{I}_n - \boldsymbol{A}_1 z - ... - \boldsymbol{A}_p z^p \big) \neq 0,  \ \ |z| \leq 1  
\end{align}
If we left multiply by the inverse of $\boldsymbol{B}$ we obtain an alternative formulation of the SVAR model such that $
\boldsymbol{A}_0 \boldsymbol{y}_t =  \boldsymbol{\mu}^{\star} + \boldsymbol{A}_1^{\star} \boldsymbol{y}_{t-1} + ... + \boldsymbol{A}_p^{\star} \boldsymbol{y}_{t-p} + \boldsymbol{\varepsilon}_t$, where we have that $\boldsymbol{A}_0  = \boldsymbol{B}^{-1}$, $ \boldsymbol{\mu}^{\star} = \boldsymbol{B}^{-1} \boldsymbol{\mu}$, $\boldsymbol{A}_j^{\star}  = \boldsymbol{B}^{-1} \boldsymbol{A}_j, \ j \in \left\{ 1,..., p \right\}$. Consider the moving average representation of the model as below: 
\begin{align}
\boldsymbol{y}_t = \boldsymbol{\nu} + \sum_{j=0}^{\infty} \boldsymbol{\Psi}_j \boldsymbol{B} \boldsymbol{\varepsilon}_{t-j}, \ \ \boldsymbol{\Psi}_0 = \boldsymbol{I}_n,     
\end{align}
where $\boldsymbol{\nu} = \boldsymbol{A}(1)^{-1} \boldsymbol{\mu}$ is the expectation of $\boldsymbol{y}_t$ and the matrices $\boldsymbol{\Psi}_j$, for $j \in \left\{ 0,1,..., \right\}$ are determined by the power series $\boldsymbol{\Psi} (z) = \boldsymbol{A}(1)^{-1} \equiv \sum_{j=0}^{ \infty } \Psi_j z^j$. An appropriate identification is needed to make the two factors in the product $\boldsymbol{B} \boldsymbol{\varepsilon}_t$, and hence the impulse responses $\boldsymbol{\Psi}_j \boldsymbol{B}$, unique.  The uniqueness property of the identification scheme is determined only up to a linear transformation of the parameter matrices. The reduced form residuals are assumed to be a linear combination of $d$ independent unobserved variables, that is, the structural shocks. Thus, the identification implies the consistent estimation of structural shocks given the reduced-form error terms (a nonprametric approach is proposed by  \cite{braun2023importance}). 
\end{example}

\newpage

Assume that the matrix $\boldsymbol{B} \equiv \boldsymbol{S} \boldsymbol{C}$ where $\boldsymbol{C} \in \mathbb{R}^{d \times d}$ is orthogonal and $\boldsymbol{S} \in \mathbb{R}^{d \times d}$ is such that $\boldsymbol{\Sigma}_u = \boldsymbol{S} \boldsymbol{S}^{\prime}$. Thus, given a large number of time series observations $\left\{ \boldsymbol{u}_t \right\}_{t=1,...,n}$, then the statistical problem is to consistently estimate the matrix product $\boldsymbol{S} \boldsymbol{C}$ and determine the distribution of the structural shocks $\boldsymbol{\varepsilon}_t$. Since, the matrix $\boldsymbol{S}$ can be directly  estimated by matching from the covariance matrix $\boldsymbol{\Sigma}_u$, using the Cholesky decomposition, then the identification problem of the structural parameter matrix, boils down to the consistent estimation of the orthogonal matrix $\boldsymbol{C} \in \mathbb{R}^{d \times d}$.  When structural shocks $\boldsymbol{\epsilon}_t$ are assumed to be Gaussian, one can only identify the space of the mixed signals without separating the individual components. The only way to ensure consistent identification is by assuming that structural shocks are non-Gaussian and then proceeding by consistently estimating the orthogonal matrix $\boldsymbol{C}$. To consistently estimate the individual components of the structural shocks we assume joint non-Gaussianity and that the orthogonal matrix $\boldsymbol{C}$ is consistently estimated. The main intuition of exploiting the non-Gaussianity of structural shocks as an identification mechanism as in \cite{lanne2017identification} can be seen via the concept of \textit{time-reversibility} (see, \cite{findley1986uniqueness}, \cite{tong2005time}, \cite{chan2006note} and \cite{lawrance1991directionality}\footnote{Professor Tony Lawrance is an Emeritus Professor at the Department of Statistics, University of Warwick, UK.}) in time series analysis. The notion of \textit{time-series reversibility} was proposed by economists and financial engineers to address the asymmetry property of business cycles. The presence of asymmetric signals can induce enough skewness to reject the null hypothesis of time-reversibility. Business cycle asymmetry can be tested with respect to whether macroeconomic fluctuations are time irreversible and can also be exploited for identification purposes as in \cite{virolainen2020structural} (see, \cite{nyberg2023risk}). 
\begin{theorem}[Uniqueness result]
Suppose that $x_t$ is a non-Gaussian, strictly stationary, zero mean time series whose $r-$th moment $\mathbb{E}( x_t^r )$ are finite for all $r = 1,2,...$. Suppose also that the spectral density of $x_t$ is positive almost everywhere. Then, $x_t$ can have at most one mean square convergent representation $x_t = \sum_{j = - \infty}^{ \infty } c_j \eta_{t-j}$, in which $\eta_t$ is a sequence of \textit{i.i.d} distributed zero mean random variables with finite moments, ignoring changes of scale and shifts in the time origin of the series $\eta_t$ (see, \cite{chan2006note}). 
\end{theorem}

\begin{example}[\cite{chan2006note}]
Consider the univariate non-Gaussian linear process representation 
\begin{align}
\label{repA}
X_t = \sum_{j=-\infty}^{\infty} A^{\prime}_j \varepsilon^{\prime}_{t-j}, \ \ \ t \in \mathbb{N}
\end{align}
There exists an integer $m$ and a matrix $B$ such that, for all $t$, (i) $A_t \equiv A_{m - t} B$ and (ii) $B \varepsilon_t$ is distributionally equivalent to $\varepsilon_t$. This uniqueness of a non-Gaussian process implies that, if $X$ admits representation \eqref{repA} and $Y_t = \sum_{j=-\infty}^{\infty} A^{\prime}_j \varepsilon^{\prime}_{t-j}$, then $Y$ is distributionally equivalent to $X$ \textit{if and only if} the previous conditions hold, with the equality concerning the innovations interpreted as equivalence in distribution. 
\end{example}

\paragraph{Proof of Proposition 1 in \cite{lanne2017identification}}

Suppose that the following two conditions hold: 
\begin{itemize}

\item[\textit{(i).}] The error process $\boldsymbol{\varepsilon}_t = \big( \varepsilon_{1,t},..., \varepsilon_{d,t} \big)$ is a sequence of \textit{i.i.d} random vectors with each component $\varepsilon_{i,t}$, for $i = \left\{ 1,..., n \right\}$, having zero mean and finite positive variance $\sigma_i^2$.  

\newpage

\item[\textit{(ii).}] The components of $\boldsymbol{\varepsilon}_t = \big( \varepsilon_{1,t},..., \varepsilon_{d,t} \big)$ are (mutually) independent and at most one of them has a Gaussian marginal distribution. 
    
\end{itemize}
A key step is to demonstrate the existence and uniqueness of an equivalent SVAR representation. In particular, one needs to show that the structural shocks of the original model can be determined as a function of the structural shocks of the observetionally equivalent SVAR process multiplied by an unknown permutation matrix $\boldsymbol{M}$ such that $
\boldsymbol{\varepsilon}_t = \boldsymbol{M} \boldsymbol{\varepsilon}_t^{*}, \ \ \ \boldsymbol{M} = \boldsymbol{B}^{-1} \boldsymbol{B}^{*}$.  To ensure a valid identification scheme, we need to examine the properties of the elements of this permutation matrix. 
Thus, we assume that at most one column of $\boldsymbol{M}$ may contain more than one nonzero element (see, Lemma A.1 in \cite{lanne2017identification} and \cite{kagan1973characterization}\footnote{For the interested reader, the property is due to the theory of \textit{locally compact abelian groups} which implies that any finitely generated group can be expressed as a direct product of cyclic groups.  As a result, the \textit{duality} result allows to construct a distributional equivalent topological space using properties from non-Gaussian topological spaces. Specifically, any compact Hausdorff abelian group can be completely described by the purely algebraic properties of its dual group. Corollary 3 in \cite{morris1979duality} defines the observetionally equivalent class of topologically isomorphic groups. These theorems remain agnostic regarding the exact properties model estimators are assumed to have ex-ante to the imposed specification, which implies that using a non-Gaussian distributional assumption is sufficient.} Suppose say that the $k-$th column of $\boldsymbol{M}$ has at least two nonzero elements, $m_{i k}$ and $m_{jk}$ such that $i \neq j$. Then, it holds that  
\begin{align}
\varepsilon_{i,t} = m_{ik} \varepsilon_{k,t}^{*} + \sum_{  \ell = 1,....n; \ell \neq k } m_{ i \ell }  \varepsilon_{\ell,t}^{*} 
\\
\varepsilon_{j,t} = m_{jk} \varepsilon_{k,t}^{*} + \sum_{  \ell = 1,....n; \ell \neq k } m_{j \ell }  \varepsilon_{\ell,t}^{*} 
\end{align}
where the random variable $\varepsilon_{k,t}^{*}$ is assumed to be Gaussian with positive variance. Moreover, for all $\ell \in \left\{ 1,..., n; \ell \neq k \right\}$, it must hold that $m_{ik} m_{j \ell} = 0$ because only the $k-$th element of $\boldsymbol{M}$ could have more than one nonzero element. Notice that the purpose of the identification scheme is to distinguish or identify the structural shocks. Due to the independence assumption of the error terms and by the fact that each row of $M$ has exactly one non-zero element, then we can conclude that there exist a permutation matrix $\boldsymbol{P}$ and a diagonal matrix $\boldsymbol{D} = \mathsf{diag} \left\{ d_1,..., d_n \right\}$ with non-zero diagonal elements such that $\boldsymbol{M} = \boldsymbol{D} \boldsymbol{P}$. Rearranging terms implies that $
\boldsymbol{\varepsilon}_t = \boldsymbol{B}^{-1} \boldsymbol{B}^{*} \boldsymbol{\varepsilon}_t^{*}$ and to have observational equivalence it should hold that 
\begin{align}
\boldsymbol{\varepsilon}_t \overset{\triangle }{\equiv} \boldsymbol{B}^{-1} \underbrace{ \boldsymbol{B} \boldsymbol{D} \boldsymbol{P} }_{ \boldsymbol{B}^{*} } \boldsymbol{\varepsilon}_t^{*},  \ \ \text{where} \ \  \boldsymbol{\varepsilon}_t^{*} = \boldsymbol{P}^{\prime} \boldsymbol{D}^{-1} \boldsymbol{\varepsilon}_t  
\end{align}
since $\boldsymbol{D}$ is a diagonal matrix then $\boldsymbol{D}^{-1}$ is also a diagonal matrix and $\boldsymbol{P} \boldsymbol{P}^{\prime} = \boldsymbol{I}$ where $\boldsymbol{P}$ represents a permutation matrix within the same equivalence class. 

Proposition 1 in \cite{lanne2017identification} characterizes a class of observationally equivalent SVAR processes that differ only with respect to the ordering and scaling of the structural shocks in the vector $\varepsilon_t$. A (rescaled) error vector of these set of permuted structural shocks, consists of exactly the same elements, whose ordering varies and thus these SVAR processes induces asymptotically equivalently impulse responses functions. Thus, the invariance transformation property allows to identify structural shocks using an observationally equivalent class of SVAR processes based on mutual independence.

\newpage

\subsubsection{Testing for Independence of Structural Shocks}

The independence property of structural shocks $\varepsilon_t$ can serve to uniquely identify the structural parameter $B$ it at most one component of $\varepsilon_t$ is Gaussian. Moreover, in non-Gaussian frameworks, one has to consider the entire dependence structure to implement a robust structural analysis. Moreover, although the original identification scheme of \cite{lanne2017identification} is based on the assumption of at least one Gaussian shock for full identification, the particular result is extended by \cite{maxand2020identification} who show that the same identification scheme is valid in the case of multiple Gaussian structural shocks while the remaining structural shocks are non-Gaussian. In addition some further studies as in \cite{crucil2023monetary} and \cite{bruns2023testing} consider the case of non-Gaussian proxy SVARs. In particular,  \cite{crucil2023monetary} extending the approach of \cite{schlaak2023monetary} employs an instrumental regression to estimate the structural shock of interest while staying agnostic for the remaining structural shocks. 

Testing for the independence of the identified structural shocks obtained from orderings (e.g., a set of permutations of all shocks) of the VAR residuals is an important validation step for the underline assumptions. When a suitable test statistic of independence for the identified shocks doesn't find evidence against the hypothesis of independence then this is an indication that using a non-Gaussian MLE approach should render valid model estimates. Motivated from the presence of such non-Gaussian and heavy-tailed error terms in SVAR models \cite{davis2023time} study the role of uncertainty in economic fluctuations. Specifically, disaster shocks are shown to have short term economic impact arising mostly from feedback dynamics. In particular, the authors show that the coefficient estimate on the infinite variance regressor would be consistent but not asymptotically normal. As a result, the coefficients of impulse response functions whether there are computed from the VAR or by local projections are likely not to be asymptotically normally distributed. A possible drawback of the particular identification scheme is that exploiting the non-Gaussianity property precludes cases of non-identification. Recently, there is a growing interest in testing not only the distributional assumptions of structural shocks but also indirectly testing the absence of non-identification. In particular,  \cite{drautzburg2023refining} mentiones that criteria such as an inverted GMM test statistic based on co-dependence of higher moments, can be employed to rule out incorrect identification schemes when the data are non-Gaussian. When the data are Gaussian, there is no co-dependence of higher moments and the estimator becomes close to the plug-in estimator in \cite{moon2012bayesian} (see, also \cite{granziera2018inference}).  Thus, the approach of \cite{drautzburg2023refining} exploits independence for identification, but does so allowing for weak identification, which indeed remains valid in the Gaussian limit where there is complete lack of identification. Further issues related to identification via non-Gaussianity is the modelling of climate-economic systems which is relevant to the structural analysis of climate shocks to macroeconomic variables. 

\begin{remark}
Under the assumption of non-Gaussianity, the fact that $\boldsymbol{\varepsilon}_t$ are white noise does not imply that they are serially independent. This is an important aspect of a non-Gaussian VAR model because even if $Corr ( \varepsilon_{it}, \varepsilon_{it+1} ) = 0$ for $i = 1,..., l$, it is still in principle possible that $\varepsilon_{it+1}$ is statistically dependent on $\varepsilon_{it}$. Thus, in this latter case $\varepsilon_{it+1}$ would not be an actual innovation term, since it can be predicted by $y_t$ (see, \cite{moneta2013causal}).      
\end{remark}

\newpage

\subsection{Identification via External Instruments}

\begin{example}
Following the framework proposed by \cite{mittnik1993asymptotic}, let $\boldsymbol{\theta}$ be the parameter vector that collects the coefficients of the original SVAR representation.

\begin{itemize}

\item  Suppose we have the model $
\boldsymbol{A}(L) \boldsymbol{y}_t = \boldsymbol{B}(L)  \boldsymbol{\varepsilon}_t$, with $\boldsymbol{\alpha} := \mathsf{vec} \left( \boldsymbol{A}_0,..., \boldsymbol{A}_p \right)$, $\boldsymbol{\beta} := \mathsf{vec} \left( \boldsymbol{B}_0,..., \boldsymbol{B}_p \right)$ and $\boldsymbol{r} := \mathsf{vec} ( \boldsymbol{R} )$, then $\boldsymbol{\theta} = \left( \boldsymbol{\alpha}^{\prime}, \boldsymbol{\beta}^{\prime}, \boldsymbol{r}^{\prime} \right)^{\prime}$. Moreover, let $\boldsymbol{\phi}$ denote the vector of parameters which are estimated, that is, the structural parameters of the system. 

\item Under the assumption that sufficient restrictions have been imposed  to identify $\boldsymbol{\phi}$, we say that these identifying restrictions induce a mapping from $\boldsymbol{\phi} \mapsto \boldsymbol{\theta}$ such that the linear map satisfies $\boldsymbol{\theta} = F( \boldsymbol{\phi} )$. Let $\hat{\boldsymbol{\phi}}$ denote an estimate of $\boldsymbol{\phi}$. Then, under stationarity, invertibility and appropriate regularity conditions on the linear map $F(\cdot)$, $\hat{\boldsymbol{\phi}}$ will be normally distributed such that $\sqrt{T} \left(    \hat{\boldsymbol{\phi}} - \boldsymbol{\phi}_0 \right) \overset{d}{\to} \mathcal{N} \left( \boldsymbol{0}, \boldsymbol{\Xi}   \right)$.    

\end{itemize}

\end{example}
Overall, SVARs are used in the applied macroeconometrics literature to estimate the dynamic causal effects of structural shocks. Parameters in these models have traditionally been point-identified using zero or long-run restrictions. A commonly used approach in empirical macroeconometrics research is to construct time series that capture exogenous changes affecting the macroeconomy. The dynamic effects of these time series can be often estimated using distributed lag regressions. However, these variables are correlated with a particular structural shock and uncorrelated with other structural shocks (e.g., target versus non-target structural shocks). In particular, external instruments impose second moment restrictions that identify SVAR shocks and associated impulse response coefficients. 

A major known issue of SVAR-IV models is the \textit{weak identification problem} which appears in the case that instruments are only weakly correlated with the variable of interest. Specifically,  building on the IV regression literature, \cite{stock2021inference} proposed weak-instrument robust inference methods for impulse response coefficients. Their framework demonstrates how an external instrument can be used to identify the structural shock of interest, its impulse response coefficients, historical effect on the variables in the VAR and contribution to forecast error variances. A relevant identification scheme combines sign restrictions with external instruments is proposed by \cite{braun2023identification}.

\begin{remark}
Recall that the invertability property implies that structural shocks can be recovered from current and past (but not future) values of the observed macro variables. On the other hand, it is likely that external shock measures are contaminated by substantial measurement errors, causing attenuation bias. In other words, external instruments can be employed as an informative tool regarding shock importance. For example, an unrestricted linear MA model using IV can be employed to construct related statistics. In such settings a key identifying assumption is the availability of external instruments that are correlated with the shock of interest but are otherwise dynamically uncorrelated with all other macro variables (see, \cite{plagborg2022instrumental}). Identification of SVAR models with external instruments is also studied by \cite{olea2021inference}. Lastly, a framework for identifying SVARs with weak proxies is proposed by  \cite{angelini2024identification}.
\end{remark}

\newpage

\subsection{Identification via Sign-Restrictions and Set-Identified SVARs}

Several studies consider the set-identification of SVAR models based on a combination of economic theory and data-driven approaches such as sign restrictions.    The key idea behind sign restrictions is to characterise a shock through imposing restrictions on the responses of some variables, while remaining agnostic about others. In particular,  \cite{inoue2013inference, inoue2016joint} study the construction of confidence sets for set-identified parameters. Moreover, \cite{giacomini2021robust} and \cite{giacomini2022robust} consider robust inference in Proxy-SVARs with external instruments.  Sign restrictions usually imply partial identification (see,  \cite{dufour2005projection}). Additionally, it has several drawbacks including the fact that different structural shocks in the model might be characterized by the same set of sign restrictions. Nevertheless, to illustrate the implementation of the particular identification scheme we follow the framework proposed by \cite{gafarov2018delta} who contribute to the literature of set-identified SVARs by proposing a novel delta-method interval for the coefficients of the IRF.

Consider the $n-$dimensional SVAR with $p-$lags such that $\boldsymbol{Y}_t = \boldsymbol{A}_1 \boldsymbol{Y}_{t-1} + ... + \boldsymbol{A}_p \boldsymbol{Y}_{t-p}  + \boldsymbol{B} \boldsymbol{\varepsilon}_t$.
Then, the object of interest is the $k-$th period ahead structural impulse response function of variable $i$ to a particular shock $j$ such that $
\lambda_{k,i,j} (A,B) \equiv e_i^{\prime} \boldsymbol{C}_k (A) \boldsymbol{B}_j$,    
where $\boldsymbol{B}_j \equiv \boldsymbol{B} \boldsymbol{e}_j$ and $e_i$ and $e_j$ denote the $i-$th and $j-$th column of $\boldsymbol{I}_n$. A second object of interest is the vector of reduced-form VAR parameters  
\begin{align}
\boldsymbol{\mu} \equiv \left( \mathsf{vec}(\boldsymbol{A})^{\prime}, \mathsf{vec}(\boldsymbol{B})^{\prime} \right)^{\prime}  \in \mathcal{M} \subset \mathbb{R}^d, \ \ \ \boldsymbol{A} \equiv \mathsf{vec} \left( \boldsymbol{A}_1, \boldsymbol{A}_2,..., \boldsymbol{A}_p \right).     
\end{align}
Overall, a common practice in empirical macroeconomics literature is to set equality and inequality restrictions to set-identify the structural IRFs. Specifically, an example of an equality restriction in a monetary VAR is that prices do not react contemporaneously to monetary policy shocks. An example of an inequality restriction is that contractionary monetary policy shock cannot increase prices. Let $\mathcal{R} (\mu) \subset \mathbb{R}^n$ be the set of values of $\boldsymbol{B}_j$ that satisfy the inequality and equality restrictions such that
\begin{align}
\mathcal{R} (\mu) \equiv \left\{ \boldsymbol{B}_j \in \mathbb{R}^n \bigg|  Z(\mu)^{\prime} B_j = \boldsymbol{0}_{m_z \times 1} \ \text{and} \ S(\mu)^{\prime} B_j \geq \boldsymbol{0}_{ m_s \times 1 } \right\},    
\end{align}
where $Z(\mu)$ is a matrix of dimension $( n \times m_z  )$ that collects the equality restrictions specified by the researcher and $S(\mu)$ is a matrix of dimension $( n \times m_s )$ that collects the inequality restrictions.

\begin{corollary}[Directional Differentiability, see \cite{gafarov2018delta}]
Consider the auxiliary function $\nu_{k,i,j} \left( \mu; r (\mu) \right)$, such that $\nu_{k,i,j} \left( \mu; r (\mu) \right) \neq 0$, then the function $\nu_{k,i,j} \left( \mu; r (\mu) \right)$ is differentiable with respect to $\big( \mathsf{vec}(\boldsymbol{A})^{\prime}, \boldsymbol{\Sigma})^{\prime} \big)^{\prime}$ with derivative given by 
\begin{align}
\frac{ \partial \nu_{k,i,j} \left( \mu; r (\mu) \right)  }{ \partial \mathsf{vec} (\boldsymbol{A}) }
&=
\frac{ \partial \left( C_k ( \boldsymbol{A}) \right) }{ \partial \mathsf{vec} (\boldsymbol{A}) } \big( x^{*} \left( \mu ; r(\mu) \right) \otimes e_i \big) 
- 
\sum_{k=1}^{\ell} w_k^* \frac{ \partial ( r_k(\mu) ) }{ \partial \mathsf{vec} (\boldsymbol{A})  } x^{*} \left( \mu ; r(\mu) \right)
\\
\frac{ \partial \nu_{k,i,j} \left( \mu; r (\mu) \right)  }{ \partial \mathsf{vec} (\boldsymbol{\Sigma}) }
&=
\lambda^{*} \Sigma^{-1} x^{*} \left( \mu ; r(\mu)  \right) \otimes \boldsymbol{\Sigma}^{-1} x^{*} \left( \mu ; r(\mu) \right)  
- 
\sum_{k=1}^{\ell} w_k^* \frac{ \partial ( r_k(\mu) ) }{ \partial \mathsf{vec} (\boldsymbol{\Sigma})  } x^{*} \left( \mu ; r(\mu) \right).
\end{align}
\end{corollary}

\newpage

\begin{theorem}[see, \cite{gafarov2018delta}]
Consider a delta-method interval of the following form  \begin{align}
CS_T (1-\alpha) \equiv \left[ \underline{\nu}_{k,i,j}(\widehat{\boldsymbol{\theta}}_T) - z_{1 - \alpha / 2} \frac{ \widehat{\sigma}_{(k,i,j)} }{\sqrt{T}},  \underline{\nu}_{k,i,j}(\widehat{\boldsymbol{\theta}}_T)  + z_{1 - \alpha / 2} \frac{ \widehat{\sigma}_{(k,i,j)} }{\sqrt{T}} \right]   
\end{align}  
where $\widehat{\boldsymbol{\theta}}_T \equiv \left( \mathsf{vec}( \widehat{\boldsymbol{A}}_T ),  \mathsf{vec}( \widehat{\boldsymbol{\Sigma}}_T ) \right)$ is the OLS estimator such that 
\begin{align}
\widehat{\boldsymbol{A}}_T \equiv \left( \frac{1}{T} \sum_{t=1}^T \boldsymbol{Y}_t \boldsymbol{X}_t^{\prime}\right) \left( \frac{1}{T} \sum_{t=1}^T \boldsymbol{X}_t \boldsymbol{X}_t^{\prime} \right)^{-1} , \ \ \ \widehat{\boldsymbol{\Sigma}}_T \equiv \frac{1}{ T - np - 1} \sum_{t=1}^T \widehat{\boldsymbol{\eta}}_t \widehat{\boldsymbol{\eta}}_t^{\prime}, \ \ \ \widehat{\boldsymbol{\eta}}_t \equiv \boldsymbol{Y}_t - \widehat{\boldsymbol{A}}_T \boldsymbol{X}_t    
\end{align}
\end{theorem}

Therefore, \cite{gafarov2018delta} suggested to exploit the specific form of the directional derivative in the SVAR framework via the following expression 
\begin{align}
\widehat{\sigma}_{(k,i,j),T} \equiv \underset{ r \in R (\widehat{\mu}_T) }{ \mathsf{max} } \left\{  \underline{ \dot{\nu} }_{k,i,j}(\widehat{\boldsymbol{\theta}}_T; r)^{\prime} \widehat{ \boldsymbol{\Omega} }_T \underline{ \dot{\nu} }_{k,i,j}(\widehat{\boldsymbol{\theta}}_T; r) \right\}^{1/2}.  
\end{align}

\paragraph{Evaluating the posterior of sign-identified VAR Models}

Consider the $d-$variate reduced-form VAR$(p)$ model such that (see, \cite{inoue2013inference}) 
\begin{align}
y_t = c + B_1 y_{t-1} + B_2 y_{t-2} + ... + B_p y_{t-p} + e_t,    
\end{align}
for $t = 1,...,T$, where $e_t \overset{ \textit{i.i.d}  }{ \sim } \mathcal{N} ( 0_{n \times 1}, \Sigma )$ and $\Sigma$ is positive definite. We write the above econometric model such that
\begin{align}
Y = X B + e    
\end{align}
where $Y = \big[ y_1 \ y_2 \ ... \ y_T \big]^{\prime}_{ T \times k }$, $X = \left[  X_1 \ X_2 \ ... \ X_T \right]$, $X_t = [ 1 \ y_{t-1}^{\prime} \ ... \ y_{t-1}^{\prime} ]$ and $B = [ c \ B_1 \ B_2 \ ... \ B_p ]^{\prime}$. Moreover, we follow the conventional approach of specifying a normal-inverse Wishart prior distribution for the reduced-form VAR parameters. 
\begin{align}
\mathsf{vec} (B) | \Sigma &\sim \mathcal{N} \big( \mathsf{vec} ( \bar{B}_0 ), \Sigma \otimes N_0^{-1}   \big),    
\\
\Sigma &\sim \mathsf{IW}_n \big( \nu_0 S_0, \nu_0 \big),    
\end{align}
where $N_0$ is an $np \times np$ positive definite matrix, $S_0$ is an $n \times n$ covariance matrix and $\nu_0 > 0$.  

\begin{remark}
A common situation in SVARs for monetary policy evaluation is that the model is only partially identified when the identifying policy variable cannot be identified by other available structural shocks. If we are concerned with a subset of impulse response functions only, what matters for constructing the posterior mode is not the joint impulse response distribution, but the marginalized distribution obtained by integrated out responses to shocks that are not identified. Since using tools from Bayesian econometrics is beyond the scope of this course, from the classical point-identification perspective \cite{carriero2023blended} proposed a blended approach which aims to sharpen identification.   
\end{remark}

\newpage

\subsection{Identification in Proxy SVARs under Possible Nonstationarity}
\label{Section4.6}

Suppose that $\left\{ Y_t: t = -p+1,...,T \right\}$
be a $( d \times 1 )$ vector of observed variables that follows a SVAR process $
\boldsymbol{Y}_t = \boldsymbol{\mu} + \sum_{j=1}^p \boldsymbol{A}_j \boldsymbol{Y}_{t-j} + \boldsymbol{u}_t, \ \ \  \boldsymbol{u}_t = \boldsymbol{B} \boldsymbol{\varepsilon}_t$, where $\boldsymbol{u}_t$ is the reduced-form error terms and $\boldsymbol{A}_j$ are $( d \times d)$ coefficient matrices for $j \in \left\{ 1,..., p \right\}$. Therefore, to study the IRFs with respect to the structural shocks $\varepsilon_{it}$, we study the estimation of the $i-$th column of the structural parameter $B$ denoted by $\boldsymbol{B} = ( \boldsymbol{b}_1,..., \boldsymbol{b}_d )^{\prime}$ with external instruments $\big\{ \boldsymbol{Z}_{t} = ( z_{1t},..., z_{kt} )^{\prime} \in \mathbb{R}^k : t = 1,..., T \big\}$.

\begin{example}
Let $z_t$ denote an external instrument for a structural shock of interest $\varepsilon_{kt}, k \in \left\{ 1,..., K \right\}$. Based on conventional assumptions from the instrumental variable estimation literature, the potential instruments must satisfy the following conditions:    
\begin{itemize}

\item[(i)] $z_t$ is relevant for the underlying structural shock $\varepsilon_{kt}$ such that $\mathbb{E} \left( \varepsilon_{kt} z_t \right) = \phi \neq 0$.

\item[(ii)] $z_t$ is exogenous from other structural shocks in the system such that $\mathbb{E} \left( \varepsilon_{jt} z_t \right) = 0, \forall \ j \in \left\{ 1,..., K \right\} \backslash \left\{ k \right\}$.    
    
\end{itemize}
Then, based on the above two conditions it follows that up to scale $\phi$, the population covariance between the instrument and VAR residuals obtains the $k-$th column of $B$, denote by $B_k$ such that
\begin{align*}
\mathbb{E} \left( u_t z_t \right) = B_k \mathbb{E} \left( \varepsilon_{kt} z_t \right) = \phi B_k.     
\end{align*}
Thus, the link between the structural shock and reduced form-error provides the condition to ensure the instrument validity and instrument relevance. These econometric specifications are also particularly important for the identification of technology shocks as in  \cite{feve2010identification} and \cite{pagan2008econometric} (see, also \cite{lovcha2021identifying}). However, according to \cite{cheng2022instrumental} long-run restrictions can lead to the problem of weak identification due to sensitivity to low-frequency correlations between variables. Furthermore, due to the fact that $y_t$ often includes components which are integrated and cointegrated it is important to develop a robust identification methodology for Proxy SVAR models which is robust to weak identification and possibly nonstationary regressors.  
\end{example}

\begin{example}[Nearly nonstationary SVAR, see \cite{gospodinov2010inference}]
\label{example1}
Suppose that we have a bivariate vector autoregressive process denoted with $\tilde{\boldsymbol{y}}_t = ( y_{1t}, y_{2t} )$  such that 
\begin{align}
\boldsymbol{\Psi} (L) = ( \boldsymbol{I} - \boldsymbol{\Phi} L ) \tilde{\boldsymbol{y}}_t = \boldsymbol{u}_t   
\end{align}
where the matrix $\boldsymbol{\Phi}$ contains the largest roots of the system and $\boldsymbol{\Psi} (L) = \left( \boldsymbol{I} - \sum_{i=1}^p \boldsymbol{\Psi}_i L^i \right)$ and where $\widetilde{\boldsymbol{A}} (L) = \boldsymbol{\Psi} (L) ( \boldsymbol{I} - \boldsymbol{\Phi} L )$. Consider the reduced form VAR model $
\tilde{\boldsymbol{y}}_t = \boldsymbol{A}_1 \tilde{\boldsymbol{y}}_{t-1} + ... + \boldsymbol{A}_{p+1} \tilde{\boldsymbol{y}}_{t-p-1} + \boldsymbol{u}_t$. Pre-multiplying both sides by the matrix $\boldsymbol{B}_0 = \begin{bmatrix}
1 & - b_{21}^{(0)}
\\
- b_{21}^{(0)} & 1
\end{bmatrix}$, yields the structural VAR model $\boldsymbol{B} (L) \boldsymbol{y}_t = \boldsymbol{\varepsilon}_t$, where $\boldsymbol{B} (L) = \boldsymbol{B}_0 \boldsymbol{A} (L)$ and $\boldsymbol{\varepsilon}_t = \boldsymbol{B}_0 \boldsymbol{u}_t$ denote the structural shocks with are assumed to be orthogonal with variances $\sigma_1^2$ and $\sigma_2^2$. Thus, writing the reduced model in the DF form \begin{align}
\boldsymbol{y}_t = \boldsymbol{A}^{\star}(1) \boldsymbol{y}_{t-1} + \sum_{j=1}^p \boldsymbol{A}^{\star \star}_j \Delta \boldsymbol{y}_{t-j} + \boldsymbol{u}_t,    
\end{align}

\newpage

such that $\boldsymbol{A}^{\star} (L) = L^{-1} \left( \boldsymbol{I} - \widetilde{\boldsymbol{A}} (L) \right)$, $\boldsymbol{A}^{\star} (1) = \boldsymbol{I} + \frac{1}{T} \boldsymbol{\Psi} (1) \boldsymbol{C}$ and $\boldsymbol{A}_i^{\star \star} = - \sum_{j = i + 1}^{p} \boldsymbol{A}_j^{\star}$. This implies that,
\begin{align}
\begin{bmatrix}
\Delta y_{1t}
\\
\Delta y_{2t}
\end{bmatrix}    
= 
\begin{bmatrix}
\psi_{12}(1) & 0 
\\
0 & \psi_{12}(1)
\end{bmatrix}
\begin{bmatrix}
\displaystyle \frac{c}{T} y_{2t}
\\
\displaystyle \frac{c}{T} y_{2t}
\end{bmatrix}
+
\begin{bmatrix}
\displaystyle  \sum_{j=1}^p a_{j,11}^{\star}   & \displaystyle \sum_{j=1}^p a_{j,12}^{\star}
\\
\displaystyle \sum_{j=1}^p a_{j,21}^{\star}   &  \displaystyle \sum_{j=1}^p a_{j,22}^{\star}
\end{bmatrix}
\begin{bmatrix}
\Delta y_{1t-j}
\\
\Delta y_{2t-j}
\end{bmatrix}
+
\begin{bmatrix}
u_{1t}
\\
u_{2t}
\end{bmatrix}
.
\end{align}
Notice that for this example following \cite{gospodinov2010inference}, the author considers a bivariate system of equations and propose an overidentifying restriction which using the VECM representation of the model reduces to the assumption that $b_{12}$ is known and is equal to zero. The relationship between the endogenous variable and the instrument is given by the second equation of the reduced VAR model as 
\begin{align}
\Delta \tilde{y}_{2t} = \psi_{22}(1) \frac{c}{T} y_{2t-1} + u_{2,t},  \ \ \ t = 1,...,n, \end{align}
Moreover, based on the long-run restrictions and the model assumptions, it can be proved 
\begin{align}
\left( \hat{b}_{12}^{(0)} - b_{12}^{(0)} \right) \Rightarrow \frac{ \sigma_1 }{ \omega_2 } \frac{\displaystyle \rho \int_0^1 J_c(s) dV_2 (s) + \sqrt{1 - \rho^2}  \int_0^1 J_c(s) dV_1(s)  }{ \displaystyle c \int_0^1 J_c(s)^2 ds + \int_0^1 J_c(s) dV_2 (s)}.   
\end{align}
\end{example}
Long-run restrictions are a popular identification method of SVAR models since the seminal contribution of \cite{blanchard1988dynamic}, which are commonly used to investigate the hours debate\footnote{In particular, the hours debate is concerned with the short-run effect of a technology shock on hours initiated by \cite{gali1999technology} and \cite{christiano2003happens}. Using a bivariate SVAR framework $Y_{1t}$ corresponds to log productivity and $Y_{2t}$ corresponds to log hours. Further studies which investigate the impact of technology shocks to business cycles include \cite{francis2005technology},   \cite{feve2009response} and \cite{chaudourne2014understanding}. For example, \cite{feve2009response} impose a long-run identifying restriction which implies that the composite non-technology shock has no long-run effect on labor productivity. This means that the upper triangular element of $\boldsymbol{A}(L)$ in the long-run must be zero, that is $A_{12} = 0$, so to uncover this restriction from the estimated VAR$(p)$ model, the matrix $\boldsymbol{A}(1)$ is obtained by the Choleski decomposition of $\boldsymbol{C}(1) \boldsymbol{\Sigma} \boldsymbol{C}(1)$.} of \cite{gali1999technology} as well as the impact of technology shocks on productivity and output (see, also \cite{chari2008structural}). 

\begin{example}
Following \cite{chaudourne2014understanding}, consider a nearly stationary persistent process, which is obtained using the local-to-unity parametrization as in \cite{phillips1987towards} who consider nearly unit-root process to investigate the asymptotic power of the unit-root tests under a sequence of local alternatives. Therefore, since we are interested to model a highly persistent process which is asymptotically stationary, we consider a sequence of local alternatives such that the process is locally nonstationary but asymptotically stationary and persistent.  
\begin{align}
\big( 1 - \rho L \big) \Delta x_t = u_t - \left( 1 - \frac{c}{ \sqrt{n} } \right) u_{t-1}, \ \ t = 1,..., n.  
\end{align}
Notice that as $n \to \infty$ and for a high value of $\rho$, this process becomes a stationary persistent process whereas, in a finite sample, the process is characterized by a unit root.

\newpage

Thus, this process is locally nonstationary but asymptotically stationary and persistent. Such a process appears to be particulary well suited to characterize the dynamics of hours worked because it implies a unit root in a finite sample but is asymptotically stationary and persistent. This is typically the case for per capita hours worked which are included in SVARs. Consider a bivariate representation where $\boldsymbol{X}_t := ( \Delta X_{1t}, X_{2t} )$ for $t = 1,..., n$, such that the variable $X_{1t}$ contains an exact unit root and $X_{2t}$ is a highly persistent variable. Notice that both variables in the vector $X_t$ are asymptotically second order stationary and they admit asymptotically the following Wold representation
\begin{align}
\boldsymbol{X}_t  = \boldsymbol{C}(L) \boldsymbol{\varepsilon}_t, \ \ \ \ \boldsymbol{C}(L)  = \sum_{j=0}^{\infty} \boldsymbol{C}_j \boldsymbol{\varepsilon}_{t-j},  
\end{align}
Consider a Structural Moving Average representation for $\boldsymbol{X}_t$ such that
\begin{align}
\boldsymbol{X}_t = \boldsymbol{A}(L) \boldsymbol{\eta}_t, \ \ \ \ \boldsymbol{A}(L)  = \sum_{j=0}^{\infty} \boldsymbol{A}_j, 
\end{align}
where $\boldsymbol{\eta}_t = \big( \eta_{1t}, \eta_{2t} \big)^{\prime}$ is the vector of orthogonal shocks with $\mathbb{E} \left( \boldsymbol{\eta}_t \boldsymbol{\eta}_t^{\prime} \right) = \boldsymbol{\Omega}$ a diagonal matrix.

A common identification assumption is that $\boldsymbol{\Omega} = \boldsymbol{I}_2$ which implies that the variance of the structural shocks is normalized to one. Then, we have that the error terms $\boldsymbol{\varepsilon}_t$ from the reduced form are related to the structural error terms $\boldsymbol{\eta}_t$ as follows: $\boldsymbol{\varepsilon}_t = \boldsymbol{A}_0 \boldsymbol{\eta}_t$ which implies that $\boldsymbol{\Sigma} = \boldsymbol{A}_0 \boldsymbol{A}_0^{\prime}$. Therefore, the identification scheme, should ensure the unique identification of the structural parameter $\boldsymbol{A}_0$ and we focus to the case of long-run restrictions as proposed by \cite{blanchard1988dynamic}. Specifically, we use the long-run variance-covariance matrix of the reduced form and the structural form which are related by $\boldsymbol{C}(1) \boldsymbol{\Sigma} \boldsymbol{C}(1)^{\prime} = \boldsymbol{A}(1) \boldsymbol{A}(1)^{\prime}$ and $\boldsymbol{A}_0 = \boldsymbol{C}(1)^{-1} \boldsymbol{A}(1)$. Typically a lower triangular structure is imposed to the long-run impact matrix $\boldsymbol{A}(1)$ which can be easily obtained using a Choleski decomposition of the long-run covariance matrix $\boldsymbol{C}(1) \boldsymbol{\Sigma} \boldsymbol{C}(1)^{\prime}$. In other words, in the case where two variables are included in $\boldsymbol{X}_t$, the first structural shock is the only one shock that can have a permanent effect on the first variable.  

Consider for example, a finite sample of size $T$ observations a structural characterization of the highly persistent variable $\boldsymbol{X}_{2t}$ as a nearly stationary persistent process such that
\begin{align}
\Delta X_{2t} 
&= 
a_{21}(L) \Delta \eta_{1t} + a_{22}(L) \left( 1 - \left( 1 - \frac{c}{ \sqrt{T} } \right) L \right) \eta_{2t}  
\nonumber
\\
&\equiv 
a_{21}(L) \Delta \eta_{1t} + \tilde{a}_{22}(L) \eta_{2t} 
\end{align}
where $\tilde{a}_{22}(L) =  a_{22}(L)\left( 1 - \left( 1 - \frac{c}{ \sqrt{T} } \right) L \right)$. Moreover, using the Beveridge-Nelson decomposition, the above model can be rewritten as below
\begin{align}
\Delta X_{2t} 
= 
\underbrace{ a_{21}(L) \Delta \eta_{1t} }_{ \text{permanent shock} } + \underbrace{ \tilde{a}_{22,T}(1) \eta_{2t} +  \tilde{a}^{*}_{22,T}(L) \Delta \eta_{2t}  }_{ \text{transitory shock}   }     
\end{align}
such that $\tilde{a}_{22,T}(1) = a_{22,T}(1) c / \sqrt{T}$ and $\tilde{a}^{*}_{22,T}(L) (1 - L) = \tilde{a}^{*}_{22,T}(L) - \tilde{a}^{*}_{22,T}(1)$.

\newpage

Putting everything together, it implies that the SMA bivariate representation contains a difference stationary process $\Delta X_{1t}$ and a highly persistent process such that 
\begin{align}
\begin{bmatrix}
\Delta X_{1t}
\\
\Delta X_{2t}
\end{bmatrix}   
=
\begin{bmatrix}
a_{11}(L) &  \tilde{a}_{12}(L) 
\\
a_{21}(L) (1- L) &  \tilde{a}_{22,T}(L) 
\end{bmatrix}
\begin{bmatrix}
\eta_{1t}
\\
\eta_{2t}
\end{bmatrix}.   
\end{align}

\paragraph{Estimation and inference}

We now consider estimation and inference with the proposed specifications of the SVAR model. The reduced-form moving average representation is retrieved by performing a finite order $\mathsf{VAR}$ on the data. Suppose that the structural moving average representation can be characterized or approximated in a small sample by a finite VAR of order $p$. Consider the reduced form $\mathsf{VAR}(p)$:
\begin{align}
\boldsymbol{D}(L) \boldsymbol{X}_t 
= 
\boldsymbol{\varepsilon}_t, \ \ \ \boldsymbol{D}(L)  = 
\boldsymbol{I} - \sum_{i=1}^p \boldsymbol{D}_i \boldsymbol{L}^i
=
\begin{bmatrix}
\displaystyle 1 - \sum_{i=1}^p d_{11}^{(i)} L^i   &  \displaystyle -  \sum_{i=1}^p d_{12}^{(i)} L^i
\\
\displaystyle -  \sum_{i=1}^p d_{12}^{(i)} L^i  &  \displaystyle 1 - \sum_{i=1}^p d_{22}^{(i)} L^i
\end{bmatrix}
\end{align}
By multiplying both sides by a matrix of the form $\boldsymbol{B}_0 = \begin{pmatrix}
1 &  \textcolor{red}{ - b_{12}^{(0)} }
\\
\textcolor{red}{ - b_{21}^{(0)} } & 1 
\end{pmatrix} = \boldsymbol{A}_0^{-1}$, we obtain  the VAR as a function of the structural shocks such that $\boldsymbol{B} (L) \boldsymbol{X}_t = \boldsymbol{\eta}_t$ with $\boldsymbol{B} (L) = \boldsymbol{B}_0 \boldsymbol{D}(L)$. 
\begin{align*}
&\begin{bmatrix}
1 & \textcolor{red}{ - b_{12}^{(0)} }
\\
\textcolor{red}{ - b_{21}^{(0)} } & 1 
\end{bmatrix}
\begin{bmatrix}
\displaystyle 1 - \sum_{i=1}^p d_{11}^{(i)} L^i   &  \displaystyle -  \sum_{i=1}^p d_{12}^{(i)} L^i
\\
\displaystyle -  \sum_{i=1}^p d_{12}^{(i)} L^i  &  \displaystyle 1 - \sum_{i=1}^p d_{22}^{(i)} L^i
\end{bmatrix}
\\
&=
\begin{bmatrix}
\left( \displaystyle 1 - \sum_{i=1}^p d_{11}^{(i)} L^i \right) + \displaystyle  \textcolor{red}{ b_{12}^{(0)} } \sum_{i=1}^p d_{12}^{(i)} L^i  &  \displaystyle - \sum_{i=1}^p d_{12}^{(i)} L^i \textcolor{red}{ - b_{12}^{(0)} }  \left( 1 - \sum_{i=1}^p d_{22}^{(i)} L^i \right)
\\
\displaystyle \textcolor{red}{ - b_{21}^{(0)} }  \left( \displaystyle 1 - \sum_{i=1}^p d_{11}^{(i)} L^i \right)  -  \sum_{i=1}^p d_{12}^{(i)} L^i  &  \displaystyle \textcolor{red}{ b_{21}^{(0)} } \sum_{i=1}^p d_{11}^{(i)} L^i  +  \left( 1 - \sum_{i=1}^p d_{22}^{(i)} L^i \right)
\end{bmatrix}.
\end{align*}
Equivalently, we have that 
\begin{align}
\begin{bmatrix}
\left( \displaystyle 1 - \sum_{i=1}^p d_{11}^{(i)} L^i \right) + \displaystyle  \textcolor{red}{ b_{12}^{(0)} } \sum_{i=1}^p d_{12}^{(i)} L^i  &  \displaystyle - \sum_{i=1}^p d_{12}^{(i)} L^i \textcolor{red}{ - b_{12}^{(0)} }  \left( 1 - \sum_{i=1}^p d_{22}^{(i)} L^i \right)
\\
\displaystyle \textcolor{red}{ - b_{21}^{(0)} }  \left( \displaystyle 1 - \sum_{i=1}^p d_{11}^{(i)} L^i \right)  -  \sum_{i=1}^p d_{12}^{(i)} L^i  &  \displaystyle \textcolor{red}{ b_{21}^{(0)} } \sum_{i=1}^p d_{11}^{(i)} L^i  +  \left( 1 - \sum_{i=1}^p d_{22}^{(i)} L^i \right)
\end{bmatrix}
\begin{bmatrix}
\Delta X_{1t}
\\
X_{2t}
\end{bmatrix}
= 
\begin{bmatrix}
\eta_{1t}
\\
\eta_{2t}
\end{bmatrix}
\end{align}
Specifically, the first variable in this bivariate system is expressed as below
\begin{align}
\Delta X_{1t} 
= 
\left( \sum_{i=1}^p d_{11}^{(i)} L^i - b_{12}^{(0)} \sum_{i=1}^p d_{21}^{(i)} L^i \right) \Delta X_{1t}  
+
\left( \sum_{i=1}^p d_{12}^{(i)} L^i + b_{12}^{(0)} \left[ 1 - \sum_{i=1}^p d_{22}^{(i)} L^i \right] \right) X_{2t} + \eta_{1t}.
\end{align}

\newpage

Moreover, imposing the structural long-run impact matrix $\boldsymbol{A}(1)$ to be lower triangular implies that $\boldsymbol{B}_0 \boldsymbol{D}(1)$ is also lower triangular by $\boldsymbol{A}(1) = \boldsymbol{D}(1)^{-1} \boldsymbol{A}_0$. The long-run multiplier of the variable $\boldsymbol{X}_{2t}$ on $\Delta \boldsymbol{X}_{1t}$ is then zero. Imposing this constraint yields that
\begin{align*}
\Delta X_{1t} 
&= 
\left( \sum_{i=1}^p d_{11}^{(i)} L^i - b_{12}^{(0)} \sum_{i=1}^p d_{21}^{(i)} L^i \right) \Delta X_{1t}  
+
b_{12}^{(0)} \Delta X_{2t}
+
\sum_{i=1}^{p-1} \tilde{b}_{12}^{(i)} L^i \Delta X_{2t} + \eta_{1t}
\\
&= 
b_{11}(L) \Delta X_{1t-1} + b_{12}^{(0)} \Delta X_{2t} + \tilde{b}_{12} (L) \Delta X_{2t-1} + \eta_{1t}.
\end{align*}
where the element $b_{12}^{(0)} = - d_{12}(1) / \left( 1 - d_{22}(1) \right)$.

\paragraph{IV Estimation}
The IV estimator of $b_{12}^{(0)}$ with $X_{2t-1}$ as instrument is given by the following expression
\begin{align*}
b_{12}^{(0)} 
= 
\left( \frac{1}{T} \sum_{t=1}^T X_{2t-1} \Delta \tilde{X}_{2t} \right)^{-1} \left( \frac{1}{T} \sum_{t=1}^T X_{2t-1} \Delta \tilde{X}_{1t} \right)
=
\left( \frac{1}{T} \sum_{t=1}^T X_{2t-1} \Delta \tilde{X}_{2t} \right)^{-1} \left( \frac{1}{T} \sum_{t=1}^T X_{2t-1} \left[ b_{12}^{(0)} \Delta \tilde{X}_{2t} + \eta_{1t} \right] \right)
\end{align*}
which implies that 
\begin{align}
\widehat{b}_{12}^{(0)} - b_{12}^{(0)} = \left( \frac{1}{T} \sum_{t=1}^T X_{2t-1} \Delta \tilde{X}_{2t} \right)^{-1} \left( \frac{1}{T} \sum_{t=1}^T X_{2t-1} \eta_{1t} \right) + o_p(1).   
\end{align}
We consider the estimation of the structural parameter $b_{21}^{(0)}$. Since $\eta_{1t}$ and $\eta_{2t}$ are orthogonal, the residuals $\widehat{\eta}_{1t} = \Delta \tilde{X}_{1t} - \widehat{b}_{12}^{(0)} \Delta \tilde{X}_{2t}$ can be used as instrument for the endogenous variable $\Delta \tilde{X}_{1t}$. Thus, it holds that $\widehat{\eta}_{1t} - \eta_{1t} = \left( \widehat{b}_{12}^{(0)} - b_{12}^{(0)} \right) \Delta \tilde{X}_{2t}$. Moreover, we shall define with $z_t = \left( \widehat{\eta}_{1t}, X_{2t-1} \right)^{\prime}$ and $x_t = \left( \Delta \tilde{X}_{1t}, \tilde{X}_{2t-1} \right)^{\prime}$. Then, the IV estimator of $b_{12}^{(0)}$ is given by 
\begin{align}
\widehat{\beta} = \left( \frac{1}{T} \sum_{t=1}^T z_t x_t^{\prime} \right)^{-1}  \left( \frac{1}{T} \sum_{t=1}^T z_t \Delta X_{2t}^{\prime} \right).    
\end{align}

\paragraph{Impulse Response Function Analysis}

Consider the following $\mathsf{VAR}(1)$ model as below
\begin{align}
\Delta X_{1t} &= \textcolor{red}{ b_{12}^{(0)} } \Delta X_{2t} + \eta_{1t}
\\
X_{2t} &= \textcolor{red}{ b_{21}^{(0)} } \Delta X_{1t} + \textcolor{red}{ b_{22} } X_{2t-1} + \eta_{2t}
\end{align}
Notice that the highly persistent variable represents hours worked. Generally, it can be proved that estimated impulse responses from LSVAR and SDVAR models are biased in a finite sample if the measure of productivity is contaminated by low frequency movements in hours. Moreover, although these estimators are asymptotically consistent they display a nonstandard limiting distribution.  
\end{example}

\newpage

\begin{example}[SVAR reduced-form representation,  \cite{cheng2022instrumental}] Suppose that we have a vector of macroeconomic variables $\boldsymbol{Y}^{*}_t$ such that $\boldsymbol{Y}^{*}_t := \boldsymbol{\mu} + \boldsymbol{Y}_t$. Consider the SVAR model formulation
\begin{align}
\underbrace{ \boldsymbol{Y}^{*}_t }_{ (k \times 1) } = 
\boldsymbol{\delta} + \sum_{j=1}^p \boldsymbol{\Phi}_j \boldsymbol{Y}^{*}_{t-j} + \boldsymbol{u}_t, \ \ \ \boldsymbol{u}_t = \boldsymbol{B} \boldsymbol{\varepsilon}_t, \ \ \ t = -p + 1,..., T.  
\end{align}
Thus, by taking the demeaned vector of regressors $\boldsymbol{Y} = \left( \boldsymbol{Y}_t^{*} - \boldsymbol{\mu} \right)$ we proceed by considering the following reduced-form representation. In other words, we allow the SVAR process to consist of a set of near-unit root and cointegrated variables.

Specifically, \cite{cheng2022instrumental} employ the reduced-form: 
\begin{align}
\boldsymbol{Y}_{1t} &=  \left( \boldsymbol{I}_{k_1} + \frac{ \boldsymbol{C} }{ n } \right) \boldsymbol{Y}_{1,t-1} + \boldsymbol{u}_{1t}, \  \boldsymbol{u}_{1t}\in \mathbb{R}^{k_1},
\\
\boldsymbol{Y}_{2t} &= \boldsymbol{\Pi} \boldsymbol{Y}_{1t} + \boldsymbol{u}_{2t}, \  \boldsymbol{u}_{2t} \in \mathbb{R}^{k_2},
\\
\boldsymbol{Y}_{3t} &= \boldsymbol{u}_{3t}, \ \boldsymbol{u}_{3t} \in \mathbb{R}^{k_3}   
\end{align}
where it holds that $\boldsymbol{\Psi} (L) \boldsymbol{u}_t = \boldsymbol{e}_t$ such that $\boldsymbol{u}_t = \big( \boldsymbol{u}_{1t}^{\prime}, \boldsymbol{u}_{2t}^{\prime}, \boldsymbol{u}_{3t}^{\prime} \big)^{\prime}$ and  $\boldsymbol{e}_t = \big( \boldsymbol{e}_{1t}^{\prime}, \boldsymbol{e}_{2t}^{\prime}, \boldsymbol{e}_{3t}^{\prime} \big)^{\prime}$ such that 
\begin{align}
\boldsymbol{\Psi} (L) = \boldsymbol{I}_k - \boldsymbol{\Psi}_1 (L) - ... - \boldsymbol{\Psi}_{p-1} (L) L^{p-1},    
\end{align}
is a $(p-1)th-$order lag polynomial and $k = k_1 + k_2 + k_3$ is the number of regressors. Based on the data generating process and SVAR process defined above, we assume that $\boldsymbol{Z}_t$ follows a linear process $
\boldsymbol{Z}_t = \mu_Z + \boldsymbol{\Xi} (L) \boldsymbol{v}_t,    
$.
\end{example}

\begin{assumption}
Suppose that the following conditions hold: 
\begin{itemize}

\item[\textit{(i).}] The roots of $\Psi (L)$ has all roots outside the unit circle.

\item[\textit{(ii).}] $\boldsymbol{\Xi}_0 = \boldsymbol{I}_K$, where $\boldsymbol{\Xi} (1)$ has full rank, where $\displaystyle \sum_{ j = 0}^{ \infty } j^2 \norm{ \boldsymbol{\Xi}_j }^2 < \infty$. 

\item[\textit{(iii).}] Suppose that $\bar{\boldsymbol{e}}_t = [ \boldsymbol{e}_t^{\prime}, \boldsymbol{v}_t^{\prime} ]^{\prime}$ is an $\textit{i.i.d} ( r + k ) \times 1$ vector with mean zero, $\mathbb{E} \big[ \bar{\boldsymbol{e}}_t  \bar{\boldsymbol{e}}_t^{\prime} \big] = \boldsymbol{\Sigma}$ is positive definite, fourth moments of $\bar{e}_t$ are finite, and $\boldsymbol{e}_t$ is homoscedastic conditional on $\boldsymbol{v}_t$.
    
\end{itemize}    
\end{assumption}
Then, it can be shown that the DGP can be reformulated as below
\begin{align}
\Delta \boldsymbol{Y}^{*}_t &= \boldsymbol{A}_1 +    \boldsymbol{A}_2 \underbrace{ ( \boldsymbol{Y}^{*}_{1,t-1} - \boldsymbol{\mu}_1 ) }_{ Y_{1,t-1} } + \boldsymbol{A}_3 \boldsymbol{D}_t + \boldsymbol{u}_t,     
\\
\boldsymbol{D}_t &:= 
\begin{bmatrix}
\left( \left( \boldsymbol{Y}^{*}_{2,t-1} - \boldsymbol{\mu}_2 \right) - \boldsymbol{\Pi} \left( \boldsymbol{Y}^{*}_{1,t-1} - \boldsymbol{\mu}_1 \right) \right)^{\prime} 
\\
\left( \boldsymbol{Y}^{*}_{3,t-1} - \boldsymbol{\mu}_3 \right)^{\prime} 
\\
\Delta \boldsymbol{Y}_{t-1}^{* \prime},...., \Delta \boldsymbol{Y}_{t-p+1}^{* \prime}
\end{bmatrix} 
, \ \boldsymbol{u}_t := \boldsymbol{P} \boldsymbol{e}_t, \ \ \boldsymbol{P} = 
\begin{bmatrix}
\boldsymbol{I}_{k_1} & \boldsymbol{0} & \boldsymbol{0} 
\\
\boldsymbol{\Pi} &  \boldsymbol{I}_{k_2} & \boldsymbol{0} 
\\
\boldsymbol{0} & \boldsymbol{0} & \boldsymbol{I}_{k_3}
\end{bmatrix}.
\end{align}

\newpage

Denote with $\boldsymbol{X}_t = \big[ \left(   \boldsymbol{Y}^{*}_{1,t-1} - \boldsymbol{\mu}_1  \right)^{\prime}, \boldsymbol{1}, \boldsymbol{D}_t^{\prime} \big]^{\prime}$, then the original SVAR process (which includes a model intercept) and the above reformulation based on the demeaned original series has an one-to-one transformation. Moreover, this equivalent representation implies that the least-squares residual $\hat{\boldsymbol{u}}_t$ from the VAR process is numerically equivalent to that obtained from regressing $\Delta \boldsymbol{Y}_t$ on $\boldsymbol{X}_t$. Therefore, we shall use the conventional SVAR representation for practical estimation purposes such as to obtain the OLS residuals of the model whereas we use the above reformulation for the asymptotic theory analysis of the corresponding model estimator under the presence of both near-unit and cointegrated regressors.  

Next, for the near unit root regressors we denote with $\boldsymbol{J}_c (s)$ an $( k_1 \times 1 )$ vector OU process such that it satisfies the following stochastic differential equation $d \boldsymbol{J}_c (s) = \boldsymbol{C} \boldsymbol{J}_c (s) + d \boldsymbol{B}_u (s)$. Denote with $\boldsymbol{\mathcal{D}}_n$ the diagonal matrix $
\boldsymbol{\mathcal{D}}_n = 
\begin{bmatrix}
\sqrt{n} \boldsymbol{I}_{k_1} & \boldsymbol{0}
\\
\boldsymbol{0}  &  \boldsymbol{I}_{ pr - k_1 + 1 }
\end{bmatrix}$. 
Moreover, it holds that 
\begin{align*}
\textcolor{blue}{ \boldsymbol{\mathcal{D}}^{-1}_n } \left( \frac{1}{n} \sum_{t=1}^n 
\boldsymbol{X}_t \boldsymbol{X}_t^{\prime}  \right) \textcolor{blue}{ \boldsymbol{\mathcal{D}}^{-1}_n }
&= 
\textcolor{blue}{ \boldsymbol{\mathcal{D}}^{-1}_n } \sum_{t=1}^n 
\begin{bmatrix}
\left(   \boldsymbol{Y}^{*}_{1,t-1} - \boldsymbol{\mu}_1  \right)
\\
\boldsymbol{1}
\\
\boldsymbol{D}_t
\end{bmatrix}
\big[ \left(   \boldsymbol{Y}^{*}_{1,t-1} - \boldsymbol{\mu}_1  \right)^{\prime}, \boldsymbol{1}, \boldsymbol{D}_t^{\prime} \big]
\textcolor{blue}{ \boldsymbol{\mathcal{D}}^{-1}_n }
\\
&=
\begin{bmatrix}
\displaystyle \sum_{t=1}^n \textcolor{blue}{ \boldsymbol{\mathcal{D}}^{-1}_n } \left( \boldsymbol{Y}^{*}_{1,t-1} - \boldsymbol{\mu}_1 \right) \left( \boldsymbol{Y}^{*}_{1,t-1} - \boldsymbol{\mu}_1 \right)^{\prime} \textcolor{blue}{ \boldsymbol{\mathcal{D}}^{-1}_n } &  \displaystyle \sum_{t=1}^n \textcolor{blue}{ \boldsymbol{\mathcal{D}}^{-1}_n } \left( \boldsymbol{Y}^{*}_{1,t-1} - \boldsymbol{\mu}_1 \right) \textcolor{blue}{ \boldsymbol{\mathcal{D}}^{-1}_n } & \boldsymbol{0}
\\
\displaystyle \sum_{t=1}^n \textcolor{blue}{ \boldsymbol{\mathcal{D}}^{-1}_n } \left( \boldsymbol{Y}^{*}_{1,t-1} - \boldsymbol{\mu}_1 \right)^{\prime} \textcolor{blue}{ \boldsymbol{\mathcal{D}}^{-1}_n } & \boldsymbol{1} &  \boldsymbol{0}
\\
\boldsymbol{0} & \boldsymbol{0} & \boldsymbol{\Gamma}_{DD}
\end{bmatrix}
\end{align*}
which converges into the following stochastic matrix
\begin{align*}
\textcolor{blue}{ \boldsymbol{\mathcal{D}}^{-1}_n } \left( \frac{1}{n} \sum_{t=1}^n 
\boldsymbol{X}_t \boldsymbol{X}_t^{\prime}  \right) \textcolor{blue}{ \boldsymbol{\mathcal{D}}^{-1}_n }
\overset{p}{\to}
\mathbb{V} :=
\begin{bmatrix}
\displaystyle \int_0^1 \boldsymbol{J}_c(s) \boldsymbol{J}_c(s)^{\prime} ds &  \displaystyle  \int_0^1 \boldsymbol{J}_c(s)^{\prime}   & \boldsymbol{0}
\\
\displaystyle \int_0^1 \boldsymbol{J}_c(s)^{\prime}
& \boldsymbol{1} &  \boldsymbol{0}
\\
\boldsymbol{0} & \boldsymbol{0} & \boldsymbol{\Gamma}_{DD}
\end{bmatrix}     
\end{align*}
Therefore, the following converges results holds
\begin{align}
\begin{bmatrix}
\displaystyle \frac{1}{ n \sqrt{n} } \sum_{t=1}^n \boldsymbol{Y}_{1,t-1}  
\\
\displaystyle \frac{1}{ n } \sum_{t=1}^n \boldsymbol{Y}_{1,t-1} \boldsymbol{u}^{\prime}_t  
\end{bmatrix}    
\overset{p}{\to}
\begin{bmatrix}
\displaystyle \int_0^1 \boldsymbol{J}_c(s) ds
\\
\displaystyle \int_0^1 \boldsymbol{J}_c(s) d \boldsymbol{B}_{u}^{\prime}(s) ds
\end{bmatrix}.
\end{align}

\paragraph{Proof of VEC Representation}

\begin{align*}
\Delta \boldsymbol{Y}^{*}_t = \boldsymbol{A}_1 +    \boldsymbol{A}_2 \underbrace{ ( \boldsymbol{Y}^{*}_{1,t-1} - \boldsymbol{\mu}_1 ) }_{ Y_{1,t-1} } + \boldsymbol{A}_3 \boldsymbol{D}_t + \boldsymbol{u}_t,     
\end{align*}

\newpage

Our goal is to demonstrate that the data generating process with three equations representing the case of both near unit roots and cointegrated regressors, is equivalent to the original SVAR formulation. 
\begin{align*}
\boldsymbol{Y}_{1t} &=  \left( \boldsymbol{I}_{k_1} + \frac{ \boldsymbol{C} }{ n } \right) \boldsymbol{Y}_{1,t-1} + \boldsymbol{u}_{1t}, \  \boldsymbol{u}_{1t}\in \mathbb{R}^{k_1},
\\
\boldsymbol{Y}_{2t} &= \boldsymbol{\Pi} \boldsymbol{Y}_{1t} + \boldsymbol{u}_{2t}, \  \boldsymbol{u}_{2t} \in \mathbb{R}^{k_2},
\\
\boldsymbol{Y}_{3t} &= \boldsymbol{u}_{3t}, \ \boldsymbol{u}_{3t} \in \mathbb{R}^{k_3}   
\end{align*}
where the above regressors correspond to their demeaned counterparts. In other words we start by rewritting the regressors specific represenation to the equivalent general SVAR process.  
\begin{align}
\begin{bmatrix}
\Delta \boldsymbol{Y}_{1t}^{*}
\\
\Delta \boldsymbol{Y}_{2t}^{*}
\\
\Delta \boldsymbol{Y}_{3t}^{*}
\end{bmatrix} 
=
\underbrace{ \begin{bmatrix}
\displaystyle \frac{ \boldsymbol{C} }{n} & \boldsymbol{0} & \boldsymbol{0}
\\
\displaystyle  \boldsymbol{\Pi} \left( \boldsymbol{I}_{k_1} + \frac{ \boldsymbol{C} }{n} \right) & - \boldsymbol{I}_{ k_2 } & \boldsymbol{0}
\\
\boldsymbol{0} & \boldsymbol{0} & - \boldsymbol{I}_{ k_3 }
\end{bmatrix} }_{ \boldsymbol{M} }
\begin{bmatrix}
\Delta \boldsymbol{Y}_{1t-1}^{*} - \boldsymbol{\mu}_1
\\
\Delta \boldsymbol{Y}_{2t-1}^{*} - \boldsymbol{\mu}_2
\\
\Delta \boldsymbol{Y}_{3t-1}^{*} - \boldsymbol{\mu}_3
\end{bmatrix} 
+
\underbrace{ \begin{bmatrix}
\boldsymbol{I}_{ k_1 }   & \boldsymbol{0} & \boldsymbol{0}
\\
\boldsymbol{\Pi} & \boldsymbol{I}_{ k_2 }  & \boldsymbol{0} 
\\
\boldsymbol{0} & \boldsymbol{0} &  \boldsymbol{I}_{ k_3 }
\end{bmatrix} }_{ \boldsymbol{P} }
\boldsymbol{\Psi} (L)^{-1} \boldsymbol{e}_t,
\end{align}
Recall that we have that $\boldsymbol{u}_t := \boldsymbol{P} \boldsymbol{e}_t$ which implies by the change of basis rule
\begin{align*}
\begin{bmatrix}
\boldsymbol{u}_{1t}
\\
\boldsymbol{u}_{2t}
\\
\boldsymbol{u}_{3t}
\end{bmatrix}
:= 
\begin{bmatrix}
\boldsymbol{I}_{k_1} & \boldsymbol{0} & \boldsymbol{0} 
\\
\boldsymbol{\Pi} &  \boldsymbol{I}_{k_2} & \boldsymbol{0} 
\\
\boldsymbol{0} & \boldsymbol{0} & \boldsymbol{I}_{k_3}
\end{bmatrix}
\begin{bmatrix}
\boldsymbol{e}_{1t}
\\
\boldsymbol{e}_{2t}
\\
\boldsymbol{e}_{3t}
\end{bmatrix}
\equiv
\begin{bmatrix}
\boldsymbol{I}_{k_1} \otimes \boldsymbol{e}_{1t}
\\
\boldsymbol{\Pi} \otimes \boldsymbol{e}_{2t} +   \boldsymbol{I}_{k_2} \otimes \boldsymbol{e}_{2t}
\\
\boldsymbol{e}_{3t}
\end{bmatrix}.
\end{align*}
Thus in matrix notation we have that $\Delta \boldsymbol{Y}^{*}_t = \boldsymbol{M} \left( \boldsymbol{Y}^{*}_{t-1} - \boldsymbol{\mu} \right) + \boldsymbol{P} \boldsymbol{\Psi} (L)^{-1} \boldsymbol{e}_t$. Multiplying both sides by $\boldsymbol{P} \boldsymbol{\Psi} (L) \boldsymbol{P}^{-1}$, we obtain
\begin{align}
\boldsymbol{P} \boldsymbol{\Psi} (L) \boldsymbol{P}^{-1} \Delta \boldsymbol{Y}^{*}_t  = \boldsymbol{P} \boldsymbol{\Psi} (L) \boldsymbol{P}^{-1}  \boldsymbol{M} \big( \boldsymbol{Y}^{*}_{t-1} - \boldsymbol{\mu} \big) + \boldsymbol{P} \boldsymbol{e}_t. 
\end{align}

\begin{remark}
In summary, \cite{cheng2022instrumental} show that that in the presence of stationary regressors, cointegration relationships, or more than one lag  variables in the VAR system, the estimation error from the nonstationary component is asymptotically negligible. Specifically, the authors prove that the asymptotic variance of the IRF only depends on the stationary component, but a consistent covariance matrix estimator is available even without knowing which series are stationary. Although the approach of \cite{cheng2022instrumental} shows that for a proxy SVAR standard asymptotic normal inference remains valid under a general form of nonstationarity in the VAR system, their theoretical result is not uniform over the entire parameter space of the roots of the autoregressive model as in the spirit of \cite{mikusheva2007uniform}.  A relevant framework that demonstrates the robustness of SVAR identification and estimation to the presence of unit root and cointegration dynamics is proposed by \cite{chevillon2020robust}. 
\end{remark}

\newpage

\subsection{Identification with Occasionally Binding Constraints}

At the Zero Lower Bound (e.g., see  \cite{eggertsson2014can} and \cite{aruoba2022svars}) the ability of policy makers to react to macroeconomic shocks is limited and therefore the response of the economy to shocks will change when the policy instrument hits the ZLB. Thus, according to \cite{mavroeidis2021identification} since the difference in the behaviour of the macroeconomic variables across the ZLB and non-ZLB regimes is only due to the impact of monetary policy, the switch across regimes provides information about the causal effects of policy. Therefore, the ZLB identifies the causal effect of policy during the unconstrained regime. However, if monetary policy remains partially effective during the ZLB regime, for example, through the use of unconventional policy instruments, then the difference in the behaviour of the economy across regimes will depend on the difference in the effectiveness of conventional and unconventional policies. In this case, we obtain only partial identification of the causal effects of monetary policy, but we can still get information bounds on the relative efficiency of unconventional policy.

From the econometric perspective a major challenge when the restriction mechanism involves OBCs is that it generates censoring of one of the dependent variables. Specifically, once the censoring mechanism is triggered then one can allows for some of the coefficients for the remaining variables to change. According to \cite{aruoba2022svars} the specification of a dynamic multivariate model with censoring and regime-dependent coefficients faces two challenges: parsimony and the existence of a unique reduced form.  These models distinguish between a shadow rate, $y_{1t}^{*}$, and the actual interest rate $y_{1,t} = \mathsf{max} \left\{ y_{1t}^{*}, c \right\}$ by defining the endogenous variable $s_t = \boldsymbol{1} \left\{ y_{1,t} > c \right\}$. This implies that the regime-dependency of the coefficients is equivalent to capturing nonlinearities in decision rules that arise in DSGE models with occasionally-binding constraints.  Moreover, \cite{duffy2022cointegration} mentions that in an autoregressive model, an occasionally binding constraint naturally give rise to nonlinearity in the form of links between unit roots and stochastic trends. A prototypical model from the VECM literature is 
\begin{align}
\Delta \boldsymbol{y}_t = \mathsf{g} \left( \boldsymbol{\beta}^{\prime} \boldsymbol{y}_{t-1} \right) + \sum_{j=1}^{k-1} \boldsymbol{\Gamma}_j \Delta \boldsymbol{y}_{t-1} + \boldsymbol{u}_t,    
\end{align}

\subsection{On Causal Discovery and Identification}

Causal ordering provides a mechanism for causal effect justification in identification and estimation of SVARs. In particular, \cite{moneta2011causal} studies statistical methodologies for causal inference in SVAR models while other authors propose a graph-theoretic method for causal analysis in vector autoregressions which allows to select the causal (contemporaneous) order of a SVAR model. Thus, a key role to the theory of causal discovery plays the framework for recursive causal models proposed by \cite{kiiveri1984recursive}. The recursive causal modelling approach utilizes the notion of conditional independence (pioneered by \cite{wright1934method} and \cite{wold1960generalization}), for system identification purposes. Moreover, \cite{swanson1997impulse} discuss the estimation of impulse response functions based on a causal approach to residual orthogonalization in vector autoregressions and propose a search algorithm as statistical mechanism for non-recursive structural models while \cite{hausman1983identification} discuss the use of IV regression as an identification mechanism for systems of simultaneous equations

\newpage

In this direction, \cite{silva2017learning} examine structural learning using instrumental variables under the non-Gaussianity assumption. Their IV causal discovery method provides a set of potential causal effects, from which a unique causal discovery is attainable if at least two instrumental variables (under the same conditioning set) are present in the true underline graph. Moreover, to facilitate computational flexibility this methodology focuses on completeness, that is, on the characterization of relation as causals when instrumental variables satisfy the given Graphical criteria, rather than the identification of all causal effects induced from faithfulness assumptions. A graphical approach to estimating high-dimensional VAR models is presented by \cite{fragetta2011effects} and  \cite{bertsche2023directed}.
\begin{example}
Consider the standard linear structural VAR model as below
\begin{align}
\boldsymbol{y}_t = \sum_{j=1}^p \boldsymbol{A}_j \boldsymbol{y}_{t-j} + \boldsymbol{A}_0 \boldsymbol{\varepsilon}_t, \ \ \ \boldsymbol{u}_t = \boldsymbol{A}_0 \boldsymbol{\varepsilon}_t 
\end{align}
We shall demonstrate the above ideas with based on the SVAR process by considering the independence properties of structural shocks. These are summarized as below: 
\begin{itemize}

\item If $\left( \varepsilon_{t,1},..., \varepsilon_{t,K}  \right)$ are mutually independent and non-Gaussian then the error equation is an ICA model and the matrix $\boldsymbol{A}_0$ can be identified up to a post-multiplication of a generalized permutation matrix (see, \cite{eriksson2004identifiability}, \cite{gourieroux2017statistical} and \cite{lanne2017identification}).

\item If for any permutation matrix $\boldsymbol{P}$, the matrix $\boldsymbol{P} \boldsymbol{A}_0^{-1} \boldsymbol{P}^{\top}$ is lower triangular, then the matrix $\boldsymbol{A}_0^{-1}$ is uniquely determined.This is a property that holds for the structural parameter $\boldsymbol{A}_0^{-1}$.

\item The \textit{recursiveness property} of $\boldsymbol{A}_0^{-1}$, implies that a recursive causal structure on the contemporaneous variables of the structural VAR model hold: if the assumption is true, one can re-order the variables entering in $\boldsymbol{y}_t$ in a "Word Causal Chain" (see, \cite{wold1960generalization}), so that each variable $y_{t,i}$ causes $y_{t,j}$ and no variable $y_{t,j}$ causes $y_{t,i}$ for $i < j$ (lower triangular elements) where $i, j \in \left\{ 1,..., K \right\}$.
  
\item  Since the structural parameter $\boldsymbol{A}_0^{-1}$ is a triangular matrix, a convenient way to represent the error vector $\boldsymbol{u}_t$ is with a directed acyclic graph (DAG) (see, discussion in  \cite{moneta2013causal}).

\end{itemize}
Suppose that $\boldsymbol{P} \boldsymbol{A}_0^{-1} \boldsymbol{P}^{\top}$ is lower triangular, without loss of generality, with $\boldsymbol{P} = \boldsymbol{I}$ and $K = 3$. Then, we have the following system of structural equations between the innovation terms: 
\begin{align*}
u_{t,1} &= \varepsilon_{t,1}
\\
u_{t,2} &= \alpha \varepsilon_{t,1} + \varepsilon_{t,2}
\\
u_{t,3} &= \beta \varepsilon_{t,1} + \gamma \varepsilon_{t,2} +  \varepsilon_{t,2}
\end{align*}
In other words, when the true DAG is known, then we are able to add zero restrictions (just-identifying or over-identifying) on the structural parameter $\boldsymbol{A}_0^{-1}$ and recover the independent shocks $\varepsilon_t$ (see,  \cite{cordoni2023identification}). Thus, knowing the causal structure is key for identification in structural econometric models. Notice that the ICA approach works well in the setting of identifying structural shocks because it preserve the causal ordering and causality of variables.
\end{example}

\newpage

\begin{example}
Consider the formulation of the multivariate time series model with respect to contemporaneous and lagged coefficients 
\begin{align}
\boldsymbol{y}_t = \boldsymbol{B} \boldsymbol{y}_t + \boldsymbol{\Gamma}_1 \boldsymbol{y}_{t-1} + ... + \boldsymbol{\Gamma}_p \boldsymbol{y}_{t-p} + \boldsymbol{\varepsilon}_t  
\end{align}
The structural parameter $\boldsymbol{B}$ has a zero diagonal by definition and $\boldsymbol{\varepsilon}_t$ is a $( k \times 1)$ vector of error terms.  The model can be equivalently expressed in the standard SVAR form 
\begin{align}
\boldsymbol{\Gamma}_0  \boldsymbol{y}_t = \boldsymbol{\Gamma}_1 \boldsymbol{y}_{t-1} + ... + \boldsymbol{\Gamma}_p \boldsymbol{y}_{t-p} + \boldsymbol{\varepsilon}_t  
\end{align}
where $\boldsymbol{\Gamma}_0  = ( \boldsymbol{I} - \boldsymbol{B} )$. Moreover, in the standard SVAR model it is assumed that the covariance matrix $\boldsymbol{\Sigma}_{\varepsilon} = \mathbb{E} \left( \boldsymbol{\varepsilon}_t \boldsymbol{\varepsilon}_t^{\prime}  \right)$. However, since the variables $( y_{1t},..., y_{kt} )$ are endogenous, these models cannot be directly estimated without biases. We can thus derive the reduced-form VAR model as below
\begin{align}
\boldsymbol{y}_t &= \boldsymbol{\Gamma}_0^{-1} \boldsymbol{\Gamma}_1 \boldsymbol{y}_{t-1} + ... + \boldsymbol{\Gamma}_0^{-1} \boldsymbol{\Gamma}_p \boldsymbol{y}_{t-p} + \boldsymbol{\varepsilon}_t  
\\
\boldsymbol{y}_t &=  \boldsymbol{A}_1 \boldsymbol{y}_{t-1} + ... + \boldsymbol{A}_p \boldsymbol{y}_{t-p} + \boldsymbol{u}_t  
\end{align}
where $\boldsymbol{u}_t$ is a vector of zero-mean white noise processes with covariance matrix such that $\boldsymbol{\Sigma}_u = \mathbb{E} \left( \boldsymbol{u}_t \boldsymbol{u}_t^{\prime} \right)$, which in general will not be diagonal. Although the VAR parameters can be easily directly estimated from the observed data, this is not sufficient to recover the parameters of the SVAR equation, whose number is larger than the number of parameters in the VAR equation. In other words, the reduced-form VAR model is only adequate for estimation and forecasting purposes and not sufficient for policy analysis (see, \cite{moneta2013causal}). 

Consider the Wold Moving Average (MA) representation of the VAR model such that 
\begin{align}
\boldsymbol{y}_t 
= 
\sum_{j=0}^{\infty} \boldsymbol{\Phi}_j \boldsymbol{u}_{t-j} 
=
\sum_{j=0}^{\infty} \boldsymbol{\Phi}_j \boldsymbol{\Gamma}_0^{-1} \boldsymbol{\Gamma}_0 \boldsymbol{u}_{t-j}  
=
\sum_{j=0}^{\infty} \boldsymbol{\Psi}_j \boldsymbol{\varepsilon}_{t-j}.
\end{align}
where the $\boldsymbol{\Phi}_j$ for $j = 0,1,2,...$ are the MA coefficient matrices and $\boldsymbol{\Psi}_j$ for $j = 0,1,2,...$ represent the impulse responses of the elements of $\boldsymbol{y}_t$ to the shocks of $\boldsymbol{\varepsilon}_{t-j}$ for $j = 0,1,2,...$
\begin{align}
\boldsymbol{\Phi}_0 = \boldsymbol{I}, \ \ \ \boldsymbol{\Phi}_i = \sum_{j=1}^i \boldsymbol{A}_j \boldsymbol{\Phi}_{i-j}, \ \ \ \boldsymbol{\Psi}_0 = \boldsymbol{\Gamma}_0^{-1}, \ \ \ \boldsymbol{\Psi}_i = \sum_{j=1}^i \boldsymbol{A}_j \boldsymbol{\Psi}_{i-j}.   
\end{align}
where $\boldsymbol{A}_j = 0$ for $j > p$. Therefore, we can see that the impulse response coefficients $\boldsymbol{\Psi}_j$ can be obtained from the reduced-form VAR parameters $\boldsymbol{A}_j$ only if we know the SVAR coefficient matrix $\boldsymbol{\Gamma}_0 = (\boldsymbol{I} - \boldsymbol{B})$.  

Since the impulse responses are crucial for policy analysis, it is clear that we need to recover the SVAR representation for this purpose. The problem is that any invertible unit-diagonal matrix $\boldsymbol{\Gamma}_0$ is compatible with the coefficient matrices we obtain by estimating the VAR reduced form. In other words, we shall find the correct $\boldsymbol{\Gamma}_0$ which produces the right transformation $\boldsymbol{\Gamma}_0 \boldsymbol{u}_t = \boldsymbol{\varepsilon}_t$ of the VAR error terms $\boldsymbol{u}_t$. The equivalence between covariance restrictions and IV was first discussed by \cite{hausman1983identification}. 
\end{example}

\newpage

\begin{example}[Impulse Response Functions]
The impulse response functions are calculated considering the system in levels. Specifically, the forecast error of the $h-$step forecast of $Y_t$ is given by 
\begin{align}
\big( Y_{t+h} - Y_t (h) \big) = u_{t+h} + \Phi_1 u_{t+h-1} + ... +  \Phi_{h-1} u_{t+1}.   
\end{align}
Notice that the $\Phi_i$ are obtained from the $A_i$ recursively by the following expression 
\begin{align}
\Phi_i = \sum_{j=1}^i \Phi_{i-j} A_j, \ \ \ i = 1,2,...    
\end{align}
with $\Phi_0 = I_k$. Since $v_t = ( I - B_0 ) u_t$, then the forecast-error of the $h-$step ahead forecast of $Y_t$ is
\begin{align}
\big( Y_{t+h} - Y_t (h) \big) = \Theta_0 v_{t+h} + \Theta_1 v_{t+h-1} + ... + \Theta_{h-1} v_{t+1}.   
\end{align}
where $\Theta_i = \Phi_i ( I - B_0 )^{-1}$. Notice that the element $(j,k)$ of $\Theta_i$ represents the response of the variable $y_j$ to a unit shock in the variable $y_k$, $i$ periods ago. 
\end{example}

\begin{example}[Uncertainty Shocks]
A growing steam of literature considers the identification of uncertainty shocks especially in relation to financial uncertainty and macroeconomic conditions. Since the work of \cite{bloom2009impact} (see, also   \cite{jurado2015measuring} and \cite{baker2016measuring}) a vast literature studies on the role of uncertainty shocks for macroeconomic fluctuations. In particular, measuring uncertainty is important for correctly measuring the impact it has on macroeconomic variables through structural analysis (e.g., see \cite{trung2019spillover}). We can show that the identification procedure proposed in the study of \cite{forni2023macroeconomic} is equivalent to the Proxy-SVAR methodology (see also \cite{carriero2023macro}).   
\begin{proof}
We can show that the OLS estimates are identical to those in the study of \textcolor{blue}{Mertens and Ravn (2013)} if the number of lags of $y_t$ included in the regression of the exogenous instruments $z_t$ is equal to the number of lags of the VAR for $y_t$. Since $A(L) y_t = \mu + \varepsilon_t$ 
\begin{align}
y_t = \mu - A_1 y_{t-1} - ... - A_p y_{t-p} + \varepsilon_t.    
\end{align}
Suppose that $k < p$ such that $k \in \left\{ 1,..., p \right\}$, then using the stack vectors forms we obtain
\begin{align}
Y_k = 
\begin{pmatrix}
y_{p+1- k}^{\prime}
\\
y_{p+2- k}^{\prime}
\\
\hdots
\\
y_{n - k}^{\prime}
\end{pmatrix}, \ \ \ 
X = \big( 1 \ Y_1 \ ... Y_p \big), \ \ \ 
\mathcal{E} =
\begin{pmatrix}
\varepsilon_{p+1- k}^{\prime}
\\
\varepsilon_{p+2- k}^{\prime}
\\
\hdots
\\
\varepsilon_{n - k}^{\prime}
\end{pmatrix}
\end{align}
Moreover, let $Y = Y_0$. Hence, the VAR equation can be written as $Y = X A + \mathcal{E}$ (proof to be completed). 
    
\end{proof}
Further studies which consider the identification of SVARs under the presence of policy uncertainties include \cite{mumtaz2018policy} (see, also  \cite{ramey2011identifying}).
\end{example}

\newpage 

\subsection{Identifying Aggregate Fluctuations}

A particular stream of literature focuses on the identifying aggregated fluctuations through sectoral and firm-specific dynamics. Specifically, \cite{gabaix2011granular} proposed a framework for investigating the granular origins of aggregate fluctuations. In addition, \cite{gabaix2023granular} illustrate when the market is sufficiently concentrated, then one can use the collection of idiosyncratic shocks to individual micro units, at each time period, as an instrument for endogenous aggregate variables. Moreover, a full-information approach to Granular Instrumental variables is proposed by \cite{baumeister2023full} while a framework for heterogeneity-robust granular instruments is presented by \cite{qian2023heterogeneity}. Lastly,  \cite{banafti2022inferential} establish inferential theory for GIVs in high dimensional settings. 

\subsubsection{Granular Instrumental Variables}

\begin{example}
Suppose we are willing to assume that the shock to country $j'$s demand consist of a global demand shock $f_{ct}$ that affects all consumers across economies in the same way and a purely idiosyncratic shock $\eta_{cjt}$. In other words, the GIV approach provides a systematic way of constructing instruments from suitably weighted idiosyncratic shocks (using observational datasets) and use them as instruments for aggregate endogenous variables. Then, the framework of \cite{banafti2022inferential} aims to generalize the method of granular instrumental variable. 
\begin{align}
d_t = \phi^d p_t + \varepsilon_t, \ \ 
y_{it} = \phi^s p_t + \boldsymbol{\lambda}_i^{\prime} \boldsymbol{\eta}_t + u_{it}    
\end{align}
\end{example}
The global market clearing condition is given by $y_{St} = d_t$, where $y_{St} := \boldsymbol{S}^{\prime} \boldsymbol{y}_{\cdot t} = \sum_{i=1}^N S_i y_{it}$, where $\boldsymbol{S}$ is the $( N \times 1)$ vector of shares that are normalized such that $\sum_{i=1}^N S_i = 1$.

\subsubsection{Feasible Granular Instrumental Variables}

Consider the following panel simultaneous equations model with factor error structure
\begin{align}
\boldsymbol{y}_{it} &= \boldsymbol{B} \boldsymbol{x}_{it} + \boldsymbol{C} \boldsymbol{\alpha}_t + \boldsymbol{v}_{it}
\\
\boldsymbol{v}_{it} &= \boldsymbol{\Lambda}_i^{\prime} \boldsymbol{F}_t + \boldsymbol{u}_{it}
\end{align}
Asymptotic distribution for structural parameters in factor augmented regressions in time series and panel models is already well-developed in the literature. When both $N$ and $T$ are large, then it can be shown that the sampling error from estimating the high dimensional precision matrix, the factors, as well as the instrument is negligible in the asymptotic distribution of the structural parameters. Moreover, in the case of supply elasticity, the estimator will additionally depend on the estimated (potentially high dimensional) precision matrix. Therefore, the contribution of \cite{banafti2022inferential} to the GIV literature focuses around relaxing the homogenous loadings assumption which is overly restrictive and propose constructing  an estimate of the instrument using PCA or iterative OLS-PCA methods.

\newpage 

\subsection{Extensions}

\subsubsection{Extension 1: VARMA Representation}

Suppose that the $k-$dimensional process $\left\{ \boldsymbol{Y}_t: t \in \mathbb{Z}\right\}$ has the following VARMA representation 
\begin{align}
Y_t - \sum_{h=1}^p \Phi_h Y_{t-h} = \varepsilon_t - \sum_{h=1}^q \Theta_h \varepsilon_{t-h},    
\end{align}
where $\left\{ \varepsilon_t \right\}$ is the innovation process, a sequence of independent and identically distributed random variables with zero-mean and an invertible covariance matrix $\Sigma$ (see, also \cite{gourieroux2020identification}).  

The following representation also holds
\begin{align}
\left( I_k - \sum_{j=1}^p \tilde{\Phi}_0 \tilde{\Phi}_j B^j \right) Y_t =  \left( I_k - \sum_{j=1}^q \tilde{\Phi}_0 \tilde{\Theta}_j B^j \right) \varepsilon_t.   
\end{align}

\begin{theorem}[see, \cite{melard2006exact}]
Let $\Phi (1)$ be a $(k \times k)$ matrix such that $\mathsf{rank} \left\{ \Phi (1)  \right\} = (k - d ) = r$, where $0 < d < k$. Let $P_1$ be any $( k \times d )$ matrix such that $\Phi (1) P_1 = 0$ and $P_2$ be any $( k \times r)$ matrix such that $P = [ P_1 , P_2 ]$ is invertible and the columns of $P_2$ are orthogonal to those of $P_1$.  Let $P^{-1} \equiv Q = \big[ Q_1^{\prime},  Q_2^{\prime}  \big]^{\prime}$ stands for the inverse of $P$, with $Q_1$ and $Q_2$ having $d$ and $r$ rows, respectively then the matrices $P_1 Q_1$ and $P_2 Q_2$ are uniquely determined. The singular value decomposition of $\Phi (1)$ can be written as $\boldsymbol{\Phi} (1) =  \boldsymbol{U} \boldsymbol{D} \boldsymbol{V}^{\prime}$, where $\boldsymbol{U}$ and $\boldsymbol{V}$ are $( k \times k)$ orthogonal matrices.
\end{theorem}
(Sketch Proof)
\begin{itemize}

\item  The null space of $\Phi (1)$, that is, the set of vectors $x$ such that $\Phi (1) x = 0$, is the space generated by the columns of the matrix $G = [ v_{r+1},..., v_k   ]$. Now let $P_1$ and $P_1^{*}$ by any $( k \times d )$ matrices such that $\Phi (1) P_1 = 0$ and $\Phi^{*} (1) P_1 = 0$.

\item  Since $P_1$ and $P_1^{*}$ are full rank matrices, there exist full rank $( d \times d)$ matrices $\alpha$ and $\alpha^{*}$ such that $P_1 = G \alpha$ and $P_1^{*} = G \alpha^{*}$. Hence, we have that $P_1^{*} = P_1 \alpha^{-1} \alpha^{*}$.   

\item  Similarly, it can be proved that $P_2^{*} = P_2 \beta^{-1} \alpha^{*}$. Then, using the inverse matrix representations, it can be proved that  $P_1^{*} Q_1^{*} = P_1 Q_1$ and $P_2^{*} Q_2^{*} = P_2 Q_2$, which demonstrates that indeed the matrices $P_1 Q_1$ and $P_2 Q_2$ are uniquely defined (see, \cite{melard2006exact}).     

\end{itemize}

\paragraph{Forecasting in VARMA models}

Modelling and forecasting a large set of macroeconomic variables can be done using VARMA representations. However, to ensure a robust estimation,  the identification of uniquely parametrized VARMA models requires imposing restrictions on the parameter space to ensure that the model is uniquely identified (see, \cite{dias2018estimation}). 
\begin{align}
\boldsymbol{A}_0 \boldsymbol{Y}_t = \underbrace{ \boldsymbol{A}_1 \boldsymbol{Y}_{t-1} + ... + \boldsymbol{A}_p \boldsymbol{Y}_{t-p} }_{ \text{AR component} }+ \boldsymbol{B}_0 \boldsymbol{u}_t + \underbrace{ \boldsymbol{B}_1 \boldsymbol{u}_{t-1} + ... + \boldsymbol{B}_q \boldsymbol{u}_{t-q} }_{ \text{MA component} } 
\end{align}
Define the lag polynomials such that 
\begin{align*}
A(L) = A_0 - A_1 L - A_2 L^2 - ... - A_p L^p
\\
B(L) = B_0 - B_1 L - B_2 L^2 - ... - B_q L^q
\end{align*}
We shall say that the model is unique if there is only one pair of stable and invertible polynomials $A(L)$ and $B(L)$, respectively which satisfies the canonical MA representation (see, also \cite{poskitt2006identification})
\begin{align*}
Y_t = A(L)^{-1} B(L) \equiv \Theta (L) u_t = \sum_{j=0}^{\infty} \Theta_j u_{t-j}.
\end{align*}


\begin{assumption}
Consider the $\mathsf{VARMA}( p_1, q_1 )$ model defined by the stochastic difference equations
\begin{align}
\label{VARMA}
\left( \boldsymbol{I}_k - \sum_{j=1}^{p_1} A_j L^j \right) X_t = \left( \boldsymbol{I}_k + \sum_{j=1}^{p_1} B_j L^j \right) \varepsilon_t
\end{align}
Moreover, suppose that all solutions of
\begin{align*}
\mathsf{det} \left( \boldsymbol{I}_k - \sum_{j=1}^{p_0} A_j z^j \right) = 0,
\ \ \ 
\mathsf{det} \left( \boldsymbol{I}_k + \sum_{j=1}^{q_0} B_j z^j \right) = 0, \ \ \text{for some} \ \ z \in \mathbb{Z},
\end{align*}
lie outside the unit ball in $\mathbb{C}$, then \eqref{VARMA} is a stable solution (see, \cite{hallin2004rank}). 
\end{assumption}

\begin{remark}
Notice that in contrast to the reduced form VAR models, setting $A_0 = B_0 = I_k$ is not a sufficient condition (only ensures existence) to ensure a unique VARMA representation. The uniqueness of the $\mathsf{VARMA}(p,q)$ representation is guaranteed by imposing restrictions on the pair of stable and invertible polynomials $A(L)$ and $B(L)$ (see, also \cite{deistler1983properties}).  
\end{remark}

\begin{remark}
A second important aspect to emphasize is that the framework proposed by \cite{hallin2004rank} presents with a mathematical rigorous way the construction of an optimal-rank test statistic for testing the adequacy of the lag orders of a VARMA model. Generally, any test statistics associated with any parameter vectors from SVAR processes, should have power functions with suitable properties (such as the monotonicity property) to detect departures from the null hypothesis. Notice that optimal-rank based tests remain valid under arbitrary elliptically symmetric innovation densities, including those with infinite variance and heavy-tails. In particular, the authors show that these optimal-rank based tests are uniformly more powerful than those based on cross-covariances. Moreover, based on the classical LAN result it allows to derive testing procedures that are locally and asymptotically optimal under a given innovation density $f$, based on a non-Gaussian form of cross-covariances. 
\end{remark}
Suppose that our aim is to construct a statistic for testing the null hypothesis $\boldsymbol{\theta} = \boldsymbol{\theta}_0$ against the alternative $\boldsymbol{\theta} \neq \boldsymbol{\theta}_0$. Choosing $p_0 < p_1$ and $q_0 < q_1$ allows one to test the adequency of the specified VARMA coefficients in $\boldsymbol{\theta}_0$, while contemplating the possibility of possibly higher-order VARMA models.

\newpage

Therefore, this null hypothesis is thus invariant under the group of affine transformation $\boldsymbol{\varepsilon}_t \mapsto \boldsymbol{M} \boldsymbol{\varepsilon}_t$ if and only if $\boldsymbol{M} \boldsymbol{A}_i \boldsymbol{M}^{-1} = \boldsymbol{A}_i$ for all $i = 1,...p_0$ and $\boldsymbol{M} \boldsymbol{B}_j \boldsymbol{M}^{-1} = \boldsymbol{B}_j$ for all $j = 1,..., q_0$, that is, \textit{iff} each $\boldsymbol{A}_i$ and each $\boldsymbol{B}_j$ commutes with any invertible matrix $\boldsymbol{M}$, which holds true \textit{iff} they are proportional to the $( k \times k )$ identify matrix (see, \cite{hallin2004rank}).

Write $\mathcal{H}^{(n)} ( \boldsymbol{\theta}_0, \boldsymbol{\Sigma}, f )$ for the hypothesis under which an observation $\boldsymbol{X}^{(n)} := \big( \boldsymbol{X}_1^{(n)},..., \boldsymbol{X}_n^{(n)} \big)^{\prime}$ is generated by the $\mathsf{VARMA}(p_0, q_0)$ model with parameter value $\boldsymbol{\theta}_0$ satisfying Assumption A and innovation process satisfying Assumption B. Our objective is to test $\mathcal{H}^{(n)} ( \boldsymbol{\theta}_0 ) := \bigcup_{\boldsymbol{\Sigma} } \bigcup_{ f }  \mathcal{H}^{(n)} ( \boldsymbol{\theta}_0, \boldsymbol{\Sigma}, f )$  against $\bigcup_{ \boldsymbol{\theta} \neq \boldsymbol{\theta}_0  } \mathcal{H}^{(n)} ( \boldsymbol{\theta} )$.   Consequently, $\boldsymbol{\Sigma}$ and $f$ play the role of nuisance parameters; note that the unions, in the definition of $\mathcal{H}^{(n)} ( \boldsymbol{\theta}_0 )$, are taken over all possible values of $\boldsymbol{\Sigma}$ and $f$.

Let $\boldsymbol{A}(L)$ and $\boldsymbol{B}(L)$ be such that $\boldsymbol{A}_i = \boldsymbol{0}$ for $i = p_0 + 1,..., p_1$, and $\boldsymbol{B}_i = \boldsymbol{0}$ for $i = q_0 + 1,..., q_1$, and consider the sequences of linear difference operators 
\begin{align}
\boldsymbol{A}^{(n)}(L) &:= \boldsymbol{I}_k - \sum_{i=1}^{p_1} \left( \boldsymbol{A}_i + n^{-1/2} \boldsymbol{\gamma}_i^{(n)}  \right) L^i 
\\
\boldsymbol{B}^{(n)}(L) &:= \boldsymbol{I}_k + \sum_{i=1}^{p_1} \left( \boldsymbol{B}_i + n^{-1/2} \boldsymbol{\delta}_i^{(n)}   \right) L^i 
\end{align}
where the vector 
\begin{align}
\boldsymbol{\uptau}^{(n)} := \left(  \mathsf{vec}  ( \boldsymbol{\gamma}_1^{(n)} )^{\prime},..., \mathsf{vec}  ( \boldsymbol{\gamma}_{p_1}^{(n)} )^{\prime}, \mathsf{vec}  ( \boldsymbol{\delta}_1^{(n)} )^{\prime},..., \mathsf{vec}  ( \boldsymbol{\delta}_{q_1}^{(n)} )^{\prime}  \right)^{\prime} \in \mathbb{R}^{ k^2 ( p_1 + q_1 )}    
\end{align}
is such that $\mathsf{sup}_n ( \boldsymbol{\uptau}^{(n)} )^{\prime} \boldsymbol{\uptau}^{(n)} < \infty$. These operators define a sequence of VARMA models such that
\begin{align}
\boldsymbol{A}^{(n)} (L) \boldsymbol{X}_t = \boldsymbol{B}^{(n)} (L) \boldsymbol{\varepsilon}_t, \ \ \ \ t \in \mathbb{Z},   
\end{align}
hence, the sequence of local alternatives $\mathcal{H}^{(n)} \left( \boldsymbol{\theta}_0 + n^{-1/2} \boldsymbol{\uptau}^{(n)}, \boldsymbol{\Sigma}, f \right)$.

\medskip

\paragraph{Testing for Invertibility}

A key assumption for valid identification, estimation and forecasting with VARMA models (see, \cite{gourieroux2020identification}, \cite{sims2012news} and  \cite{lutkepohl2002forecasting}), is the invertability of underline polynomial representations. In particular, \cite{chen2017testing} propose a test for invertability or fundamentalness of structural vector autoregressive moving average models generated by non-Gaussian independent and identically distributed structural shocks. Moreover, it can be proved that in these models under certain regularity conditions the Wold innovations are a martingale difference sequence \textit{if and only if} the structural shocks are fundamental. Thus, the particular representation provides a mechanism for testing the presence of invertability. In other words, \cite{chen2017testing} convert the statistical problem of polynomial invertability of VARMA processes into a statistical testing problem for the martingale difference sequence property of the Wold innovations.     An MLE approach for non-invertible ARMA models is studied by \cite{meitz2013maximum}. Further issues on identification and estimation of non-invertible SVARMA models are presented by \cite{funovits2020identifiability}.

\newpage

\subsubsection{Extension 2: Non-Linear Dynamic System}

A more challenging task is the econometric identification in nonlinear systems. In addition, for a non-linear dynamic system excited by a non-Gaussian disturbances, the system equation governing the probability distribution of the system response contains an infinite number of terms which makes identification and estimation considerable complex. Therefore, a pertrubation scheme is proposed to obtain an approximate solution to the system, thereby the effects of excitation non-Gaussianity can be evaluated (see, \cite{cai1992response}, \cite{grigoriu1995linear} and \cite{sengupta1989efficient}).

Consider the unit-root VAR model, $\boldsymbol{A} (z) \boldsymbol{y}_t = \boldsymbol{\varepsilon}_t$, which crucially depends on $\boldsymbol{A} (z)^{-1} (L)$, thus we address the issue of of the inversion of $\boldsymbol{A} (z)$ around a pole $z = z_0$. Denote with
\begin{align}
\boldsymbol{A} (z) = \sum_{j=0}^K \boldsymbol{A}_j z^j, \ \ \ \boldsymbol{A}_j \neq 0,     
\end{align}
be a matrix polynomial of order $n$ and degree $K$ and $z_0$ such that $\mathsf{deg} \big[ \boldsymbol{A} (z) \big] = 0$. Therefore, 
\begin{align}
\boldsymbol{A} (z) = \sum_{j=0}^K \frac{1}{j!} \boldsymbol{A}^{(j)} ( z - z_0 )^j    
\end{align}
Then, an analytical expression for the matrix function $\boldsymbol{A} (z)^{-1}$ is given by 
\begin{align}
\boldsymbol{A} (z)^{-1} = \sum_{ j = -m }^{ + \infty } \boldsymbol{N}^{j} ( z - z_0 )^j = \sum_{ j = -m }^{ -1} \boldsymbol{N}^{j} ( z - z_0 )^j  + \boldsymbol{M}(z).  
\end{align}
where $m$ is the order of the pole and the first term above represents the principal part while $\boldsymbol{M}(z) = \sum_{ j = 0 }^{ + \infty  } \boldsymbol{N}^{j} ( z - z_0 )^j$, represents the regular part. Using the fact that $\boldsymbol{A} (z) \boldsymbol{A} (z)^{-1} = \boldsymbol{I}_n$ 
\begin{align}
\boldsymbol{I}_n  =  \left(  \sum_{ j = -m }^{ + \infty } \boldsymbol{N}^{j} ( z - z_0 )^j  \right)  \left( \sum_{k=0}^K \frac{1}{k!} \boldsymbol{A}^{(k)} ( z - z_0 )^k \right) 
&=  
\sum_{ k = -m }^{ + \infty } \left( \sum_{ j = 0 }^{ m + k }  \boldsymbol{N}_{k-j} \boldsymbol{A}^{(j)} \right) ( z - z_0 )^k 
\\
&=  
\sum_{ h =  0 }^{ + \infty }  \left( \sum_{j=0}^h \frac{1}{j!} \boldsymbol{A}^{(j)} \boldsymbol{N}_{h-m-j} \right)  ( z - z_0 )^{ h - m }
\end{align}

Notice that the Unit-Root VAR model is a linear non-homogeneous difference-equation system in matrix form whose solution can be formally written as below
\begin{align}
\boldsymbol{y}_t = \boldsymbol{A}^{-1} (L) \boldsymbol{\varepsilon}_t + \boldsymbol{A}^{-1} (L)     
\end{align}
where the first term is a particular solution of the non-homogeneous equation and the second is the so-called complementary solution. However, both depend on the operator $\boldsymbol{A}^{-1} (L)$, and thus on the matrix $\boldsymbol{A}^{-1} (z)$ since the algebras of the polynomial functions of the lag operator $L$ and of the complex variable $z$ are isomorphic. These expressions are useful when deriving statistical tests for invertability.

\newpage

\subsubsection{Extension 3: Automatic Inference for Finite Order VARs}

Selecting lag length is useful when $p$ is unknown, which can directly impact the estimation of the effects of structural shocks. Notice that we only discuss lag selection procedures for stationary settings\footnote{A relevant framework for model selection in explosive time series regressions is presented by \cite{tao2020model}. }. Roughly speaking the selection of the optimal lag length is achieved via the use of information criteria such as the AIC, SIC and HQC (e.g., see \cite{ivanov2005practitioner} and \cite{choi2012model}).

In this section, we focus on the framework proposed by \cite{kuersteiner2005automatic} who consider automatic inference for infinite order vector autoregressions. In particular, infinite order dynamics in cointegration models has been examined by several studies. Specifically, $y_t$ has an infinite order VAR representation  
\begin{align}
y_t = \mu + \sum_{j=1}^{\infty} \xi_j y_{t-j} + v_t, \ \ \ C(L) =  \sum_{j=0}^{\infty} C_j L^j.    
\end{align}
where $\mu = C(1)^{-1} \mu_y$ and $C(L)^{-1} = \xi (L)$ such that $\xi (L) = \displaystyle \left(  I - \sum_{j=1}^{\infty} \xi_j L^j  \right)$. Therefore, the approximate model with VAR coefficient matrices $\left\{ \xi_{1,p},..., \xi_{p,p}  \right\}$ has the following formulation
\begin{align}
y_t = \mu_{y,p} + \xi_{1,p} y_{t-1} + ... + \xi_{p,p} y_{t-p} + v_{t,p},    
\end{align}
where $\Sigma_{v,p} = \mathbb{E} \big[ v_{t,p} v_{t,p}^{\prime} \big]$ the mean squared prediction error of the approximating model. Moreover, it can be shown that the parameters $( \xi_{1,p},..., \xi_{p,p} )$ are root$-n$ consistent and asymptotically normal for $\xi (p) = ( \xi_1,..., \xi_p )$ under the condition that the lag order $p$ does not increase too quickly, that is, $p$ is chosen such that $p^3 / n \to 0$. 

Consider the sequential testing procedure proposed by \cite{ng2001lag}, then a Wald test for the null hypothesis that the coefficients of the last lag $p$ are jointly zero, is expressed by
\begin{align}
\mathcal{W} (p,p) = \left( \mathsf{vec} ( \hat{\xi}_{p,p} ) \right)^{\prime} \left[ \hat{\Sigma}_{v,p} \otimes M_p^{-1} (1) \right]^{-1} \left( \mathsf{vec} ( \hat{\xi}_{p,p} ) \right).    
\end{align}

\begin{definition}
\label{Definition1}
The general-to-specific procedure proposed by \cite{ng2001lag} is based on:
\begin{itemize}
\item[\textit{(i).}] $\hat{p}_n = p$ if, at significance level $\alpha \in (0,1)$, $\mathcal{W} (p,p)$ is the first statistic in the sequence $\mathcal{W} (j,j)$ such that $\left\{ j = p_{\mathsf{max}},..., 1  \right\}$, which is significantly different than zero or

\item[\textit{(ii).}] $\hat{p}_n = 0$ if, $\mathcal{W} (j,j)$ is not significantly different from zero for all $\left\{ j = p_{\mathsf{max}},..., 1  \right\}$, where $p_{\mathsf{max}}$ is such that $p_{\mathsf{max}} / n \to 0$ and $n^{1/2} \sum_{j= p_{\mathsf{max}} + 1}^{\infty} \norm{ \xi_j } \to 0$ as $n \to \infty$.
     
\end{itemize}
\end{definition}

\begin{lemma}
Let $Y_t = \big( y_t^{\prime} - \mu_y^{\prime},  y_{t-1}^{\prime} - \mu_y^{\prime},...,  y_{t-p+1}^{\prime} - \mu_y^{\prime} \big)^{\prime}$, and let 
\begin{align}
\Gamma_{1,p} = \frac{1}{ (n - p) } \sum_{t=1}^{n-1} Y_{t,p} ( y_{t+1} - \mu_y )^{\prime}  
\ \ \ \ 
\Gamma_{p} = \frac{1}{ (n - p) } \sum_{t=p}^{n-1} Y_{t,p} Y_{t,p}^{\prime}  
\end{align}

\newpage

Denote with $\xi ( \hat{p}_n )^{\prime} = \Gamma_{ 1, \hat{p}_n }^{\prime} \Gamma_{\hat{p}_n }^{\prime}$ such that $\xi ( \hat{p}_n )$ is given by Definition \ref{Definition1} above, then it holds that 
\begin{align}
\norm{  \hat{\xi} ( \hat{p}_n )  -  \xi^o ( \hat{p}_n )  } = o_p( n^{ - 1/2} ).    
\end{align}
\end{lemma}

\begin{proof}
Choose $\delta$ such that $0 < \delta < \frac{1}{3}$, and pick a sequence $p_{\mathsf{min} }^{*}$ such that
\begin{align*}
p_{ \mathsf{max} } \geq \textcolor{blue}{ p^{*}_{\mathsf{min}} } \ \ \ \text{and} \ \ \ ( p_{ \mathsf{max} } - \textcolor{blue}{ p^{*}_{\mathsf{min}} } ) \to \infty,
\end{align*}
where $n^{1/2} \sum_{j= \textcolor{blue}{ p^{*}_{\mathsf{min}} } + 1}^{\infty} \norm{ \xi_j } \to 0$. Define with 
\begin{align}
p_{ \mathsf{min} } = \mathsf{max} \left\{ \textcolor{blue}{ p^{*}_{\mathsf{min}} },  p_{\mathsf{max}} - \left( n^{1/2} \sum_{ j = \textcolor{blue}{ p^{*}_{\mathsf{min}} } + 1 }^{\infty} \norm{ \xi_j } \right)^{-1 } \right\}.    
\end{align}
Moreover, since $p_{ \mathsf{min} } \leq p_{ \mathsf{max} }$, and by changing sign of the expression above it follows that 
\begin{align}
\mathsf{min} \left\{ p_{\mathsf{max}} - \textcolor{blue}{ p^{*}_{\mathsf{min}} },   \left( n^{1/2} \sum_{ j = \textcolor{blue}{ p^{*}_{\mathsf{min}} } + 1 }^{\infty} \norm{ \xi_j } \right)^{-1 } \right\}  
=    
p_{\mathsf{max}} - p_{\mathsf{min}} \leq \left( n^{1/2} \sum_{ j = \textcolor{blue}{ p^{*}_{\mathsf{min}} } + 1 }^{\infty} \norm{ \xi_j }^{-1} \right)
\end{align}
such that $\left( p_{\mathsf{max}} - p_{\mathsf{min}} \right) \to \infty$ and $\left( p_{\mathsf{max}} - p_{\mathsf{min}} \right) = \mathcal{O}_p ( n^{\delta} )$. Furthermore, we have that 
\begin{align*}
\left( \hat{\Gamma}_{ 1, \hat{p}_n } -  \Gamma^o_{ 1, \hat{p}_n }  \right) 
&= 
\frac{1}{ \left( n - \hat{p}_n \right) } \sum_{ t = \hat{p}_n }^{ (n - 1) }  Y_{t, \hat{p}_n } \left( \mu_y - \bar{y} \right)^{\prime} 
+
\big( \boldsymbol{1}_{ \hat{p}_n } \otimes ( \mu_y - \bar{y} ) \big) \frac{1}{ \left( n - \hat{p}_n \right) } \sum_{ t = \hat{p}_n }^{ (n - 1) } ( y_{t+1} - \mu_y )^{\prime}
\\
&+
\frac{ ( n - 1 - \hat{p}_n ) }{ ( n - \hat{p} ) } \big( \boldsymbol{1}_{ \hat{p}_n } \otimes ( \mu_y - \bar{y} ) \big) ( y_{t} - \mu_y )^{\prime}.
\end{align*}
Therefore, we can prove that 
\begin{align*}
\mathbb{E} \underset{ p \in  P_n }{ \mathsf{max} }  \norm{  \frac{1}{(n - p)} \sum_{t = p}^{ (n-1)} \big( y_{t+1} - \mu_y \big) }^2  
&\leq
\sum_{ p = p_{\mathsf{min} }  }^{ p_{ \mathsf{max} } } \frac{1}{ ( n - p )^2 } \sum_{ t,s = p }^{(n-1)} \mathsf{trace} \left\{ \mathbb{E} \big[ \big( y_{t+1} - \mu_y \big) \big( y_{t+1} - \mu_y \big)^{\prime} \big] \right\}
\\
&\leq
\mathcal{O}_p ( n^{\delta} ) \frac{1}{ ( n - p_{\mathsf{max} } )^2 } \sum_{ t,s = p_{ \mathsf{min} } }^{(n-1)} \left| \mathsf{trace} \left\{ \mathbb{E} \big[ \big( y_{t+1} - \mu_y \big) \big( y_{t+1} - \mu_y \big)^{\prime} \big] \right\} \right|
\\
&= \mathcal{O}_p ( n^{ -1 + \delta } ).
\end{align*}

\end{proof}

\newpage 

\subsubsection{Extension 4: Locally Trimmed OLS Procedure for VARs}

A particular methodology which has been proposed in the literature to deal with the problem of outliers when estimating vector autoregression models, is the least-trimmed squares procedure. In particular, \cite{agullo2008multivariate} proposed a multivariate least-trimmed squares estimator and prove Fisher-consistency under the assumption of elliptically symmetric error distributions. Moreover, \cite{cavaliere2009robust} consider robust methods for estimation and unit root testing in autoregressions with infrequent outliers whose number, size, and location can be random and unknown\footnote{In particular, the authors show that standard inferene based on ordinary least squares estimation of an augmented Dickey-Fuller (ADF) regression may not be reliable, because (a) clusters of outliers may lead to inconsistent estimation of the autoregressive parameters, and (b) large outliers induce a jump component in the asymptotic distribution of UR test statistics.}.

Let $\left\{ y_t | t \in \mathbb{Z} \right\}$ be a $d-$dimensional stationary time series. The vector autoregressive model of order $p$, denoted by VAR$(p)$, is given by the following expression
\begin{align}
y_t = \mathcal{B}_0^{\prime} + \mathcal{B}_1^{\prime} y_{t-1} + ... + \mathcal{B}_p^{\prime} y_{t-p} + \epsilon_t, 
\end{align} 
with $y_t$ a $d-$dimensional vector, the intercept parameter  $\mathcal{B}_0$ a vector in $\mathbb{R}^d$ and the slope parameters $\left\{ \mathcal{B}_1, ...., \mathcal{B}_p \right\}$ are matrices in $\mathbb{R}^{d \times d}$. Throughout the paper, $B^{\prime}$ will stand for the transpose of the matrix $B$. Moreover, the $d-$dimensional error terms $\epsilon_t$ are assumed to be independently and identically distributed with a density function of the form 
\begin{align}
f_{ \epsilon t } = \frac{ g \left( u^{\prime} \Sigma u \right) }{ \left( \text{det} \Sigma \right)^{1 / 2} },
\end{align}
with $\Sigma$ a positive definite matrix, called the scatter matrix, and $g$ a positive function. If the second moment of $\epsilon_t$ exists, $\Sigma$ will be proportional to the covariance matrix of the error terms. Existence of a second moment, will not be required for the robust estimator. We focus on the unrestricted VAR$(p)$ model, where no restrictions are put on the parameters $\left\{ \mathcal{B}_1, ...., \mathcal{B}_p \right\}$. Suppose that the multivariate time series $y_t$ is observed for $t=1,...,T$. Then, the vector autoregression model has the following multivariate regression representation 
\begin{align}
y_t = \mathcal{B}^{\prime} x_t + \epsilon_t
\end{align}
for $t \in \left\{ p+1,..., T \right\}$ with $x_t = \left( 1, y_{t-1}^{\prime}, ..., y_{t-p}^{\prime} \right)^{\prime} \in \mathbb{R}^q$, where $q = pd + 1$. Moreover, the matrix $\mathcal{B} = \left( \mathcal{B}_1, ...., \mathcal{B}_p \right)^{\prime} \in \mathbb{R}^{q \times d}$ contains all the unknown matrices of coefficients. In matrix notation, we have that $X = \left( x_{p+1},...,x_{T} \right)^{\prime} \in \mathbb{R}^{ n \times q}$ is the matrix of regressors and $Y = \left( y_{p+1},..., y_{T} \right)^{\prime}$ is the matrix of predictants, where $n = T - p$ is the sample size after excluding the estimation for the lag components of the model. The classical least squares estimator for the regression parameter $\mathcal{B}$ is estimated by 
\begin{align}
\widehat{ \mathcal{B} }_{ \text{OLS} } = \left( X^{\prime} X \right)^{-1} \left( X^{\prime} Y \right), 
\end{align} 

\newpage

and the scatter matrix $\Sigma$ is estimated by the following expression 
\begin{align}
\widehat{ \Sigma }_{ \text{OLS} } = \frac{1}{n - d} \left(  Y - X \widehat{ \mathcal{B} }_{ \text{OLS} } \right)^{\prime} \left(  Y - X \widehat{ \mathcal{B} }_{ \text{OLS} } \right). 
\end{align}
\begin{remark}
Outliers can affect parameter estimates, model specification and forecasts based on the selected model. Therefore, a careful examination of the nature of outliers is necessary to ensure robust statistical inference. Moreover, outliers can be of different nature, the most well known types being additive outliers and innovational outliers. More specifically, for the vector autoregressive model $y_t$ is an additive outlier if only its own value has been affected by contamination. On the other hand, an outlier is said to be innovational if the error term $\epsilon_t$ is contaminated. Innovational outliers will therefore have an effect on the next observations as well, due to the dynamic structure of the series. Additive outliers have an isolated effect on the time series, but they still may seriously affect the parameter estimates. 
\end{remark}
\begin{remark}
Notice that for the Residual Autocovariance (RA) estimators are considered to be affine equivariant version of the estimators of \cite{li1989robust}. The particular estimators are considered to be resistant to outliers and therefore we can consider employing them for in the case of constructing a robust estimator to outliers for the Vector Autoregressive Model. For this purpose, we use the Multivariate Least Trimmed Squares (MLTS) estimator proposed by \cite{agullo2008multivariate}.  The specific estimator is defined by minimizing a trimmed sum of squared Machalanobis distances, and can be computed by a fast algorithm. Moreover, the procedure also provides a natural estimator for the scatter matrix of the residuals, which can then be used for model selection criteria. 
\end{remark}

\paragraph{The Multivariate Least Trimmed Squares Estimator}

The unknown parameters of the VAR$(p)$ will be estimated via the multivariate regression model. To allow for the robust estimation of the model under the presence of outliers the Multivariate Least Trimmed Squares (MLTS) estimator is employed, which is based on the idea of the Minimum Covariance Determinant estimator. The way that the MLTS estimator works is that it selects the subset of $h$ observations which minimize the determinant of the covariance matrix of residuals based on a least squares fit from the particular subsample. 

Consider the dataset $Z = \left\{ \left( x_t, y_t \right), t = p+1,..., T \right\} \subset \mathbb{R}^{d+q}$. We denote with $\mathcal{H}$ the collection of all subsets of size $h$, such that: 
\begin{align}
\mathcal{H} = \big\{ H \subset \left\{ p+1,...., T  \right\} \ | \ \mathcal{N} \left\{ H \right\} = h \big\}
\end{align}
Then, for any subset $H \in \mathcal{H}$, we denote with $\widehat{\mathcal{B}}_{\text{OLS}}$ the classical least squares estimator based on the observations of the particular subset given by the following expression 
\begin{align}
\widehat{\mathcal{B}}_{\text{OLS}} (H) =  \left( X^{\prime}_H X_H \right)^{-1} \left( X^{\prime}_H Y_H \right), 
\end{align}
where $X_H$ and $Y_H$ are submatrices of $X$ and $Y$ respectively, consisting of the rows of $X$ and $Y$ that corresponds to the subsample $H$.

\newpage

Then, the corresponding scatter matrix estimator computed from this subset is constructed as below
\begin{align}
\widehat{\Sigma}_{\text{OLS}} = \frac{1}{h - d} \left( Y_H - X_H \widehat{ \mathcal{B} }_{\text{OLS}} (H) \right)^{\prime} \left( Y_H - X_H \widehat{ \mathcal{B} }_{\text{OLS}} (H) \right).
\end{align}
The MLTS estimator is now defined as below
\begin{align}
\label{condition}
\widehat{ \mathcal{B} }_{\text{MLTS}} \left( k \right) = \widehat{ \mathcal{B} }_{\text{OLS}} \left( \hat{h} \right), \ \text{where} \ \ \hat{H} = \underset{ H \in \mathcal{H} }{ \text{argmin} } \ \text{det} \left[  \widehat{\Sigma}_{\text{OLS}} \left( H \right) \right]
\end{align}
Moreover, the associated estimator of the scatter matrix of the error terms is given by 
\begin{align}
\widehat{\Sigma}_{\text{MLTS}} \left( H \right) = c_{\alpha} \widehat{\Sigma}_{\text{OLS}} \left( \hat{H} \right)
\end{align}
where $c_{\alpha}$ is the correction factor to obtain consistent estimation of $\Sigma$ for the error distribution of the model, and $\alpha$ the trimming proportion for the MLTS estimator, that is, $\alpha \approx \left( 1 - \frac{h}{n} \right)$. In the case of multivariate error terms it has been shown that $
c_{\alpha} = (1 - \alpha ) \big/ F_{ \chi^2_{d+2} \left( q_{\alpha} \right)}$, 
where $F_{\chi^2_{d+2}}$ is the cumulative distribution function of a $\chi^2$ distribution with $q$ degrees of freedom, and $q_{\alpha} = \chi^2_{q, 1 - \alpha}$ is the upper $\alpha$ quantile of this distribution.

Equivalent characterizations of the MLTS estimator are given by \cite{agullo2008multivariate}. The particular study shows that any $\widetilde{\mathcal{B}} \in \mathbb{R}^{d \times q}$ minimizing the sum of the $h$ smallest squared Mahlanobis distance of its residuals, subject to the condition that $\text{det} \left[ \Sigma \right] = 1$ is a solution of \ref{condition}. More formally, 
\begin{align}
\widetilde{\mathcal{B}} = \underset{ \left( \mathcal{B}, \Sigma \right) }{ \text{argmin} } \left\{ \sum_{j=1}^h d_{s:n}^2 \left( \mathcal{B}, \Sigma \right)  \right\} \  \ \text{such that} \ \ |\Sigma| = 1
\end{align}
where $d_{1:n} \left( \mathcal{B}, \Sigma \right),...,d_{n:n} \left( \mathcal{B}, \Sigma \right)$ is the ordered sequence of the residual Mahalanobis distances is 
\begin{align}
d_s \left( \mathcal{B}, \Sigma \right) = \left( y_t - \mathcal{B}^{\prime} x_t \right)^{\prime} \Sigma^{-1} \left( y_t - \mathcal{B}^{\prime} x_t \right)^{1 / 2}
\end{align}
for some $\mathcal{B} \in \mathbb{R}^{d \times q}$. Notice that the MLTS estimator minimizes the sum of the $h$ smallest squared distances of its residuals, and is therefore the multivariate extension of the Least Trimmed Squares (LTS) estimator of \cite{rousseeuw1984least}. The particular MLTS estimator has been found to have low efficiency so the literature has proposed the reweighted version of the estimator to improve the performance of the MLTS estimator. Therefore, since the efficiency of the MLTS estimator is quite low, the re-weighted version of the particular estimator is employed to improve the performance of the MLTS. 

The Reweighted Multivariate Least Trimmed Squares (RMLTS) estimates are defined as
\begin{align}
\widehat{\mathcal{B}}_{ \text{RMLTS} } &= \widehat{\mathcal{B}}_{\text{OLS} } \left( J \right) \ \ \ \text{and} \ \ \ \widehat{\Sigma}_{\text{RMLTS} } = c_{ \delta } \widehat{\Sigma}_{OLS} \left( J \right)
\\
J &= \left\{ j \in \left\{ 1,...,n \right\} | d_j^2 \left( \widehat{B}_{MLTS}, \widehat{\Sigma}_{MLTS} \right)  \leq q_{\delta} \right\}
\end{align}
and $q_{\delta} = \chi^2_{q, 1 - \delta}$.

\newpage

Since outliers tend to have large residuals with respect to the initial robust MLTS estimator, then it implies a large residual Mahalanobis distance $d_j^2 \left( \mathcal{ \widehat{ \mathcal{B}} }_{\text{MLTS}}, \widehat{\Sigma}_{\text{MLTS}} \right)$. If the later is above the critical value $q_{\delta}$, then the observation is flagged as an outlier. The final RMLTS is then based on those observations not having been detected as outliers. For instance, one can set $\delta = 0.01$ and take as trimming proportion for the initial MLTS estimator $\alpha = 25 \%$.

\paragraph{Determining the autoregressive order}

To select the order $p$ of a vector autoregressive model, information criteria are computed for several values of $p$ and an optimal order is selected by minimizing the criterion.  Most information criteria are in terms of the value of the log likelihood $\ell_k$ of the VAR$(p)$ model. Using the model assumption for the distribution of the error terms, we obtain that 
\begin{align}
\ell_k = \sum_{t = k+1}^T g \left( \epsilon_t \Sigma^{-1} \epsilon_t  \right) - \frac{n}{2} \text{log} \big\{ \text{det} \big[ \Sigma \big] \big\}, 
\end{align} 
with $n = T - k$. When error terms are multivariate normal the above leads to the following expression 
\begin{align}
\ell_k = - \frac{n}{2} \big\{ \text{det} \big[ \Sigma \big] \big\} - \frac{np}{2} \text{log} \left( 2 \pi \right) - \frac{1}{2} \sum_{t = k+1}^T \epsilon^{\prime}_t \Sigma^{-1} \epsilon_t 
\end{align}
The log likelihood will depend on the autoregressive order via the estimate of the covariance matrix of the residuals. In particular, for the classical least squares estimator we have that 
\begin{align}
\widehat{ \Sigma }_{ \text{OLS} } = \frac{1}{n - p} \sum_{t = k+1}^T \widehat{\epsilon}_t (k)  \widehat{\epsilon}^{\prime}_t (k)
\end{align}
where $\widehat{\epsilon}^{\prime}_t (k)$ are the residuals corresponding to the VAR$(p)$ model. Using trace properties, the last term of the log-likelihood function corresponds to $(n - p) p / 2$ for the OLS estimator. In order to prevent that outliers might affect the optimal selection of the information criteria, we estimate $\Sigma$ by the RMLTS estimator given by the following expression: 
\begin{align}
\widehat{ \Sigma }_{ \text{RMLTS} } = \frac{ c_{\delta} }{ m(k) - p } \sum{ t \in J(k) } \widehat{\epsilon}_t (k) \widehat{\epsilon}_t (k)^{\prime}, 
\end{align} 
with $J(k)$ as above and $m(k)$ the number of elements in $J(k)$. In this case, the last term of the log-likelihood function equals now $\left( m(k) - p \right) p \big/ \left( 2 c_{\delta} \right)$.

\paragraph{Open Problems} Some possible extensions of the particular approach is to propose a framework which employes the LTS method that accounts for outliers in multivariate time series models, specifically to the setting of structural vector autoregressions.

\newpage

\subsubsection{Extension 5: Time-Varying VAR Models}

The identification and estimation of time-varying VAR models is useful when modelling dynamically macroeconomic events such as monetary policy regimes (see, \cite{primiceri2005time}, \cite{koop2013large}) and \cite{chan2016large} as well as time-varying instrumental variable estimation as in \cite{kapetanios2019large}, \cite{petrova2019quasi} and \cite{giraitis2021time}. A relevant setting is the framework of factor augmented  SVAR models (e.g., see \cite{eickmeier2015classical}).  

\textbf{Financial Connectedness}

The study of network topology and financial connectedness is used for modelling systemic risk, spillover effects and financial contagion using the VAR framework (see, discussion in \cite{katsouris2023limit}). Estimating systemic risk measures implies identifying shocks which are transmitted via system-wide connectedness. In particular, \cite{diebold2014network, diebold2012better} proposed a framework for modelling the network topology via the forecast error variance decomposition (FEVD) of an estimated VAR process, as a transformation method that summarizes connectedness between a set of institutions. Thus this approach allows to construct measures of directional connectedness which are analogous to bilateral imports and exports for each of a set of $N$ countries. Recent implementations are presented by  \cite{barigozzi2017network}, \cite{barunik2018measuring}, \cite{demirer2018estimating}, \cite{korobilis2018measuring}, \cite{geraci2018measuring},  \cite{rehman2018precious}, \cite{yoon2019network}, \cite{barigozzi2021time}, \cite{laborda2021volatility}, \cite{mensi2023time} and \cite{barunik2024persistence}. 

Consider, an $N$ dimensional covariance-stationary series $\boldsymbol{Y_t}=(y_{1,t},...,y_{N,t})$, which is generated from a VAR$(p)$ process for $t=1,...,T$. Then, based on the Wold decomposition theorem we have that
\begin{align*}
\boldsymbol{\Phi(L)} \boldsymbol{Y}_t = \boldsymbol{\epsilon_t},\ \ \  \text{with} \ \boldsymbol{\Phi(L)}=\sum_{h = 0}^{\infty} \boldsymbol{\Phi}_h L^h,
\end{align*}
such that $|\Phi(z)|$ lie outside the unit-circle, then the VAR process has the MA($\infty$) representation expressed as $\boldsymbol{Y_t}= \boldsymbol{\Theta(L)}\boldsymbol{\epsilon_t}$, where $\boldsymbol{\Theta(L)}$ is an $(N x N)$ matrix of infinite lag polynomials and $\boldsymbol{\epsilon}_t$ is considered to be a white-noise generated by the covariance matrix $\boldsymbol{\Sigma}$. Therefore, the $H-$step generalized variance decomposition\footnote{See also, \cite{lanne2016generalized} who employ a GFEV decomposition for linear and nonlinear multivariate models. Moreover,  \cite{katsouris2023estimating} develops methods for the identification and estimation of financial networks.} matrix D$^{gH}$ = $[d_{ij}^{gH}]$ is given by 
\begin{align}
d_{ij} = \frac{ \sigma_{ij}^{-1} \sum_{h=0}^{H-1} (e'_i\boldsymbol{\Theta}_h \boldsymbol{\Sigma} e_j   )^2   }{ \sum_{h=0}^{H-1} (e'_i\boldsymbol{\Theta}_h \boldsymbol{\Sigma} \boldsymbol{\Theta}^{'}_h e_i   )  },
\end{align}
where $e_j$ are selection matrices e.g., $e_j=(0,...1,...0)^{'}$, $\boldsymbol{\Theta}_h$ is the coefficient matrix evaluated at the $h$-lagged shock vector and $\sigma_{ij}$ is the $j$th diagonal element of $\boldsymbol{\Sigma}$. Modelling the asymmetric responses of the economy to monetary policy shocks over the business cycle due to the impact of the monetary transmission mechanism  (e.g., see \cite{eickmeier2015classical},  \cite{morley2012asymmetric} and \cite{veldkamp2005slow}), requires to employ specifications which allow for time-varying effects (see, \cite{paul2020time}). 

\newpage

\section{Dynamic Causal Effects}

The two main types of identification presented so far in the literature consists of point-identification and set-identification. Specifically, point-identified SVARs are interpreted as achieving identification through internal instruments. Consequently, in these models, structural shocks are the interventions of interest, and thus the goal is to estimate the dynamic causal effect of these shocks on macroeconomic outcomes. In contrast, modern microeconometric identification stategies rely heavily on external sources of variation that provide quasi-experiemnts to identify causal effects.  Such external variation might be found, for example, in firm-specific characteristics that introduce as-if randomness in the variable of interest (the treatment). The use of such external instruments in microeconometrics has proven highly productive and has yielded compelling estimates of causal effects. Therefore, the use of external instruments has opened a new and rapidly growing research avenue in macroeconometrics, in which credible identification is obtained using as-if random variation in the shock of interest that is distinct from - external to - the macroeconomic shocks hitting the economy (\cite{stock2018identification}). 

\subsection{Identification and Estimation of Causal Effects}

\subsubsection{Causal Effects and IV Regression}

Following the paper of \cite{stock2018identification}, a starting point for formulating the theory on identification and estimation of dynamic causal effects in macroeconomics, is that the expected difference in outcomes between the treatment and control groups in a randomized experiment with a binary treatment is the average treatment effect. Roughly speaking, if a binary treatment $X$ is randomly assigned, then all other determinants of $Y$ are independent of $X$, which implies that the (average) treatment effect is 
\begin{align}
\mathbb{E} \big[ Y |  X = 1 \big] -  \mathbb{E} \big[ Y |  X = 0 \big]    
\end{align}
In the linear model $Y = \alpha + \beta X + u$, where $\beta$ is the treatment effect, random assignment implies that $\mathbb{E} \left( u | X  \right) = 0$ so that the population regression coefficient is the treatment effect. When randomization is conditional on covariates $W$, then the treatment effect for an individual with covariates $W = w$ is estimated by the outcome of a random experiment on a group of subjects with the same value of $W$, 
\begin{align}
\mathbb{E} \big[ Y |  X = 1, W = w \big] -  \mathbb{E} \big[ Y |  X = 0, W = w \big].    
\end{align}

\begin{example}
Usually the path of observed macroeconomic variables arising from current and past shocks and measurement error are collected into the $( m \times 1 )$, $\varepsilon_t$ error vector $\varepsilon_t$. Therefore, the $( n \times 1 )$ vector of macroeconomic variables $Y_t$ can be written in terms of current and past innovation terms  
\begin{align}
Y_t = \Theta (L) \varepsilon_t,     
\end{align}
where $\Theta (L) = \Theta_0 + \Theta_1 L + \Theta_2 L^2 + ...$, where $\Theta_h$ is an $( n \times m )$ matrix of coefficients.

\newpage

Then the variance matrix of the error terms $\boldsymbol{\Sigma} = \mathbb{E} [ \boldsymbol{\epsilon}_t \boldsymbol{\epsilon}_t ]$ is assumed to be positive-definite to ensure the existence of a non ill-conditioned covariance matrix. Moreover these disturbance terms (shocks) are assumed to be mutually uncorrelated. The expression $Y_t = \Theta (L) \varepsilon_t$, corresponds to the structural moving average representation of $Y_t$. Specifically, the coefficients of $\Theta (L)$ are the structural impulse response functions, which are the dynamic causal effects of the shocks (see, also \cite{bacchiocchi2018gimme}).   

Under the null of invertability, the following SVAR representation applies: 
\begin{align}
\boldsymbol{A} (L) \boldsymbol{Y}_t = \boldsymbol{\Theta} (L) \boldsymbol{\varepsilon}_t,     
\end{align}
Therefore, under invertability, the structural moving average is $\boldsymbol{Y}_t = \boldsymbol{C}(L) \boldsymbol{\Theta}_0 \boldsymbol{\varepsilon}_t$, where $\boldsymbol{C}(L) = \boldsymbol{A} (L)^{-1}$, such that $\boldsymbol{\Theta} (L) = \boldsymbol{C}(L) \boldsymbol{\Theta}_0$. The null and alternative hypotheses are then
\begin{align}
H_0:  \boldsymbol{C}_h \boldsymbol{\Theta}_{0,1} = \boldsymbol{\Theta}_{h,1} \ \ \ \text{against} \ \ \ H_1: \boldsymbol{C}_h \boldsymbol{\Theta}_{0,1} \neq \boldsymbol{\Theta}_{h,1}, \ \ \text{for some} \ h.     
\end{align}
Recall that the SVAR can be also written in state-space form as below: 
\begin{align}
\boldsymbol{Y}_t &= \boldsymbol{\mathcal{B}} \boldsymbol{X}_t  
\\
\boldsymbol{X}_t &= \boldsymbol{A} \boldsymbol{X}_{t-1} + \boldsymbol{G} \boldsymbol{\varepsilon}_t,  
\end{align}
where $\boldsymbol{X}_t = \big( \boldsymbol{Y}_t^{\top}, \boldsymbol{Y}_{t-1}^{\top},....,  \boldsymbol{Y}_{t-p+1}^{\top} \big)^{\top}$, such that $\boldsymbol{A}$ is the companion matrix and $\boldsymbol{\mathcal{B}} = \big( \boldsymbol{I}_n \ \boldsymbol{0} \cdots \boldsymbol{0} \big)$ is a selection matrix. Then, the local projection regression equation is written as below: 
\begin{align}
\boldsymbol{Y}_{t+h} = \boldsymbol{\Theta}_{h,1} \boldsymbol{Y}_{1,t} + \boldsymbol{\Gamma}_h \boldsymbol{W}_t + \boldsymbol{u}_{t+h},     
\end{align}

\begin{remark}
The study of \cite{stock2018identification} verifies important insights regarding the identification of dynamic causal effects. In particular,  under the assumption of Gaussian errors, every invertible model  has multiple observationally equivalent non-invertible representations, which imply that to identify a unique representation some external information regarding the system is required. If we assume that the structural shocks are independent and non-Gaussian then using information from higher-order restrictions the causal structure of the system can be identified. In practice, external instruments can be employed to estimate dynamic causal effects directly without using an indirect VAR identification step. 
\end{remark}

\subsubsection{Causal Effects Identification with SVARs}

\end{example}

\begin{example}[Causal effects of Lockdown Policies on Health and Macro Outcomes]

\

We follow the framework proposed by \cite{arias2023causal} who consider the causal impact of pandemic-induced lockdowns and other nonpharmaceutical policy interventions (NPIs), on health and macroeconomic outcomes. To do this, identification assumptions are needed to access causality. During the first step the authors  use the Bayesian approach to estimate an epidiomiological model with time variation in the parameters controlling an infectious disease's dynamics.

\newpage 

In particular, time variation in the parameters of the model allows to: $\textit{(i)}$ capture changes in the behaviour of individuals as they respond to public health conditions and  $\textit{(ii)}$ to include shifts in the transmission and clinical outcomes of the pandemic. 

We write our SVAR model such that 
\begin{align}
\boldsymbol{y}_t^{\prime} \otimes \boldsymbol{A}_0 = \boldsymbol{x}_t^{\prime}  \otimes \boldsymbol{A}_{+} + \varepsilon_t^{\prime}, \ \ 1 \leq t \leq T    
\end{align}
\begin{itemize}

\item $\boldsymbol{y}_t$ is an $( n \times 1)$ vector of endogenous variables, $\boldsymbol{x}_t^{\prime} = \big[ \boldsymbol{y}_t^{\prime}, ...,  \boldsymbol{y}_{t-p}^{\prime}, \boldsymbol{z}_t, \boldsymbol{1} \big]$, where $\boldsymbol{z}_t$ is a $( z \times 1 )$ vector of exogenous variables, and $\boldsymbol{\varepsilon}_t$ is an $( n \times 1)$ vector of structural shocks. 

\item Model parameters: $\boldsymbol{A}_0$ is an $(n \times n)$ invertible matrix of parameters, and $\boldsymbol{A}_{+}$ is an $(np + z + 1) \times n$ matrix of parameters such that $p$ is the lag length and $T$ is the sample size.  
\end{itemize}
In addition, we assume that the vector $\boldsymbol{\varepsilon}_t$, conditional on past information and the initial conditions $\boldsymbol{y}_0,..., \boldsymbol{y}_{1-p}$, is Gaussian with mean zero and covariance matrix $\boldsymbol{I}_n$. Without loss of generality, we assume that the first equation of the SVAR characterizes the policy rule. This implies that
\begin{align}
\boldsymbol{y}_t^{\prime} \boldsymbol{\alpha}_{0,1} =  \boldsymbol{x}_t^{\prime} \boldsymbol{\alpha}_{+,1} + \boldsymbol{\varepsilon}_{1t}, \ \ 1 \leq t \leq T   
\end{align}
is the policy equation such that 
\begin{itemize}

\item $\boldsymbol{\varepsilon}_{1t}$ denotes the first entry of the vector $\boldsymbol{\varepsilon}_{t}$. 

\item $\boldsymbol{\alpha}_{+,1}$ denotes the first column of $\boldsymbol{A}_{+}$ for $\ell \in \left\{ 0,..., p \right\}$ and $\alpha_{s,ij}$ denotes the $(i,j)$ entry of $\boldsymbol{A}_{s}$, where $s \in \left\{ 0, + \right\}$, and describes the systematic component of the policy rule.    
    
\end{itemize}
Restricting the systematic component of the policy rule is equivalent to restricting $\alpha_{s,ij}$ and identifying a policy shock that we call the \textit{stringency shock}. Recall that the dynamic causal effect of a unit intervention in $\varepsilon_{j,t} \in \varepsilon_{t}$ on $z_{t+h}$ is given by $\big( \mathbb{E} [ z_{t+h} |  \varepsilon_{j,t} = 1; \varepsilon_{t-1}, ....  ]  - \mathbb{E} [ z_{t+h} |  \varepsilon_{j,t} = 0; \varepsilon_{t-1}, ....  ] \big)$. Then, the structural impulse response functions can be derived by $\partial z_{t+h} / \partial \varepsilon_{j,t}$ for $h = 0,1,2,...$. 

\end{example}

\medskip

\begin{example}[ECB monetary policy and bank default risk, see \cite{soenen2022ecb}]

\

In this example, the primary focus is the impact of ECB monetary policy on bank risk. There is an ongoing debate on the effect of accommodative monetary policy on bank risk taking and financial stability in general. Specifically, one concern relates to the potential increase in risk taking and the possible under-pricing of risk. The relevant question is whether or not monetary policy causes excessive risk taking by banks, since this could hamper financial stability which can be reflected in higher bank CDS spreads. Next, regarding the choice of the variable of interest: we cannot use the policy rate because of the zero lower bound constraint and, similarly, we cannot use the ECB balance sheet because some important monetary policy measures did not affect the balance sheet.

\newpage 

Therefore, when assessing the causal impact of monetary policy, we decide to employ a structural VAR because incorporating a broad set of financial market indicators allows us not only to identify actual ECB monetary policy decisions, but also to capture anticipation effects and instances in which financial markets judge that monetary policy actions were insufficient, given the prevailing market conditions. In other words, we can estimate a time series of exogeneous monetary policy shocks by modelling a set of relevant financial market variables in a structural VAR model which is given by the following econometric specification:
\begin{align}
\boldsymbol{y}_t = \boldsymbol{A}_1 \boldsymbol{y}_{t-1} + ... +  \boldsymbol{A}_p \boldsymbol{y}_{t-p} + \boldsymbol{R} \boldsymbol{v}_t,
\end{align}    
where $\boldsymbol{y}_t$ is an $n-$dimensional vector of endogenous variables, $\boldsymbol{v}_t$ is an $n-$dimensional vector of orthogonal structural innovations with mean zero and $\left\{ \boldsymbol{A}_1,..., \boldsymbol{A}_p   \right\}$ and $\boldsymbol{R}$ are $(n \times n)$ time-invariant parameter matrices. The reduced-form residuals corresponding to this structural model are given by $\boldsymbol{\varepsilon}_t = \boldsymbol{R} \boldsymbol{v}_t$. 
\end{example}

\begin{remark}
A macroeconomic puzzle relevant to the optimal economic policy decision-making with respect to macro-prudential restructuring and debt collection is the aspect of the tax policy conducted over the Business Cycle. In particular, it is well known that government spending has typically been procyclical in developing economics but exhibit acyclical or countercyclical behaviour in industrial economies. Moreover, an interesting avenue of further research is the counterfactual distribution of the cyclical behaviour of tax rates, as opposed to tax revenues which are endogenous to the business cycle and hence cannot shed light on the cyclicality of tax policy (see, \cite{vegh2015tax}). 
\end{remark}
Therefore the identification and estimation of dynamic treatment effects allows to investigate the VAT behaviour of firms in relation to imposed by the government tax rates. The intuition goes as follows. Governmental taxation is the pillar of fiscal policy and decisions on the optimal tax policy are usually taken in relation to their effectiveness on various macroeconomic outcomes, such as output fluctuations as well as structural changes in agent's behaviour over the business cycle (aggregate fluctuations). From the econometrics perspective to facilitate inference requires to derive the asymptotic behaviour of the dynamic treatment effect estimator across a panel of economics which have undergone a period of fiscal adaption to new tax policies. In other words, one would be interested to obtain the asymptotic behaviour of these estimators based on the trajectory of the control against the experimental groups in relation to a common taxation policy as well as to obtain counterfactual outcomes when a control group has adapted an updated fiscal policy across the business cycle. In this direction, we are interested to incorporate tax policy variables such as taxation revenues which are often considered as policy outcomes. Furthermore, due to the fact that these variables might be positive correlated during the business cycle, suitable instrumental variables that can ensure consistent model estimation and inference procedures can be employed. In particular, \cite{vegh2015tax} use the  VAT rates across economic during the business cycle as a valid instrumental variable. Moreover, the authors use a novel tax database which combines both corporate and personal income tax rates as well as data on value-added taxes. A relevant study which discusses the effects of global tax reform is presented by \cite{gomez2023measuring}.

\newpage 

\subsection{Counterfactual Analysis}

\subsubsection{Causal Effect Inference with Time Series Data}

In particular, the parameter of interest is the ATE $:= \mathbb{E} \left[ Y_1 - Y_0  \right]$. Moreover, the Conditional Average Treatment Effect (CATE) defined by 
\begin{align}
\tau(x) := \mathbb{E} \left[ Y_1 - Y_0 | X = x \right],     
\end{align}
captures how the average treatment effect changes across sub-populations with covariate value $X = x$. Furthermore, under the unconfoundness assumption $\tau(x)$ can be nonparametrically identified by 
\begin{align}
\tau(x) := \mathbb{E} \left[ Y | D = 1, X = x \right] - \mathbb{E} \left[ Y | D = 0, X = x \right]      
\end{align}

\begin{assumption}
Suppose that the following conditions hold 
\begin{itemize}
\item[(i)]  $\big( Y(1), Y(0) \big) \ \textit{independent} \  D | X$.
    
\item[(ii)] $X$ has a compact support $\mathsf{Supp}(X)$ and there exists $c > 0$ such that $c < \Pi (x) < 1 - c$ for all $x \in \mathsf{Supp}(X)$. 
    
\end{itemize}
\end{assumption}

\begin{remark}
Notice that the above assumptions are fairly standard in the causal inference literature and we assume that they hold throughout the remaining of the paper unless otherwise states. In particular, Assumption 1 (i), corresponds to the unconfoundedness condition while Assumption 1 (ii) corresponds to the common support condition for the set of covariates in the model. 
\end{remark}
Consider the $k-$step ahead forecast of $Y_t$ given the information available up to time $t^{*}$ which is given by the following expression 
\begin{align}
Y_{ t^* + k | t^* } (0) := \widehat{\mathbb{E}} \big[ Y_{ t^* + k } (0) | \mathcal{F}_{ t^* }  \big] \equiv X_{ t^* + k }^{\prime} \widehat{\beta}   + \widehat{\mathbb{E}} \big[ Z_{ t^* + k } | \mathcal{F}_{ t^* }  \big].     
\end{align}
Notice that by definition, $Y_{ t^{*} + k } (0)$ is the potential outcome expected at time $t^{*} + k$ in case the intervention does not occur. Therefore, an estimator of the point causal effect can be defined as 
\begin{align*}
\widehat{\tau}_{ t^{*} + k }  
:=  Y_{ t^* + k } (1)  - Y_{ t^* + k | t^* } (0)
= \tau_{ t^{*} + k } + X_{ t^* + k } \big( \beta - \widehat{\beta} \big) + Z_{t^* + k } - \widehat{\mathbb{E}} \big[ Z_{ t^* + k } | \mathcal{F}_{ t^* }  \big]
\end{align*}

\subsubsection{Potential Outcome Framework}

Let $D \in \left\{ 0, 1 \right\}$ be a binary treatment indicator and $Y_d$ be the potential outcomes for $d = 0,1$ in the status quo environment (see, also \cite{hsu2022counterfactual}). In particular, $Y_1$ is the outcome if an individual is exogenously assigned to the treatement $(D = 1)$ and $Y_0$ is the outcome in the absence of treatement $( D = 0)$. Then, the actual observed outcome is 
\begin{align}
Y = D Y_1 + (1-D) Y_0.    
\end{align}

\newpage

We observe a $d-$dimensional vector of pretreatment covariates $X = \left( X_1,..., X_k   \right)$ in the status quo environment and $X^{\star} = \left( X^{\star}_1,..., X^{\star}_k   \right)$ in the counterfactual environment which is of the same dimension as the set of covariates $X$. 
\begin{itemize}

\item[(i)] $X$ and $X^{\star}$ are statistically independent 

\item[(ii)] $X^{\star}$ is a deterministic transformation of $X$, such that $X^{\star} = \pi (X)$ for some known function $\pi$. 

\end{itemize}

The corresponding treatment indicator, outcome, and potential outcomes in the counterfactual environment are denoted as $D^{*}, Y^{*}$ and $Y_d^{*}$ for $d = 0,1$, respectively with 
\begin{align}
Y^{\star} = D^{\star} Y^{\star}_1 + (1-D^{\star}) Y^{\star}_0. 
\end{align}
However, notice that since the treatment has not been implemented in the counterfactual environment yet, neither $D^{*}$ nor $Y^{*}$ is observed in our model (see, \cite{hsu2022counterfactual}).  We adopt the former approach such that a closed-form solution exists
\begin{align}
\widehat{\delta}^{*} = \frac{1}{n} \sum_{j=1}^n \left[  \widehat{ \mathbb{E}} \left( Y_1 | X = X_j^{*} \right) - \widehat{ \mathbb{E}} \left( Y_0 | X = X_j^{*} \right)  \right]
\end{align}
where $\widehat{ \mathbb{E}} \left( Y_d | X = X_j^{*} \right)$ is the Nadaraya-Watson estimator such that 
\begin{align}
\widehat{ \mathbb{E}} \left( Y_d | X = X_j^{*} \right) = \frac{ \displaystyle  \sum_{i=1}^n Y_i \mathbf{1} \left\{ D_i = d \right\} K_{x,h} \left( X_i - x \right) }{ \displaystyle \sum_{i=1}^n \mathbf{1} \left\{ D_i = d \right\} K_{x,h} \left( X_i - x \right) }    
\end{align}
The asymptotic properties for $\widehat{\delta}^{*}$ can be derived under weaker regularity conditions stated below. Furthermore based on regularity conditions it can be proved that 
\begin{align}
\sqrt{n} \left( \widehat{\delta}^{*} - \delta^{*} \right) \overset{d}{\to} \mathcal{N} \left( 0, \sigma^2_{\delta^{*}} \right).    
\end{align}
Therefore, compared to the semiparametric efficiency bound of the ATE estimator we have that 
\begin{align}
\mathbb{E} \left\{  \frac{\mathsf{Var} ( Y_1 | X )}{ p(X) } +  \frac{\mathsf{Var} ( Y_0 | X )}{ 1 - p(X) } + \mathbb{E} \big[ Y_1 - Y_0 | X \big] - \mathbb{E} \big[ \left( Y_1 - Y_0 \right) \big]^2 \right\}.
\end{align}

\newpage

\subsection{Time Series Experiments and Causal Effects}

\subsubsection{Estimating Contemporaneous Effect of Treatment}

Time series causal effects are defined as a comparison between the potential outcomes at a fixed point in time. Thus, the primary object of interest is the temporal average of these causal effects. Therefore, the only source of randomness in our formulation is the randomization of the treatment. Once we begin the experiment, the randomization reveals a particular outcome path; however, this does not change the set of potential outcomes, and therefore, does not alter what we would have observed had we administered a different treatment. Thus, in this sense, we are treating the potential outcomes as predetermined, and the only source of randomness comes from which path we observe. 
\begin{definition}[General Causal Effects] For paths $w_{1:t}$ and $w_{1:t}^{\prime}$, the $t-$th causal effect is given by 
\begin{align}
\tau ( w_{1:t}, w_{1:t}^{\prime} ) = \mathbb{E} \big[ Y_{it} ( w_{1:t} ) -  Y_{it} ( w_{1:t}^{\prime} ) \big].        
\end{align}
\end{definition}
Then, the temporal average treatment effect of the paths $w_{1:T}$ and $w_{1:T}^{\prime}$ is
\begin{align}
\tau ( w_{1:t}, w_{1:t}^{\prime} ) = \frac{1}{T} \sum_{t=1}^T \tau_t \left( w_{1:t}, w_{1:T}^{\prime}   \right).    
\end{align}
Moreover, we define the potential outcomes just intervening on treatment the last $j$ periods as 
\begin{align}
Y_{it} ( x_{t-j:t} ) = Y_{it} \big( X_{i, 1:t-j-1}, x_{t-j:t} \big)    
\end{align}
In other words, this "marginal" potential outcome represents the potential or counterfactual level of regime A in country $i$, if we let welfare spending run its natural course up to $t-j-1$ and just set the last $j$ lags of spending to $x_{t-j:t}$. Thus, we can define an important quantity of interest, that is, the \textit{contemporaneous effect of treatment} (CET) of $X_{it}$ on $Y_{it}$ such that:
\begin{align*}
\tau_c (t) 
&:= 
\mathbb{E} \big[ Y_{it} \big( X_{i,1:t-1}, 1 \big) - Y_{it} \big( X_{i,1:t-1}, 0 \big) \big]
\equiv \mathbb{E} \big[ Y_{it}(1) -  Y_{it}(0)  \big].
\end{align*}
Furthermore, following \cite{bojinov2019time} specifically in longitudinal studies and survey studies with a follow-up wave of cross-sectional datasets, we are particularly interested to investigate the treatment effect of survey participants as well as their membership in certain endogenously generated groups. This implies that we are modelling the presence of time-varying exposure to dynamic causal estimands. Therefore, to construct a feasible structural econometric framework we consider the treatment and potential outcome path with time-varying effects. At each time step, $t = 1,...,T$, we assume that each individual (sampling unit), is exposed to either treatment, $W_t = 1$, or control, $W_t = 0$, and then we measure the corresponding outcome variable of interest.

\newpage 

Although we focus on binary treatments, our results can be generalized to multiple treatments as well. Thus, the random "treatment path" is given by
$W_{1:t} = \big( W_1,..., W_t \big)$. As a result, the potential path for the treatment path $w_{1:t}$ can corresponds to the following sequence
\begin{align}
Y_{1:t} ( w_{1:t} ) = \big\{ Y_1 (w_1), Y_2 (w_{1:2}),..., Y_2 (w_{1:t}) \big\}.  
\end{align}

\subsubsection{Null Hypothesis of Temporal Causal Effects}

Following the framework proposed by \cite{bojinov2019time}, to access whether the treatment has a statistically significant effect, we consider the sharp null of no temporal causal effects
\begin{align}
H_0: Y_t ( w_{1:t} ) = Y_t ( w_{1:t}^{\prime} ), \ \ \ \text{for all} \ w_{1:t} = w_{1:t}^{\prime}, \ \ \ t = 1,2,...,T.    
\end{align}
In particular, this will null hypothesis will be tested against a portmanteu alternative. Thus, invoking the sharp null hypothesis implies that $Y_t ( w_{1:t}^{\text{obs}} ) = Y_t ( w_{1:t}^{\prime}  )$ for all $w_{1:t}^{\prime}$. Moreover, the particular formulation of the null hypothesis means "the null of no temporal causality at lag $p \geq 0$ of the treatment on the outcome" holds, such that $H_{0,p}: \tau_{t,p} = 0$, for all $t = 1,2,...,T$. 

In the case where $p = 0$, then $\widehat{\tau}_{t,0}$ estimation error can be written as below
\begin{align}
u_{t,0} = \left(  \frac{ \boldsymbol{1} \left\{ W_t = 1 \right\} }{ p_t(1) } -  \frac{ \boldsymbol{1} \left\{ W_t = 0 \right\} }{ p_t(0) }  \right) Y_t ( w_{1:t}^{\text{obs}} )   
\end{align}
and under the sharp null hypothesis we obtain that
\begin{align}
\mathbb{E}^R \big[ u_{t,0} | \mathcal{F}_{T,t-1} \big] = 0 \ \ \ \text{and} \ \ \ Var \big[ u_{t,0} | \mathcal{F}_{T,t-1} \big] = \frac{ Y_t ( w_{1:t}^{\text{obs}} )  }{ p_t(1) p_t(0) }.    
\end{align}
In particular, whenever $p_t(1) = p_t(0)$ the estimator of the upper bound to the variance for the contemporaneous causal effect, is equal to the variance under the sharp null of no treatment effect. Then, the $p$ lagged error variance $Var^R \big( u_{t-p,p} | \mathcal{F}_{T, t-p-1} \big)$ equals to 
\begin{align}
\nu_p^2 = Y_t ( w_{1:t}^{\text{obs}} )^2 \frac{1}{ 2^{2p} } \sum_{ w \in \left\{0,1 \right\}^p } \left( \frac{1}{ p_t( 1,w ) } + \frac{1}{ p_t(0,w) } \right).    
\end{align}
We begin by rewriting the within-unit matching estimator as below
\begin{align}
\hat{\tau}_{ \text{match} } = \frac{1}{ \frac{1}{N} \sum_{i=1}^N C_i } . \frac{1}{N} \sum_{i=1}^N C_i \left\{  \frac{ \displaystyle  \sum_{t=1}^T X_{it} Y_{it}  }{ \displaystyle \sum_{t=1}^T X_{it} }  - \frac{ \displaystyle  \sum_{t=1}^T X_{it} Y_{it}  }{ \displaystyle \sum_{t=1}^T (1 - X_{it} ) }  \right\}.    
\end{align}

\newpage

Consider the conditional independence property which implies that $\big\{ Y_{it}(1), Y_{it}(0) \big\}  \perp \boldsymbol{X}_i | \boldsymbol{U}_i$. By the law of iterated expectations this implies that
\begin{align}
\mathbb{E} \big[ Y_{it} (x) | C_i = 1 \big] = \mathbb{E} \bigg[ \mathbb{E} \big[ Y_{it} (x) | \boldsymbol{U}_i, C_i = 1 \big] \big| C_i = 1 \bigg]    
\end{align}
for $x = 0,1$.

\subsection{Bias-Reduced Doubly Robust Estimation}

According to \cite{vermeulen2015bias} the proposed approach may more generally lend itself better to small-sample inference. For instance, suppose that interest lies in the marginal causal effect $\tau = \mathbb{E} [ Y(1) - Y(0) ]$. Because the proposed estimation strategy does not require acknowledging the uncertainty of the estimated nuisance parameters (up to first order), we foresee that it may potentially lend itself better to randomization inference. How such randomization inference could be accomplished and how it performs in small to large samples is indeed relevant questions for further investigation.

\subsubsection{Marginal Treatment Effects}

Consider $\textit{i.i.d}$ data $\left\{ \boldsymbol{Z}_i = \left( Y_i, A_i, \boldsymbol{X}_i \right), i = 1,..., n \right\}$, where $Y_i$ is the outcome of interest, $A_i$ is a dichotomous treatment taking values zero and one and $\boldsymbol{X}_i$ is a sufficient set of covariates to control for confounding of the treatment effect, in the sence that $Y( a ) \perp A| \boldsymbol{X}$ for $a \in \left\{ 0, 1\right\}$. Specifically, $Y(a)$ denotes the counterfactual outcome for treatment level $a \in \left\{ 0, 1\right\}$, which is linked to the observed data through the consistency assumption (i.e., $Y(a) = Y$ iff $A=a$).

\begin{assumption}[Covariates-Treatment Independence] Let $X_{j,t}$ be a vector of covariates that are predictive of the outcome of unit $j$, for all $t \in \left\{ t^{*}+1,..., T \right\}$ such covariates are not affected by the intervention, such that, $X_{j,t}(1) = X_{j,t}(0)$.     
\end{assumption}

\begin{definition}
Let $t^{*}$ be the time of the intervention and define with $k \in \left\{ 1,..., K \right\}$ such that $t^{*} + K = T$. For any $k$, the point, cumulative and temporal average causal effects on unit $j$ at time $t^{*} + k$ are defined, respectively, such that
\begin{align}
\tau_{j, t^{*} + k} (1;0) &= Y_{j, t^{*} + k}(1) - Y_{j, t^{*} + k}(0)
\\
\Delta_{j, t^{*} + k} (1;0) &= \sum_{h=1}^k \tau_{j, t^{*} + k} (1;0)
\end{align}
\end{definition}

\begin{assumption}[Selection on observables]
The following conditional independence assumption ensures that model selection is conditioned on observables (see, also \cite{angrist2011causal})    
\begin{align}
Y_{t,1}(d), Y_{t,2}(d),..., \perp D_t | z_t, \ \ \ \text{for all} \ \ d \ \text{and} \ \ \psi \in \Psi_t.   
\end{align}
\end{assumption}

\newpage

Notice that the selection on observables assumption says that policies are independent of potential outcomes after appropriate conditioning. Moreover, the conditioning set includes the unknown error vector of the model, $\varepsilon_t$, such that $Y_{t,1}(d), Y_{t,2}(d),..., \perp \varepsilon_t | z_t$. This is because, conditional on the set of instruments $z_t$, randomness in $D_t$ is due exclusively due to randomness induced from the unknown $\varepsilon_t$ vector. Therefore, in order to convert the above assumption into a testable hypothesis which has an identifiable parameter space, we consider the sharp null hypothesis such that $Y_{t,j} ( d^{\prime} ) = Y_{t,j} (d) = Y_{t+j}$.

Then, if we substitute the observed with the potential outcomes we obtain the following testable conditional independence assumption: 
\begin{align}
Y_{t+1},..., Y_{t+j},..., \perp D_t | z_t,    
\end{align}
Notice that $D_t$ is not necessarily a binary treatment status but can be a set of different taxation policies which implies that we have a continuum of treatment status.

\begin{definition}[Granger noncausality]
Under suitable regularity conditions it holds that 
\begin{align}
Y_{t+1} \perp D_t, \bar{D}_{t+1} | \bar{Y}_t.   
\end{align}
\end{definition}

The following theorem of Sims and Granger causality holds. 

\medskip

\begin{theorem}
Let $\chi_t$ be a stochastic process defined on a probability space $\big( \Omega, \mathcal{F}, \mathbb{P} \big)$, assuming that conditional probability measures $\mathbb{P} \big( Y_{t+1} , D_t | z_t \big)$ are well defined for all $t$ expect possibly on a set of measure zero.     
\end{theorem}

\begin{remark}
A relevant example is presented by \cite{menchetti2023combining} who consider estimating  the effect of policy interventions in observational time series settings in the absence of untreated units, by  combining counterfactual outcomes and ARIMA models for prediction purposes. Specifically, when
focusing on observational time series data, commonly used methodologies for policy evaluation within the Rubin Causal Model (RCM) framework include Differences-in-Differences (DiD) and synthetic control methods. However, a main drawback of these two approaches is that they require the presence of controls that did not experience the treatment, and therefore finding untreated units is a challenging task \citep{menchetti2023combining}. Moreover, the particular method is not guaranteed to reveal relevant causal effects especially due to the presence of unobserved factors.  For example, if our main goal is to estimate the causal effect of a new intervention policy such as due to the adoption of new tax rates, as well as its indirect impact on a common fiscal policy across similar economies (e.g., VAT policy), then a causal time series modelling approach can be employed to  obtain the causal estimands of interest. However, one may find that these causal effects have no negative or positive impact on the spending trajectory of similar economies that did not adopt new fiscal policies, suggesting that unobserved factors may drive spending decision making more than taxation policy. Other relevant counterfactual frameworks can be constructed (see, \cite{kilian2011does}), around issues from the international finance literature such as the modelling sovereign default risk (e.g., see \cite{hatchondo2016debt}).  
\end{remark}

\newpage

\subsection{The Augmented Synthetic Control Method}

The synthetic control method (SCM) is a popular approach for estimating the impact of a treatment on a single unit in panel data settings (see, \cite{ben2021augmented}). In other words is a weighted average of control units that balances the treated units pre-treatment outcomes and other covariates as closely as possible. We define the treated potential outcome as $Y_{it} = Y_{it} (0) + \tau_{it}$, where the treatment effects $\tau_{it}$ are fixed parameters. Since the first unit is treated, the key estimand of interest is 
\begin{align}
\tau = \tau_{1T} = Y_{1T} (1) - Y_{1T} (0).    
\end{align}
Then, the observed outcomes are as below
\begin{align}
Y_{it} = 
\begin{cases}
Y_{it}(0) & \ \text{if} \ W_i = 0 \ \text{or} \ t \leq T_0,
\\
Y_{it}(1) & \ \text{if} \ W_i = 1 \ \text{or} \ t > T_0.
\end{cases}
\end{align}

\begin{remark}
Notice that pre-treatment outcomes serve as covariates in SCM, we use $X_{it}$, for $t \leq T_0$, to represent pre-treatment outcomes.     
\end{remark}

\begin{itemize}

\item We are interested on a potential outcome system as a foundational framework for analyzing dynamic causal effects on outcomes in observational time series settings. We consider settings in which there is a single unit observed over time. At each time period $t \geq 1$, the unit receives a vector of assignments $W_t$, and an associated vector of outcomes $Y_t$ are generated. The outcomes are causally related to the assignments through a potential outcome process, which is a stochastic process that describes what would be observed along counterfactual assignment paths. 

\item A dynamic causal effect is generally defined as the comparison of the potential outcome process along different assignment paths at a fixed point in time. 
    
\end{itemize}

\begin{assumption}[Assignment and Potential Outcome]
The assignment process $\left\{ W_t \right\}_{ t \geq 1}$ satisfies 
\begin{align}
W_t \in \mathcal{W} := \times_{k=1}^{  d_w } \mathcal{W}_k       
\end{align}
The potential outcome process, is for any deterministic sequence $\left\{ w_s \right\}_{ s \geq 1 }$ with $w_s \in \mathcal{W}$ for all $s \geq 1$, $\big\{  Y_t \big( \left\{ w_s \right\}_{ s \geq 1}  \big) \big\}_{ t \geq 1}$, where the time$-t$ potential outcomes satisfies $Y_t \big( \left\{ w_s \right\}_{ s \geq 1}  \big) \in \mathcal{Y} \subset \mathbb{R}^{ d_y }$.
\end{assumption}

\medskip

\begin{remark}
The simplest case is when the assignment is scalar and binary $\mathcal{W} = \left\{ 0, 1 \right\}$, in which case $W_t = 1$ corresponds to "treatment" and $W_t = 0$ is "control". Moreover, the potential outcome $Y_t \big( \left\{ w_s \right\}_{ s \geq 1}  \big)$ may depend on future assignments $\left\{ w_s \right\}_{ s \geq t + 1}$. The next assumption considers the dependence structure we allow in the proposed econometric framework, restricting the potential outcome to only depend on past and contemporaneous assignments.        
\end{remark}

\medskip

\begin{assumption}[Non-anticipating Potential Outcomes]
For each $t \geq 1$, and all deterministic $\left\{ w_t \right\}_{t \geq 1}$ and $\left\{ w_t^{\prime} \right\}_{t \geq 1}$ with $w_t, w_t^{\prime} \in \mathcal{W}$, 
\begin{align}
Y_t \big( w_{1:t}, \left\{ w_s \right\}_{ s \geq t + 1} \big) = Y_t \big( w_{1:t}, \left\{ w_s^{\prime} \right\}_{ s \geq t + 1} \big)  \ \text{almost surely}.  
\end{align}
\end{assumption}

\begin{remark}
The above assumption is a stochastic process analogue of non-interference. Furthermore, it still allows for rich dependence on past and contemporaneous assignments. Under Assumption 2, we drop references to the future assignmnets in the potential outcome process and write
\begin{align}
\big\{ Y_t \big(  \left\{ w_s \right\}_{ s \geq 1} \big) \big\}_{t \geq 1} = \big\{ Y_t \big( w_{1:t} \big) \big\}_{t \geq 1}    
\end{align}
Then, the set $\big\{ Y_t \big( w_{1:t} \big): w_{1:t} \in \mathcal{W}^t \big\}$ collects all the potential outcomes at time $t$.
\end{remark}

\begin{assumption}[Output]
The output is $\left\{ W_t, Y_t \right\}_{ t \geq 1} = \left\{ W_t, Y_t \big( W_{1:t} \big)_{ t \geq 1} \right\}$. The $\left\{ Y_t \right\}_{ t \geq 1}$ is called the outcome process.     
\end{assumption}

Furthermore, we assume that the assignmemnt process is sequentially probabilistic, meaning that any assignment vector may be realized with positive probability at time $t$ given the history of the observable stochastic process up to time $t-1$. Let $\left\{ \mathcal{F}_t  \right\}_{ t \geq 1}$ denote the natural filtration generated by the realized $\left\{ w_t, y_t \right\}_{ t \geq 1}$.

\begin{assumption}[Sequentially probabilistic assignment process]
The assignment process satisfies $0 < \mathbb{P} \big( W_t = w \big| \mathcal{F}_{t-1} \big) < 1$ with probability one for all $w \in \mathcal{W}$. Moreover, the probabilities are determined by a filtered probability space of $\big\{ W_t, \big\{ Y_t ( w_{1:t} ), w_{1:t} \in \mathcal{W}^t \big\} \big\}_{t \geq 1}$.    
\end{assumption}

\begin{definition}
For $t \geq 1$, $h \geq 0$ and any fixed $w, w^{\prime} \in \mathcal{W}$, the time$-t$, $h-$period ahead:
\begin{itemize}

\item $Y_{t+h}( w ) - Y_{t+h}( w^{\prime} )$ (impulse causal effect). 

\item $\mathbb{E} \big[ Y_{t+h}( w ) - Y_{t+h}( w^{\prime} ) \big| \mathcal{F}_{t-1}  \big]$ (impulse causal effect). 

$\mathbb{E} \big[ Y_{t+h}( w ) - Y_{t+h}( w^{\prime} ) \big]$ (impulse causal effect). 
    
\end{itemize}
    
\end{definition}

\medskip

\begin{remark}
Various further applications of the aforementioned frameworks can be considered starting with the pioneered work of \cite{aigner1988optimal} who consider the optimal experimental design for error correction models, the econometric framework of \cite{de2020two} which corresponds to the two-way fixed effects estimators with heterogeneous treatment effects as well as the novel framework of \cite{lewis2020double} which corresponds to a time series setting based on machine learning and statistical learning techniques. However, we consider all aforementioned settings as more advanced topics and beyond the scope of our current discussion.
\end{remark}

\newpage

\subsection{Panel Data Regression for Dynamic Causal Effects}

\subsubsection{Panel Experiments and Dynamic Causal Effects}

Following the framework proposed by \cite{bojinov2021panel} we consider a design-based framework which provides a generalization of the finite population literature in cross-sectional causal inference and time series experiments to panel experiments. 

\begin{definition}
For $p < T$ and $\boldsymbol{w}, \tilde{\boldsymbol{w}} \in \mathcal{W}^{p+1}$,
\begin{itemize}
    \item[(i)] the time$-t$ lag$-p$ average dynamic causal effect is  
    \begin{align}
        \bar{\tau}_{ . t } ( \boldsymbol{w} , \tilde{\boldsymbol{w}} ; p ) := \frac{1}{N} \sum_{i=1}^N \tau_{i,t} ( \boldsymbol{w} , \tilde{\boldsymbol{w}} ; p ). 
    \end{align}

    \item[(ii)] the unit$-i$ lag$-p$ average dynamic causal effect is
    \begin{align}
        \bar{\tau}_{ i . } ( \boldsymbol{w} , \tilde{\boldsymbol{w}} ; p ) := \frac{1}{T-p} \sum_{t=p+1}^T \tau_{i,t} ( \boldsymbol{w} , \tilde{\boldsymbol{w}} ; p ). 
    \end{align}

    \item[(iii)] the total lag$-p$ average dynamic causal effect is
    \begin{align}
        \tau ( \boldsymbol{w} , \tilde{\boldsymbol{w}} ; p ) := \frac{1}{N(T-p)} \sum_{t=p+1}^T \sum_{i=1}^N \tau_{i,t} ( \boldsymbol{w} , \tilde{\boldsymbol{w}} ; p ).
    \end{align}
\end{itemize}
\end{definition}

\begin{example}
Suppose that we aim to generate the potential outcomes for the panel experiment using an autoregressive model formulated as below
\begin{align}
Y_{i,t} = \phi_{i,t,1} Y_{i,t-1} ( w_{i,1:t-1} ) + ... + \phi_{i,t,t-1} Y_{i,1} ( w_{i,1} )  + \beta_{i,t,0} w_{i,t} + ... + \beta_{i,t,t-1} w_{i,1} + \epsilon_{i,t}, \ \ \ \forall \ t > 1. 
\end{align}
In a more compact form the above specification can be written as $Y_{i,1} ( w_{i,1} ) = \beta_{i,1,0} w_{i,1} + \epsilon_{i,1}$, where
\begin{align}
\phi_{i,t,s} =
\begin{cases}
 \phi,  & \text{for} \ s = 1
\\
0,  & \text{for} \ s > 1. 
\end{cases}
\ \ \ \ \text{and} \ \ \ \
\beta_{i,t,s}
\begin{cases}
 \beta,  & \text{for} \ s = 0
\\
0,  & \text{for} \ s > 0. 
\end{cases}
\end{align}
Within a simulation setting we can vary the choice of $\phi$, which governs the persistence of the process, and $\beta$, which governs the size of the contemporaneous causal effects. Varying degrees of persistence in conjecture with the size of contemporaneous causal effects might cause size distortions for values close to the boundary of the parameter space. A robust method that accommodates for these the presence of these features in the aforementioned econometric environment is an interesting avenue for further research. In addition, we can vary the probability of treatment $p_{i,t-p} (w) = p(w)$ as well as the distribution of the errors $\epsilon_{i,t}$, which can be either sampled from a standard normal or a Cauchy distribution. 
\end{example}

\newpage 

\subsubsection{Identifying Dynamic Causal Effects with Panel Data: Applications}

\begin{example}[see, \cite{goes2016institutions}]
Consider the following panel SVAR(1) model defined as below
\begin{align}
\boldsymbol{B} \boldsymbol{y}_{i,t} = \boldsymbol{f}_i + \boldsymbol{A} (L) \boldsymbol{y}_{i,t-1} + \boldsymbol{e}_{i,t}, \ \ \ i \in \left\{ 1,..., N \right\}, \ \ t \in \left\{ 1,..., T \right\}   
\end{align}
where $y_{i,t} \equiv [ c_{i,t}, k_{i,t} ]^{\prime}$ is a bi-dimensional vector of stacked endogenous variables, such that $c_{i,t}$ is the log of GDP per capita and $k_{i,t}$ is the proxy for institutional quality, $\boldsymbol{f}_i$ is a diagonal matrix of time-invariant individual-specific intercepts. Moreover, $\boldsymbol{A} (L) = \sum_{j=0}^p \boldsymbol{A}_j L^j$ is a polynomial of lagged coefficients, $\boldsymbol{A}_j$ is a matrix of coefficients, and $\boldsymbol{e}_{i,t}$ is a vector of stacked residuals, and $\boldsymbol{B}$ is a matrix of contemporaneous coefficients.  However, since $\boldsymbol{f}_i$ is correlated to the error terms, estimation through OLS leads to biased coefficients. As proposed by \cite{baltagi2008econometric}, a strategy to obtain consistent parameters and eliminate individual fixed-effects when $N$ is large and $T$ is fixed, is to apply first-differencing and use lagged instruments. We consider the GMM/IV technique using a system of $m=2$ equations. Each equation in the system has the first difference of an endogenous variable on the left hand side, $p$ lagged first differences of all $m$ endogenous variables on the right hand side, and no constant. 
\begin{align}
\Delta y_{1,i,t} &= \sum_{j=1}^p \textcolor{blue}{\gamma_{11}^j} \Delta y_{1,i,t-j} + ... + \sum_{j=1}^p \textcolor{blue}{\gamma_{1m}^j } \Delta y_{m,i,t-j} + e_{1,i,t}
\\  
\nonumber
\vdots \ \ &= \ \ \vdots 
\\
\Delta y_{m,i,t} &= \sum_{j=1}^p \textcolor{blue}{ \gamma_{m1}^j } \Delta y_{1,i,t-j} + ... + \sum_{j=1}^p \textcolor{blue}{ \gamma_{mm}^j } \Delta y_{m,i,t-j} + e_{m,i,t}
\end{align}
Moreover, the model has an equivalent vector moving average $(VMA)$ representation which implies that the Panel SVAR model can be formulated as follows
\begin{align}
\boldsymbol{B} \boldsymbol{y}_{i,t} = \boldsymbol{\Phi}(L) \boldsymbol{e}_{i,t}, \ \ \  \boldsymbol{\Phi} (L) := \sum_{j=0}^{ \infty } \boldsymbol{\Phi}_j L^j \equiv \sum_{j=0}^{ \infty } \boldsymbol{A}_1^j L^j
\end{align}
is a polynomial of reduced-form responses to stochastic innovations and $\boldsymbol{\Phi}_0 = \boldsymbol{A}_1^0 \equiv \boldsymbol{I}_m$. 

Thus, to recover the $\boldsymbol{B}$ matrix and ensure robust identification, we first retrieve the variance-covariance matrix $\boldsymbol{\Sigma}_e = \mathbb{E} \left[ \boldsymbol{e}_{i,t} \boldsymbol{e}_{i,t}^{\prime} \right]$. Since, it holds that $\boldsymbol{B}^{-1} \boldsymbol{e}_{i,t} = \boldsymbol{u}_{i,t}$, then $\boldsymbol{\Sigma}_e = \mathbb{E} \left[ \boldsymbol{B} \boldsymbol{u}_{i,t} \boldsymbol{u}_{i,t}^{\prime} \boldsymbol{B}^{\prime}  \right]$. 
The structural shocks of the model are assumed to be uncorrelated, such that, $\boldsymbol{u}_{i,t} \boldsymbol{u}_{i,t}^{\prime} = \boldsymbol{I}_m$, we identify the matrix $\boldsymbol{B}$ by decomposing the variance-covariance matrix into two triangular matrices.  
Therefore, to identify the model we impose one restriction in order to orthogonalize the contemporaneous responses. 

In particular, using the Cholesky ordering, and based on the variables of interest, institutional quality is set to have no contemporaneous effect on GDP per capita while the latter is allowed to contemporaneously impact the former. Overall, the study of \cite{goes2016institutions} investigates the relation between institutions and economic growth using a structural panel VAR(1) model. Next, we focus on the approach to recover the IRs from the coefficient matrices.

\newpage 

Specifically, we take the following VMA representation of the Panel SVAR such that 
\begin{align}
\boldsymbol{B} \boldsymbol{M} (L) = \boldsymbol{e}_{i,t}.   
\end{align}
\end{example}

\begin{example}
We briefly discuss the econometric specification proposed by \cite{huber2023bayesian} which is suitable for dealing with cross-country heterogeneity in panel VARs using finite mixture models. Specifically, a global econometric identification and estimation model (see, also \cite{feldkircher2020factor}) is 
\begin{align}
\boldsymbol{y}_{it} = \boldsymbol{\beta}_i + \boldsymbol{A}_{i1} \boldsymbol{y}_{it-1} + ... +  \boldsymbol{A}_{ip} \boldsymbol{y}_{it-p} + \boldsymbol{B}_{i1} \boldsymbol{y}_{-it-1}  + ... + \boldsymbol{B}_{ip} \boldsymbol{y}_{-it-p} + \boldsymbol{\varepsilon}_{it}.
\end{align}
Under the Gaussian "state-of-the-world" one can impose Gaussian assumptions regarding the error covariances across countries. Thus, we can stack the country-specific errors $\boldsymbol{\varepsilon}_{it}$ in a $K-$dimensional vector $\boldsymbol{\varepsilon}_{t} = \left( \boldsymbol{\varepsilon}_{1t}^{\prime},..., \boldsymbol{\varepsilon}_{Nt}^{\prime} \right)^{\prime}$ and assume that  $\boldsymbol{\varepsilon}_t \sim \mathcal{N} \big( \boldsymbol{0}, \boldsymbol{\Sigma}_t \big)$, where $\boldsymbol{\Sigma}_t$ is a $( K \times K )$ dimensional covariance matrix. Furthermore, according to \cite{huber2023bayesian}, there are three important dimensions of model uncertainty which have been identified in the literature as summarized below:
\begin{itemize}

\item The first source of model uncertainty is concerned with modelling contemporaneous relations across the shocks in the system \textcolor{blue}{(static interdependancies)}.

\item The second source of model uncertainty focuses on the question whether coefficients associated with lagged domestic variables are homogeneous across countries \textcolor{blue}{(homogeneity restrictions)}. In particular, if such domestic coefficients are similar, then the so-called homogeneity restrictions might be imposed, effectively introducing the same set of coefficients for various countries and thus reducing the number of free parameters. 

\item The third source of model uncertainty deals with the question whether to allow for lagged dependencies between countries \textcolor{blue}{(dynamic interdependencies)}. Therefore, if we are interested to capture the international effects of climate shocks in globalized markets then capturing cross-country interdependencies on the macroeconomic level is crucial. 

\item  From the empirical contribution perspective one can focus on how climate shocks impact country-specific macroeconomic fundamentals. Then the empirical analysis of the novel framework proposed by   demonstrate that climate shocks are have a substantial impact on short-term interest rates and inflation, which is the primary conventional tool and target variable of central banks, and to a lesser degree on output and exchange rates (see, \cite{huber2023bayesian}). 

\end{itemize}
\end{example}

\begin{example}[Macro-finance Application] A causal effect question of interest from the macro-finance literature is the linkage between financial crises and economic contraction (see, \cite{huber2018disentangling}, \cite{romer2017new} and \cite{baron2021banking}). In particular, \cite{mei2023implicit} employ a local projection panel data estimator with fixed effects to study the link between financial distress and economic conditions deterioration. Thus, an econometric issue of interest is to overcome the so-called Nickell bias when considering classical FE estimators. Consider an $h-$period ahead panel LP model
\begin{align}
y_{i,t+h} = \mu_i^{(h)} + \beta^{(h)} x_{i,t} + \varepsilon_{i,t+h}^{(h)}, \ \ \text{for} \ t = 1,...., n-h, \ \ h = 0,1,...,H. 
\end{align}

\end{example}

\newpage

\section{Advanced Topics}

\subsection{Testing for Unstable Root in Structural ECMs}

We follow the framework proposed by \cite{boswijk1994testing}. Consider the single-equation error correction model of a time series $\left\{ y_t \right\}$ conditional upon the $( k \times 1 )$ vector time series $\left\{ z_t \right\}$, for $t \in \left\{ 1,... T \right\}$ as
\begin{align}
\Delta y_t 
= 
\beta_0^{\prime} \Delta z_t + \lambda \big( y_{t-1} - \theta^{\prime} \big) + \sum_{j=1}^{p-1} \big( \gamma_j \Delta y_{t-j} - \beta_j^{\prime} \Delta z_{t-j}  \big) + v_t,  
\end{align}
where $\left\{ v_t \right\}$ is an innovation process relative to $\left\{ z_t, y_{t-j}, z_{t-j}, j = 1,2,... \right\}$ with positive variance $\omega^2$. Notice that $\theta$ and $\beta_j$ are $( k \times 1)$ parameter vectors such that $j \in \left\{ 1,..., p-1 \right\}$, such that $\theta$ defines the long-run equilibrium relation $y = \theta^{\prime} z$, the deviations from which lead to a correction of $y_t$ by a proportion of $\lambda$, the adjustment or error correction coefficient. 

The conditional model is said to be stable if all roots of the characteristic equation
\begin{align}
\varphi ( \zeta ) = ( 1 - \zeta ) \left( 1 - \sum_{j=1}^{p-1} \gamma_j \zeta^j      \right) - \lambda \zeta = 0,   
\end{align}
are outside the unit circle. In other words, stability of the model implies that the disequilibrium error $\left( y_t - \theta^{\prime} z_t \right)$ is a stationary process, even though $z_t$ and $y_t$ are integrated of order one and hence nonstationary. Specifically, if the model is stable, then $x_t = (  y_t, z_t^{\prime} )^{\prime}$ is cointegrated with cointegrating vector $( 1, - \theta )^{\prime}$. Thus, the purpose of this exercise is to develop a class of tests for the null hypothesis that the characteristic equation has a unit root, so that the model is unstable, against the alternative hypothesis of stability. Furthermore, the single-equation conditional model can be seen as a special case if a structural error correction model. This is a system of Error Correction Equations for a $( g \times 1 )$ vector of time series $\left\{ y_t \right\}$ conditional upon $\left\{ z_t \right\}$ such that (see, \cite{boswijk1994testing}):
\begin{align}
\boldsymbol{\Gamma}_0 \Delta \boldsymbol{y}_t = \boldsymbol{B}_0 \Delta \boldsymbol{z}_t + \boldsymbol{\Lambda} \big( \boldsymbol{\Gamma} \boldsymbol{y}_{t-1} + \boldsymbol{B} \boldsymbol{z}_{t-1} \big) + \sum_{ j = 1}^{p-1} \big( \boldsymbol{\Gamma}_j \Delta \boldsymbol{y}_{t-j}  + \boldsymbol{B}_j \Delta \boldsymbol{z}_{t-j} \big)  + \boldsymbol{v}_t, 
\end{align}
where $\left\{ \boldsymbol{v}_t \right\}$ is an innovation process with a positive-definite covariance matrix $\boldsymbol{\Omega}$ such that the above expression corresponds to a parametrization of a conditional model of $y_t$ given $z_t$ (if $\boldsymbol{\Gamma}_0 \neq \boldsymbol{I}_g)$. Next we consider the identification of these matrix parameters. In particular, the above model implies $g$ cointegrating relationships $\boldsymbol{\Gamma} \boldsymbol{y} + \boldsymbol{B} \boldsymbol{z} = \boldsymbol{0}$, provided that it is stable, such that the characteristic equation is expressed as below
\begin{align}
\varphi ( \zeta ) = \left|  (1 - \zeta) \left( \boldsymbol{\Gamma}_0 - \sum_{j=1}^{p-1} \boldsymbol{\Gamma}_j \zeta^j \right)  - \boldsymbol{\Lambda} \boldsymbol{\Gamma} \zeta \right| = 0. 
\end{align}
has all roots outside the unit circle.

\newpage 

In other words, unless the parameters are restricted in some way, the structural model is not identified. Specifically, we identify the long-run relations by imposing restrictions of the usual form such that 
\begin{align}
\boldsymbol{\Gamma}_{ii} = 1 \ \ \ \boldsymbol{R}_i [ \boldsymbol{\Gamma}_i \  \  \boldsymbol{B}_i  ] =  \boldsymbol{0}, \ \ \ i \in \left\{ 1,..., g \right\},  
\end{align}
where $\boldsymbol{\Gamma}_i$ and $\boldsymbol{B}_i$ denote the $i-$th row of $\boldsymbol{\Gamma}$ and $\boldsymbol{B}$, respectively, and where $\boldsymbol{R}_i$ is a known matrix of approximate order. Thus, the rank condition for identification of the $i-$th long-run relation 
\begin{align}
\mathsf{rank} \big( \boldsymbol{R}_i \big[ \boldsymbol{\Gamma} \  \ \boldsymbol{B} \big]^{\prime} \big) = (g - 1)    
\end{align}
Then, the remaining parameters are identified by the normalization $\boldsymbol{\Gamma}_{0,ii} = 1$, and the restriction that $\boldsymbol{\Lambda} = \mathsf{diag} \big( \lambda_1,..., \lambda_g   \big)$, that is, where $\boldsymbol{\Lambda}$ is a diagonal matrix. In practice, this means that only the disequilibrium error of the $i-$th long-run relation appears in the $i-$th structural error correction equation which means that the $i-$th equation is (over)-identified if the $i-$th long-run equation is.  

\begin{remark}
The reason for restricting the error correction matrix $\boldsymbol{\Lambda}$ to be diagonal is twofold. Firstly, it allows for an interpretation of these separate equations as representing economic behaviour of a group of agents, whose target consists in a particular long-run relationship, such as a money demand relation or a consumption function. Furthermore, notice that the possibility that all endogenous variables are affected by each disequilibrium error is not excluded. However, this is considered to be a property of the reduced form of the system rather than the structural form. Moreover, imposing a symmetric matrix $\boldsymbol{\Lambda} = \mathsf{diag} \big( \lambda_1,..., \lambda_g   \big)$, facilitates the implementation and interpretation of a test for instability. Notice that we assume that the number of stable relationships is equal to the number of cointegrating relationships.   
\end{remark}
Consequently, under this null hypothesis, there is no error correction in the $i-$th equation, which suggests that the $i-$th  row of the system $\boldsymbol{\Gamma} y + \boldsymbol{B} \boldsymbol{z} = 0$ is not a cointegrating relationship. Specifically, let $z_t$ be generated by the following expression: 
\begin{align}
\Delta \boldsymbol{z}_t 
= 
\alpha \big( \boldsymbol{\Gamma} \boldsymbol{y}_{t-1}  + \boldsymbol{B} \boldsymbol{z}_{t-1} \big) + \sum_{j=1}^{p-1} \big( \boldsymbol{A}_{1j} \Delta \boldsymbol{y}_{t-j}  \boldsymbol{A}_{2j} \Delta \boldsymbol{z}_{t-j} \big) + \boldsymbol{\varepsilon}_t,     
\end{align}
Assumption 1, can be interpreted as a form of exogeneity condition, because it states that the cointegration properties of the conditional model carry over to the full VAR system. Understanding well these concepts are essential especially if one is interested to apply advanced identification techniques in SVARs such as the Markov switching regime (see, \cite{lanne2002threshold} and \cite{lanne2010structural}) as well as identification in Gaussian Mixture Vector Autoregressive Models (see,  \cite{kalliovirta2016gaussian} and \cite{meitz2023mixture}). In particular, \cite{kalliovirta2016gaussian} develop a framework for Gaussian mixture vector autoregression which is a nonlinear vector autoregression and is designed for analyzing time series that exhibit regime-switching dynamics\footnote{The property of long memory and regime switching is examined in the study of \cite{diebold2001long} while a regime switching panel data regression model is proposed by \cite{cheng2019regime}. For a forecasting application see \cite{nyberg2018forecasting}.} (see, \cite{kalliovirta2015gaussian} and \cite{virolainen2020structural}).  Next, we consider another example motivated from the structural break literature specifically in the context of cointegrated models (see, also \cite{hansen2002testing}).

\newpage

\subsubsection{Testing Cointegration Rank when some Cointegrating Directions are Changing}

Consider the following VECM of order $p$ with a structural break such that 
\begin{align*}
\Delta \boldsymbol{X}_t = \left( \boldsymbol{\alpha}_0 \boldsymbol{\beta}_0^{\prime} \boldsymbol{X}_{t-1} + \boldsymbol{\Phi}_0 \boldsymbol{D}_t \right) \boldsymbol{1} \left\{ t \leq k_0 \right\} + \left( \boldsymbol{\alpha}_1 \boldsymbol{\beta}_1^{\prime} \left( \boldsymbol{X}_{t-1} - \boldsymbol{X}_{k_0}  \right) + \boldsymbol{\Phi}_1 \boldsymbol{D}_t \right) \boldsymbol{1} \left\{ t > k_0 \right\} 
+ \sum_{j=1}^p \boldsymbol{\Gamma}_j \Delta \boldsymbol{X}_{t-j} + \boldsymbol{\varepsilon}_t,
\end{align*}
Then, the date of the break is characterized by the fraction $\lambda_0$ of the sample size $n$ such that $k_0 = \floor{\lambda_0}$, where $\lambda_0 \in [ \underline{\lambda}, \bar{\lambda} ]$ and $0 < \underline{\lambda} < \bar{\lambda}  < 1$. Moreover, the VECM is assumed to satisfy the stability condition over the two regimes as given by the Assumption below. 
\begin{assumption}
The roots of the characteristic polynomials of the $\mathsf{VECM}(1)$ 
\begin{align}
\mathsf{det} \left[  ( 1 - z ) \boldsymbol{I}_n - \alpha_i \beta_i^{\prime} z - \sum_{j=1}^p \Gamma_j ( 1 - z ) z^j \right], \ \ i = 0,1,   
\end{align}
satisfy $z = 1$ or $| z | > 1$.
\end{assumption}

\begin{assumption}
Suppose that $\left\{ \varepsilon_t \right\}$ is a vector martingale difference sequence with respect to $\mathcal{F}_{t-1}$ such that the variance-covariance matrix $\boldsymbol{\Sigma}_{\varepsilon} = \mathbb{E} [ \varepsilon_t \varepsilon_t^{\prime} | \mathcal{F}_{t-1} ]$, for $t = 1,2,...$.    
\end{assumption}

\begin{remark}
Following \cite{andrade2005testing}, we emphasize that the assumption above does not allow for a shift in the innovation covariance matrix $\Sigma_{\varepsilon}$ across the two regimes. This would split the model into two separate ones if the short-term coefficient matrices were also shifting with the break and would thus require a separate analysis over the two sub-samples. Similarly we do not allow for a different number of cointegration relationships across the two two regimes. In particular, we consider the following cases: 
\begin{itemize}

\item[Case 1.] The break (structural change) does not affect the loading factors. In this case their estimation requires an identification constraint. In other words, $\alpha_0$ and $\alpha_1$ span the same vector space and thus a natural identification scheme for the two regimes is given by $\alpha^{\prime} \Sigma_{\varepsilon} \alpha = I_r$.  

\item[Case 2.] The break affects the cointegration vectors and the loading factors. This implies that the spaces spanned by $\alpha_0$ and $\alpha_1$ differ across the two regimes.  
    
\end{itemize}
\end{remark}

\begin{theorem}[\cite{andrade2005testing}]
Consider the model defined by the first expression with $\alpha_0 = \alpha_1 = \alpha$ and Assumption 2. Under $H_0$, when $T \to \infty$ it holds that 
\begin{align*}
\xi_T ( \lambda_0 ) &\overset{d}{\to} \frac{1}{\lambda_0} \mathsf{trace} \left\{ \left( \int_0^{ \lambda_0 } d \boldsymbol{W}_0 \boldsymbol{G}_0^{\prime}  \right) \left( \int_0^{ \lambda_0 } \boldsymbol{G}_0 \boldsymbol{G}_0^{\prime} du \right)^{-1}  \left( \int_0^{ \lambda_0 } \boldsymbol{G}_0 d\boldsymbol{W}_0^{\prime}  \right) \right\}    
\\
&+
\frac{1}{ \left( 1 - \lambda_0 \right)} \mathsf{trace} \left\{ \left( \int_{ \lambda_0 }^1 d\boldsymbol{W}_0 \boldsymbol{G}_1^{\prime}  \right) \left( \int_{ \lambda_0 }^1 \boldsymbol{G}_1 \boldsymbol{G}_1^{\prime} du \right)^{-1}  \left( \int_{ \lambda_0 }^1 \boldsymbol{G}_1 d\boldsymbol{W}_0^{\prime}  \right) \right\}    
\end{align*}
\end{theorem}


\newpage 

\subsection{Empirical Likelihood Estimation Approach}

\begin{example}
Consider the following constant coefficient autoregressive model with time-varying variances as below: 
\begin{align}
Y_t &= \beta_0 + \beta_1 Y_{t-1} + \beta_2 Y_{t-2}  + ... + \beta_p Y_{t-p} + u_t,
\\
Y_t &= \boldsymbol{X}^{\top}_{t-1} \boldsymbol{\beta}_o + u_t, \ u_t = \sigma_t \varepsilon_t,  \ t = 1,...,T,
\end{align}
where the vector of covariates is denoted by $\boldsymbol{X}_t = \big( 1, Y_{t-1},..., Y_{t-p}  \big)^{\top} \in \mathbb{R}^{ p + 1}$ and the true model parameter of interest is denoted by $\boldsymbol{\beta}_o = \big( \beta_0, \beta_1,..., \beta_p \big)^{\top} \in \mathbb{R}^{ p + 1}$, where the lag order is finite and known.  
\end{example}

Based on the aforementioned assumptions the estimation of the unknown parameter vector $\boldsymbol{\beta}_o$ based on the OLS estimator $\hat{\boldsymbol{\beta}}$ is given by 
\begin{align}
\sqrt{T} \left( \boldsymbol{\beta} - \boldsymbol{\beta}_o \right) 
= 
\left( \frac{1}{T} \sum_{t=1}^T \boldsymbol{X}_{t-1}^{\top}  \boldsymbol{X}_{t-1} \right)^{-1} \left( \frac{1}{T} \sum_{t=1}^T \boldsymbol{X}_{t-1}^{\top}  \varepsilon_t \right)  \overset{d}{\to} \mathcal{N} ( \boldsymbol{0}, \boldsymbol{\Lambda} ),
\end{align}
where $\boldsymbol{\Lambda} = \boldsymbol{\Omega}_1^{-1} \boldsymbol{\Omega}_2 \boldsymbol{\Omega}_1^{-1}$, are defined as $(p+1) \times (p+1)$ matrices.

To construct an empirical likelihood function, the estimation equations are defined as below: 
\begin{align}
\boldsymbol{W}_t ( \boldsymbol{b} )  = \boldsymbol{X}_{t-1} \cdot \left(  Y_t - \boldsymbol{X}_{t-1}^{\top} \boldsymbol{b} \right),  
\end{align}
for a generic parameter $\boldsymbol{b} \in \mathbb{R}^{p+1}$. By using the Lagrange multipliers method, we have that $\hat{\boldsymbol{\lambda}} = \hat{\boldsymbol{\lambda}} ( \boldsymbol{b} ) \in \mathbb{R}^{p+1}$ is the solution of the following set of equations: 
\begin{align}
\frac{1}{T} \sum_{t=1}^T \frac{ \boldsymbol{W}_t ( \boldsymbol{b} )   }{ 1 + \hat{\boldsymbol{\lambda}}^{\prime} \cdot \boldsymbol{W}_t ( \boldsymbol{b} )   } = \boldsymbol{0}.    
\end{align}
Then, the corresponding empirical log-likelihood ratio is given by
\begin{align}
\ell ( \boldsymbol{b} ) = 2 \sum_{t = 1}^T \mathsf{log} \big[ 1 + \hat{\boldsymbol{\lambda}}^{\prime} \cdot \boldsymbol{W}_t ( \boldsymbol{b} )   \big]    
\end{align}
and it holds that $\ell ( \boldsymbol{\beta}_0 )  \overset{d}{\to} \chi_p^2$, as $T \to \infty$.

\newpage

\subsubsection{High Dimensional Generalized Empirical Likelihood Estimation}

In this section, we discuss relevant aspects to the empirical likelihood estimation for high dimensional dependent data, which is applicable to time series regression models (see, \cite{chang2015high}). Further frameworks related to identification and estimation of high-dimensional time series with SVAR models are  proposed by \cite{krampe2023structural}, \cite{zhang2023statistical} and \cite{adamek2022local}. In particular, denote with $\boldsymbol{\theta} = ( \theta_1,..., \theta_p )^{\prime}$ be a $p-$dimensional parameter taking values in a parameter space $\Theta$. Consider a sequence of $r-$dimensional estimating equation such that
\begin{align}
\mathsf{g} ( X_t, \boldsymbol{\theta} ) = \big( \mathsf{g}_1 ( X_t, \boldsymbol{\theta} ),....,  \mathsf{g}_r ( X_t, \boldsymbol{\theta} )  \big)    
\end{align}
for some $r \geq p$, Then, the model information regarding the data and the data parameter is summarized by moment restrictions below:
\begin{align}
\mathbb{E} \big[  \mathsf{g} ( X_t, \boldsymbol{\theta}_0 )  \big] = \boldsymbol{0}.
\end{align}
where $\theta_0 \in \Theta$ is the true parameter. Furthermore, in order to preserve the dependence structure among the underlying data, we employ the blocking technique. Let $M$ and $L$ be two integers denoting the block length and separation between adjacent blocks, respectively. Then, the total number of blocks is 
\begin{align}
Q = \floor{ \frac{ (n - M )}{L} }  + 1   
\end{align}
Then, the EL estimator $\boldsymbol{\theta}_o$ is $\widehat{\boldsymbol{\theta}}_{EL} = \mathsf{arg max}_{ \theta in \Theta } \ \mathsf{log} \ \mathcal{L} ( \boldsymbol{\theta} )$. Consequently, the maximization problem can be carried out more efficiently by solving the corresponding dual problem, which implies that $\widehat{\boldsymbol{\theta}}_{EL}$ can be obtained as below: 
\begin{align}
\widehat{\boldsymbol{\theta}}_{EL} = \underset{ \boldsymbol{\theta} \in \Theta   }{ \mathsf{arg \ min} } \ \underset{ \lambda \in \widehat{\Lambda}_n ( \boldsymbol{\theta} ) }{ \mathsf{max} } \ \sum_{q = 1}^Q \mathsf{log} \big[ 1 + \boldsymbol{\lambda}^{\top} \phi_M (B_q, \boldsymbol{\theta} ) \big],  
\end{align}
\begin{align}
\widehat{\Lambda}_n ( \boldsymbol{\theta} ) := \left\{ \lambda \in \mathbb{R}^r : \lambda^{\top} \cdot \phi_M (B_q, \boldsymbol{\theta} ) \in \mathcal{N}, q = 1,..., Q \right\}    
\end{align}
for any $\boldsymbol{\theta} \in Q$ and $\mathcal{N}$ an open interval containing zero. 

\begin{example}[Time Series Regression]
Consider a structural model $s-$dimensional time series $Y_t$ which involve unknown parameter $\boldsymbol{\theta} \in \mathbb{R}^p$ of interest as well as time innovations with unknown distributional form. Thus, suppose we have that
\begin{align}
h \big( Y_t,..., Y_{t-m}; \boldsymbol{\theta}_0 \big) = \boldsymbol{\varepsilon}_t \in \mathbb{R}^r    
\end{align}
where $m \geq 1$ is some constant. In particular, for conventional vector autoregressive models such that 
\begin{align}
Y_t = \boldsymbol{A}_1 Y_{t-1} + ... + \boldsymbol{A}_m Y_{t-m} + \eta_t,   
\end{align}

\newpage

Notice that these set of model parameters $\left\{ \boldsymbol{A}_1 ,..., \boldsymbol{A}_m \right\}$ correspond to coefficient matrices and are usually estimated based on some maximum likelihood approach (usually the MLS under Gaussianity) and according to the distributional assumptions of the model. Moreover, suppose that $\eta_t$ is the white noise series such that 
\begin{align}
h \big( Y_t,..., Y_{t-m}; \boldsymbol{\theta}_0 \big) 
= 
\big( Y_t - \boldsymbol{A}_1 Y_{t-1} - ... - \boldsymbol{A}_m Y_{t-m}  \big) \otimes \big( Y_t^{\top},...., Y_{t-m}^{\top} \big)^{\top}.
\end{align}
High dimensional time series analysis, requires to assume that the dimensionality of $Y_t$ is large in relation to sample size, that is, $s \to \infty$ as $n \to \infty$. Within a high-dimensional environment, the number of estimating equation and unknown parameters are both $s^2 m$. On the other hand, if we replace $\big( Y_t^{\top},...., Y_{t-m}^{\top} \big)^{\top}$ by $\big( Y_t^{\top},...., Y_{t-m - \ell}^{\top} \big)^{\top}$ for some fixed $\ell \geq 1$, then the model will be over-identified. This phenomenon of over-parametrization in such models is well-known to the literature. Thus, to implement a consistent estimation approach, the sparsity assumption allows to employ a penalized estimation methodology. 
\end{example}

\subsection{High Dimensional VARs with Common factors}

We follow the framework proposed by \cite{miao2023high} who study high-dimensional vector autoregressions (VARs) augmented with common factors that allow for strong cross-sectional dependence. This approach allows to incorporate in a unified framework a convinient mechanism for accommodating the interconnectedness and temporal co-variability that are often present in large dimensional systems.  

Consider the $N-$dimensional vector-valued time series $\left\{ Y_t \right\} = \left\{ \left( y_{1t},..., y_{NT} \right)^{\prime} \right\}$, the high-dimensional VAR model of order $p$ with CFs given by 
\begin{align}
\boldsymbol{Y}_t = \sum_{j=1}^p \boldsymbol{A}_j^0 \boldsymbol{Y}_{t-j} +  \boldsymbol{\Lambda}^0 \boldsymbol{f}_t^0 + \boldsymbol{u}_t, \ \ \ t = 1,...,T,   
\end{align}
where $\boldsymbol{A}_1^0,..., \boldsymbol{A}_p^0$ are the $( N \times N )$ transition matrices and $\boldsymbol{u}_t$ is an $N-$dimensional vector of unobserved idiosyncratic errors. Moreover, the analytical framework allows for both the number of cross-sectional units $N$ and the number of time periods $T$ to pass to infinity. The lag length is also allowed to (slowly) grow to infinity with $(N,T)$. Estimation then is a natural high-dimensional problem.  

Then, the $N-$dimensional VAR$(p)$ process $\left\{ Y_t \right\}$ can be rewritten in a companion form as an $Np-$dimensional VAR$(1)$ process with common factors such that: 
\begin{align}
\underbrace{
\begin{bmatrix}
Y_t
\\
Y_{t-1}
\\
\vdots
\\
Y_{t-p-1}
\end{bmatrix} 
}_{ \boldsymbol{X}_{t+1}  }
= 
\underbrace{
\begin{bmatrix}
A_1^0 & A_2^0 & \hdots &  A_{p-1}^0  &  A_p^0
\\
I_N & \boldsymbol{0} & \hdots & \boldsymbol{0} & \boldsymbol{0}
\\
\vdots & \vdots & \ddots & \vdots & \vdots
\\
\boldsymbol{0} & \boldsymbol{0}  & \hdots & \boldsymbol{0} & \boldsymbol{0} 
\end{bmatrix}
}_{ \boldsymbol{\Phi} }
\underbrace{
\begin{bmatrix}
Y_{t-1}
\\
Y_{t-2}
\\
\vdots
\\
Y_{t-p}
\end{bmatrix}}_{ \boldsymbol{X}_t } 
+
\underbrace{
\begin{bmatrix}
\boldsymbol{\Lambda}^0 \boldsymbol{f}_t^0
\\
\boldsymbol{0}
\\
\vdots
\\
\boldsymbol{0}
\end{bmatrix} 
}_{ \boldsymbol{F}_t }
+ 
\underbrace{
\begin{bmatrix}
\boldsymbol{u}_t
\\
\boldsymbol{0}
\\
\vdots
\\
\boldsymbol{0}
\end{bmatrix}.
}_{ \boldsymbol{U}_t }
\end{align}

\newpage

As a result, the reverse characteristic polynomial of $Y_t$ can be written as below: 
\begin{align}
\mathcal{A}(z) \equiv I_N - \sum_{j=1}^p A_j^0 z^p.     
\end{align}
In the low-dimensional framework, the process is stationary if $\mathcal{A}(z)$ has no roots in and on the complex unit circle, or equivalently the largest modules of the eigenvalues of $\boldsymbol{\Phi}$ is less than 1. Therefore, to achieve identification, we shall study the Gram or signal matrix $\boldsymbol{S}_X = \boldsymbol{X}^{\prime} \boldsymbol{X}/ T$ and its population counterpart $\Sigma_X = \mathbb{E} \left( \boldsymbol{X}^{\prime}_t \boldsymbol{X}_t \right)$. In other words, one can study the deviation bounds for the Gram matrix, under the Gaussianity assumption and boundedness of the spectral density function. 

In order to ensure that the matrix $\Sigma_X$ is well-behaved, we write $\boldsymbol{X}_{t+1}$ as a moving average process of infinite order MA$( \infty )$ such that 
\begin{align}
X_{t+1} \equiv \sum_{j=0}^{\infty} \boldsymbol{\Phi}^j \big(  \boldsymbol{F}_{t-j} +  \boldsymbol{U}_{t-j}  \big)
= 
\sum_{j=0}^{\infty} \boldsymbol{\Phi}^j \boldsymbol{F}_{t-j}  +  \sum_{j=0}^{\infty} \boldsymbol{\Phi}^j \boldsymbol{U}_{t-j}  
\end{align}

\textbf{Eigenvalue Analysis:}

\begin{itemize}

\item First, consider  $X_{t+1}^{(f)} = \sum_{j=0}^{\infty}  \boldsymbol{\Phi}^j \boldsymbol{F}_{t-j}$, the component due to the common factors. The covariance matrix of $\boldsymbol{F}_{t}$ is a high-dimensional matrix with rank $R^0$ and explosive non-zero eigenvalues. In other words, even if the largest modules of the eigenvalues of $\boldsymbol{\Phi}$ is smaller than 1, the variances of the entries of $X_{t+1}^{(f)}$ are not assumed to be uniformly bounded. 

\item Consider $y_{it}^{(f)}$, which is the $i-$th entry of $X_{t+1}^{(f)}$. Let $e_{j,M}$ be the $j-$th column of $I_M$. Noting that $y_{it}^{(f)} = \big( e_{1,p} \otimes e_{i,N} \big)^{\prime} X_{t+1}^{(f)}$, we can write $y_{it}^{(f)}$ as the MA$(\infty)$ process given below: 
\begin{align}
y_{it}^{(f)} = \sum_{j=0}^{\infty} \big( e_{1,p} \otimes e_{i,N} \big)^{\prime} \Phi^j  \big( e_{1,p} \otimes e_{i,N} \big) f_{t-j}^0 \equiv \sum_{j=0}^{\infty} \alpha_{iN}^{(f)} (j) f_{t-j}^0,  
\end{align}
in which $f_t^0$ are allowed to be serially correlated.     
\end{itemize}

\begin{assumption}[\cite{miao2023high}]
Consider that the following conditions hold: 
\begin{itemize}

\item[\textit{(i)}.] Let $u_t = C^{(u)} \varepsilon_t^{(u)}$, where $\varepsilon_t^{(u)} = \big( \varepsilon_{1,t}^{(u)},..., \varepsilon_{m,t}^{(u)} \big)^{\prime}$ such that $\varepsilon_{i,t}^{(u)}$ are $\textit{i.i.d}$ random variables across $(i,t)$ with mean zero and variance 1. 

\item[\textit{(ii)}.] $\left\{ f_t^0 \right\}$ follows a strictly stationary linear process given as below:
\begin{align}
f_t^0 - \mu_f = \sum_{j=0}^{\infty} C_j^{(f)} \epsilon^{(f)}_{t-j},    
\end{align}
    
\end{itemize}

\end{assumption}

\newpage

\subsection{Fast Algorithm for Detection of Breaks in Large VAR models}

Currently in the literature there are two approaches in terms of the identification of break-points. Firstly, there is a so-called "top down" procedure, in the sense that one tests all the data to determine if there is at least one change-point and iterates the procedure in the intervals immediately to the "left" and "right" of the most recently detected change-point. However, the particular methodology has certain weaknesses compared to a number of other thresholding procedures. A second approach is the so-called "bottom up" procedure motivated by the observation in the presence of multiple change-points or to mitigate the effects of inadequately controlled drift in the "baseline" mean value, it may be useful to compare a candidate change-point at $j$ to an appropriate "local" background $(j,k)$, where $i < j < k$. Similar approaches are the Wild Binary Segmentation (WBS) methodology which uses a random set of possible backgrounds and an apparently determined threshold   (see, \cite{fang2020segmentation}). 

\subsubsection{Illustrative Examples}

\begin{example}
Suppose we are interested in testing whether there is a change in the covariance matrix of the stationary multivariate time series $\left\{ \boldsymbol{Y}_{n1},..., \boldsymbol{Y}_{n1} \right\}$. Then, this implies that 
\begin{align}
\boldsymbol{Y}_{ni} := \boldsymbol{Y}_{n1}^{(0)} \boldsymbol{1} \left\{ i \leq k \right\} +  \boldsymbol{Y}_{n1}^{(1)} \boldsymbol{1} \left\{ i > k \right\}, \ \ \ 1 \leq i \leq n,   
\end{align}
for some change-point in the interval $1 \leq k \leq n$, where $\boldsymbol{Y}_{n1}^{(0)}$ denotes the stationary time series before and change-point and $\boldsymbol{Y}_{n1}^{(1)}$ denotes the stationary time series after the change-point. Therefore, we assume that it induces a change to the structure of the covariance matrix which implies the following expressions 
\begin{align}
\boldsymbol{\Sigma}_n^{(0)} := \mathsf{Var} \left( Y_{n,1}^{(0)} \right) \ \ \ \text{and} \ \ \  \boldsymbol{\Sigma}_n^{(1)} := \mathsf{Var} \left( Y_{n, k+1}^{(1)} \right) 
\end{align}
Denote with $\boldsymbol{\Sigma}_n( [j] ) := \mathsf{Var} \left( \boldsymbol{Y}_{nj} \right)$,  $1 \leq j \leq n$. Testing the change-point problem at a single location 
\begin{align}
H_0:  \boldsymbol{\Sigma}_n( [j] ) \equiv \boldsymbol{\Sigma}_n^{(0)} \ \ \ \forall \ j = 1,...,n \ \ \ \text{and} \ \ \ H_1: \exists \ k \in \left\{ 2,..., n \right\}: \boldsymbol{\Sigma}_n( [k] ) \neq \boldsymbol{\Sigma}_n^{(0)},    
\end{align}
which tests the null hypothesis of stability of the covariance structure against the alternative hypothesis of a change.  An equivalent formulation of the statistical testing problem under consideration is given as 
\begin{align}
H_0: k = n \ \ \ \text{against} \ \ \ H_1: k < n.     
\end{align}
The proposed change-point detection test statistics are based are based on the partial sums of the stationary multivariate time series such that $\boldsymbol{S}_{nk} = \sum_{i \leq k} \boldsymbol{Y}_{ni} \boldsymbol{Y}_{ni}^{\prime}, \ \ \ k \geq 1$.  Under the assumption that the true change-point $k$ is proportional to the sample size, $k = \floor{ \lambda n }$, where $0 < r < \lambda < 1$, the unknown break-point can be consistently estimated using the following estimator
\begin{align}
\hat{k}_n := \underset{ n_0 \leq k \leq n }{ \mathsf{arg max} } \ \frac{1}{n} \left| \boldsymbol{v}_n^{\prime} \left( \boldsymbol{S}_{nk} - \frac{k}{n} \boldsymbol{S}_{nn} \right) \boldsymbol{w}_n \right|    
\end{align}
\end{example}

\newpage 

The formulation of the test statistic presented in the above example, implies that is the (smallest) time index $k \geq n_0 := \floor{nr}$, for some small $r > 0$, at which the maximum in the definition of the test statistic $C_n$ is attained. On the other hand, in order to robustify these CUSUM-type test statistics for detecting changes not only in the middle of the sample, the following statistic is proposed
\begin{align}
C_{n,md} = \frac{ T_{n, md } }{ \hat{\alpha}_n }, \ \ \ \text{with} \ \   T_{n, md } = \underset{ 1 \leq j \leq n }{ \mathsf{max} } \ \frac{1}{ \sqrt{n} } \left|  \sum_{1 \leq \ell \leq j} \boldsymbol{v}_n^{\prime} \left( \boldsymbol{S}_{n \ell} - \frac{( j - i)}{n} \boldsymbol{S}_{nn} \right) \boldsymbol{w}_n  \right|  
\end{align}
which considers the maximal deviation from the average of the cumulated sums, the maximum being taken over all possible subsamples $\boldsymbol{Y}_{ n \ell}, i \leq \ell \leq j$, for $1 \leq i < j \leq n$. To be more precise, the particular algorithm implies that each subsample induces a corresponding CUSUM statistic and then the maximal test statistic is taken over all these subsamples. Therefore, it can be proved 
\begin{align}
T_{n,md} \overset{d}{\to} \ \underset{ 0 < s < t < 1 }{ \mathsf{sup} } \ \left| B^o(s) - B^o(t) \right|, \ \ \text{as} \ n \to \infty.  
\end{align}
Further studies related to testing for structural breaks in covariance structures include \cite{kao2018testing}. In particular, $\widehat{\boldsymbol{\Sigma}}$ be the sample covariance matrix such that $\widehat{\boldsymbol{\Sigma}} = \frac{1}{T} \sum_{t=1}^T y_t y_t^{\prime}$. Then, for a given $\uptau \in [0,1]$, we define a point in time $\floor{ T \uptau }$, and we use the subscripts  $\uptau$ and $1 - \uptau$ to denote quantities calculated using the subsamples $t = 1,..., \floor{ T \uptau }$ and $t = \floor{ T \uptau } + 1,...,T$, respectively. In particular, we consider the sequence of partial sample estimators 
\begin{align}
\widehat{\boldsymbol{\Sigma}}_{ \uptau } = \displaystyle \frac{1}{T \uptau} \sum_{t=1}^{ \floor{ T \uptau } } y_t y_t^{\prime} \ \ \ \ \text{and} \ \ \ \ \widehat{\boldsymbol{\Sigma}}_{ 1 - \uptau } = \displaystyle \frac{1}{T (1 - \uptau )} \sum_{t = \floor{ T \uptau } + 1}^{ T } y_t y_t^{\prime}
\end{align}
We denote with $w_t = \mathsf{vec} \big( y_t y_t^{\prime} \big)$ and $\bar{w}_t = \mathsf{vec} \big( y_t y_t^{\prime} - \boldsymbol{\Sigma} \big)$.

\subsubsection{Main Asymptotic theory results}

\begin{assumption}
\label{Assumption1}
Let $\mathsf{sup}_t \mathbb{E} \norm{ y_t }^{2r} < \infty$ for some $r > 2$. Then, we define with 
\begin{align}
V_{ \Sigma, T} = \frac{1}{T} \mathbb{E} \left[ \left( \sum_{t=1}^T \bar{w}_t \right) \left( \sum_{t=1}^T \bar{w}_t \right)^{\prime} \right]
\end{align}
\end{assumption}

\begin{theorem}
Under Assumption \ref{Assumption1}, as $T \to \infty$
\begin{align}
\frac{1}{\sqrt{T}} \sum_{t=1}^{ \floor{T \uptau} } \bar{w}_t \to \big[ V_{ \Sigma } \big]^{1/2} W_{n^2} (\uptau), 
\end{align} 
uniformly in $\uptau$. Redefining $\bar{w}_t$ in a richer probability space, under Assumptions 1, there exists a $\delta > 0$ such that $
\sum_{t=1}^{ \floor{T \uptau} } \bar{w}_t = \sum_{t=1}^{ \floor{T \uptau} } X_t + O_{a.s} \left( \floor{T \uptau}^{\frac{1}{2} - \delta} \right)$, uniformly in $\uptau$, where $X_t$ is a zero mean, $i.i.d$ Gaussian sequence with $\mathbb{E} \left( X_t X_t^{\prime}  \right) = V_{\Sigma}$. Next, we focus on the estimation of $\boldsymbol{V}_{\Sigma}$.
\end{theorem}

\newpage

 In particular, if no serial dependence is present, a possible choice is the full sample estimator 
\begin{align}
\widehat{\boldsymbol{V}}_{\Sigma} = \frac{1}{T} \sum_{t=1}^T w_t w_t^{\prime} - \left[ \mathsf{vec} \big( \widehat{\boldsymbol{\Sigma}} \big)  \right] \left[ \mathsf{vec} \big( \widehat{\boldsymbol{\Sigma}} \big)  \right]^{\prime}
\end{align}
Alternatively, one could use the sequence of partial sample estimators such that 
\begin{align*}
\widehat{\boldsymbol{V}}_{\Sigma} = \frac{1}{T} \sum_{t=1}^T w_t w_t^{\prime}
- 
\bigg\{ \uptau \left[ \mathsf{vec} \big( \widehat{\boldsymbol{\Sigma}}_{\uptau} \big)  \right] \left[ \mathsf{vec} \big( \widehat{\boldsymbol{\Sigma}}_{\uptau} \big)  \right]^{\prime}  
+ (1 - \uptau ) \left[ \mathsf{vec} \big( \widehat{\boldsymbol{\Sigma}}_{1-\uptau} \big)  \right] \left[ \mathsf{vec} \big( \widehat{\boldsymbol{\Sigma}}_{1-\uptau} \big)  \right]^{\prime}
\bigg\}.
\end{align*}
Furthermore, to accommodate for the case in which $\Psi_{\ell} \equiv \mathbb{E} \big( \bar{w}_t \bar{w}_{t- \ell} \big) \neq 0$ for some $\ell$, we propose a weighted sum-of-covariance estimator with bandwidth $m$, such that
\begin{align}
\tilde{\boldsymbol{V}}_{\Sigma} 
&= 
\widehat{\Psi}_0 + \sum_{\ell = 1}^m \left( 1 - \frac{\ell}{m} \right) \left[ \widehat{\Psi}_{\ell} +  \widehat{\Psi}_{\ell}^{\prime} \right],
\\
\widehat{\Psi}_{\ell} 
&= 
\frac{1}{T} \sum_{t = \ell + 1}^T \big[ w_t - \mathsf{vec} \left( \widehat{\boldsymbol{\Sigma}}  \right) \big] \big[ w_{t-\ell} - \mathsf{vec} \left( \widehat{\boldsymbol{\Sigma}}  \right) \big]^{\prime}
\end{align}
Alternatively, we can employ the following matrix
\begin{align}
\tilde{\boldsymbol{V}}_{\Sigma, \uptau} 
&= 
\left( \widehat{\Psi}_{0, \uptau} +  \widehat{\Psi}_{0, 1 - \uptau} \right) \sum_{\ell = 1}^m \left( 1 - \frac{\ell}{m} \right) \left[ \left( \widehat{\Psi}_{\ell, \uptau} +  \widehat{\Psi}_{\ell, \uptau}^{\prime} \right) + \left( \widehat{\Psi}_{\ell, 1 - \uptau} +  \widehat{\Psi}_{\ell, 1 - \uptau}^{\prime} \right) \right]
\\
\widehat{\Psi}_{\ell, \uptau} 
&= 
\frac{1}{T} \sum_{t = \ell + 1}^{ \floor{ T \uptau } }  \big[ w_t - \mathsf{vec} \left( \widehat{\boldsymbol{\Sigma}}_{\uptau}   \right) \big] \big[ w_{t-\ell} - \mathsf{vec} \left( \widehat{\boldsymbol{\Sigma}}_{\uptau}  \right) \big]^{\prime}
\end{align}
A similar expression holds for $\widehat{\Psi}_{\ell, 1 - \uptau}$. Furthermore, denote with 
\begin{align}
\Omega_{T} = \frac{1}{T} \mathbb{E} \left\{ \sum_{t=1}^T \mathsf{vec} \big[ \bar{w}_t \bar{w}_t^{\prime} - \mathbb{E} \left( \bar{w}_t \bar{w}_t^{\prime} \right) \big] \times \mathsf{vec} \big[ \bar{w}_t \bar{w}_t^{\prime} - \mathbb{E} \left( \bar{w}_t \bar{w}_t^{\prime} \right) \big]^{\prime} \right\}
\end{align}
where $\Omega_T$ is positive definite uniformly in $T$, and $\Omega_T \to \Omega$ with $\norm{\Omega} < \infty$. 
\begin{theorem}
Under the null hypothesis of no structural break in $\Sigma$ then, if Assumption \ref{Assumption1} holds, as $T \to \infty$, there exists a $\delta_o > 0$ such that 
\begin{align}
\underset{ 1 \leq \floor{T \uptau } \leq T }{ \mathsf{sup} } \norm{ \hat{\boldsymbol{V}}_{\Sigma, \uptau} - \boldsymbol{V}_{\Sigma} } = o_p \left(  \frac{1}{ T^{ \delta_o } }   \right).
\end{align}  
\end{theorem}

\begin{remark}
\cite{aue2009break} consider the detection of structural breaks in the covariance structure of multivariate time series under the assumption of stationarity and ergodicity. Their approach identifies breaks in the unconditional mean of a $p-$dimensional multivariate time series. Moreover, \cite{lee2023change}  propose a change point test for structural vector autoregressive models via independent component analysis. The test statistic is constructed based on the fitted residuals and the CUSUM functional.    
\end{remark}

\newpage 

Furthermore, a growing literature develops methodologies for testing  structural breaks in factor models as in \cite{breitung2011testing}, \cite{baltagi2021estimating} and  \cite{bai2024likelihood} among others (see, also \cite{bai1998estimating}). Regarding the literature relevant to structural break testing around the Vector Autoregressive framework of interest there are two main approaches. On the one hand, the approach of \cite{safikhani2022fast} implies a fast and scalable algorithm for detection of structural breaks in large VAR models. The particular modelling environment corresponds to a piecewise linear Vector Autoregression model. Thus, testing for structural breaks takes the form of a time-varying VAR model and therefore in practise the statistical problem of interest is to test for the presence of time-varying vector autoregressive matrix coefficients across the blocks. Moreover, since the model is parametrized in a form of a large VAR model with a corresponding companion matrix, then the focus is on the fast and accurate break detection algorithm with respect to the dimensionality of the model. On the other hand, \cite{gourieroux2020identification} proposed an econometric framework for identification and estimation of the structural vector autoregressive model where the main focus is how the distributional assumptions on the shocks affect the identification procedure. In practice, based on the non-Gaussianity assumption of the structural shocks the model is identifiable and the estimation of model parameters is achieved via the use of the likelihood function. Therefore, testing for structural breaks in such a setting could be implemented using an appropriate specification which captures time-variation. 

\subsubsection{Large VAR Models: Model Formulation}

 For the remaining of this section, we follow the framework proposed by \cite{safikhani2022fast}. Their procedure is found to be scalable to big high-dimensional time series datasets with a computational complexity that can achieve $O ( \sqrt{n} )$, where $n$ is the length of the time series (sample size), compared to an exhaustive procedure that requires $O(n)$ steps.  

Suppose that there exist $m_0$ break points $0 < t_1 < ... t_{m_0} < T$, with $t_0 = 0$ and $t_{ m_0 + 1 } = T$ where $j \in \left\{ 1,..., m_0 + 1\right\}$. Notice that if there are no break points, then $m_0 = 0$ and so $( t_0 = 1, t_1 = T )$ which corresponds to the full sample period. 

Thus, we can assume that $m_0 \geq 1$, which implies that $( t_0 = 1, t_2 = T )$ and so $j = \left\{ 1, 2 \right\}$. In other words, for $t_{j-1} \leq t < t_j$ (e.g., $t_0 \leq t < t_1$, i.e., $t = 1,..., T$) we have the following model:
\begin{align}
y_t = \Phi_{ (1,j) } y_{t-1} + ... + \Phi_{ (q,j) } y_{t-q} + \Sigma_j^{1/2} \varepsilon_t 
\end{align}
Since in the case that $m_0 = 1$ then $j \in \left\{ 1,2\right\}$ then we have a two-regime model estimation:
\begin{align}
M_0: \boldsymbol{y}_t &= \boldsymbol{\Phi}_{ (1,1) } \boldsymbol{y}_{t-1} + ... + \boldsymbol{\Phi}_{ (q,1) } \boldsymbol{y}_{t-q} + \boldsymbol{\Sigma}_1^{1/2} \boldsymbol{\varepsilon}_t, \ \ \ t_0 \leq t < t_1, \ \ t \in \left\{ 0,..., t_1 \right\},
\\
M_1: \boldsymbol{y}_t &= \boldsymbol{\Phi}_{ (1,2) } \boldsymbol{y}_{t-1} + ... + \boldsymbol{\Phi}_{ (q,2) } \boldsymbol{y}_{t-q} + \boldsymbol{\Sigma}_2^{1/2} \boldsymbol{\varepsilon}_t, \ \ \ t_1 \leq t < t_2, \ \ t \in \left\{ t_1,..., T \right\},
\end{align}
There is a single break-point and we fit a different VAR model before the break, i.e. $M_0$ pre-break based on the observations up to the unknown break-point, and a we fit $M_1$ post-break based on the observations from the unknown break-point and onward until the end of the full sample (see, \cite{safikhani2022fast}).

\newpage 

\begin{itemize}

\item In each segment we estimate the VAR process. Let $\boldsymbol{y}_t$ be a $p-$dimensional vector of observations at time $t$ and $\boldsymbol{\Phi}_{(\ell, j)}$ be an $( p \times p )$ sparse coefficient matrix corresponding to the $\ell-$th lag of a VAR process of order $q$ during the $j-$th stationary segment, and $\varepsilon_t$ is a white-noise process with zero mean and variance matrix $\boldsymbol{\Sigma}_j$.

\item In each segment $[ t_{j-1}, t_j )$, all model parameters are assumed to be fixed. However, the auto-regressive (AR) parameters $\boldsymbol{\Phi}_{( \ell, j)}$ will change values across segments, while the error covariance across all segments is assumed to be $\boldsymbol{\Sigma}_j = \sigma^2 \boldsymbol{I}$. 

\end{itemize}

Based on the aforementioned setting proposed by \cite{safikhani2022fast}, the unknown parameters in each segment are: the number of break-points $m_0$, their locations $t_j, j = 1,..., m_0$, the VAR parameters $\boldsymbol{\Phi}_{(q, j)}$ together with the covariance matrix. Therefore, according to , the statistical problem is then to detect the break points $t_j$, in a computational efficient manner that is also scalable for very large values of $T$. Additionally, we are also interested to estimate accurately the VAR parameters $\boldsymbol{\Phi}_{( \ell, j)}$, under a high-dimensional environment $( p >> T )$. 

\subsubsection{A Block Segmentation Scheme Based Algorithm}

The main idea of the BSS is to partition the time points into blocks of size $b_n$ and fix the VAR parameters within each block. To this end, define a sequence of time points $q = r_0 < r_1 < ... < r_{ k_n } = T + 1$ which play the role of the end points for the blocks. In other words, the total number of blocks is calculated as: 
\begin{align}
r_{ i + 1 } - r_i = b_n, \ \ \ \text{for} \ i \in \left\{ 0, ...., k_n - 2 \right\}, \ k_n = \floor{  \frac{ n }{ b_n } }.     
\end{align}
where $k_n$ is the total number of blocks based on the (effective) sample size and $b_n$ is the length size of each block which is kept fixed across all blocks. Notice that throughout, the dimensions of the vector $\boldsymbol{y}_t$ remains fixed having $p$ elements as well as the (effective) sample size being $n$.

The theoretical results presented in \cite{safikhani2022fast} deal with the high-dimensional case such that the sparsity levels $d_{ k_j }$ increase with the sample size, $T$. Specifically, we define with $p \equiv p(n)$ and $m_0 \equiv m_0 (n)$ and $d_{ k_j } \equiv d_{ k_j } (n)$, where $n = T - q + 1$. In other words, $n$ is the effective sample size or the length of the time series for estimation of model parameters purposes, since the number of lags used for constructing additional regressors affects the sample size (see, \cite{safikhani2022fast}). 

The minimum distance between two consecutive break points is denoted by 
\begin{align}
\Delta_n = \underset{ 1 \leq j \leq m_0 + 1 }{ \mathsf{min}  } | t_j - t_{j-1} |.    
\end{align}
Let $\boldsymbol{\mathcal{Y}} \in \mathbb{R}^{ n \times p}$ and $\boldsymbol{\mathcal{X}} \in \mathbb{R}^{ n \times k_n p q }$ and $\boldsymbol{\theta} \in \mathbb{R}^{k_n p q \times p }$ and $\boldsymbol{E} \in \mathbb{R}^{ n \times p }$. In other words, based on this parametrization, $\theta_i \neq 0$ for $i \geq 2$ implies a change in the VAR coefficients. Therefore, for $j \in \left\{ 1,..., m_0 \right\}$, the structural break points $t_j$ can be estimated as block-end time point $r_{i-1}$, where $i \geq 2$ and $\theta_i \neq 0$.

\newpage 

Then, the model is formulated as below
\begin{align}
\boldsymbol{Y} = \boldsymbol{Z} \boldsymbol{\Theta} + \boldsymbol{E}, \ \ \boldsymbol{Y} = \mathsf{vec} ( \boldsymbol{\mathcal{Y}} ) \in \mathbb{R}^{ np \times 1 }, \ \ \boldsymbol{Z} = \boldsymbol{I}_p \otimes \boldsymbol{\mathcal{X}} \in \mathbb{R}^{ n p \times \pi_b }
\end{align}
where $\pi_b = k_n p^2 q$ and $\boldsymbol{\Theta} = \mathsf{vec} ( \boldsymbol{\theta} ) \in \mathbb{R}^{ \pi_b \times 1 }$. The initial estimate of the parameter $\boldsymbol{\Theta}$ is obtained: 
\begin{align}
\widehat{\boldsymbol{\Theta} } = \underset{ \boldsymbol{\Theta} \in \in \mathbb{R}^{ \pi_b \times 1 } }{ \mathsf{arg \ min} }  \ \frac{1}{n} \norm{ \boldsymbol{Y} - \boldsymbol{Z} \boldsymbol{\Theta} }_2^2 + \lambda_{1,n} \norm{ \boldsymbol{\Theta} }_1 + \lambda_{2,n} \sum_{i=1}^{ k_n } \norm{ \sum_{j=1}^{i} \theta_j }_1. 
\end{align}
Notice that the above high dimensional problem (convexity is preserved) uses fused lasso penalty with two $\ell_1$ penalties controlling the number of break points and the sparsity of the VAR model. 

Denote the set of indices of blocks with nonzero jumps and corresponding estimated change points obtained from solving the high-dimensional problem as below:
\begin{align}
\widehat{I}_n 
&= 
\left\{  \hat{i}_1, \hat{i}_2,..., \hat{i}_{ \hat{m} } \right\} =  \left\{ i: \norm{ \widehat{\boldsymbol{\Theta}}_i }_F \neq 0, i = 1,..., k_n \right\},  
\\
\mathcal{A} 
&=
\left\{  \hat{t}_1, \hat{t}_2,..., \hat{t}_{ \hat{m} } \right\} = \left\{ r_{i-1}: i \in \widehat{I}_n \right\}.
\end{align}
The total number of estimated change points in this step corresponds to the cardinality of the set $\mathcal{A}_n$, where $\hat{m} = \left| \mathcal{A}_n \right|$. Moreover, the block size $b_n$ acts as a tuning parameter that regulates the number of model parameters to be estimated. In addition to this, the algorithm applies a local screening step. Since the set $\mathcal{A}_n$ of candidate change points tends to overestimates their numbers, a screening step to eliminate redundant break points is required. According to \cite{safikhani2022fast}, the main idea goes is to estimate the VAR parameters \textit{locally} on the left and right side of each selected break point in the first step and compare them to one VAR parameter estimated from combining the left and right of the selected break point as one large stationary segment. Now, if the selected break point is close to a true break point, the sum of squared errors calculated assuming stationarity around the true break point will be much larger compared to the sum of squared errors calculated from two separate VAR parameter estimates on the right and left of the selected break point. Therefore, we can get consistent estimates of the number of break points by minimizing a localized information criterion (LIC) comprising of the sum of squared errors and a penalty term on the number of break points.   

Next, we describe the localized screening step in more details. Recall that $\mathcal{A} 
= \left\{  \hat{t}_1, \hat{t}_2,..., \hat{t}_{ \hat{m} } \right\}$ is the set of candidate break points selected in the first step. Then, for each subset $A \subset \mathcal{A}_n$, we define the following local VAR parameter estimates for some $\widehat{t}_i \in A$ such that 
\begin{align}
\widehat{\psi}_{  \widehat{t}_i, 1 } = \underset{ \widehat{t}_i, 1 }{ \mathsf{  arg \ min } }  \left\{  \frac{1}{ a_n } \sum_{ t = \widehat{t}_i - a_n }^{  \widehat{t}_i - 1 } \norm{ y_t - \psi_{  \widehat{t}_i, 1 } Y_{t-1} }_2^2 + \eta_{  \widehat{t}_i, 1 } \norm{  \psi_{  \widehat{t}_i, 1 } }_1 \right\} 
\end{align}
where $a_n$ corresponds to the neighbourhood size in which the VAR parameters are estimated. The procedure of  \cite{safikhani2022fast}, can be directly implemented for identifying and dating the presence of multiple break-points in multivariate time series. In contrast to the literature of SVAR models, the assumption is that the model parameters are identifiable but a possible parameter instability could exist.

\newpage

\subsubsection{Consistency of the BSS Estimator}

To establish the consistency properties of the BSS-based estimator relevant regularity conditions are required which we briefly present below. 
\begin{assumption}[see,  \cite{safikhani2022fast}]
For each $j \in \left\{ 1,2,..., m_0 + 1 \right\}$, the process
\begin{align}
\boldsymbol{y}_{(j),t} = \boldsymbol{\Phi}_{ (1,j) } \boldsymbol{y}_{(j),t-1} + ... + \boldsymbol{\Phi}_{ (q,j) } \boldsymbol{y}_{(j),t-q} + \boldsymbol{\Sigma}_j^{1/2} \boldsymbol{\varepsilon}_t    
\end{align}
is a stationary Gaussian time series. Denote the covariance matrices with 
\begin{align}
\boldsymbol{\Gamma}_{(j)}(h) = \mathsf{Cov} \left( \boldsymbol{y}_{(j),t}, \boldsymbol{y}_{(j),t+h} \right) , \ \ \text{for} \ \ t, h \in \mathbb{Z}.  
\end{align}
\end{assumption}

\begin{assumption}[see,  \cite{safikhani2022fast}]
The matrices $\boldsymbol{\Phi}_{ (1,j) }$ are sparse such that for all $k \in \left\{ 1,... p \right\}$ and $j \in \left\{ 1,..., m_0 \right\}$, $d_{ k_j } << p$, that is, $d_{ k_j } / p = o(1)$. There exists a positive constant $M_{ \Phi } > 0$ such that
\begin{align}
\underset{ 1 \leq j \leq m_0 + 1 }{ \mathsf{max} } \norm{ \boldsymbol{\Phi}_{ (\cdot,j) }  }_{ \infty } \leq M_{ \Phi }    
\end{align}
\end{assumption}

\begin{example}
Consider the CUSUM series of $y_{\ell t}$ over a generic segment $[ s, e ]$ for some $1 \leq s < e < T$,
\begin{align}
\mathcal{Y}_{s,b,e}^{\ell} = \frac{1}{ \sigma_{\ell} } \sqrt{ \frac{( b - s + 1 ) ( e - b ) }{ e - s + 1} } \left(  \frac{1}{b-s+1} \sum_{t=s}^b y_{\ell t} - \frac{1}{e - b}  \sum_{t=b+1}^{e} y_{\ell t} \right),
\end{align}
for $b = s,..., e-1$, where $\sigma_{\ell}$ denotes a scaling constant for treating all rows of the sequence $y_{\ell t}$ such that $1 \leq \ell \leq N$ on equal footing. Notice that if $\epsilon_{\ell t}$ were $\textit{i.i.d}$  Gaussian random variables, the maximum likelihood estimator of the change-point location for $\big\{ y_{\ell t}, s \leq t \leq e \big\}$ would coincide with $\underset{ b \in [ s, e) }{ \mathsf{arg \max} } \left| \mathcal{Y}_{s,b,e}^{\ell}  \right|$. Moreover, $\mathcal{D}_{s, n e} (m)$ takes the contrast between the $m$ largest CUSUM values $\left| \mathcal{Y}_{s,b,e}^{\ell} \right|, 1 \leq \ell \leq m$ and the rest at each $b$, and thus partitions the coordinates into the $m$ that are the most likely to contain a change-point and those which are not in a point-wise manner. Then, the test statistic is derived by maximizing the two-dimensional array of DC statistics over both time and cross-sectional indices as   
\begin{align}
\mathcal{T}_{s,e} = \underset{ b \in [s, e) }{ \mathsf{max} } \underset{ 1 \leq m \leq N }{ \mathsf{max} } \mathcal{D}_{s, n e} (m).
\end{align}
Therefore, the presence of multiple break-points implies that there are multiple regimes within the full sample, where the corresponding time series sequences appear to have structural breaks.  
\end{example}

\begin{theorem}[\cite{fang2020segmentation}]
Let $( X_1,..., X_m )$ be an independent sequence of normally distributed random variables with mean $\mu$ and variance 1. Then for $\mathcal{Z}_{i,j,k}$ we have for $b \to  \infty$ and $m \approx b^2$, 
\begin{align}
\mathbb{P} \left(  \underset{ 0 \leq i < j < k \leq m  }{  \mathsf{max} } \left| \mathcal{Z}_{i,j,k} \right|  \geq b \right).    
\end{align}
\end{theorem}

\newpage

\subsection{Sieve Bootstrap for Functional Vector Autoregressions}

According to \cite{paparoditis2018sieve}, bootstrap procedures for Hilbert space-valued time series proposed so far in the literature, are mainly attempts to adapt, to the infinite dimensional functional framework, of bootstrap methods that have been developed for the finite dimensional time series case. Specifically, \cite{paparoditis2018sieve} shows, that under quite general assumptions, the stochastic process of Fourier coefficients obeys a so-called vector autoregressive representation and this representation plays a key role in developing a bootstrap procedure for the functional time series at hand. To capture the essential driving functional parts of the underlying infinite dimensional process, the first $m$ functional principal components are used and the corresponding $m-$dimensional time series of Fourier coefficients is bootstrapped using a $p$th order vector autoregression fitted to the vector time series of sample Fourier coefficients. In this way, a $m-$dimensional pseudo-time series of Fourier coefficients is generated which imitates the temporal dependence structure of the vector time series of sample Fourier coefficients. Thus, using the truncated Karhunen-Loeve expansion, these pseudo-coefficients are then transformed to functional bootstrap replicates of the main driving, principal components, of the observed functional time series.

\subsubsection{The bootstrap procedure}

The main building block of the procedure proposed by \cite{paparoditis2018sieve} is to generate pseudo-replicates $X_1^{*},...,X_n^{*}$ of the functional time series at hand by first bootstrapping the $m-$dimensional time series of Fourier coefficients $\xi_t = ( \xi_{1,t}, \xi_{2,t},..., \xi_{m,t} )$, $t =1,2,...n$, corresponding to the first $m$ principal components. This $m-$dimensional time series of Fourier coefficients is bootstrapped using the autoregressive representation $\xi_t$. Then, the generated $m-$dimensional pseudo-time series of Fourier coefficients is then transformed to functional principal pseudo-components by means of truncated Karhunen-Loeve expansion $\sum_{j=1}^n \xi_{j,t} v_j$. Adding to this, an appropriately resampled functional noise leads to the functional pseudo-time series $X_1^{*}, X_2^{*} ,..., X_n^{*}$. However, since the $\xi_t'$s are not observed, we work with the time series of estimates scores. The functional sieve bootstrap algorithm of \cite{paparoditis2018sieve} is:

\begin{itemize}
\item[Step 1:] Select a number $m = m(n)$ of functional principal components and an autoregressive order $p = p(n)$, both finite and depending on $n$

\item[Step 2:] Let
\begin{align}
\widehat{ \boldsymbol{\xi}} = \left( \widehat{\xi}_{j,t} = \langle X_t, \widehat{v}_j \rangle, j = 1,...,m \right)^{\top}, \ \ \ t =1,...,n,
\end{align}
be the $m-$dimensional series of estimated Fourier coefficients, where we denote with $\widehat{v}_j$, $j = 1,...,m$ the estimated eigenfunctions corresponding to the estimated eigenvalues $\widehat{\lambda}_1 > \widehat{\lambda}_2 > ... > \widehat{\lambda}_m$ of the sample covariance operator 
\begin{align}
\widehat{C}_0 = n^{-1} \sum_{t=1}^n \left( X_t - \bar{X}_n \right) \otimes \left( X_t - \bar{X}_n \right), \ \ \ \bar{X}_n = n^{-1} \sum_{t=1}^n X_t.
\end{align}

\newpage

\item[Step 3:] Let $\widehat{X}_{t,m} = \sum_{j=1}^m \widehat{\xi}_{j,t} \widehat{v}_j$ and define the functional residuals $\widehat{U}_{t,m} = X_t - \widehat{X}_{t,m}$ and define the functional residuals $\widehat{U}_{t,m} = X_t - \widehat{X}_{t,m}$, $t=1,...,n$.

\item[Step 4:] Fit a $p$th order vector autoregressive process to the $m-$dimensional time series $\widehat{\xi}_t$, $t = 1,...,n$, denote by $\widehat{A}_{j,p} (m)$, $j = 1,...,p$, the estimates of the autoregressive matrices and by $\widehat{e}_{t,p}$ the residuals,
\begin{align}
\widehat{e}_{t,p} = \widehat{\xi}_t - \sum_{j=1}^p \widehat{A}_{j,p} (m) \widehat{\xi}_{t-j}, \ \ \ t = p+1,p+2,...,n.
\end{align}
We focus on the following Yule-Walker estimators for $\widehat{A}_{j,p} (m)$, $j = 1,...,p$. 

\item[Step 5:] Generate a $m-$dimensional pseudo time series of scores $\xi^{*}_t = \left( \xi^{*}_{1,t},...,\xi^{*}_{m,t}  \right)$, $t =1,2,...,n$,
using 
\begin{align}
\xi^{*}_t = \sum_{j=1}^p \widehat{A}_{j,p}(m) \xi^{*}_{t-j} + e^{*}_t,
\end{align}
where $e^{*}_t$, $t=1,...,n$ are i.i.d random vectors having as distribution the empirical distribution of the centered residual vectors $\widetilde{e}_{t,p} = \widehat{e}_{t,p} - \bar{\widehat{e}}_{n,t}$, for $t = p + 1, p+2,...,n$ and $\bar{\widehat{e}} = (n - p)^{-1} \sum_{t=p+1}^n \widehat{e}_{t,p}$. 

\item[Step 6:] Generate a pseudo-functional time series $X_1^{*},...,X_n^{*}$ ,where 
\begin{align}
\xi^{*}_t = \sum_{j=1}^m \xi_{j,t}^{*} \widehat{v}_j + U_t^{*}, \ \ \ t=1,...,n.
\end{align}
and $U_1^{*},...,U_n^{*}$ are i.i.d random functions obtained by choosing with replacement from the set of centered functional residuals $\widehat{U}_{t,m} - \bar{\widehat{U}}_{t}$ and $\bar{\widehat{U}}_{t} = n^{-1} \sum_{t=1}^n \widehat{U}_{t,m}$, for $t=1,...,n$.
\end{itemize}

\begin{remark}
Notice that $X_1^{*},...,X_n^{*}$ are functional pseudo-random variables and the the autoregressive representation of the vector time series of Fourier coefficients is solely used as a toll to bootstrap the $m$ main functional principal components of the functional time series. In fact, it is this autoregressive representation which allows the generation of the pseudo-time series of Fourier coefficients $\xi^{*}_1,...,\xi^{*}_n$ in Step 4 and Step 5 in a way that imitates the dependence structure of the sample Fourier coefficients $\xi_1,...,\xi_n$. These pseudo-Fourier coefficients are transformed to bootstrapped main principal components by means of the truncated and estimated Karhunen-Loeve expansion which together with the additive functional noise $U_t^{*}$, lead to the new functional pseudo-observations $X_1^{*},...,X_n^{*}$ (see, \cite{paparoditis2018sieve}).  
\end{remark}

Furthermore, an important issue for bootstrap-based inference methods is to  investigate the validity of the functional sieve bootstrap applied in order to approximate the distribution of some statistic $T_n = T \left( X_1,...,X_n \right)$ of interest, when the bootstrap analogue $T_n^{*} = T \left( X_1^{*},...,X_n^{*} \right)$ is used. Notice that establishing validity of a bootstrap procedure for time series heavily depends on two issues. On the dependence structure of the underlying process which affects the distribution of the statistic of interest and on the capability of the bootstrap procedure used to mimic appropriately this dependence structure. 

\newpage    

Recall the definition given by \cite{paparoditis2018sieve} such that
\begin{align}
X_t^{*} = \sum_{j=1}^m \textbf{1}_j^{\top} \xi_j^{*} \widehat{v}_j + U_t^{*}
\end{align}
and observe that $\xi^{*}_t = \sum_{l=0}^{\infty} \widehat{\Psi}_{l,p}(m) e^{*}_{t-l}$, where $\widehat{\Psi}_{0,p}(m) = I_m$ and the power series 
\begin{align}
\widehat{\Psi}_{m,p}(z) = I_m + \sum_{l=1}^{\infty} \widehat{\Psi}_{l,p}(m) z^l = \left( I_m - \sum_{j=1}^p \widehat{A}_{j,p} (m) z^j \right)^{-1}
\end{align}    
converges for $|z| \leq 1$. Thus, following \cite{paparoditis2018sieve}  we can decompose each sequence as below
\begin{align}
X_t^{*} &= \sum_{l=0}^{\infty} \sum_{j=1}^m \textbf{1}_j^{\top} \widehat{\Psi}_{l,p} (m) e^{*}_{t-l} \widehat{v}_j + U_t^{*}
\\
X_{t,M}^{*} &= \sum_{l=0}^{M-1} \sum_{j=1}^m \textbf{1}_j^{\top} \widehat{\Psi }_{l,p} (m) e_{t-l}^{*} \widehat{v}_j +  \sum_{l=M}^{\infty} \sum_{j=1}^m \textbf{1}_j^{\top} \widehat{\Psi}_{l,p} (m) e_{t-l}^{*} \widehat{v}_j +  U_t^{*}
\end{align}
where for each $t \in \mathbb{Z}, \left\{ e_{s,t}^{*} , s \in \mathbb{Z} \right\}$ is an independent copy of $\left\{ e_{s}^{*} , s \in \mathbb{Z} \right\}$. Notice that, 
\begin{align}
X_M^{*} - X_{M,M}^{*} = \sum_{l = M}^{\infty} \sum_{j=1}^m \textbf{1}_j^{\top} \widehat{\Psi }_{l,p} (m) \big( e_{M-l}^{*} - e_{M-l,M}^{*} \big) \widehat{v}_j 
\end{align}
By Minskowski's inequality, we have that 
\begin{align}
\label{expression1}
\sqrt{ \mathbb{E} \norm{X_{M}^{*} - X_{M,M}^{*}}^2 } 
\leq
\sqrt{ \mathbb{E} \norm{\sum_{l = M}^{\infty} \sum_{j=1}^m \textbf{1}_j^{\top} \widehat{\Psi }_{l,p} (m) e_{M-l}^{*} \widehat{v}_j}^2 }  
\nonumber
+ \sqrt{ \mathbb{E} \norm{\sum_{l = M}^{\infty} \sum_{j=1}^m \textbf{1}_j^{\top} \widehat{\Psi }_{l,p} (m) e_{M-l,M}^{*}  \widehat{v}_j}^2 }.
\end{align}
Thus, evaluating the first expectation term, we obtain that using $\norm{A}_F^2 = tr \left( A A^{\top} \right)$ and the submultiplicative property of the Frobenius matrix norm, that 
\begin{align*}
\mathbb{E} \norm{\sum_{l = M}^{\infty} \sum_{j=1}^m \textbf{1}_j^{\top} \widehat{\Psi }_{l,p} (m) e_{M-l,M}^{*}\widehat{v}_j}^2
= 
\sum_{l=M}^{\infty} tr \left( \widehat{\Psi}_{l,p}(m) \Sigma^{*}(m) \widehat{\Psi}_{l,p}(m) \right) 
\leq 
\norm{ \widehat{\Sigma}_{e,p}^{1 / 2} (m) }_F^2 \sum_{l=M}^{\infty} \norm{ \widehat{\Psi}_{l,p}(m) }_F^2,
\end{align*}
where $\widehat{\Sigma}_{e,p}(m) = \widehat{\Sigma}_{e,p}^{1 / 2}(m) \widehat{\Sigma}_{e,p}^{1 / 2}(m)$. An identical expression appears for the second expectation term on the right-hand side of \ref{expression1}. Applying Minkowski's inequality again, we get by the exponential decay of $\norm{ \widehat{\Psi}_{l,p} (m) }_F$, such that the following probability bound holds (see, \cite{paparoditis2018sieve})
\begin{align*}
\sum_{M=1}^{\infty} \sqrt{ \mathbb{E} \norm{ X_M^{*} - X_{M,M}^{*}}^2} 
\leq 
2 \norm{ \widehat{\Sigma}_{e,p}^{1 / 2} (m) }_F^2 \sum_{l = M}^{\infty} \sum_{j=1}^m \norm{ \widehat{\Psi}_{l,p}(m) }_F
= 
2 \norm{ \widehat{\Sigma}_{e,p}^{1 / 2} (m) }_F^2 \sum_{l = 1}^{\infty} l \norm{ \widehat{\Psi}_{l,p}(m) }_F = \mathcal{O}_p(1).
\end{align*}

\newpage

\section{Conclusion}

There are various statistical identification procedures for SVAR models such as imposing sign-restrictions, short-run restrictions, long-run restrictions, identification by higher-order moments as well as identification by conditional heteroscedasticity and changes in volatility (e.g., see,  \cite{lanne2008identifying, lanne2008statistical}. More precisely, by exploiting the heteroscedasticity of the errors $\varepsilon_t$, then the identification of the structural parameters of the system can be achieved. In particular, \cite{lanne2010structural} assume Markov switching and a smooth transition in the covariance matrix of the error term $\varepsilon_t$ of the model. Specifically, the structural analysis of VARs in this case is based on the Markov regime switching property to identify shocks, which holds due to the fact that the reduced form error covariance matrix varies across states. Furthermore, a different stream of literature considers identification in SVAR models by assuming that the error terms are non-Gaussian and mutually independent (e.g., see \cite{lanne2017identification} and \cite{gourieroux2020identification}). For example, \cite{lanne2010structural} assume that the errors of the model are independent over time with a distribution that is a mixture of two Gaussian distributions with zero means and diagonal covariance matrices, one of which is an identify matrix and the other one has positive diagonal elements, which for identifiability have to be distinct. Roughly speaking there are specific systems where the use of the non-Gaussianity property is a more suitable identification strategy of structural shocks of the SVAR system which we leave for discussion in a follow-up paper.

In this note we have reviewed key statistical properties of VAR processes, cointegrated VAR processes and SVAR processes as well as suitable identification strategies for the structural parameters of these structural econometric specifications (see, also \cite{mumtaz2018vars}). Moreover, we have discusses several empirical-driven applications from the applied macroeconometrics literature and we studied the proposed identification and estimation methods employed to those cases. Then, motivated from the literature of dynamic causal effects (see, \cite{mertens2013dynamic}, \cite{stock2018identification} and \cite{budnik2023identifying}), we have motivated the main statistical and econometric tools for  identification of structural econometric models, commonly used for counterfactual analysis purposes, and then presented some examples of time series regression models, both important for understanding suitable techniques for policy-making evaluation in macroeconometrics. We leave any further presentation of other relevant topics for further research discussions. Lastly, some advanced topics were also presented such as break testing for unstable roots in structural error correction models, estimation techniques for the optimal lag order, the estimation of a high dimensional VAR model as well as a fast algorithmic procedure presented in the literature for detecting structural breaks in large VAR models. 

In the Appendix of this set of lecture notes we present an example of the maximum likelihood estimation for $\alpha-$stable autoregressive processes (see, \cite{andrews2009maximum}) as well as the main asymptotic theory results for inference for the VEC(1) model from \cite{guo2023inference}.

\newpage

\appendix
\numberwithin{equation}{section}
\makeatletter 

\section{Maximum Likelihood Theory}

\subsection{MLE for $\alpha-$Stable Autoregressive Processes}

Following the framework of \cite{andrews2009maximum}, we consider the MLE estimation for both causal and non-causal autoregressive time series processes with non-Gaussian $\alpha-$stable noise. In particular, the estimators for the autoregressive parameters are $n^{- 1 / \alpha}-$consistent (parametric rate) and converge in distribution to the maximizer of a random function. Although the form of this limiting distribution is intractable, the shape of the distribution can be examined using the bootstrap procedure. Moreover, the estimators for the parameters of the stable noise distribution have the traditional $n^{1/2}$ rate of convergence and are asymptotically normal. In addition, the behavior of the estimators for finite samples is studied via simulation, and we use the maximum likelihood estimation to fit a noncausal autoregressive model (see,  \cite{andrews2009maximum} and \cite{andrews2013model}). 

Let $\left\{ X_t \right\}$ be the AR process which satisfies the difference equations
\begin{align}
\phi_0 (B) X_t = Z_t,    
\end{align}
where the AR polynomial is given by 
\begin{align}
\phi_0 (z) &:= 1 - \phi_{01} z - ... -  \phi_{0p} z^p \neq 0, \ \ \text{for} \ |z| = 1,    
\end{align}
Furthermore, because $\phi_0 (z) \neq 0$ for $|z| = 1$, the Laurent series expansion of $1 / \phi_0 (z)$ is given by 
\begin{align}
\frac{1}{ \phi_0 (z) }  = \sum_{ j = - \infty }^{ + \infty} \psi_j z^j,  
\end{align}
and the unique, strictly stationary solution is given by 
\begin{align}
X_t =  \sum_{ j = - \infty }^{ + \infty} \psi_j Z_{t-j},   
\end{align}

\begin{itemize}

\item If $\phi_0(z) \neq 0$ for $|z| \leq 1$, then $\psi_j = 0$ for $j < 0$ and so $\left\{ X_t \right\}$ is a \textit{causal process}, since $X_t =  \sum_{ j = - \infty }^{ + \infty} \psi_j Z_{t-j}$, is a function of the past and present $\left\{ Z_t \right\}$. 

\item If $\phi_0(z) \neq 0$ for $|z| > 1$, then $X_t =  \sum_{ j = 0 }^{ + \infty} \psi_{-j} Z_{t+j}$, and  $\left\{ X_t \right\}$ is called a \textit{purely noncausal process}.
    
\end{itemize}

Moreover, in the purely non-causal process case, the coefficients $\left\{ \psi_j \right\}$ satisfiy
\begin{align}
\big( 1 - \phi_{01} z - ... -  \phi_{0p} z^p  \big) \big( \psi_0 + \psi_{-1} z^{-1} + ... \big) = 1,    
\end{align}
which, if $\phi_{0p} \neq 0$, implies that $\psi_0 = \psi_{-1} = ... = \psi_{1-p} = 0$ and $\psi_{1-p} = - \phi_{0p}^{-1}$.

\newpage

Therefore, to express $\phi_0(z)$ as the product of causal and purely noncausal polynomials, suppose
\begin{align}
\phi_0(z) = \textcolor{red}{ \big( 1 - \theta_{01} z - ... -  \theta_{0 r_0} z^{  r_0 } \big) }  \textcolor{blue}{ \big( 1 - \theta_{0, r_0 + 1} z - ... -  \theta_{0, r_0 + s_0 } z^{  s_0 } \big)  }, 
\end{align}
where $( r_0 + s_0 ) = p$ and 
\begin{align}
\theta_0 = 
\begin{cases}
\textcolor{red}{ \big( 1 - \theta_{01} z - ... -  \theta_{0 r_0} z^{  r_0 } \big) } =: \theta_0^{+} \neq 0 & ,\text{if} \ |z| \geq 1
\\
\textcolor{blue}{ \big( 1 - \theta_{0, r_0 + 1} z - ... -  \theta_{0, r_0 + s_0 } z^{  s_0 } \big)  } =: \theta_0^{*}  \neq 0  & ,\text{if} \ |z| \leq 1  
\end{cases}    
\end{align}
which implies that $\theta_0^{+}$ is a causal polynomial and $\theta_0^{*}$ is a purely noncausal polynomial. Therefore, $\phi_0(z)$ has a unique representation as the product of causal and noncausal polynomials if the true order of the AR polynomial $\phi_0(z)$ is less than $p$.  

Consider the following sequence $\left\{ Z_t \right\}$ such that
\begin{align}
Z_t = \textcolor{red}{ \big( 1 - \theta_{01} B - ... -  \theta_{0 r_0} B^{  r_0 } \big) } \textcolor{blue}{ \big( 1 - \theta_{0, r_0 + 1} B - ... -  \theta_{0, r_0 + s_0 } B^{  s_0 } \big)  }  X_t 
\end{align}
Consider the arbitraty autoregressive polynomials as below:
\begin{align}
\theta^{+}(z) &= \textcolor{red}{ \big( 1 - \theta_{1} z - ... -  \theta_{r} z^{r} \big) }
\\
\theta^{*}(z) &= \textcolor{blue}{ \big( 1 - \theta_{r + 1} z - ... -  \theta_{r + s} z^{  s } \big)  }
\end{align}
Then, we define with 
\begin{align}
Z_t ( \boldsymbol{\theta}, s ) :=  \textcolor{red}{ \big( 1 - \theta_{1} B - ... -  \theta_{r} B^{r} \big) } \textcolor{blue}{ \big( 1 - \theta_{r + 1} B - ... -  \theta_{r + s} B^{  s } \big)  } X_t  
\end{align}
where $\boldsymbol{\theta} = ( \theta_1,..., \theta_p )^{\prime}$ and $\boldsymbol{\theta}  = ( \theta_{01},..., \theta_{0p} )^{\prime}$. Moreover, let $\boldsymbol{\eta} = ( \eta_1,..., \eta_{p+4}  ) = ( \theta_1,..., \theta_p, \alpha, \beta, \sigma, \mu )^{\prime} = ( \boldsymbol{\theta}^{\prime}, \boldsymbol{\tau}^{\prime} )^{\prime}$. In other words, given a realization $\left\{ X_t \right\}_{t=1}^n$, then the log-likelihood of $\boldsymbol{\eta}$ can be approximated by the conditional log-likelihood such that 
\begin{align}
\mathcal{L} ( \boldsymbol{\eta}, s ) = \sum_{t = p+1}^s \big[ \mathsf{ln} f \big( Z_t ( \boldsymbol{\theta}, s ); \boldsymbol{\tau} \big) + \mathsf{ln} | \theta_p | \boldsymbol{1} \left\{ s > 0 \right\}  \big],    
\end{align}
where the component $\left\{ Z_t ( \boldsymbol{\theta}, s ) \right\}_{t=p+1}^n$ is computed separately. 
\begin{remark}
The use of causal and non-causal autoregressive processes can be helpful in understanding concepts such as identification, granger non-causality and weak exogeneity while ensuring that these properties are not violated when learning for dynamic causal effects based on macroeconomic models. Representations for such univariate processes when modelling explosive bubbles/regimes are studied by \cite{fries2019mixed}, \cite{davis2020noncausal}, \cite{cavaliere2020bootstrapping} and \cite{hecq2021forecasting} while related representations for multivariate processes are studied by \cite{lanne2013noncausal}, \cite{velasco2018frequency} and \cite{rygh2022causal}.  A framework for optimal forecasting with noncausal autoregressive time series is presented by \cite{lanne2012optimal}. 
\end{remark}

\newpage 

\subsection{Inference for the VEC(1) Model}

Consider an $m-$dimensional vector $\boldsymbol{Y}_t$ satisfying the following equation (see, \cite{guo2023inference}) 
\begin{align}
\Delta \boldsymbol{Y}_t = \boldsymbol{\Pi}_0 \boldsymbol{Y}_{t-1} + \boldsymbol{u}_t \equiv \boldsymbol{\alpha}_0 \boldsymbol{\beta}_0^{\prime} \boldsymbol{Y}_{t-1} + \boldsymbol{u}_t,     
\end{align}
where $\boldsymbol{\Pi}_0 = \boldsymbol{\alpha}_0 \boldsymbol{\beta}_0^{\prime}$ has matrix rank $0 \leq r_0 \leq m$, with $\boldsymbol{\alpha}_0$ and $\boldsymbol{\beta}_0$ being $( m \times r_0)$ full-rank matrices. 

Moreover, \cite{johansen1995identifying} proposed the following partial sum Granger representation
\begin{align}
\boldsymbol{Y}_t &= \boldsymbol{C} \sum_{s=1}^t \boldsymbol{u}_s +  \boldsymbol{\alpha}_0 \left( \boldsymbol{\beta}_0^{\prime}  \boldsymbol{\alpha}_0 \right)^{-1} \boldsymbol{\Xi} (L) \boldsymbol{\beta}_0^{\prime} \boldsymbol{u}_t + \boldsymbol{C} \boldsymbol{Y}_0,
\\
\boldsymbol{C} &=  \boldsymbol{\alpha}_{0 \perp} \left( \boldsymbol{\alpha}_{0 \perp} \boldsymbol{\beta}^{\prime}_{0 \perp}   \right)^{-1} \boldsymbol{\alpha}_{0 \perp}^{\prime}  \ \ \ \boldsymbol{\Xi}(L) = \sum_{s=0}^{\infty} \boldsymbol{\Xi}_s L^s.  
\end{align}
Denote with 
\begin{align}
\boldsymbol{Q} = \big[ \boldsymbol{\beta}_0,  \boldsymbol{\alpha}_{0 \perp} \big]^{\prime} \ \ \ \boldsymbol{Q}^{-1} = \big[ \boldsymbol{\alpha}_0 \left( \boldsymbol{\beta}^{\prime}_0 \boldsymbol{\alpha}_0 \right)^{-1}, \boldsymbol{\beta}^{\prime}_{0 \perp} \left( \boldsymbol{\alpha}_{0 \perp} \boldsymbol{\beta}^{\prime}_{0 \perp}  \right)^{-1} \big].    
\end{align}
Thus, we have that
\begin{align}
\boldsymbol{Q}  \boldsymbol{\Pi}_0  =
\begin{bmatrix}
\boldsymbol{\beta}^{\prime}_{0} \boldsymbol{\alpha}_{0}  \boldsymbol{\beta}^{\prime}_{0} 
\\
\boldsymbol{0}
\end{bmatrix}
\ \ \ 
\boldsymbol{Q}  \boldsymbol{\Pi}_0 \boldsymbol{Q}^{-1} 
=
\begin{bmatrix}
\boldsymbol{\beta}^{\prime}_{0} \boldsymbol{\alpha}_{0}  & \boldsymbol{0}
\\
\boldsymbol{0} & \boldsymbol{0}
\end{bmatrix}
\end{align}
Denote with 
\begin{align}
\boldsymbol{Z}_t = 
\begin{pmatrix}
\boldsymbol{\beta}_0^{\prime} \boldsymbol{Y}_t
\\
\boldsymbol{\alpha}^{\prime}_{0 \perp} \boldsymbol{Y}_t
\end{pmatrix}
\equiv 
\begin{pmatrix}
\boldsymbol{Z}_{1,t}
\\
\boldsymbol{Z}_{2,t}
\end{pmatrix},
\ \ \ \ \ 
\boldsymbol{W}_t = 
\begin{pmatrix}
\boldsymbol{\beta}_0^{\prime} \boldsymbol{u}_t
\\
\boldsymbol{\alpha}^{\prime}_{0 \perp} \boldsymbol{u}_t
\end{pmatrix}
\equiv 
\begin{pmatrix}
\boldsymbol{W}_{1,t}
\\
\boldsymbol{W}_{2,t}
\end{pmatrix}
\end{align}
which implies that the econometric specification can be formulated as below
\begin{align}
\Delta \boldsymbol{Z}_t = \boldsymbol{Q} \boldsymbol{\Pi}_0    \boldsymbol{Q}^{-1} \boldsymbol{Z}_{t-1} + \boldsymbol{W}_t.
\end{align}
Using $\boldsymbol{Z}_{1,t} = \boldsymbol{\beta}_0^{\prime} \boldsymbol{Y}_t$ and $\boldsymbol{Z}_{2,t} = \boldsymbol{\alpha}^{\prime}_{0 \perp} \boldsymbol{Y}_t$ we obtain the following expressions
\begin{align}
\textcolor{blue}{ \boldsymbol{Z}_{1,t} } &\equiv  \textcolor{blue}{ \boldsymbol{\beta}_0^{\prime} } \boldsymbol{Y}_t 
= 
\textcolor{blue}{ \boldsymbol{\beta}_0^{\prime} } \boldsymbol{C} \sum_{s=1}^t \boldsymbol{u}_s +  \textcolor{blue}{ \boldsymbol{\beta}_0^{\prime} } \boldsymbol{\alpha}_0 \left( \boldsymbol{\beta}_0^{\prime}  \boldsymbol{\alpha}_0 \right)^{-1} \boldsymbol{\Xi} (L) \boldsymbol{\beta}_0^{\prime} \boldsymbol{u}_t + \textcolor{blue}{ \boldsymbol{\beta}_0^{\prime} } \boldsymbol{C} \boldsymbol{Y}_0,  
\\
\textcolor{red}{ \boldsymbol{Z}_{2,t} } &\equiv \textcolor{red}{ \boldsymbol{\alpha}^{\prime}_{0 \perp} } \boldsymbol{Y}_t = \textcolor{red}{ \boldsymbol{\alpha}^{\prime}_{0 \perp} }  \boldsymbol{C} \sum_{s=1}^t \boldsymbol{u}_s + \textcolor{red}{ \boldsymbol{\alpha}^{\prime}_{0 \perp} }  \boldsymbol{\alpha}_0 \left( \boldsymbol{\beta}_0^{\prime}  \boldsymbol{\alpha}_0 \right)^{-1} \boldsymbol{\Xi} (L) \boldsymbol{\beta}_0^{\prime} \boldsymbol{u}_t + \textcolor{red}{ \boldsymbol{\alpha}^{\prime}_{0 \perp} }  \boldsymbol{C} \boldsymbol{Y}_0, 
\end{align}
Therefore, we obtain that 
\begin{align}
\textcolor{blue}{ \boldsymbol{Z}_{1,t} }  &=  \boldsymbol{\Xi} (L) \boldsymbol{\beta}_0^{\prime} \boldsymbol{u}_t  = \sum_{j=0} \left( \sum_{i=0}^j \boldsymbol{\Xi}_i \boldsymbol{\beta}_0^{\prime} \boldsymbol{D}_{j-i} \right) \boldsymbol{\epsilon}_{t-j} \equiv \sum_{j=0}^{ \infty } \boldsymbol{C}_j \boldsymbol{\varepsilon}_{t-j},
\end{align}

\newpage

\begin{align}
\textcolor{red}{ \boldsymbol{Z}_{2,t} }  &=  \textcolor{red}{ \boldsymbol{\alpha}^{\prime}_{0 \perp} } \boldsymbol{C} \left( \sum_{s=1}^t \boldsymbol{u}_s + \boldsymbol{Y}_0 \right) \equiv \big[ \boldsymbol{0}, \boldsymbol{I}_d \big] \boldsymbol{Q} \left(  \sum_{s=1}^t \boldsymbol{u}_s + \boldsymbol{Y}_0 \right) \equiv \big[ \boldsymbol{0}, \boldsymbol{I}_d \big] \sum_{s=1}^t \boldsymbol{\gamma}_s,
\end{align}
where $d = ( m - r_0 )$, $\boldsymbol{C}_j = \sum_{i=0}^j \boldsymbol{\Xi}_i \boldsymbol{\beta}^{\prime} \boldsymbol{D}_{j-i} = \mathcal{O}_p ( \rho^j )$ and
\begin{align}
\boldsymbol{\gamma}_s = \sum_{i=0}^{\infty} \boldsymbol{Q} \boldsymbol{D}_i \boldsymbol{\varepsilon}_{s-i} \equiv \sum_{i=0}^{\infty} \boldsymbol{\phi}_i \boldsymbol{\varepsilon}_{s-i}, \ \ \text{with} \   \boldsymbol{\phi}_i = \boldsymbol{Q} \boldsymbol{D}_i = \mathcal{O}_p(\rho^i) \ \text{for some} \ \rho \in (0,1).     
\end{align}
Therefore, estimating the matrix $\boldsymbol{\Pi}_0$ by OLS gives 
\begin{align}
\widehat{ \boldsymbol{\Pi} }_{ols} 
:= 
\underset{ \Pi \in \mathbb{R}^{ m \times m }  }{ \mathsf{arg min} } \ \sum_{t=1}^n \norm{ \Delta \boldsymbol{Y}_t - \boldsymbol{\Pi} \boldsymbol{Y}_{t-1}  }^2 
= 
\left( \sum_{t=1}^n  \Delta \boldsymbol{Y}_t \boldsymbol{Y}_{t-1}^{\prime}  \right) \left(  \sum_{t=1}^n  \boldsymbol{Y}_{t-1} \boldsymbol{Y}_{t-1}^{\prime}  \right)^{-1}.
\end{align}

\begin{theorem}[\cite{guo2023inference}]
Let $\bar{\boldsymbol{\beta}}_{\perp} = \boldsymbol{\beta}_{0 \perp} \left( \boldsymbol{\alpha}^{\prime}_{0 \perp} \boldsymbol{\beta}_{0 \perp} \right)^{-1}$, $\bar{\boldsymbol{\beta}} = \boldsymbol{\alpha}_0 \left( \boldsymbol{\beta}^{\prime}_0 \boldsymbol{\alpha}_0 \right)^{-1}$ and $\boldsymbol{\phi} = \boldsymbol{Q} \sum_{i=0}^{\infty} \boldsymbol{D}_i$. Suppose that Assumptions hold and that $\varepsilon_t$ has a symmetric distribution, then it follows
\begin{itemize}
\item[(a).] $\left( \widehat{\boldsymbol{\Pi}}_{ols} - \boldsymbol{\Pi}_0 \right) \bar{\boldsymbol{\beta}} \to_d \boldsymbol{R}_1 \boldsymbol{\Gamma}_{11}^{-1}$, 

\item[(b).]  $n \left( \widehat{\boldsymbol{\Pi}}_{ols} - \boldsymbol{\Pi}_0 \right) \bar{\boldsymbol{\beta}}_{\perp} \to_d \boldsymbol{R}_2^{*}  \boldsymbol{\Gamma}_{22}^{-1} - \boldsymbol{R}_1^{*} \boldsymbol{\Gamma}_{11}^{-1} \boldsymbol{\Gamma}_{21} \boldsymbol{\Gamma}_{22}^{-1}$,
\end{itemize}
\end{theorem}
Denote with $\boldsymbol{R}_2$ such that 
\begin{align}
\boldsymbol{R}_2 
&= 
\left[ \int_0^1 \boldsymbol{P}(r) d \boldsymbol{P}^{\prime} (r) \right]^{\prime} \boldsymbol{\phi}^{\prime} \big[ \boldsymbol{0}, \boldsymbol{I}_d \big]^{\prime},   
\\
\boldsymbol{\Gamma}_{22} 
&= 
\big[ \boldsymbol{0}, \boldsymbol{I}_d \big] \boldsymbol{\phi} \left[ \int_0^1 \boldsymbol{P}(r) \boldsymbol{P}^{\prime} (r) dr \right] \boldsymbol{\phi}^{\prime} \big[ \boldsymbol{0}, \boldsymbol{I}_d \big]^{\prime},   
\\
\boldsymbol{\Gamma}_{11} 
&= 
\sum_{j=0}^{\infty} \boldsymbol{C}_j \boldsymbol{S}_1 \boldsymbol{C}_j^{\prime}, 
\\
\boldsymbol{\Gamma}_{21} 
&= 
\left\{ \boldsymbol{R}_2^{\prime} \sum_{i=0}^{\infty} \boldsymbol{C}_i^{\prime} + \big[ \boldsymbol{0}, \boldsymbol{I}_d \big]  \boldsymbol{Q} \sum_{i=0}^{\infty} \sum_{j=0}^i \boldsymbol{D}_j \boldsymbol{S}_1 \boldsymbol{C}_i^{\prime} \right\}. 
\end{align}

\begin{remark}
Notice that in order to establish the asymptotic distribution of the OLS estimator for the VEC(1) model with heavy tailed errors, we shall consider the weak convergence of the matrix functionals based on the above limit results. In particular, the OLS related to the short-term parameters is not consistent and converges to a functional of stable matrix-processes. Moreover, the inconsistency is because $\left\{ \boldsymbol{u}_t \right\}$ is a series of dependent random vectors which is a critical difference from the $\textit{i.i.d}$ white noises cases as in \cite{she2022whittle}. 
\end{remark}

\newpage

By defining suitable functionals as above we can establish the joint weak convergence given below
\begin{align}
\left( \frac{1}{ a_n^2 } \sum_{t=1}^n \boldsymbol{Y}_{t-i} \boldsymbol{Y}_{t-j}^{\prime}, \frac{1}{ \tilde{a}_n^2 } \sum_{t=1}^n \boldsymbol{\varepsilon}_t \boldsymbol{Y}_{t-k}^{\prime} \right) \overset{d}{\to} \left( \sum_{\ell = 0}^{\infty} \boldsymbol{\Psi}_{\ell} \boldsymbol{S}_1 \boldsymbol{\Psi}^{\prime}_{\ell + i - j },  \sum_{\ell = 0}^{\infty}  \boldsymbol{S}_{ \ell + k + 1 } \boldsymbol{\Psi}_{\ell}^{\prime} \right)  
\end{align}
where $1 \leq j \leq i \leq p$, $1 \leq k \leq p$ and $\boldsymbol{Y}_t = \displaystyle \sum_{ \ell = 0 }^{ \infty } \boldsymbol{\Psi}_{\ell}  \boldsymbol{\varepsilon}_{t - \ell}$. 

Moreover, denoting with $\boldsymbol{B}_n = \left( \boldsymbol{D}_n \boldsymbol{Q} \right)^{-1}$ and $\boldsymbol{H}$ such that
\begin{align*}
\boldsymbol{H} = \mathsf{diag} \left\{ \frac{ \tilde{a}_n }{ a_n } \boldsymbol{I}_{ r_0 }, n^{-1/2} a_n \boldsymbol{I}_{m - r_0} \right\}.
\end{align*}
Then, based on the above asymptotic results and the joint weak convergence we obtain that
\begin{align*}
\boldsymbol{H}^{-1} \boldsymbol{D}_n \left( \sum_{t=1}^n \boldsymbol{Z}_{t-1} \boldsymbol{Z}_{t-1}^{\prime}  \right) \boldsymbol{D}_n \boldsymbol{H}^{-1} 
= 
\begin{bmatrix}
\displaystyle \frac{1}{ a_n^2 } \boldsymbol{S}_{11} & \displaystyle \frac{1}{ \sqrt{n} } a_n^{-2} \boldsymbol{S}_{12n}  
\\
\displaystyle \frac{1}{ \sqrt{n} } \frac{1}{ a_n^2 } \boldsymbol{S}_{11} & \displaystyle \frac{1}{ \sqrt{n} } a_n^{-2} \boldsymbol{S}_{12n}  
\end{bmatrix}
&\overset{d}{\to} 
\begin{bmatrix}
\displaystyle \sum_{ \ell = 0 }^{\infty} \boldsymbol{C}_{\ell} \boldsymbol{S}_1  \boldsymbol{C}_{\ell}^{\prime}  & \boldsymbol{0}
\\
\boldsymbol{0} & \displaystyle \boldsymbol{\phi} \left( \int_0^1 \boldsymbol{P}(r) \boldsymbol{P}^{\prime}(r) dr \right) \boldsymbol{\phi}
\end{bmatrix}.
\end{align*}
Therefore, it holds that 
\begin{align}
\mathsf{vec} \left( \widetilde{\boldsymbol{\Pi} }_{n,f} -  \widehat{\boldsymbol{\Pi}}_{n} \right)^{\prime} \left( \sum_{t=1}^n \boldsymbol{Y}_{t-1} \boldsymbol{Y}_{t-1}^{\prime} \otimes \boldsymbol{I}_m \right) \mathsf{vec} \left( \widetilde{\boldsymbol{\Pi} }_{n,f} - \widehat{\boldsymbol{\Pi}}_{n}  \right)^{\prime} \geq \mu_{ 2n. \mathsf{min} } \norm{ ( \widehat{\boldsymbol{\Pi}}_n - \widetilde{\boldsymbol{\Pi}}_{n,f} ) \boldsymbol{B}_n \boldsymbol{H} }^2. 
\end{align}

\subsubsection{Permanent-Transitory Shocks Decomposition}

Understanding the existing methods for decomposing permanent and transitory effects in cointegrating regressions is useful also in correctly implementing relevant identification strategies for SVAR models. In particular, \cite{gonzalo2001systematic} propose a framework for analyzing the dynamic effects of permanent and transitory shocks. Moreover \cite{hecq2000permanent} apply the permanent-transitory decomposition (see, \cite{beveridge1981new} and \cite{ quah1992relative}) in VAR models with cointegration and common cycles (see, also \cite{garratt2006permanent} and \cite{myers2018long}).  

Specifically decomposing which shocks are temporary and which shocks have a permanent effects is important for cointegration analysis purposes. Consider a bivariate SVAR(1) system, then if there as at least one permanent shock a cointegrated VAR representation can be identified or a Vector Error Correction Model while if all shocks have a permanent shock then there is no underline cointegration dynamics and thus the model represents a VAR system in first differences. In other words, under the absence of transitory shocks then the system involves no cointegration dynamics. On the other hand, there should be at least one shock that is not transitory but has a permanent impact in order to ensure that a cointegrated VAR representation holds.

\newpage 

Consider the following three-variable VEC
\begin{align*}
\Delta y_t = \mu + \alpha \beta^{\prime} y_{t-1} + \sum_{j=1}^q \Gamma_j \Delta y_{t-j} + \varepsilon_t,
\end{align*} 
such that there are $r \leq 2$ cointegrating relationships, $\beta$ contains the cointegrating vectors and $\mu, \alpha, \Gamma_j$ are unknown parameters to be estimated. Moreover, the error terms $\varepsilon_t$ are assumed to be serially uncorrelated but may be contemporaneously correlated. Based on the aforementioned formulation, \cite{gonzalo1995estimation} showed that given the three variable VEC representation with two cointegrating vectors, $y_t$ can be decomposed into permanent and transitoru components as below
\begin{align*}
y_t := A a^{\prime}_{\perp} y_t + B \beta^{\prime} y_{t} 
\equiv  
\underbrace{ \beta_{\perp} \left( \alpha_{\perp}^{\prime} \beta_{\perp} \right)^{-1} a^{\prime}_{\perp} }_{ \text{permanent component} } y_t 
+ 
\underbrace{ \alpha \left( \beta^{\prime} \alpha \right)^{-1} \beta^{\prime} }_{ \text{transitory component}  }y_{t}
\end{align*} 
where $A$ and $B$ are loading matrices which scale the two components accordingly.
Notice that the decomposition works because the term $a^{\prime}_{\perp} y_t$ represents the only linear combination of $y_t$ that ensures the transitory component  $B \beta^{\prime} y_{t}$ has no long-run impact on $y_t$. Furthermore, the particular decomposition when modelling the joint behaviour of multiple time series, is also useful in applications related to permanent and temporary income shock dynamics as well as to household consumption patterns. 

\begin{remark}
Long-run restrictions were pioneered with the seminal study of \cite{blanchard1988dynamic}. Moreover, the advantage of considering cointegration dynamics is that it can be shown that the long-run behaviour of an $n$ dimensional system with $r$ cointegrating relationships can be described by only $(n - r)$ independent stochastic trends. In other words, the presence of common trends can be identified as structural innovations to the system which implies that the remaining $r$ structural innovations cannot have a permanent impact to the system. Therefore, in order to conduct proper inference in these structural models the asymptotic distribution of the parameters of interest has to be derived. 
\end{remark}

\begin{example}[Multivariate Beveridge-Nelson Decomposition with $I(1)$ and $I(2)$ series]

The study of \cite{murasawa2015multivariate} considers that the consumption Euler equation implies that the output growth rate and the real interest rate are of the same order of integration. To estimate the natural rates and gaps of macroeconomic variables jointly, we need to develop a multivariate Beveridge-Nelson Decomposition when some series are $I(1)$ and other are $I(2)$. These common sets of macro variables as below: 
\begin{itemize}
\item output: let $Y_t$ be the output.  

\item inflation rate: let $P_t$ be the price level and $\pi_t = ln( P_t / P_{t-1} )$ be the inflation rate. Moreover, let $\hat{r}_t := i_t - \pi_{t+1}$ be the ex-post interest rate. 

\item interest rate: let $I_t$ be the 3-month interest rate (annual rate in percent). 

\item unemployment rate: let $L_t$ be the labor force, $E_t$ be employment and $U_t := - ln (E_t / L_t)$ be the unemployment rate. 
\end{itemize}

Assume that $\left\{ \Delta ln Y_t \right\}$ is $I(1)$ without a drift, $\left\{ \pi_t \right\}$ is $I(1)$ and $r_t$ is also $I(1)$.  

\end{example}

\newpage

\subsubsection{Automated Estimation of VECMs}

A related stream of literature includes the notion of \textit{irreducible cointegrating relation}, which implies that no variable can be omitted without loss of the cointegration property (see, \cite{davidson1998structural}) as well as the notion of causal ordering of cointegrating relations (see, \cite{hoover2020discovery}). In addition, the "minimal" algorithm proposed by \cite{davidson1998structural} is extended to the literature of automated cointegration models as in \cite{krolzig2003general} 
and \cite{liao2015automated}. Note that the "minimal" algorithm of \cite{davidson1998structural} is different from other methods since it tries to find all the smallest subsets of variables which are cointegrating. Furthermore, his method is based on the Wald test statistic developed by \cite{davidson1998wald} and thus the algorithm is faster since the Wald test does not require restricted optimization. Moreover, the approach proposed by \cite{liao2015automated} corresponds to an automated estimation of vector error correction model to tackle model selection and associared issues of post-model selection inference which present well known challenges in empirical econometric research.

Following the automated estimation of vector error correction model approach proposed by \cite{liao2015automated},  we consider the parametric VEC representation of a cointegrated system
\begin{align}
\Delta \boldsymbol{Y}_t = \boldsymbol{\Pi}_0 \boldsymbol{Y}_{t-1} + \sum_{j=1}^p \boldsymbol{B}_{0,j} \Delta \boldsymbol{Y}_{t-j} + \boldsymbol{u}_t,     
\end{align}
where $\Delta \boldsymbol{Y}_t = \boldsymbol{Y}_t - \boldsymbol{Y}_{t-1}$, where $\boldsymbol{Y}_t$ is an $m-$dimensional vector-valued time series such that $\boldsymbol{\Pi}_0 = \boldsymbol{\alpha}_0 \boldsymbol{\beta}_0^{\prime}$ has rank $0 \leq r_0 \leq m$ and $\boldsymbol{B}_{0,j}$ for $j = 1,...,p$ are $( m \times m)$ (transient) coefficient matrices, where $\boldsymbol{u}_t$ is an $m-$vector error term with mean zero and non-singular covariance matrix. 

Specifically, the OLS shrinkage estimator of $( \boldsymbol{\Pi}_0, \boldsymbol{B}_0 )$, where $\boldsymbol{B}_0 = ( \boldsymbol{B}_{0,1},..., \boldsymbol{B}_{0,p} )$ is defined as 
\begin{align*}
( \widehat{\boldsymbol{\Pi}}_0, \widehat{\boldsymbol{B}}_0 ) = \underset{ \Pi, B_1,..., B_p \in \mathbb{R}^{m \times m} }{ \mathsf{arg \ min} } &\sum_{t=1}^n \norm{ \Delta \boldsymbol{Y}_t - \Pi \boldsymbol{Y}_{t-1} - \sum_{ j \leq p} \boldsymbol{B}_j \Delta \boldsymbol{Y}_{t-j} }^2
\\
&+ n \sum_{j=1}^p \lambda_{b,j,n} \norm{ \boldsymbol{B}_j } + n \sum_{k=1}^m \lambda_{r,k,n} \norm{ \Phi_{n,k} (\boldsymbol{\Pi}) },  
\end{align*}
where $\lambda_{b,j,n}$ and $\lambda_{r,k,n}$ are tuning parameters that directly control the penalization.

\paragraph{First Order VECM Estimation}

Consider the following formulation of the VECM as below
\begin{align}
\Delta \boldsymbol{Y}_t = \boldsymbol{\Pi}_0 \boldsymbol{Y}_{t-1} + u_t \equiv \boldsymbol{\alpha}_0 \boldsymbol{\beta}_0^{\prime} \boldsymbol{Y}_{t-1} + \boldsymbol{u}_t.      
\end{align}
In other words, the model based on the above formulation contains no deterministic trend and no lagged differences. Therefore, our focus in this simplified system is to outline the approach to cointegrating rank selection and develop key elements in the limit theory, showing consistency in rank selection and reduced rank coefficient matrix estimation.

\newpage 

\begin{assumption}[WN] Suppose that $\left\{ \boldsymbol{u}_t \right\}_{ t \geq 1 }$ is an $m-$dimensional $\textit{i.i.d}$ process with zero mean and non-singular covariance matrix $\boldsymbol{\Omega}_u$.    
\end{assumption}

Under the above assumption then partial sums of $\boldsymbol{u}_t$ satisfy the functional law below
\begin{align}
\frac{1}{ \sqrt{n} } \sum_{ t = 1}^{ [ n \cdot ] } \boldsymbol{u}_t \Rightarrow \boldsymbol{B}_u( \cdot ),   
\end{align}
where $\boldsymbol{B}_u( \cdot )$ is a vector of Brownian motion with variance matrix $\boldsymbol{\Omega}_u$. 

Consider the following partial sum Granger representation as below
\begin{align*}
\boldsymbol{Y}_t =  \boldsymbol{C} \sum_{j=1}^t \boldsymbol{u}_j + \boldsymbol{\alpha}_0 \left( \boldsymbol{\beta}_0^{\prime} \boldsymbol{\alpha}_0  \right)^{-1} \boldsymbol{R}(L) \boldsymbol{\beta}_0^{\prime} \boldsymbol{u}_t + \boldsymbol{C} \boldsymbol{Y}_0,    
\end{align*}
Define with $\boldsymbol{Q} = [ \boldsymbol{\beta_0}, \boldsymbol{\alpha}_{0 \perp} ]^{\prime}$, and note that $\boldsymbol{C} = \boldsymbol{\beta}_{0 \perp}  \left(  \boldsymbol{\alpha}_{0 \perp }^{\prime}  \boldsymbol{\beta}_{0 \perp }^{\prime} \right)^{-1} \boldsymbol{\alpha}_{0 \perp}^{\prime}$. Then, using the matrix $\boldsymbol{Q}$ we obtain the following equivalent expression 
\begin{align}
\Delta \boldsymbol{Z}_t &= \underbrace{ \boldsymbol{Q} \boldsymbol{\Pi}_0 \boldsymbol{Q}^{-1} }_{  \boldsymbol{\Xi}_0 }  \boldsymbol{Z}_{t-1} + \boldsymbol{\eta}_t 
\\
\boldsymbol{Z}_t &= 
\begin{pmatrix}
\boldsymbol{\beta}_0^{\prime} \boldsymbol{Y}_t
\\
\boldsymbol{\alpha}_{0 \perp}^{\prime} \boldsymbol{Y}_t
\end{pmatrix}
\equiv
\begin{pmatrix}
\boldsymbol{Z}_{1t}
\\
\boldsymbol{Z}_{2t}
\end{pmatrix}, \ \ 
\boldsymbol{\eta}_t
=
\begin{pmatrix}
\boldsymbol{\beta}_0^{\prime} \boldsymbol{u}_t
\\
\boldsymbol{\alpha}_{0 \perp}^{\prime} \boldsymbol{u}_t   
\end{pmatrix}
\equiv 
\begin{pmatrix}
\boldsymbol{\eta}_{1,t}
\\
\boldsymbol{\eta}_{2,t}
\end{pmatrix}
\end{align}
Under the above assumptions, we have the following functional law as below
\begin{align}
\frac{1}{ \sqrt{n} } \sum_{t=1}^{ [ n \cdot ] } \boldsymbol{\eta}_t \Rightarrow \boldsymbol{B}_{\eta} (\cdot) 
= 
\boldsymbol{Q} \boldsymbol{B}_{\eta} (\cdot) 
=
\begin{bmatrix}
\boldsymbol{\beta}_0^{\prime} \boldsymbol{B}_u (\cdot) 
\\
\boldsymbol{\alpha}_{0 \perp}^{\prime} \boldsymbol{B}_u (\cdot) 
\end{bmatrix}
\equiv
\begin{bmatrix}
\boldsymbol{B}_{ \eta_1 } (\cdot) 
\\
\boldsymbol{B}_{ \eta_2 } (\cdot) 
\end{bmatrix}.
\end{align}
Then, the shrinkage LS estimator $\hat{\boldsymbol{\Pi}}_n$ of $\boldsymbol{\Pi}_0$ is
\begin{align*}
\widehat{\boldsymbol{\Pi}}_n = \underset{ \boldsymbol{\Pi} \in \mathbb{R}^{m \times m} }{ \mathsf{arg \ min} } \ \sum_{t=1}^n \norm{ \Delta \boldsymbol{Y}_t - \boldsymbol{\Pi} \boldsymbol{Y}_{t-1} }^2
+ n \sum_{k=1}^m \lambda_{r,k,n} \norm{ \Phi_{n,k} (\boldsymbol{\Pi}) }.  
\end{align*}
Notice that the unrestricted LS estimator $\widehat{\boldsymbol{\Pi}}$ of $\boldsymbol{\Pi}_0$ is
\begin{align*}
\widehat{\boldsymbol{\Pi}}^{unrest}_n 
= 
\underset{ \boldsymbol{\Pi} \in \mathbb{R}^{m \times m} }{ \mathsf{arg \ min} } \ \sum_{t=1}^n \norm{ \Delta \boldsymbol{Y}_t - \boldsymbol{\Pi} \boldsymbol{Y}_{t-1} }^2
=
\left( \sum_{t=1}^n \Delta \boldsymbol{Y}_t \boldsymbol{Y}_{t-1}^{\prime} \right) \left( \sum_{t=1}^n \Delta \boldsymbol{Y}_{t-1} \boldsymbol{Y}_{t-1}^{\prime} \right)^{-1}.
\end{align*}
Decomposing nearly nonstationary from nearly stationary roots in the set of regressors is discussed in \cite{paruolo1997asymptotic} who considers that the cointegrated system $\boldsymbol{X}_t = P_{\beta} \boldsymbol{X}_t + P_{ \beta_{ \perp}  } \boldsymbol{X}_t$ in two systems of $r$ stationary $\boldsymbol{\beta}^{\prime} \boldsymbol{X}_t$ and $(d-r)$ nonstationary $\boldsymbol{ \beta^{\prime}_{ \perp}  } \boldsymbol{X}_t$ components (see, also \cite{stock1988testing}).

\newpage

\section{Algebraic Theory of Identification in Structural Models}

\subsection{Relation between Equivalence and Orbit}

Following, \cite{kociecki2011algebraic} we shall avoid using the definition that only the structural parameters that are in one-to-one correspondence with the reduced form parameters are identified. The reason behind this is that this is not always the case. In fact, there are other structural parameters, which are identified, but can not be uniquely recovered from the reduced form parameters. In other words, an one-to-one correspondence is a necessary condition for identification but is not a unique one. 

There is a close connection between equivalence class and orbit. For a number of econometric models, equivalence classes are simply orbits. Therefore, when equivalence class is an orbit the approach to identification based on checking local properties of the likelihood (information matrix), is rather misplaced. We present a relevant definition which is a fundamental tool in statistical invariance theory. 

\begin{definition}[\cite{kociecki2011algebraic}]
A function $f: \Theta \to Y$ is said to be invariant under some action of a group $G$ on $\Theta$, if $f ( \theta ) = f( g \circ  \theta )$ for any $g \in G, \theta \in \Theta$. Moreover, a function $f: \Theta \to Y$ is called \textit{maximal invariant} $G-$invariant if f is $G-$invariant and for any $\theta_1, \theta_2 \in \Theta$, $f(\theta_1) = f(\theta_2)$ implies $\theta_1 = g \circ \theta_2$ for some $g \in G$, that is, $\theta_1$ and $\theta_2$ lie on the same orbit.   
\end{definition}

\begin{example}[Finite Mixture Models]
Suppose that the pdf of a finite mixture of two normal distributions is given by 
\begin{align}
f_{y_t} (y_t) = p_1 \cdot (2 \pi)^{- \frac{1}{2} } \mathsf{exp} \left\{ - \frac{1}{2} (y_t - \mu)^2 \right\} + p_2 \cdot (2 \pi)^{- \frac{1}{2} } \mathsf{exp} \left\{ - \frac{1}{2} (y_t - \mu)^2 \right\}.    
\end{align}
where $y_t$ is an one-dimensional endogenous variable, and $0 \leq p_i \leq 1$, $p_1 + p_2 = 1$ and $\mu_1, \mu_2 \in \mathbb{R}$. 

Let $\theta = \left( p_1, \mu_1, p_2, \mu_2 \right) \in \Theta$, then $C_{\theta} = \mathsf{Orb}_{\theta} := \big\{ g \circ \left( p_1, \mu_1, p_2, \mu_2 \right) | g \in S_2 \big\}$, where $S_2$ denotes the symmetric group of degree 2. Notice that in general, $S_n$ is the group of permutations which has $n!$ elements, that is, $| S_n| = n!$. Therefore, $g \circ \left( p_1, \mu_1, p_2, \mu_2 \right) := \big( p_{g(1)}, \mu_{g(1)}, p_{g(2)}, \mu_{g(2)} \big)$. 
\end{example}

\subsection{Iterative Schemes for Large Linear Systems}

\begin{definition}
A square matrix $B$ is irreducible if there exists no permutation matrix $Q$ for which 
\begin{align}
Q B Q^{\top} = 
\begin{bmatrix}
B_{11} & B_{12}
\\
0 & B_{22}
\end{bmatrix}
\end{align}
where $B_{11}$ and $B_{22}$ are square matrices. 
\end{definition}

\begin{definition}
A matrix $B = (\beta_{ij}) \in \mathbb{R}^{ m \times m }$ is \textit{irreducibly diagonally dominant} if it is irreducible, 
\begin{align}
| \beta_{ii} | \geq \sum_{ j = 1, j \neq i } | \beta_{ij} |, \ \ \textit{with strict inequality holding at least for one} \ i.    
\end{align}
\end{definition}

\newpage



\newpage

\bibliographystyle{apalike}
\addcontentsline{toc}{section}{References}
\bibliography{myreferences1}

\begin{thebibliography}{}

\bibitem[Adamek et~al., 2022]{adamek2022local}
Adamek, R., Smeekes, S., and Wilms, I. (2022).
\newblock Local projection inference in high dimensions.
\newblock {\em arXiv preprint arXiv:2209.03218}.

\bibitem[Agullo et~al., 2008]{agullo2008multivariate}
Agullo, J., Croux, C., and Van~Aelst, S. (2008).
\newblock The multivariate least-trimmed squares estimator.
\newblock {\em Journal of Multivariate Analysis}, 99(3):311--338.

\bibitem[Aigner and Balestra, 1988]{aigner1988optimal}
Aigner, D.~J. and Balestra, P. (1988).
\newblock Optimal experimental design for error components models.
\newblock {\em Econometrica: Journal of the Econometric Society}, pages
  955--971.

\bibitem[Alloza et~al., 2020]{alloza2020dynamic}
Alloza, M., Gonzalo, J., and Sanz, C. (2020).
\newblock Dynamic effects of persistent shocks.
\newblock {\em arXiv preprint arXiv:2006.14047}.

\bibitem[Andrade et~al., 2005]{andrade2005testing}
Andrade, P., Bruneau, C., and Gregoir, S. (2005).
\newblock Testing for the cointegration rank when some cointegrating directions
  are changing.
\newblock {\em Journal of Econometrics}, 124(2):269--310.

\bibitem[Andrews et~al., 2009]{andrews2009maximum}
Andrews, B., Calder, M., and Davis, R.~A. (2009).
\newblock Maximum likelihood estimation for $\alpha$-stable autoregressive
  processes.
\newblock {\em The Annals of Statistics}, 37(4):1946--1982.

\bibitem[Andrews and Davis, 2013]{andrews2013model}
Andrews, B. and Davis, R.~A. (2013).
\newblock Model identification for infinite variance autoregressive processes.
\newblock {\em Journal of Econometrics}, 172(2):222--234.

\bibitem[Angelini et~al., 2024]{angelini2024identification}
Angelini, G., Cavaliere, G., and Fanelli, L. (2024).
\newblock An identification and testing strategy for proxy-svars with weak
  proxies.
\newblock {\em Journal of Econometrics}, 238(2):105604.

\bibitem[Angrist and Kuersteiner, 2011]{angrist2011causal}
Angrist, J.~D. and Kuersteiner, G.~M. (2011).
\newblock Causal effects of monetary shocks: Semiparametric conditional
  independence tests with a multinomial propensity score.
\newblock {\em Review of Economics and Statistics}, 93(3):725--747.

\bibitem[Antol{\'\i}n-D{\'\i}az and Rubio-Ram{\'\i}rez,
  2018]{antolin2018narrative}
Antol{\'\i}n-D{\'\i}az, J. and Rubio-Ram{\'\i}rez, J.~F. (2018).
\newblock Narrative sign restrictions for svars.
\newblock {\em American Economic Review}, 108(10):2802--2829.

\bibitem[Anttonen et~al., 2023]{anttonen2023statistically}
Anttonen, J., Lanne, M., and Luoto, J. (2023).
\newblock Statistically identified svar model with potentially skewed and
  fat-tailed errors.
\newblock {\em Available at SSRN 3925575}.

\bibitem[Arias et~al., 2023]{arias2023causal}
Arias, J.~E., Fern{\'a}ndez-Villaverde, J., Rubio-Ram{\'\i}rez, J.~F., and
  Shin, M. (2023).
\newblock The causal effects of lockdown policies on health and macroeconomic
  outcomes.
\newblock {\em American Economic Journal: Macroeconomics}, 15(3):287--319.

\bibitem[Aruoba et~al., 2022]{aruoba2022svars}
Aruoba, S.~B., Mlikota, M., Schorfheide, F., and Villalvazo, S. (2022).
\newblock Svars with occasionally-binding constraints.
\newblock {\em Journal of Econometrics}, 231(2):477--499.

\bibitem[Aue et~al., 2009]{aue2009break}
Aue, A., H{\"o}rmann, S., Horv{\'a}th, L., Reimherr, M., et~al. (2009).
\newblock Break detection in the covariance structure of multivariate time
  series models.
\newblock {\em The Annals of Statistics}, 37(6B):4046--4087.

\bibitem[Bacchiocchi et~al., 2018]{bacchiocchi2018gimme}
Bacchiocchi, E., Castelnuovo, E., and Fanelli, L. (2018).
\newblock Gimme a break! identification and estimation of the macroeconomic
  effects of monetary policy shocks in the united states.
\newblock {\em Macroeconomic Dynamics}, 22(6):1613--1651.

\bibitem[Bacchiocchi and Fanelli, 2011]{bacchiocchi2011new}
Bacchiocchi, E. and Fanelli, L. (2011).
\newblock New identification strategies in structural vector autoregressive
  models with structural changes, with an application to us monetary policy.
\newblock {\em hypothesis}, page~2.

\bibitem[Bai et~al., 2024]{bai2024likelihood}
Bai, J., Duan, J., and Han, X. (2024).
\newblock The likelihood ratio test for structural changes in factor models.
\newblock {\em Journal of Econometrics}, 238(2):105631.

\bibitem[Bai and Perron, 1998]{bai1998estimating}
Bai, J. and Perron, P. (1998).
\newblock Estimating and testing linear models with multiple structural
  changes.
\newblock {\em Econometrica}, pages 47--78.

\bibitem[Bai and Wang, 2014]{bai2014identification}
Bai, J. and Wang, P. (2014).
\newblock Identification theory for high dimensional static and dynamic factor
  models.
\newblock {\em Journal of Econometrics}, 178(2):794--804.

\bibitem[Baillie and Kapetanios, 2013]{baillie2013estimation}
Baillie, R.~T. and Kapetanios, G. (2013).
\newblock Estimation and inference for impulse response functions from
  univariate strongly persistent processes.
\newblock {\em The Econometrics Journal}, 16(3):373--399.

\bibitem[Baker et~al., 2016]{baker2016measuring}
Baker, S.~R., Bloom, N., and Davis, S.~J. (2016).
\newblock Measuring economic policy uncertainty.
\newblock {\em The quarterly journal of economics}, 131(4):1593--1636.

\bibitem[Baltagi, 2008]{baltagi2008econometric}
Baltagi, B.~H. (2008).
\newblock {\em Econometric analysis of panel data}, volume~4.
\newblock Springer.

\bibitem[Baltagi et~al., 2021]{baltagi2021estimating}
Baltagi, B.~H., Kao, C., and Wang, F. (2021).
\newblock Estimating and testing high dimensional factor models with multiple
  structural changes.
\newblock {\em Journal of Econometrics}, 220(2):349--365.

\bibitem[Banafti and Lee, 2022]{banafti2022inferential}
Banafti, S. and Lee, T.-H. (2022).
\newblock Inferential theory for granular instrumental variables in high
  dimensions.
\newblock {\em arXiv preprint arXiv:2201.06605}.

\bibitem[Barigozzi and Hallin, 2017]{barigozzi2017network}
Barigozzi, M. and Hallin, M. (2017).
\newblock A network analysis of the volatility of high dimensional financial
  series.
\newblock {\em Journal of the Royal Statistical Society Series C: Applied
  Statistics}, 66(3):581--605.

\bibitem[Barigozzi et~al., 2021a]{barigozzi2021time}
Barigozzi, M., Hallin, M., Soccorsi, S., and von Sachs, R. (2021a).
\newblock Time-varying general dynamic factor models and the measurement of
  financial connectedness.
\newblock {\em Journal of Econometrics}, 222(1):324--343.

\bibitem[Barigozzi et~al., 2021b]{barigozzi2021large}
Barigozzi, M., Lippi, M., and Luciani, M. (2021b).
\newblock Large-dimensional dynamic factor models: Estimation of
  impulse--response functions with i (1) cointegrated factors.
\newblock {\em Journal of Econometrics}, 221(2):455--482.

\bibitem[Barigozzi and Trapani, 2022]{barigozzi2022testing}
Barigozzi, M. and Trapani, L. (2022).
\newblock Testing for common trends in nonstationary large datasets.
\newblock {\em Journal of Business \& Economic Statistics}, 40(3):1107--1122.

\bibitem[Barnichon and Brownlees, 2019]{barnichon2019impulse}
Barnichon, R. and Brownlees, C. (2019).
\newblock Impulse response estimation by smooth local projections.
\newblock {\em Review of Economics and Statistics}, 101(3):522--530.

\bibitem[Baron et~al., 2021]{baron2021banking}
Baron, M., Verner, E., and Xiong, W. (2021).
\newblock Banking crises without panics.
\newblock {\em The Quarterly Journal of Economics}, 136(1):51--113.

\bibitem[Barun{\'\i}k and Ellington, 2024]{barunik2024persistence}
Barun{\'\i}k, J. and Ellington, M. (2024).
\newblock Persistence in financial connectedness and systemic risk.
\newblock {\em European Journal of Operational Research}, 314(1):393--407.

\bibitem[Barun{\'\i}k and K{\v{r}}ehl{\'\i}k, 2018]{barunik2018measuring}
Barun{\'\i}k, J. and K{\v{r}}ehl{\'\i}k, T. (2018).
\newblock Measuring the frequency dynamics of financial connectedness and
  systemic risk.
\newblock {\em Journal of Financial Econometrics}, 16(2):271--296.

\bibitem[Baumeister and Hamilton, 2015]{baumeister2015sign}
Baumeister, C. and Hamilton, J.~D. (2015).
\newblock Sign restrictions, structural vector autoregressions, and useful
  prior information.
\newblock {\em Econometrica}, 83(5):1963--1999.

\bibitem[Baumeister and Hamilton, 2019]{baumeister2019structural}
Baumeister, C. and Hamilton, J.~D. (2019).
\newblock Structural interpretation of vector autoregressions with incomplete
  identification: Revisiting the role of oil supply and demand shocks.
\newblock {\em American Economic Review}, 109(5):1873--1910.

\bibitem[Baumeister and Hamilton, 2021]{baumeister2021advances}
Baumeister, C. and Hamilton, J.~D. (2021).
\newblock Advances in using vector autoregressions to estimate structural
  magnitudes.
\newblock {\em Econometric Theory}, pages 1--39.

\bibitem[Baumeister and Hamilton, 2023]{baumeister2023full}
Baumeister, C. and Hamilton, J.~D. (2023).
\newblock A full-information approach to granular instrumental variables.
\newblock Technical report, Working Paper, UCSD.

\bibitem[Ben-Michael et~al., 2021]{ben2021augmented}
Ben-Michael, E., Feller, A., and Rothstein, J. (2021).
\newblock The augmented synthetic control method.
\newblock {\em Journal of the American Statistical Association},
  116(536):1789--1803.

\bibitem[Bertsche and Braun, 2022]{bertsche2022identification}
Bertsche, D. and Braun, R. (2022).
\newblock Identification of structural vector autoregressions by stochastic
  volatility.
\newblock {\em Journal of Business \& Economic Statistics}, 40(1):328--341.

\bibitem[Bertsche et~al., 2023]{bertsche2023directed}
Bertsche, D., Br{\"u}ggemann, R., and Kascha, C. (2023).
\newblock Directed graphs and variable selection in large vector autoregressive
  models.
\newblock {\em Journal of Time Series Analysis}, 44(2):223--246.

\bibitem[Beveridge and Nelson, 1981]{beveridge1981new}
Beveridge, S. and Nelson, C.~R. (1981).
\newblock A new approach to decomposition of economic time series into
  permanent and transitory components with particular attention to measurement
  of the ‘business cycle’.
\newblock {\em Journal of Monetary economics}, 7(2):151--174.

\bibitem[Bhansali, 2002]{bhansali2002multi}
Bhansali, R.~J. (2002).
\newblock Multi-step forecasting.
\newblock {\em A Companion to Economic Forecasting}, 1:207--221.

\bibitem[Bianchi et~al., 2023]{bianchi2023fiscal}
Bianchi, F., Faccini, R., and Melosi, L. (2023).
\newblock A fiscal theory of persistent inflation.
\newblock {\em The Quarterly Journal of Economics}, page qjad027.

\bibitem[Bilgili, 2012]{bilgili2012impact}
Bilgili, F. (2012).
\newblock The impact of biomass consumption on co2 emissions: cointegration
  analyses with regime shifts.
\newblock {\em Renewable and Sustainable Energy Reviews}, 16(7):5349--5354.

\bibitem[Blanchard and Perotti, 2003]{blanchard2003empirical}
Blanchard, O. and Perotti, R. (2003).
\newblock An empirical investigation of the dynamic effects of shocks to
  government spending and taxes on output.
\newblock {\em Quarterly Journal of Economics}, 117:1--329.

\bibitem[Blanchard and Quah, 1988]{blanchard1988dynamic}
Blanchard, O.~J. and Quah, D. (1988).
\newblock The dynamic effects of aggregate demand and supply disturbances.

\bibitem[Bloom, 2009]{bloom2009impact}
Bloom, N. (2009).
\newblock The impact of uncertainty shocks.
\newblock {\em econometrica}, 77(3):623--685.

\bibitem[Bognanni, 2018]{bognanni2018class}
Bognanni, M. (2018).
\newblock A class of time-varying parameter structural vars for inference under
  exact or set identification.

\bibitem[Boiciuc, 2015]{boiciuc2015effects}
Boiciuc, I. (2015).
\newblock The effects of fiscal policy shocks in romania. a svar approach.
\newblock {\em Procedia Economics and Finance}, 32:1131--1139.

\bibitem[Bojinov et~al., 2021]{bojinov2021panel}
Bojinov, I., Rambachan, A., and Shephard, N. (2021).
\newblock Panel experiments and dynamic causal effects: A finite population
  perspective.
\newblock {\em Quantitative Economics}, 12(4):1171--1196.

\bibitem[Bojinov and Shephard, 2019]{bojinov2019time}
Bojinov, I. and Shephard, N. (2019).
\newblock Time series experiments and causal estimands: exact randomization
  tests and trading.
\newblock {\em Journal of the American Statistical Association},
  114(528):1665--1682.

\bibitem[Boswijk, 1994]{boswijk1994testing}
Boswijk, H.~P. (1994).
\newblock Testing for an unstable root in conditional and structural error
  correction models.
\newblock {\em Journal of econometrics}, 63(1):37--60.

\bibitem[Braun, 2023]{braun2023importance}
Braun, R. (2023).
\newblock The importance of supply and demand for oil prices: Evidence from
  non-gaussianity.
\newblock {\em Quantitative Economics}, 14(4):1163--1198.

\bibitem[Braun and Br{\"u}ggemann, 2023]{braun2023identification}
Braun, R. and Br{\"u}ggemann, R. (2023).
\newblock Identification of svar models by combining sign restrictions with
  external instruments.
\newblock {\em Journal of Business \& Economic Statistics}, 41(4):1077--1089.

\bibitem[Breitung and Eickmeier, 2011]{breitung2011testing}
Breitung, J. and Eickmeier, S. (2011).
\newblock Testing for structural breaks in dynamic factor models.
\newblock {\em Journal of Econometrics}, 163(1):71--84.

\bibitem[Br{\"u}ggemann et~al., 2016]{bruggemann2016inference}
Br{\"u}ggemann, R., Jentsch, C., and Trenkler, C. (2016).
\newblock Inference in vars with conditional heteroskedasticity of unknown
  form.
\newblock {\em Journal of econometrics}, 191(1):69--85.

\bibitem[Bruneau and De~Bandt, 2003]{bruneau2003monetary}
Bruneau, C. and De~Bandt, O. (2003).
\newblock Monetary and fiscal policy in the transition to emu: what do svar
  models tell us?
\newblock {\em Economic Modelling}, 20(5):959--985.

\bibitem[Bruns and Keweloh, 2023]{bruns2023testing}
Bruns, M. and Keweloh, S. (2023).
\newblock Testing for strong exogeneity in proxy-vars.
\newblock {\em Available at SSRN 4674721}.

\bibitem[Bruns et~al., 2020]{bruns2020multicointegration}
Bruns, S.~B., Csereklyei, Z., and Stern, D.~I. (2020).
\newblock A multicointegration model of global climate change.
\newblock {\em Journal of Econometrics}, 214(1):175--197.

\bibitem[Budnik and R{\"u}nstler, 2023]{budnik2023identifying}
Budnik, K. and R{\"u}nstler, G. (2023).
\newblock Identifying structural vars from sparse narrative instruments:
  Dynamic effects of us macroprudential policies.
\newblock {\em Journal of Applied Econometrics}, 38(2):186--201.

\bibitem[Cai and Lin, 1992]{cai1992response}
Cai, G. and Lin, Y. (1992).
\newblock Response distribution of non-linear systems excited by non-gaussian
  impulsive noise.
\newblock {\em International journal of non-linear mechanics}, 27(6):955--967.

\bibitem[Caldara and Kamps, 2017]{caldara2017analytics}
Caldara, D. and Kamps, C. (2017).
\newblock The analytics of svars: a unified framework to measure fiscal
  multipliers.
\newblock {\em The Review of Economic Studies}, 84(3):1015--1040.

\bibitem[Carriero et~al., 2023a]{carriero2023macro}
Carriero, A., Marcellino, M., and Tornese, T. (2023a).
\newblock Macro uncertainty in the long run.
\newblock {\em Economics Letters}, 225:111067.

\bibitem[Carriero et~al., 2023b]{carriero2023blended}
Carriero, A., Marcellino, M.~G., and Tornese, T. (2023b).
\newblock Blended identification in structural vars.
\newblock {\em BAFFI CAREFIN Centre Research Paper}, (200).

\bibitem[Carrion-i Silvestre and Kim, 2021]{carrion2021statistical}
Carrion-i Silvestre, J.~L. and Kim, D. (2021).
\newblock Statistical tests of a simple energy balance equation in a synthetic
  model of cotrending and cointegration.
\newblock {\em Journal of Econometrics}, 224(1):22--38.

\bibitem[Castelnuovo and Surico, 2010]{castelnuovo2010monetary}
Castelnuovo, E. and Surico, P. (2010).
\newblock Monetary policy, inflation expectations and the price puzzle.
\newblock {\em The Economic Journal}, 120(549):1262--1283.

\bibitem[Cavaliere and Georgiev, 2009]{cavaliere2009robust}
Cavaliere, G. and Georgiev, I. (2009).
\newblock Robust inference in autoregressions with multiple outliers.
\newblock {\em Econometric Theory}, pages 1625--1661.

\bibitem[Cavaliere et~al., 2020]{cavaliere2020bootstrapping}
Cavaliere, G., Nielsen, H.~B., and Rahbek, A. (2020).
\newblock Bootstrapping noncausal autoregressions: with applications to
  explosive bubble modeling.
\newblock {\em Journal of Business \& Economic Statistics}, 38(1):55--67.

\bibitem[Cevik and Miryugin, 2023]{cevik2023s}
Cevik, S. and Miryugin, F. (2023).
\newblock It's never different: Fiscal policy shocks and inflation.

\bibitem[Chan et~al., 2016]{chan2016large}
Chan, J.~C., Eisenstat, E., and Koop, G. (2016).
\newblock Large bayesian varmas.
\newblock {\em Journal of Econometrics}, 192(2):374--390.

\bibitem[Chan et~al., 2023]{chan2023large}
Chan, J.~C., Matthes, C., and Yu, X. (2023).
\newblock Large structural vars with multiple sign and ranking restrictions.

\bibitem[Chan et~al., 2006]{chan2006note}
Chan, K.-S., Ho, L.-H., and Tong, H. (2006).
\newblock A note on time-reversibility of multivariate linear processes.
\newblock {\em Biometrika}, 93(1):221--227.

\bibitem[Chang et~al., 2015]{chang2015high}
Chang, J., Chen, S.~X., and Chen, X. (2015).
\newblock High dimensional generalized empirical likelihood for moment
  restrictions with dependent data.
\newblock {\em Journal of Econometrics}, 185(1):283--304.

\bibitem[Chari et~al., 2008]{chari2008structural}
Chari, V.~V., Kehoe, P.~J., and McGrattan, E.~R. (2008).
\newblock Are structural vars with long-run restrictions useful in developing
  business cycle theory?
\newblock {\em Journal of Monetary Economics}, 55(8):1337--1352.

\bibitem[Chatfield, 1993]{chatfield1993calculating}
Chatfield, C. (1993).
\newblock Calculating interval forecasts.
\newblock {\em Journal of Business \& Economic Statistics}, 11(2):121--135.

\bibitem[Chaudourne et~al., 2014]{chaudourne2014understanding}
Chaudourne, J., F{\`e}ve, P., and Guay, A. (2014).
\newblock Understanding the effect of technology shocks in svars with long-run
  restrictions.
\newblock {\em Journal of Economic Dynamics and Control}, 41:154--172.

\bibitem[Chen et~al., 2017]{chen2017testing}
Chen, B., Choi, J., and Escanciano, J.~C. (2017).
\newblock Testing for fundamental vector moving average representations.
\newblock {\em Quantitative Economics}, 8(1):149--180.

\bibitem[Chen, 2010]{chen2010new}
Chen, H.-F. (2010).
\newblock New approach to recursive identification for armax systems.
\newblock {\em IEEE Transactions on Automatic Control}, 55(4):868--879.

\bibitem[Chen and Zhao, 2014]{chen2014recursive}
Chen, H.-F. and Zhao, W. (2014).
\newblock {\em Recursive identification and parameter estimation}.
\newblock CRC Press.

\bibitem[Cheng et~al., 2019]{cheng2019regime}
Cheng, T., Gao, J., and Yan, Y. (2019).
\newblock Regime switching panel data models with interactive fixed effects.
\newblock {\em Economics Letters}, 177:47--51.

\bibitem[Cheng et~al., 2022]{cheng2022instrumental}
Cheng, X., Han, X., and Inoue, A. (2022).
\newblock Instrumental variable estimation of structural var models robust to
  possible nonstationarity.
\newblock {\em Econometric Theory}, 38(5):845--874.

\bibitem[Chevillon et~al., 2020]{chevillon2020robust}
Chevillon, G., Mavroeidis, S., and Zhan, Z. (2020).
\newblock Robust inference in structural vector autoregressions with long-run
  restrictions.
\newblock {\em Econometric Theory}, 36(1):86--121.

\bibitem[Chinco et~al., 2019]{chinco2019sparse}
Chinco, A., Clark-Joseph, A.~D., and Ye, M. (2019).
\newblock Sparse signals in the cross-section of returns.
\newblock {\em The Journal of Finance}, 74(1):449--492.

\bibitem[Choi, 2002]{choi2002structural}
Choi, I. (2002).
\newblock Structural changes and seemingly unidentified structural equations.
\newblock {\em Econometric Theory}, 18(3):744--775.

\bibitem[Choi and Kurozumi, 2012]{choi2012model}
Choi, I. and Kurozumi, E. (2012).
\newblock Model selection criteria for the leads-and-lags cointegrating
  regression.
\newblock {\em Journal of Econometrics}, 169(2):224--238.

\bibitem[Christiano et~al., 2003]{christiano2003happens}
Christiano, L., Eichenbaum, M.~S., and Vigfusson, R.~J. (2003).
\newblock What happens after a technology shock?

\bibitem[Cordoni et~al., 2023]{cordoni2023identification}
Cordoni, F., Doremus, N., and Moneta, A. (2023).
\newblock Identification of vector autoregressive models with nonlinear
  contemporaneous structure.

\bibitem[Crucil et~al., 2023]{crucil2023monetary}
Crucil, R., Hambuckers, J., and Maxand, S. (2023).
\newblock Do monetary policy shocks affect financial uncertainty? a
  non-gaussian proxy svar approach.
\newblock {\em A Non-gaussian Proxy SVAR Approach (June 5, 2023)}.

\bibitem[Davidson, 1998a]{davidson1998structural}
Davidson, J. (1998a).
\newblock Structural relations, cointegration and identification: some simple
  results and their application.
\newblock {\em Journal of econometrics}, 87(1):87--113.

\bibitem[Davidson, 1998b]{davidson1998wald}
Davidson, J. (1998b).
\newblock A wald test of restrictions on the cointegrating space based on
  johansen's estimator.
\newblock {\em Economics Letters}, 59(2):183--187.

\bibitem[Davis and Ng, 2023]{davis2023time}
Davis, R. and Ng, S. (2023).
\newblock Time series estimation of the dynamic effects of disaster-type
  shocks.
\newblock {\em Journal of Econometrics}, 235(1):180--201.

\bibitem[Davis and Song, 2020]{davis2020noncausal}
Davis, R.~A. and Song, L. (2020).
\newblock Noncausal vector ar processes with application to economic time
  series.
\newblock {\em Journal of Econometrics}, 216(1):246--267.

\bibitem[De~Chaisemartin and d’Haultfoeuille, 2020]{de2020two}
De~Chaisemartin, C. and d’Haultfoeuille, X. (2020).
\newblock Two-way fixed effects estimators with heterogeneous treatment
  effects.
\newblock {\em American Economic Review}, 110(9):2964--2996.

\bibitem[Deistler, 1983]{deistler1983properties}
Deistler, M. (1983).
\newblock The properties of the parameterization of armax systems and their
  relevance for structural estimation and dynamic specification.
\newblock {\em Econometrica: Journal of the Econometric Society}, pages
  1187--1207.

\bibitem[Demirer et~al., 2018]{demirer2018estimating}
Demirer, M., Diebold, F.~X., Liu, L., and Yilmaz, K. (2018).
\newblock Estimating global bank network connectedness.
\newblock {\em Journal of Applied Econometrics}, 33(1):1--15.

\bibitem[Dhrymes, 2013]{dhrymes2013mathematics}
Dhrymes, P. (2013).
\newblock {\em Mathematics for econometrics}.
\newblock Springer.

\bibitem[Dias and Kapetanios, 2018]{dias2018estimation}
Dias, G.~F. and Kapetanios, G. (2018).
\newblock Estimation and forecasting in vector autoregressive moving average
  models for rich datasets.
\newblock {\em Journal of Econometrics}, 202(1):75--91.

\bibitem[Dickey et~al., 1986]{dickey1986unit}
Dickey, D.~A., Bell, W.~R., and Miller, R.~B. (1986).
\newblock Unit roots in time series models: Tests and implications.
\newblock {\em American statistician}, pages 12--26.

\bibitem[Diebold and Inoue, 2001]{diebold2001long}
Diebold, F.~X. and Inoue, A. (2001).
\newblock Long memory and regime switching.
\newblock {\em Journal of econometrics}, 105(1):131--159.

\bibitem[Diebold and Yilmaz, 2012]{diebold2012better}
Diebold, F.~X. and Yilmaz, K. (2012).
\newblock Better to give than to receive: Predictive directional measurement of
  volatility spillovers.
\newblock {\em International Journal of forecasting}, 28(1):57--66.

\bibitem[Diebold and Y{\i}lmaz, 2014]{diebold2014network}
Diebold, F.~X. and Y{\i}lmaz, K. (2014).
\newblock On the network topology of variance decompositions: Measuring the
  connectedness of financial firms.
\newblock {\em Journal of econometrics}, 182(1):119--134.

\bibitem[Dolado, 1992]{dolado1992note}
Dolado, J.~J. (1992).
\newblock A note on weak exogeneity in var cointegrated models.
\newblock {\em Economics Letters}, 38(2):139--143.

\bibitem[Drautzburg and Wright, 2023]{drautzburg2023refining}
Drautzburg, T. and Wright, J.~H. (2023).
\newblock Refining set-identification in vars through independence.
\newblock {\em Journal of Econometrics}.

\bibitem[Duffy et~al., 2022]{duffy2022cointegration}
Duffy, J.~A., Mavroeidis, S., and Wycherley, S. (2022).
\newblock Cointegration with occasionally binding constraints.
\newblock {\em arXiv preprint arXiv:2211.09604}.

\bibitem[Dufour, 2003]{dufour2003identification}
Dufour, J.-M. (2003).
\newblock Identification, weak instruments, and statistical inference in
  econometrics.
\newblock {\em Canadian Journal of Economics/Revue canadienne
  d'{\'e}conomique}, 36(4):767--808.

\bibitem[Dufour and Taamouti, 2005]{dufour2005projection}
Dufour, J.-M. and Taamouti, M. (2005).
\newblock Projection-based statistical inference in linear structural models
  with possibly weak instruments.
\newblock {\em Econometrica}, 73(4):1351--1365.

\bibitem[Eggertsson et~al., 2014]{eggertsson2014can}
Eggertsson, G., Ferrero, A., and Raffo, A. (2014).
\newblock Can structural reforms help europe?
\newblock {\em Journal of Monetary Economics}, 61:2--22.

\bibitem[Eickmeier et~al., 2015]{eickmeier2015classical}
Eickmeier, S., Lemke, W., and Marcellino, M. (2015).
\newblock Classical time varying factor-augmented vector auto-regressive
  models—estimation, forecasting and structural analysis.
\newblock {\em Journal of the Royal Statistical Society Series A: Statistics in
  Society}, 178(3):493--533.

\bibitem[Engle et~al., 1983]{engle1983exogeneity}
Engle, R.~F., Hendry, D.~F., and Richard, J.-F. (1983).
\newblock Exogeneity.
\newblock {\em Econometrica: Journal of the Econometric Society}, pages
  277--304.

\bibitem[Eriksson and Koivunen, 2004]{eriksson2004identifiability}
Eriksson, J. and Koivunen, V. (2004).
\newblock Identifiability, separability, and uniqueness of linear ica models.
\newblock {\em IEEE signal processing letters}, 11(7):601--604.

\bibitem[Fang et~al., 2020]{fang2020segmentation}
Fang, X., Li, J., and Siegmund, D. (2020).
\newblock Segmentation and estimation of change-point models: false positive
  control and confidence regions.
\newblock {\em The Annals of Statistics}, 48(3):1615--1647.

\bibitem[Feldkircher et~al., 2020]{feldkircher2020factor}
Feldkircher, M., Huber, F., and Pfarrhofer, M. (2020).
\newblock Factor augmented vector autoregressions, panel vars, and global vars.
\newblock {\em Macroeconomic Forecasting in the Era of Big Data: Theory and
  Practice}, pages 65--93.

\bibitem[F{\`e}ve and Guay, 2009]{feve2009response}
F{\`e}ve, P. and Guay, A. (2009).
\newblock The response of hours to a technology shock: A two-step structural
  var approach.
\newblock {\em Journal of Money, Credit and Banking}, 41(5):987--1013.

\bibitem[Feve and Guay, 2010]{feve2010identification}
Feve, P. and Guay, A. (2010).
\newblock Identification of technology shocks in structural vars.
\newblock {\em The Economic Journal}, 120(549):1284--1318.

\bibitem[Findley, 1986]{findley1986uniqueness}
Findley, D.~F. (1986).
\newblock The uniqueness of moving average representations with independent and
  identically distributed random variables for non-gaussian stationary time
  series.
\newblock {\em Biometrika}, 73(2):520--521.

\bibitem[Forni and Gambetti, 2014]{forni2014sufficient}
Forni, M. and Gambetti, L. (2014).
\newblock Sufficient information in structural vars.
\newblock {\em Journal of Monetary Economics}, 66:124--136.

\bibitem[Forni et~al., 2023]{forni2023macroeconomic}
Forni, M., Gambetti, L., and Sala, L. (2023).
\newblock Macroeconomic uncertainty and vector autoregressions.
\newblock {\em Econometrics and Statistics}.

\bibitem[Forni et~al., 2017]{forni2017dynamic}
Forni, M., Hallin, M., Lippi, M., and Zaffaroni, P. (2017).
\newblock Dynamic factor models with infinite-dimensional factor space:
  Asymptotic analysis.
\newblock {\em Journal of Econometrics}, 199(1):74--92.

\bibitem[Fragetta and Melina, 2011]{fragetta2011effects}
Fragetta, M. and Melina, G. (2011).
\newblock The effects of fiscal policy shocks in svar models: a graphical
  modelling approach.
\newblock {\em Scottish Journal of Political Economy}, 58(4):537--566.

\bibitem[Francis and Ramey, 2005]{francis2005technology}
Francis, N. and Ramey, V.~A. (2005).
\newblock Is the technology-driven real business cycle hypothesis dead? shocks
  and aggregate fluctuations revisited.
\newblock {\em Journal of Monetary Economics}, 52(8):1379--1399.

\bibitem[Freyberger et~al., 2020]{freyberger2020dissecting}
Freyberger, J., Neuhierl, A., and Weber, M. (2020).
\newblock Dissecting characteristics nonparametrically.
\newblock {\em The Review of Financial Studies}, 33(5):2326--2377.

\bibitem[Fries and Zakoian, 2019]{fries2019mixed}
Fries, S. and Zakoian, J.-M. (2019).
\newblock Mixed causal-noncausal ar processes and the modelling of explosive
  bubbles.
\newblock {\em Econometric Theory}, 35(6):1234--1270.

\bibitem[Fry and Pagan, 2011]{fry2011sign}
Fry, R. and Pagan, A. (2011).
\newblock Sign restrictions in structural vector autoregressions: A critical
  review.
\newblock {\em Journal of Economic Literature}, 49(4):938--960.

\bibitem[Funovits, 2020]{funovits2020identifiability}
Funovits, B. (2020).
\newblock Identifiability and estimation of possibly non-invertible svarma
  models: A new parametrisation.
\newblock {\em arXiv preprint arXiv:2002.04346}.

\bibitem[Gabaix, 2011]{gabaix2011granular}
Gabaix, X. (2011).
\newblock The granular origins of aggregate fluctuations.
\newblock {\em Econometrica}, 79(3):733--772.

\bibitem[Gabaix and Koijen, 2023]{gabaix2023granular}
Gabaix, X. and Koijen, R.~S. (2023).
\newblock Granular instrumental variables.
\newblock {\em Journal of Political Economy (forthcoming)}.

\bibitem[Gafarov et~al., 2018]{gafarov2018delta}
Gafarov, B., Meier, M., and Montiel~Olea, J.~L. (2018).
\newblock Delta-method inference for a class of set-identified svars.
\newblock {\em Journal of Econometrics}, 203(2):316--327.

\bibitem[Gali, 1999]{gali1999technology}
Gali, J. (1999).
\newblock Technology, employment, and the business cycle: do technology shocks
  explain aggregate fluctuations?
\newblock {\em American economic review}, 89(1):249--271.

\bibitem[Ganics et~al., 2021]{ganics2021confidence}
Ganics, G., Inoue, A., and Rossi, B. (2021).
\newblock Confidence intervals for bias and size distortion in iv and local
  projections-iv models.
\newblock {\em Journal of Business \& Economic Statistics}, 39(1):307--324.

\bibitem[Garratt et~al., 2006]{garratt2006permanent}
Garratt, A., Robertson, D., and Wright, S. (2006).
\newblock Permanent vs transitory components and economic fundamentals.
\newblock {\em Journal of Applied Econometrics}, 21(4):521--542.

\bibitem[Geraci and Gnabo, 2018]{geraci2018measuring}
Geraci, M.~V. and Gnabo, J.-Y. (2018).
\newblock Measuring interconnectedness between financial institutions with
  bayesian time-varying vector autoregressions.
\newblock {\em Journal of Financial and Quantitative Analysis},
  53(3):1371--1390.

\bibitem[Giacomini and Kitagawa, 2021]{giacomini2021robust}
Giacomini, R. and Kitagawa, T. (2021).
\newblock Robust bayesian inference for set-identified models.
\newblock {\em Econometrica}, 89(4):1519--1556.

\bibitem[Giacomini et~al., 2022]{giacomini2022robust}
Giacomini, R., Kitagawa, T., and Read, M. (2022).
\newblock Robust bayesian inference in proxy svars.
\newblock {\em Journal of Econometrics}, 228(1):107--126.

\bibitem[Giraitis et~al., 2021]{giraitis2021time}
Giraitis, L., Kapetanios, G., and Marcellino, M. (2021).
\newblock Time-varying instrumental variable estimation.
\newblock {\em Journal of Econometrics}, 224(2):394--415.

\bibitem[Goes, 2016]{goes2016institutions}
Goes, C. (2016).
\newblock Institutions and growth: A gmm/iv panel var approach.
\newblock {\em Economics letters}, 138:85--91.

\bibitem[Gomez~Cram and Olbert, 2023]{gomez2023measuring}
Gomez~Cram, R. and Olbert, M. (2023).
\newblock Measuring the expected effects of the global tax reform.
\newblock {\em Review of Financial Studies}.

\bibitem[Gonzalo and Granger, 1995]{gonzalo1995estimation}
Gonzalo, J. and Granger, C. (1995).
\newblock Estimation of common long-memory components in cointegrated systems.
\newblock {\em Journal of Business \& Economic Statistics}, 13(1):27--35.

\bibitem[Gonzalo and Ng, 2001]{gonzalo2001systematic}
Gonzalo, J. and Ng, S. (2001).
\newblock A systematic framework for analyzing the dynamic effects of permanent
  and transitory shocks.
\newblock {\em Journal of Economic dynamics and Control}, 25(10):1527--1546.

\bibitem[Gorodnichenko and Lee, 2020]{gorodnichenko2020forecast}
Gorodnichenko, Y. and Lee, B. (2020).
\newblock Forecast error variance decompositions with local projections.
\newblock {\em Journal of Business \& Economic Statistics}, 38(4):921--933.

\bibitem[Gospodinov, 2010]{gospodinov2010inference}
Gospodinov, N. (2010).
\newblock Inference in nearly nonstationary svar models with long-run
  identifying restrictions.
\newblock {\em Journal of Business \& Economic Statistics}, 28(1):1--12.

\bibitem[Gourieroux and Monfort, 1997]{gourieroux1997time}
Gourieroux, C. and Monfort, A. (1997).
\newblock {\em Time series and dynamic models}, volume~3.
\newblock Cambridge University Press.

\bibitem[Gourieroux et~al., 2017]{gourieroux2017statistical}
Gourieroux, C., Monfort, A., and Renne, J.-P. (2017).
\newblock Statistical inference for independent component analysis: Application
  to structural var models.
\newblock {\em Journal of Econometrics}, 196(1):111--126.

\bibitem[Gouri{\'e}roux et~al., 2020]{gourieroux2020identification}
Gouri{\'e}roux, C., Monfort, A., and Renne, J.-P. (2020).
\newblock Identification and estimation in non-fundamental structural varma
  models.
\newblock {\em The Review of Economic Studies}, 87(4):1915--1953.

\bibitem[Granziera et~al., 2018]{granziera2018inference}
Granziera, E., Moon, H.~R., and Schorfheide, F. (2018).
\newblock Inference for vars identified with sign restrictions.
\newblock {\em Quantitative Economics}, 9(3):1087--1121.

\bibitem[Greenaway-McGrevy, 2013]{greenaway2013multistep}
Greenaway-McGrevy, R. (2013).
\newblock Multistep prediction of panel vector autoregressive processes.
\newblock {\em Econometric Theory}, 29(4):699--734.

\bibitem[Greenaway-McGrevy, 2020]{greenaway2020multistep}
Greenaway-McGrevy, R. (2020).
\newblock Multistep forecast selection for panel data.
\newblock {\em Econometric Reviews}, 39(4):373--406.

\bibitem[Grigoriu, 1995]{grigoriu1995linear}
Grigoriu, M. (1995).
\newblock Linear and nonlinear systems with non-gaussian white noise input.
\newblock {\em Probabilistic engineering mechanics}, 10(3):171--179.

\bibitem[Guay, 2021]{guay2021identification}
Guay, A. (2021).
\newblock Identification of structural vector autoregressions through higher
  unconditional moments.
\newblock {\em Journal of Econometrics}, 225(1):27--46.

\bibitem[Guo and Ling, 2023]{guo2023inference}
Guo, F. and Ling, S. (2023).
\newblock Inference for the vec (1) model with a heavy-tailed linear process
  errors.
\newblock {\em Econometric Reviews}, 42(9-10):806--833.

\bibitem[Hallin and Mehta, 2015]{hallin2015r}
Hallin, M. and Mehta, C. (2015).
\newblock R-estimation for asymmetric independent component analysis.
\newblock {\em Journal of the American Statistical Association},
  110(509):218--232.

\bibitem[Hallin and Paindaveine, 2004]{hallin2004rank}
Hallin, M. and Paindaveine, D. (2004).
\newblock Rank-based optimal tests of the adequacy of an elliptic varma model.
\newblock {\em The Annals of Statistics}, 32(6):2642--2678.

\bibitem[Hamilton, 1994]{hamilton1994time}
Hamilton, J.~D. (1994).
\newblock {\em Time series analysis}.
\newblock Princeton university press.

\bibitem[Han, 2015]{han2015tests}
Han, X. (2015).
\newblock Tests for overidentifying restrictions in factor-augmented var
  models.
\newblock {\em Journal of Econometrics}, 184(2):394--419.

\bibitem[Han, 2018]{han2018estimation}
Han, X. (2018).
\newblock Estimation and inference of dynamic structural factor models with
  over-identifying restrictions.
\newblock {\em Journal of Econometrics}, 202(2):125--147.

\bibitem[Hannadige et~al., 2023]{hannadige2023forecasting}
Hannadige, S.~B., Gao, J., Silvapulle, M.~J., and Silvapulle, P. (2023).
\newblock Forecasting a nonstationary time series using a mixture of stationary
  and nonstationary factors as predictors.
\newblock {\em Journal of Business \& Economic Statistics}, pages 1--13.

\bibitem[Hansen and Seo, 2002]{hansen2002testing}
Hansen, B.~E. and Seo, B. (2002).
\newblock Testing for two-regime threshold cointegration in vector
  error-correction models.
\newblock {\em Journal of econometrics}, 110(2):293--318.

\bibitem[Hatchondo et~al., 2016]{hatchondo2016debt}
Hatchondo, J.~C., Martinez, L., and Sosa-Padilla, C. (2016).
\newblock Debt dilution and sovereign default risk.
\newblock {\em Journal of Political Economy}, 124(5):1383--1422.

\bibitem[Hausman and Taylor, 1983]{hausman1983identification}
Hausman, J.~A. and Taylor, W.~E. (1983).
\newblock Identification in linear simultaneous equations models with
  covariance restrictions: An instrumental variables interpretation.
\newblock {\em Econometrica: Journal of the Econometric Society}, pages
  1527--1549.

\bibitem[Hecq et~al., 2000]{hecq2000permanent}
Hecq, A., Palm, F.~C., and Urbain, J.-P. (2000).
\newblock Permanent-transitory decomposition in var models with cointegration
  and common cycles.
\newblock {\em Oxford Bulletin of Economics and Statistics}, 62(4):511--532.

\bibitem[Hecq and Voisin, 2021]{hecq2021forecasting}
Hecq, A. and Voisin, E. (2021).
\newblock Forecasting bubbles with mixed causal-noncausal autoregressive
  models.
\newblock {\em Econometrics and Statistics}, 20:29--45.

\bibitem[Hendry, 1994]{hendry1994interactions}
Hendry, D.~F. (1994).
\newblock On the interactions of unit roots and exogeneity.
\newblock {\em Econometric Reviews}, 14(4):383--419.

\bibitem[Hodoshima, 1988]{hodoshima1988estimation}
Hodoshima, J. (1988).
\newblock Estimation of a single structural equation with structural change.
\newblock {\em Econometric Theory}, 4(1):86--96.

\bibitem[Hoover, 2020]{hoover2020discovery}
Hoover, K.~D. (2020).
\newblock The discovery of long-run causal order: A preliminary investigation.
\newblock {\em Econometrics}, 8(3):31.

\bibitem[Hsiao, 1997]{hsiao1997statistical}
Hsiao, C. (1997).
\newblock Statistical properties of the two-stage least squares estimator under
  cointegration.
\newblock {\em The Review of Economic Studies}, 64(3):385--398.

\bibitem[Hsu et~al., 2022]{hsu2022counterfactual}
Hsu, Y.-C., Lai, T.-C., and Lieli, R.~P. (2022).
\newblock Counterfactual treatment effects: Estimation and inference.
\newblock {\em Journal of Business \& Economic Statistics}, 40(1):240--255.

\bibitem[Huber et~al., 2023]{huber2023bayesian}
Huber, F., Krisztin, T., and Pfarrhofer, M. (2023).
\newblock A bayesian panel vector autoregression to analyze the impact of
  climate shocks on high-income economies.
\newblock {\em The Annals of Applied Statistics}, 17(2):1543--1573.

\bibitem[Huber, 2018]{huber2018disentangling}
Huber, K. (2018).
\newblock Disentangling the effects of a banking crisis: Evidence from german
  firms and counties.
\newblock {\em American Economic Review}, 108(3):868--898.

\bibitem[Inoue and Kilian, 2013]{inoue2013inference}
Inoue, A. and Kilian, L. (2013).
\newblock Inference on impulse response functions in structural var models.
\newblock {\em Journal of Econometrics}, 177(1):1--13.

\bibitem[Inoue and Kilian, 2016]{inoue2016joint}
Inoue, A. and Kilian, L. (2016).
\newblock Joint confidence sets for structural impulse responses.
\newblock {\em Journal of Econometrics}, 192(2):421--432.

\bibitem[Inoue and Kilian, 2020]{inoue2020uniform}
Inoue, A. and Kilian, L. (2020).
\newblock The uniform validity of impulse response inference in
  autoregressions.
\newblock {\em Journal of Econometrics}, 215(2):450--472.

\bibitem[Ivanov and Kilian, 2005]{ivanov2005practitioner}
Ivanov, V. and Kilian, L. (2005).
\newblock A practitioner's guide to lag order selection for var impulse
  response analysis.
\newblock {\em Studies in Nonlinear Dynamics \& Econometrics}, 9(1).

\bibitem[Johansen, 1995]{johansen1995identifying}
Johansen, S. (1995).
\newblock Identifying restrictions of linear equations with applications to
  simultaneous equations and cointegration.
\newblock {\em Journal of econometrics}, 69(1):111--132.

\bibitem[Jorda, 2005]{jorda2005estimation}
Jorda, O. (2005).
\newblock Estimation and inference of impulse responses by local projections.
\newblock {\em American economic review}, 95(1):161--182.

\bibitem[Jurado et~al., 2015]{jurado2015measuring}
Jurado, K., Ludvigson, S.~C., and Ng, S. (2015).
\newblock Measuring uncertainty.
\newblock {\em American Economic Review}, 105(3):1177--1216.

\bibitem[Kagan et~al., 1973]{kagan1973characterization}
Kagan, A.~M., Linnik, I., and Rao, C.~R. (1973).
\newblock Characterization problems in mathematical statistics.
\newblock {\em (No Title)}.

\bibitem[Kalliovirta et~al., 2015]{kalliovirta2015gaussian}
Kalliovirta, L., Meitz, M., and Saikkonen, P. (2015).
\newblock A gaussian mixture autoregressive model for univariate time series.
\newblock {\em Journal of Time Series Analysis}, 36(2):247--266.

\bibitem[Kalliovirta et~al., 2016]{kalliovirta2016gaussian}
Kalliovirta, L., Meitz, M., and Saikkonen, P. (2016).
\newblock Gaussian mixture vector autoregression.
\newblock {\em Journal of econometrics}, 192(2):485--498.

\bibitem[Kang, 2003]{kang2003multi}
Kang, I.-B. (2003).
\newblock Multi-period forecasting using different models for different
  horizons: an application to us economic time series data.
\newblock {\em International Journal of Forecasting}, 19(3):387--400.

\bibitem[Kao et~al., 2018]{kao2018testing}
Kao, C., Trapani, L., and Urga, G. (2018).
\newblock Testing for instability in covariance structures.
\newblock {\em Bernoulli}, 24(1):740--771.

\bibitem[Kapetanios et~al., 2019]{kapetanios2019large}
Kapetanios, G., Marcellino, M., and Venditti, F. (2019).
\newblock Large time-varying parameter vars: A nonparametric approach.
\newblock {\em Journal of Applied Econometrics}, 34(7):1027--1049.

\bibitem[Katsouris, 2021]{katsouris2021forecast}
Katsouris, C. (2021).
\newblock Forecast evaluation in large cross-sections of realized volatility.
\newblock {\em arXiv preprint arXiv:2112.04887}.

\bibitem[Katsouris, 2023a]{katsouris2023estimating}
Katsouris, C. (2023a).
\newblock Estimating conditional value-at-risk with nonstationary quantile
  predictive regression models.
\newblock {\em arXiv preprint arXiv:2311.08218}.

\bibitem[Katsouris, 2023b]{katsouris2023limit}
Katsouris, C. (2023b).
\newblock Limit theory under network dependence and nonstationarity.
\newblock {\em arXiv preprint arXiv:2308.01418}.

\bibitem[Keweloh, 2021]{keweloh2021generalized}
Keweloh, S.~A. (2021).
\newblock A generalized method of moments estimator for structural vector
  autoregressions based on higher moments.
\newblock {\em Journal of Business \& Economic Statistics}, 39(3):772--782.

\bibitem[Keweloh et~al., 2023]{keweloh2023monetary}
Keweloh, S.~A., Hetzenecker, S., and Seepe, A. (2023).
\newblock Monetary policy and information shocks in a block-recursive svar.
\newblock {\em Journal of International Money and Finance}, page 102892.

\bibitem[Khalaf and Urga, 2014]{khalaf2014identification}
Khalaf, L. and Urga, G. (2014).
\newblock Identification robust inference in cointegrating regressions.
\newblock {\em Journal of Econometrics}, 182(2):385--396.

\bibitem[Kiiveri et~al., 1984]{kiiveri1984recursive}
Kiiveri, H., Speed, T.~P., and Carlin, J.~B. (1984).
\newblock Recursive causal models.
\newblock {\em Journal of the australian Mathematical Society}, 36(1):30--52.

\bibitem[Kilian, 2009]{kilian2009not}
Kilian, L. (2009).
\newblock Not all oil price shocks are alike: Disentangling demand and supply
  shocks in the crude oil market.
\newblock {\em American Economic Review}, 99(3):1053--1069.

\bibitem[Kilian and Kim, 2011]{kilian2011reliable}
Kilian, L. and Kim, Y.~J. (2011).
\newblock How reliable are local projection estimators of impulse responses?
\newblock {\em Review of Economics and Statistics}, 93(4):1460--1466.

\bibitem[Kilian and Lee, 2014]{kilian2014quantifying}
Kilian, L. and Lee, T.~K. (2014).
\newblock Quantifying the speculative component in the real price of oil: The
  role of global oil inventories.
\newblock {\em Journal of International Money and Finance}, 42:71--87.

\bibitem[Kilian and Lewis, 2011]{kilian2011does}
Kilian, L. and Lewis, L.~T. (2011).
\newblock Does the fed respond to oil price shocks?
\newblock {\em The Economic Journal}, 121(555):1047--1072.

\bibitem[Kilian and L{\"u}tkepohl, 2017]{kilian2017structural}
Kilian, L. and L{\"u}tkepohl, H. (2017).
\newblock {\em Structural vector autoregressive analysis}.
\newblock Cambridge University Press.

\bibitem[Kilian and Park, 2009]{kilian2009impact}
Kilian, L. and Park, C. (2009).
\newblock The impact of oil price shocks on the us stock market.
\newblock {\em International economic review}, 50(4):1267--1287.

\bibitem[Kim and Roubini, 2000]{kim2000exchange}
Kim, S. and Roubini, N. (2000).
\newblock Exchange rate anomalies in the industrial countries: A solution with
  a structural var approach.
\newblock {\em Journal of Monetary economics}, 45(3):561--586.

\bibitem[Kleibergen and Van~Dijk, 1994]{kleibergen1994shape}
Kleibergen, F. and Van~Dijk, H.~K. (1994).
\newblock On the shape of the likelihood/posterior in cointegration models.
\newblock {\em Econometric theory}, 10(3-4):514--551.

\bibitem[Kociecki, 2011]{kociecki2011algebraic}
Kociecki, A. (2011).
\newblock Algebraic theory of indentification in parametric models.
\newblock {\em National Bank of Poland Working Paper}, (88).

\bibitem[Kocikecki and Kolasa, 2018]{kocikecki2018global}
Kocikecki, A. and Kolasa, M. (2018).
\newblock Global identification of linearized dsge models.
\newblock {\em Quantitative Economics}, 9(3):1243--1263.

\bibitem[Kocikecki and Kolasa, 2023]{kocikecki2023solution}
Kocikecki, A. and Kolasa, M. (2023).
\newblock A solution to the global identification problem in dsge models.
\newblock {\em Journal of Econometrics}, 236(2):105477.

\bibitem[Koistinen and Funovits, 2022]{koistinen2022estimation}
Koistinen, J. and Funovits, B. (2022).
\newblock Estimation of impulse-response functions with dynamic factor models:
  A new parametrization.
\newblock {\em arXiv preprint arXiv:2202.00310}.

\bibitem[Komunjer and Ng, 2011]{komunjer2011dynamic}
Komunjer, I. and Ng, S. (2011).
\newblock Dynamic identification of dynamic stochastic general equilibrium
  models.
\newblock {\em Econometrica}, 79(6):1995--2032.

\bibitem[Kong et~al., 2019]{kong2019weak}
Kong, J., Phillips, P.~C., and Sul, D. (2019).
\newblock Weak $\sigma$-convergence: Theory and applications.
\newblock {\em Journal of Econometrics}, 209(2):185--207.

\bibitem[Koo et~al., 2022]{koo2022impulse}
Koo, B., Lee, S., and Seo, M.~H. (2022).
\newblock What impulse response do instrumental variables identify?
\newblock {\em arXiv preprint arXiv:2208.11828}.

\bibitem[Koop and Korobilis, 2013]{koop2013large}
Koop, G. and Korobilis, D. (2013).
\newblock Large time-varying parameter vars.
\newblock {\em Journal of Econometrics}, 177(2):185--198.

\bibitem[Korobilis and Yilmaz, 2018]{korobilis2018measuring}
Korobilis, D. and Yilmaz, K. (2018).
\newblock Measuring dynamic connectedness with large bayesian var models.
\newblock {\em Available at SSRN 3099725}.

\bibitem[Krampe et~al., 2023]{krampe2023structural}
Krampe, J., Paparoditis, E., and Trenkler, C. (2023).
\newblock Structural inference in sparse high-dimensional vector
  autoregressions.
\newblock {\em Journal of Econometrics}, 234(1):276--300.

\bibitem[Krolzig, 2003]{krolzig2003general}
Krolzig, H.-M. (2003).
\newblock General-to-specific model selection procedures for structural vector
  autoregressions.
\newblock {\em Oxford Bulletin of Economics and Statistics}, 65:769--801.

\bibitem[Kuersteiner, 2005]{kuersteiner2005automatic}
Kuersteiner, G.~M. (2005).
\newblock Automatic inference for infinite order vector autoregressions.
\newblock {\em Econometric Theory}, 21(1):85--115.

\bibitem[Kyle et~al., 2023]{kyle2023beliefs}
Kyle, A.~S., Obizhaeva, A.~A., and Wang, Y. (2023).
\newblock Beliefs aggregation and return predictability.
\newblock {\em The Journal of Finance}, 78(1):427--486.

\bibitem[Laborda and Olmo, 2021]{laborda2021volatility}
Laborda, R. and Olmo, J. (2021).
\newblock Volatility spillover between economic sectors in financial crisis
  prediction: Evidence spanning the great financial crisis and covid-19
  pandemic.
\newblock {\em Research in International Business and Finance}, 57:101402.

\bibitem[Lai and Wei, 1982]{lai1982least}
Lai, T.~L. and Wei, C.~Z. (1982).
\newblock Least squares estimates in stochastic regression models with
  applications to identification and control of dynamic systems.
\newblock {\em The Annals of Statistics}, 10(1):154--166.

\bibitem[Lanne, 2000]{lanne2000near}
Lanne, M. (2000).
\newblock Near unit roots, cointegration, and the term structure of interest
  rates.
\newblock {\em Journal of Applied Econometrics}, 15(5):513--529.

\bibitem[Lanne et~al., 2023]{lanne2023identifying}
Lanne, M., Liu, K., and Luoto, J. (2023).
\newblock Identifying structural vector autoregression via leptokurtic economic
  shocks.
\newblock {\em Journal of Business \& Economic Statistics}, 41(4):1341--1351.

\bibitem[Lanne and Luoto, 2021]{lanne2021gmm}
Lanne, M. and Luoto, J. (2021).
\newblock Gmm estimation of non-gaussian structural vector autoregression.
\newblock {\em Journal of Business \& Economic Statistics}, 39(1):69--81.

\bibitem[Lanne et~al., 2012]{lanne2012optimal}
Lanne, M., Luoto, J., and Saikkonen, P. (2012).
\newblock Optimal forecasting of noncausal autoregressive time series.
\newblock {\em International Journal of Forecasting}, 28(3):623--631.

\bibitem[Lanne and L{\"u}tkepohl, 2008a]{lanne2008identifying}
Lanne, M. and L{\"u}tkepohl, H. (2008a).
\newblock Identifying monetary policy shocks via changes in volatility.
\newblock {\em Journal of Money, Credit and Banking}, 40(6):1131--1149.

\bibitem[Lanne and L{\"u}tkepohl, 2008b]{lanne2008statistical}
Lanne, M. and L{\"u}tkepohl, H. (2008b).
\newblock A statistical comparison of alternative identification schemes for
  monetary policy shocks.

\bibitem[Lanne et~al., 2010]{lanne2010structural}
Lanne, M., L{\"u}tkepohl, H., and Maciejowska, K. (2010).
\newblock Structural vector autoregressions with markov switching.
\newblock {\em Journal of Economic Dynamics and Control}, 34(2):121--131.

\bibitem[Lanne et~al., 2017]{lanne2017identification}
Lanne, M., Meitz, M., and Saikkonen, P. (2017).
\newblock Identification and estimation of non-gaussian structural vector
  autoregressions.
\newblock {\em Journal of Econometrics}, 196(2):288--304.

\bibitem[Lanne and Nyberg, 2016]{lanne2016generalized}
Lanne, M. and Nyberg, H. (2016).
\newblock Generalized forecast error variance decomposition for linear and
  nonlinear multivariate models.
\newblock {\em Oxford Bulletin of Economics and Statistics}, 78(4):595--603.

\bibitem[Lanne and Saikkonen, 2002]{lanne2002threshold}
Lanne, M. and Saikkonen, P. (2002).
\newblock Threshold autoregressions for strongly autocorrelated time series.
\newblock {\em Journal of Business \& Economic Statistics}, 20(2):282--289.

\bibitem[Lanne and Saikkonen, 2013]{lanne2013noncausal}
Lanne, M. and Saikkonen, P. (2013).
\newblock Noncausal vector autoregression.
\newblock {\em Econometric Theory}, 29(3):447--481.

\bibitem[Lastrapes, 2005]{lastrapes2005estimating}
Lastrapes, W.~D. (2005).
\newblock Estimating and identifying vector autoregressions under diagonality
  and block exogeneity restrictions.
\newblock {\em Economics letters}, 87(1):75--81.

\bibitem[Lawrance, 1991]{lawrance1991directionality}
Lawrance, A. (1991).
\newblock Directionality and reversibility in time series.
\newblock {\em International statistical review/revue internationale de
  statistique}, pages 67--79.

\bibitem[Lee et~al., 2023]{lee2023change}
Lee, S., Lee, S., and Maekawa, K. (2023).
\newblock Change point test for structural vector autoregressive model via
  independent component analysis.
\newblock {\em Journal of Statistical Computation and Simulation},
  93(5):687--707.

\bibitem[Leeper et~al., 2013]{leeper2013fiscal}
Leeper, E.~M., Walker, T.~B., and Yang, S.-C.~S. (2013).
\newblock Fiscal foresight and information flows.
\newblock {\em Econometrica}, 81(3):1115--1145.

\bibitem[Lewis and Syrgkanis, 2020]{lewis2020double}
Lewis, G. and Syrgkanis, V. (2020).
\newblock Double/debiased machine learning for dynamic treatment effects via
  g-estimation.
\newblock {\em arXiv preprint arXiv:2002.07285}.

\bibitem[LI and Hui, 1989]{li1989robust}
LI, W.~K. and Hui, Y. (1989).
\newblock Robust multiple time series modelling.
\newblock {\em Biometrika}, 76(2):309--315.

\bibitem[Liao and Phillips, 2015]{liao2015automated}
Liao, Z. and Phillips, P.~C. (2015).
\newblock Automated estimation of vector error correction models.
\newblock {\em Econometric Theory}, 31(3):581--646.

\bibitem[Lovcha and Perez-Laborda, 2021]{lovcha2021identifying}
Lovcha, Y. and Perez-Laborda, A. (2021).
\newblock Identifying technology shocks at the business cycle via spectral
  variance decompositions.
\newblock {\em Macroeconomic Dynamics}, 25(8):1966--1992.

\bibitem[Lusompa, 2023]{lusompa2023local}
Lusompa, A. (2023).
\newblock Local projections, autocorrelation, and efficiency.
\newblock {\em Quantitative Economics}, 14(4):1199--1220.

\bibitem[L{\"u}tkepohl, 1982]{lutkepohl1982non}
L{\"u}tkepohl, H. (1982).
\newblock Non-causality due to omitted variables.
\newblock {\em Journal of Econometrics}, 19(2-3):367--378.

\bibitem[L{\"u}tkepohl, 1990]{lutkepohl1990asymptotic}
L{\"u}tkepohl, H. (1990).
\newblock Asymptotic distributions of impulse response functions and forecast
  error variance decompositions of vector autoregressive models.
\newblock {\em The review of economics and statistics}, pages 116--125.

\bibitem[L{\"u}tkepohl, 2002]{lutkepohl2002forecasting}
L{\"u}tkepohl, H. (2002).
\newblock Forecasting cointegrated varma processes.
\newblock {\em A Companion to Economic Forecasting, Blackwell, Oxford}, pages
  179--205.

\bibitem[L{\"u}tkepohl, 2005]{lutkepohl2005new}
L{\"u}tkepohl, H. (2005).
\newblock {\em New introduction to multiple time series analysis}.
\newblock Springer Science \& Business Media.

\bibitem[L{\"u}tkepohl et~al., 2021]{lutkepohl2021testing}
L{\"u}tkepohl, H., Meitz, M., Net{\v{s}}unajev, A., and Saikkonen, P. (2021).
\newblock Testing identification via heteroskedasticity in structural vector
  autoregressive models.
\newblock {\em The Econometrics Journal}, 24(1):1--22.

\bibitem[L{\"u}tkepohl and Saikkonen, 1997]{lutkepohl1997impulse}
L{\"u}tkepohl, H. and Saikkonen, P. (1997).
\newblock Impulse response analysis in infinite order cointegrated vector
  autoregressive processes.
\newblock {\em Journal of Econometrics}, 81(1):127--157.

\bibitem[Magnusson and Mavroeidis, 2014]{magnusson2014identification}
Magnusson, L.~M. and Mavroeidis, S. (2014).
\newblock Identification using stability restrictions.
\newblock {\em Econometrica}, 82(5):1799--1851.

\bibitem[Marcellino et~al., 2006]{marcellino2006comparison}
Marcellino, M., Stock, J.~H., and Watson, M.~W. (2006).
\newblock A comparison of direct and iterated multistep ar methods for
  forecasting macroeconomic time series.
\newblock {\em Journal of econometrics}, 135(1-2):499--526.

\bibitem[Mariano, 2002]{mariano2002testing}
Mariano, R.~S. (2002).
\newblock Testing forecast accuracy.
\newblock {\em A companion to economic forecasting}, 2:284--298.

\bibitem[Mavroeidis, 2021]{mavroeidis2021identification}
Mavroeidis, S. (2021).
\newblock Identification at the zero lower bound.
\newblock {\em Econometrica}, 89(6):2855--2885.

\bibitem[Maxand, 2020]{maxand2020identification}
Maxand, S. (2020).
\newblock Identification of independent structural shocks in the presence of
  multiple gaussian components.
\newblock {\em Econometrics and Statistics}, 16:55--68.

\bibitem[Mei et~al., 2023]{mei2023implicit}
Mei, Z., Sheng, L., and Shi, Z. (2023).
\newblock Implicit nickell bias in panel local projection.
\newblock {\em arXiv preprint arXiv:2302.13455}.

\bibitem[Meitz et~al., 2023]{meitz2023mixture}
Meitz, M., Preve, D., and Saikkonen, P. (2023).
\newblock A mixture autoregressive model based on student’st--distribution.
\newblock {\em Communications in Statistics-Theory and Methods},
  52(2):499--515.

\bibitem[Meitz and Saikkonen, 2013]{meitz2013maximum}
Meitz, M. and Saikkonen, P. (2013).
\newblock Maximum likelihood estimation of a noninvertible arma model with
  autoregressive conditional heteroskedasticity.
\newblock {\em Journal of Multivariate Analysis}, 114:227--255.

\bibitem[Meitz and Saikkonen, 2021]{meitz2021testing}
Meitz, M. and Saikkonen, P. (2021).
\newblock Testing for observation-dependent regime switching in mixture
  autoregressive models.
\newblock {\em Journal of Econometrics}, 222(1):601--624.

\bibitem[M{\'e}lard et~al., 2006]{melard2006exact}
M{\'e}lard, G., Roy, R., and Saidi, A. (2006).
\newblock Exact maximum likelihood estimation of structured or unit root
  multivariate time series models.
\newblock {\em Computational statistics \& data analysis}, 50(11):2958--2986.

\bibitem[Menchetti et~al., 2023]{menchetti2023combining}
Menchetti, F., Cipollini, F., and Mealli, F. (2023).
\newblock Combining counterfactual outcomes and arima models for policy
  evaluation.
\newblock {\em The Econometrics Journal}, 26(1):1--24.

\bibitem[Mensi et~al., 2023]{mensi2023time}
Mensi, W., Aslan, A., Vo, X.~V., and Kang, S.~H. (2023).
\newblock Time-frequency spillovers and connectedness between precious metals,
  oil futures and financial markets: Hedge and safe haven implications.
\newblock {\em International Review of Economics \& Finance}, 83:219--232.

\bibitem[Mertens and Ravn, 2010]{mertens2010measuring}
Mertens, K. and Ravn, M.~O. (2010).
\newblock Measuring the impact of fiscal policy in the face of anticipation: a
  structural var approach.
\newblock {\em The Economic Journal}, 120(544):393--413.

\bibitem[Mertens and Ravn, 2013]{mertens2013dynamic}
Mertens, K. and Ravn, M.~O. (2013).
\newblock The dynamic effects of personal and corporate income tax changes in
  the united states.
\newblock {\em American economic review}, 103(4):1212--1247.

\bibitem[Miao et~al., 2023]{miao2023high}
Miao, K., Phillips, P.~C., and Su, L. (2023).
\newblock High-dimensional vars with common factors.
\newblock {\em Journal of Econometrics}, 233(1):155--183.

\bibitem[Mikusheva, 2007]{mikusheva2007uniform}
Mikusheva, A. (2007).
\newblock Uniform inference in autoregressive models.
\newblock {\em Econometrica}, 75(5):1411--1452.

\bibitem[Mittnik and Zadrozny, 1993]{mittnik1993asymptotic}
Mittnik, S. and Zadrozny, P.~A. (1993).
\newblock Asymptotic distributions of impulse responses, step responses, and
  variance decompositions of estimated linear dynamic models.
\newblock {\em Econometrica: Journal of the Econometric Society}, pages
  857--870.

\bibitem[Moneta et~al., 2011]{moneta2011causal}
Moneta, A., Chla{\ss}, N., Entner, D., and Hoyer, P. (2011).
\newblock Causal search in structural vector autoregressive models.
\newblock In {\em NIPS Mini-Symposium on Causality in Time Series}, pages
  95--114.

\bibitem[Moneta et~al., 2013]{moneta2013causal}
Moneta, A., Entner, D., Hoyer, P.~O., and Coad, A. (2013).
\newblock Causal inference by independent component analysis: Theory and
  applications.
\newblock {\em Oxford Bulletin of Economics and Statistics}, 75(5):705--730.

\bibitem[Montiel~Olea and Plagborg-Moller, 2021]{montiel2021local}
Montiel~Olea, J.~L. and Plagborg-Moller, M. (2021).
\newblock Local projection inference is simpler and more robust than you think.
\newblock {\em Econometrica}, 89(4):1789--1823.

\bibitem[Montiel~Olea et~al., 2021]{olea2021inference}
Montiel~Olea, J.~L., Stock, J.~H., and Watson, M.~W. (2021).
\newblock Inference in structural vector autoregressions identified with an
  external instrument.
\newblock {\em Journal of Econometrics}, 225(1):74--87.

\bibitem[Moon and Schorfheide, 2002]{moon2002minimum}
Moon, H.~R. and Schorfheide, F. (2002).
\newblock Minimum distance estimation of nonstationary time series models.
\newblock {\em Econometric Theory}, 18(6):1385--1407.

\bibitem[Moon and Schorfheide, 2012]{moon2012bayesian}
Moon, H.~R. and Schorfheide, F. (2012).
\newblock Bayesian and frequentist inference in partially identified models.
\newblock {\em Econometrica}, 80(2):755--782.

\bibitem[Morley and Piger, 2012]{morley2012asymmetric}
Morley, J. and Piger, J. (2012).
\newblock The asymmetric business cycle.
\newblock {\em Review of Economics and Statistics}, 94(1):208--221.

\bibitem[Morris, 1979]{morris1979duality}
Morris, S.~A. (1979).
\newblock Duality and structure of locally compact abelian groups.
\newblock {\em Math Chronicle}, 8.

\bibitem[M{\"u}ller and Petalas, 2010]{muller2010efficient}
M{\"u}ller, U.~K. and Petalas, P.-E. (2010).
\newblock Efficient estimation of the parameter path in unstable time series
  models.
\newblock {\em The Review of Economic Studies}, 77(4):1508--1539.

\bibitem[Mumtaz et~al., 2018]{mumtaz2018vars}
Mumtaz, H., Pinter, G., and Theodoridis, K. (2018).
\newblock What do vars tell us about the impact of a credit supply shock?
\newblock {\em International Economic Review}, 59(2):625--646.

\bibitem[Mumtaz and Surico, 2018]{mumtaz2018policy}
Mumtaz, H. and Surico, P. (2018).
\newblock Policy uncertainty and aggregate fluctuations.
\newblock {\em Journal of Applied Econometrics}, 33(3):319--331.

\bibitem[Murasawa, 2015]{murasawa2015multivariate}
Murasawa, Y. (2015).
\newblock The multivariate beveridge--nelson decomposition with i (1) and i (2)
  series.
\newblock {\em Economics Letters}, 137:157--162.

\bibitem[Myers et~al., 2018]{myers2018long}
Myers, R.~J., Johnson, S.~R., Helmar, M., and Baumes, H. (2018).
\newblock Long-run and short-run relationships between oil prices, producer
  prices, and consumer prices: What can we learn from a permanent-transitory
  decomposition?
\newblock {\em The Quarterly Review of Economics and Finance}, 67:175--190.

\bibitem[Ng and Young, 1990]{ng1990recursive}
Ng, C.~N. and Young, P.~C. (1990).
\newblock Recursive estimation and forecasting of non-stationary time series.
\newblock {\em Journal of Forecasting}, 9(2):173--204.

\bibitem[Ng and Perron, 2001]{ng2001lag}
Ng, S. and Perron, P. (2001).
\newblock Lag length selection and the construction of unit root tests with
  good size and power.
\newblock {\em Econometrica}, 69(6):1519--1554.

\bibitem[Noh, 2017]{noh2017impulse}
Noh, E. (2017).
\newblock Impulse-response analysis with proxy variables.
\newblock {\em Available at SSRN 3070401}.

\bibitem[Nyberg, 2018]{nyberg2018forecasting}
Nyberg, H. (2018).
\newblock Forecasting us interest rates and business cycle with a nonlinear
  regime switching var model.
\newblock {\em Journal of Forecasting}, 37(1):1--15.

\bibitem[Nyberg and Rauhala, 2022]{nyberg2022structural}
Nyberg, H. and Rauhala, S. (2022).
\newblock A structural vector autoregression containing positive-valued
  components.
\newblock {\em Available at SSRN}.

\bibitem[Nyberg and Savva, 2023]{nyberg2023risk}
Nyberg, H. and Savva, C.~S. (2023).
\newblock Risk-return trade-off in international stock returns: Skewness and
  business cycles.
\newblock {\em Econometrics and Statistics}.

\bibitem[Pagan and Pesaran, 2008]{pagan2008econometric}
Pagan, A.~R. and Pesaran, M.~H. (2008).
\newblock Econometric analysis of structural systems with permanent and
  transitory shocks.
\newblock {\em Journal of Economic Dynamics and control}, 32(10):3376--3395.

\bibitem[Paparoditis, 2018]{paparoditis2018sieve}
Paparoditis, E. (2018).
\newblock Sieve bootstrap for functional time series.
\newblock {\em The annals of Statistics}, 46(6B):3510--3538.

\bibitem[Paruolo, 1997]{paruolo1997asymptotic}
Paruolo, P. (1997).
\newblock Asymptotic inference on the moving average impact matrix in
  cointegrated 1 (1) var systems.
\newblock {\em Econometric Theory}, 13(1):79--118.

\bibitem[Paruolo, 2000]{paruolo2000asymptotic}
Paruolo, P. (2000).
\newblock Asymptotic efficiency of the two stage estimator in i (2) systems.
\newblock {\em Econometric Theory}, 16(4):524--550.

\bibitem[Paruolo and Rahbek, 1999]{paruolo1999weak}
Paruolo, P. and Rahbek, A. (1999).
\newblock Weak exogeneity in i (2) var systems.
\newblock {\em Journal of Econometrics}, 93(2):281--308.

\bibitem[Paul, 2020]{paul2020time}
Paul, P. (2020).
\newblock The time-varying effect of monetary policy on asset prices.
\newblock {\em Review of Economics and Statistics}, 102(4):690--704.

\bibitem[Pesaran et~al., 2000]{pesaran2000structural}
Pesaran, M.~H., Shin, Y., and Smith, R.~J. (2000).
\newblock Structural analysis of vector error correction models with exogenous
  i (1) variables.
\newblock {\em Journal of econometrics}, 97(2):293--343.

\bibitem[Petrova, 2019]{petrova2019quasi}
Petrova, K. (2019).
\newblock A quasi-bayesian local likelihood approach to time varying parameter
  var models.
\newblock {\em Journal of Econometrics}, 212(1):286--306.

\bibitem[Petrova, 2022]{petrova2022asymptotically}
Petrova, K. (2022).
\newblock Asymptotically valid bayesian inference in the presence of
  distributional misspecification in var models.
\newblock {\em Journal of Econometrics}, 230(1):154--182.

\bibitem[Petursson and Slok, 2001]{petursson2001wage}
Petursson, T.~G. and Slok, T. (2001).
\newblock Wage formation and employment in a cointegrated var model.
\newblock {\em The Econometrics Journal}, 4(2):191--209.

\bibitem[Phillips, 1987]{phillips1987towards}
Phillips, P. C.~B. (1987).
\newblock Towards a unified asymptotic theory for autoregression.
\newblock {\em Biometrika}, 74(3):535--547.

\bibitem[Phillips, 1989]{phillips1989partially}
Phillips, P. C.~B. (1989).
\newblock Partially identified econometric models.
\newblock {\em Econometric Theory}, 5(2):181--240.

\bibitem[Phillips, 1991]{phillips1991optimal}
Phillips, P. C.~B. (1991).
\newblock Optimal inference in cointegrated systems.
\newblock {\em Econometrica: Journal of the Econometric Society}, pages
  283--306.

\bibitem[Phillips, 1998]{phillips1998impulse}
Phillips, P. C.~B. (1998).
\newblock Impulse response and forecast error variance asymptotics in
  nonstationary vars.
\newblock {\em Journal of econometrics}, 83(1-2):21--56.

\bibitem[Phillips and Hansen, 1990]{phillips1990statistical}
Phillips, P. C.~B. and Hansen, B.~E. (1990).
\newblock Statistical inference in instrumental variables regression with i (1)
  processes.
\newblock {\em The review of economic studies}, 57(1):99--125.

\bibitem[Phillips et~al., 2020]{phillips2020econometric}
Phillips, P. C.~B., Leirvik, T., and Storelvmo, T. (2020).
\newblock Econometric estimates of earth’s transient climate sensitivity.
\newblock {\em Journal of Econometrics}, 214(1):6--32.

\bibitem[Plagborg-M{\o}ller and Wolf, 2021]{plagborg2021local}
Plagborg-M{\o}ller, M. and Wolf, C.~K. (2021).
\newblock Local projections and vars estimate the same impulse responses.
\newblock {\em Econometrica}, 89(2):955--980.

\bibitem[Plagborg-M{\o}ller and Wolf, 2022]{plagborg2022instrumental}
Plagborg-M{\o}ller, M. and Wolf, C.~K. (2022).
\newblock Instrumental variable identification of dynamic variance
  decompositions.
\newblock {\em Journal of Political Economy}, 130(8):2164--2202.

\bibitem[Poskitt, 2006]{poskitt2006identification}
Poskitt, D.~S. (2006).
\newblock On the identification and estimation of nonstationary and
  cointegrated armax systems.
\newblock {\em Econometric Theory}, 22(6):1138--1175.

\bibitem[Pretis, 2020]{pretis2020econometric}
Pretis, F. (2020).
\newblock Econometric modelling of climate systems: The equivalence of energy
  balance models and cointegrated vector autoregressions.
\newblock {\em Journal of Econometrics}, 214(1):256--273.

\bibitem[Pretis, 2021]{pretis2021exogeneity}
Pretis, F. (2021).
\newblock Exogeneity in climate econometrics.
\newblock {\em Energy Economics}, 96:105122.

\bibitem[Primiceri, 2005]{primiceri2005time}
Primiceri, G.~E. (2005).
\newblock Time varying structural vector autoregressions and monetary policy.
\newblock {\em The Review of Economic Studies}, 72(3):821--852.

\bibitem[Qian, 2023]{qian2023heterogeneity}
Qian, E. (2023).
\newblock Heterogeneity-robust granular instruments.
\newblock {\em arXiv preprint arXiv:2304.01273}.

\bibitem[Quah, 1992]{quah1992relative}
Quah, D. (1992).
\newblock The relative importance of permanent and transitory components:
  identification and some theoretical bounds.
\newblock {\em Econometrica: Journal of the Econometric Society}, pages
  107--118.

\bibitem[Ramey, 2011]{ramey2011identifying}
Ramey, V.~A. (2011).
\newblock Identifying government spending shocks: It's all in the timing.
\newblock {\em The Quarterly Journal of Economics}, 126(1):1--50.

\bibitem[Rehman et~al., 2018]{rehman2018precious}
Rehman, M.~U., Shahzad, S. J.~H., Uddin, G.~S., and Hedstr{\"o}m, A. (2018).
\newblock Precious metal returns and oil shocks: A time varying connectedness
  approach.
\newblock {\em Resources Policy}, 58:77--89.

\bibitem[Romer and Romer, 2017]{romer2017new}
Romer, C.~D. and Romer, D.~H. (2017).
\newblock New evidence on the aftermath of financial crises in advanced
  countries.
\newblock {\em American Economic Review}, 107(10):3072--3118.

\bibitem[Rousseeuw, 1984]{rousseeuw1984least}
Rousseeuw, P.~J. (1984).
\newblock Least median of squares regression.
\newblock {\em Journal of the American statistical association},
  79(388):871--880.

\bibitem[Rubio-Ramirez et~al., 2010]{rubio2010structural}
Rubio-Ramirez, J.~F., Waggoner, D.~F., and Zha, T. (2010).
\newblock Structural vector autoregressions: Theory of identification and
  algorithms for inference.
\newblock {\em The Review of Economic Studies}, 77(2):665--696.

\bibitem[Rubio-Ramirez et~al., 2005]{rubio2005markov}
Rubio-Ramirez, J.~F., Waggoner, D.~F., and Zha, T.~A. (2005).
\newblock Markov-switching structural vector autoregressions: theory and
  application.

\bibitem[Rygh~Swensen, 2022]{rygh2022causal}
Rygh~Swensen, A. (2022).
\newblock On causal and non-causal cointegrated vector autoregressive time
  series.
\newblock {\em Journal of Time Series Analysis}, 43(2):178--196.

\bibitem[Safikhani et~al., 2022]{safikhani2022fast}
Safikhani, A., Bai, Y., and Michailidis, G. (2022).
\newblock Fast and scalable algorithm for detection of structural breaks in big
  var models.
\newblock {\em Journal of Computational and Graphical Statistics},
  31(1):176--189.

\bibitem[Saikkonen, 1995]{saikkonen1995problems}
Saikkonen, P. (1995).
\newblock Problems with the asymptotic theory of maximum likelihood estimation
  in integrated and cointegrated systems.
\newblock {\em Econometric Theory}, 11(5):888--911.

\bibitem[Samuelson, 1941]{samuelson1941conditions}
Samuelson, P.~A. (1941).
\newblock Conditions that the roots of a polynomial be less than unity in
  absolute value.
\newblock {\em The Annals of Mathematical Statistics}, 12(3):360--364.

\bibitem[Sargan, 1983]{sargan1983identification}
Sargan, J.~D. (1983).
\newblock Identification and lack of identification.
\newblock {\em Econometrica: Journal of the Econometric Society}, pages
  1605--1633.

\bibitem[Sargent, 1976]{sargent1976observational}
Sargent, T.~J. (1976).
\newblock The observational equivalence of natural and unnatural rate theories
  of macroeconomics.
\newblock {\em Journal of Political Economy}, 84(3):631--640.

\bibitem[Schlaak et~al., 2023]{schlaak2023monetary}
Schlaak, T., Rieth, M., and Podstawski, M. (2023).
\newblock Monetary policy, external instruments, and heteroskedasticity.
\newblock {\em Quantitative Economics}, 14(1):161--200.

\bibitem[Sengupta and Kay, 1989]{sengupta1989efficient}
Sengupta, D. and Kay, S. (1989).
\newblock Efficient estimation of parameters for non-gaussian autoregressive
  processes.
\newblock {\em IEEE Transactions on Acoustics, Speech, and Signal Processing},
  37(6):785--794.

\bibitem[She et~al., 2022]{she2022whittle}
She, R., Mi, Z., and Ling, S. (2022).
\newblock Whittle parameter estimation for vector arma models with heavy-tailed
  noises.
\newblock {\em Journal of Statistical Planning and Inference}, 219:216--230.

\bibitem[Silva and Shimizu, 2017]{silva2017learning}
Silva, R. and Shimizu, S. (2017).
\newblock Learning instrumental variables with structural and non-gaussianity
  assumptions.
\newblock {\em Journal of Machine Learning Research}, 18(120):1--49.

\bibitem[Sims, 2012]{sims2012news}
Sims, E.~R. (2012).
\newblock News, non-invertibility, and structural vars.
\newblock In {\em DSGE Models in Macroeconomics: Estimation, Evaluation, and
  New Developments}, pages 81--135. Emerald Group Publishing Limited.

\bibitem[Soderstrom et~al., 1978]{soderstrom1978theoretical}
Soderstrom, T., Ljung, L., and Gustavsson, I. (1978).
\newblock A theoretical analysis of recursive identification methods.
\newblock {\em Automatica}, 14(3):231--244.

\bibitem[Soenen and Vander~Vennet, 2022]{soenen2022ecb}
Soenen, N. and Vander~Vennet, R. (2022).
\newblock Ecb monetary policy and bank default risk.
\newblock {\em Journal of International Money and Finance}, 122:102571.

\bibitem[Stock and Watson, 2021]{stock2021inference}
Stock, J. and Watson, M. (2021).
\newblock Inference in svars identified with external instruments.
\newblock {\em Journal of Econometrics}, 225:74--87.

\bibitem[Stock and Watson, 1988]{stock1988testing}
Stock, J.~H. and Watson, M.~W. (1988).
\newblock Testing for common trends.
\newblock {\em Journal of the American statistical Association},
  83(404):1097--1107.

\bibitem[Stock and Watson, 2001]{stock2001vector}
Stock, J.~H. and Watson, M.~W. (2001).
\newblock Vector autoregressions.
\newblock {\em Journal of Economic perspectives}, 15(4):101--115.

\bibitem[Stock and Watson, 2018]{stock2018identification}
Stock, J.~H. and Watson, M.~W. (2018).
\newblock Identification and estimation of dynamic causal effects in
  macroeconomics using external instruments.
\newblock {\em The Economic Journal}, 128(610):917--948.

\bibitem[Swanson and Granger, 1997]{swanson1997impulse}
Swanson, N.~R. and Granger, C.~W. (1997).
\newblock Impulse response functions based on a causal approach to residual
  orthogonalization in vector autoregressions.
\newblock {\em Journal of the American Statistical Association},
  92(437):357--367.

\bibitem[Tao and Yu, 2020]{tao2020model}
Tao, Y. and Yu, J. (2020).
\newblock Model selection for explosive models.
\newblock In {\em Essays in honor of Cheng Hsiao}, volume~41, pages 73--103.
  Emerald Publishing Limited.

\bibitem[Tchatoka and Dufour, 2013]{tchatoka2013finite}
Tchatoka, F.~D. and Dufour, J.-M. (2013).
\newblock On the finite-sample theory of exogeneity tests with possibly
  non-gaussian errors and weak identification.

\bibitem[Tong and Zhang, 2005]{tong2005time}
Tong, H. and Zhang, Z. (2005).
\newblock On time-reversibility of multivariate linear processes.
\newblock {\em Statistica Sinica}, pages 495--504.

\bibitem[Trung, 2019]{trung2019spillover}
Trung, N.~B. (2019).
\newblock The spillover effects of us economic policy uncertainty on the global
  economy: A global var approach.
\newblock {\em The North American Journal of Economics and Finance},
  48:90--110.

\bibitem[Vegh and Vuletin, 2015]{vegh2015tax}
Vegh, C.~A. and Vuletin, G. (2015).
\newblock How is tax policy conducted over the business cycle?
\newblock {\em American Economic Journal: Economic Policy}, 7(3):327--370.

\bibitem[Velasco and Lobato, 2018]{velasco2018frequency}
Velasco, C. and Lobato, I.~N. (2018).
\newblock Frequency domain minimum distance inference for possibly
  noninvertible and noncausal arma models.
\newblock {\em The Annals of Statistics}, 46(2):555--579.

\bibitem[Veldkamp, 2005]{veldkamp2005slow}
Veldkamp, L.~L. (2005).
\newblock Slow boom, sudden crash.
\newblock {\em Journal of Economic theory}, 124(2):230--257.

\bibitem[Vermeulen and Vansteelandt, 2015]{vermeulen2015bias}
Vermeulen, K. and Vansteelandt, S. (2015).
\newblock Bias-reduced doubly robust estimation.
\newblock {\em Journal of the American Statistical Association},
  110(511):1024--1036.

\bibitem[Virolainen, 2020]{virolainen2020structural}
Virolainen, S. (2020).
\newblock Structural gaussian mixture vector autoregressive model with
  application to the asymmetric effects of monetary policy shocks.
\newblock {\em arXiv preprint arXiv:2007.04713}.

\bibitem[Vlaar, 2004]{vlaar2004asymptotic}
Vlaar, P.~J. (2004).
\newblock On the asymptotic distribution of impulse response functions with
  long-run restrictions.
\newblock {\em Econometric Theory}, 20(5):891--903.

\bibitem[Watson, 1994]{watson1994vector}
Watson, M.~W. (1994).
\newblock Vector autoregressions and cointegration.
\newblock {\em Handbook of econometrics}, 4:2843--2915.

\bibitem[White and Pettenuzzo, 2014]{white2014granger}
White, H. and Pettenuzzo, D. (2014).
\newblock Granger causality, exogeneity, cointegration, and economic policy
  analysis.
\newblock {\em Journal of Econometrics}, 178:316--330.

\bibitem[Wold, 1960]{wold1960generalization}
Wold, H.~O. (1960).
\newblock A generalization of causal chain models (part iii of a triptych on
  causal chain systems).
\newblock {\em Econometrica: journal of the Econometric Society}, pages
  443--463.

\bibitem[Wright, 1934]{wright1934method}
Wright, S. (1934).
\newblock The method of path coefficients.
\newblock {\em The annals of mathematical statistics}, 5(3):161--215.

\bibitem[Xu, 2023]{xu2023local}
Xu, K.-L. (2023).
\newblock Local projection based inference under general conditions.
\newblock {\em Available at SSRN 4372388}.

\bibitem[Yamamoto and Horie, 2023]{yamamoto2023cross}
Yamamoto, Y. and Horie, T. (2023).
\newblock A cross-sectional method for right-tailed panic tests under a
  moderately local to unity framework.
\newblock {\em Econometric Theory}, 39(2):389--411.

\bibitem[Yoon et~al., 2019]{yoon2019network}
Yoon, S.-M., Al~Mamun, M., Uddin, G.~S., and Kang, S.~H. (2019).
\newblock Network connectedness and net spillover between financial and
  commodity markets.
\newblock {\em The North American Journal of Economics and Finance},
  48:801--818.

\bibitem[Zha, 1999]{zha1999block}
Zha, T. (1999).
\newblock Block recursion and structural vector autoregressions.
\newblock {\em Journal of Econometrics}, 90(2):291--316.

\bibitem[Zhang, 2023]{zhang2023statistical}
Zhang, Y. (2023).
\newblock Statistical inference of high-dimensional vector autoregressive time
  series with non-iid innovations.
\newblock {\em arXiv preprint arXiv:2310.07364}.

\end{thebibliography}

\newpage

\end{document}